\documentclass[10pt,journal,twocolumn,a4paper]{IEEEtran}
\usepackage{amsfonts,amsmath,amssymb,amstext,amsthm,cite,color,dsfont,enumerate,hyperref,makeidx,verbatim,xfrac,xr,url}
\usepackage{pict2e}
\usepackage[mathscr]{euscript}
\hypersetup{
pdftitle		={The R\'{e}nyi Capacity and Center},
pdfauthor		={Bar{\i}s Nakiboglu},}
\title{The \renyi Capacity and Center}
\author{Bar\i\c{s} Nakibo\u{g}lu \thanks{e-mail:\href{mailto:bnakib@metu.edu.tr}{bnakib@metu.edu.tr}}}
\IEEEaftertitletext{ 
\flushright
\vspace{-2\baselineskip}
{\it 
Can{\i}m halam Fatma Nakibo\u{g}lu Aydi\c{c}'in an{\i}s{\i}na adanm{\i}{\c{s}}t{\i}r. \hspace{1.6cm}~
\\
Dedicated to the memory of my dear aunt Fatma Nakibo\u{g}lu Aydi\c{c}.}\qquad~\\
~\\
}
\theoremstyle{plain}
\newtheorem{lemma}{Lemma} 
\newtheorem{theorem}{Theorem} 
 
\newtheorem{conjecture}{Conjecture}
\newtheorem*{conjecture*}{Conjecture}

\theoremstyle{definition}
\newtheorem{definition}{Definition} 

\newtheorem{example}{Example}

\newtheorem{remark}{Remark}
\newtheorem*{remark*}{Remark}
\definecolor{mygray}{gray}{0.4}
\newcommand{\set} [1]			{{\mathscr{{#1}}}}
\newcommand{\alg}[1]			{{\mathcal{{#1}}}}
\newcommand{\rndv}[1]      {{\mathsf{{#1}}}}
\newcommand{\oper}[1]      {{\mathtt{{#1}}}}
\newcommand{\msr}[1]       {{\it    {{#1}}}}
\newcommand{\cnst}[1]      {{\mathit{{#1}}}}
\newcommand{\cnsA}[1]      {{\mathcal{{#1}}}}

\newcommand{\sss}[1]		{{\mathit{2}^{{#1}}}}
\newcommand{\integers}[1]	{{\mathbb{Z}}_{^{{#1}}}}
\newcommand{\rationals}[1]	{{\mathbb{Q}}_{^{{#1}}}}
\newcommand{\reals}[1]		{{\mathbb{R}}_{^{{#1}}}}

\newcommand{\clos}[1]      {{\mathtt{cl}{{#1}}}}
\newcommand{\conv}[1]      {{\mathtt{ch}{{#1}}}}

\newcommand{\pmf}[0]       {{p.m.f.~}}

\DeclareMathOperator*{\essup}{ess\,sup}
\DeclareMathOperator*{\esinf}{ess\,inf}
\newcommand{\dif}[1]       {{\mathrm{d}{#1}}}  
\newcommand{\der}[2]        {\tfrac{\dif{#1}}{\dif{#2}}}  
\newcommand{\pder}[2]       {\tfrac{\partial{#1}}{\partial{#2}}}  
\newcommand{\supp}[1]       {\mathtt{supp}({{#1}})}       

\newcommand{\DEF}[0]			{{\!\!~\triangleq\!~}}  
\newcommand{\mtimes}[0]			{{\circledast}}

\newcommand{\AC}[0]            {{\prec}}  
\newcommand{\UAC}[0]           {{\mathop{\prec}\nolimits^{uni}}}  
  
\newcommand{\abs}[1]           {{\left\lvert{{#1}}\right\lvert}}

\newcommand{\lon}[1]           {{{\left\lVert{{#1}}\right\lVert}}} 
\newcommand{\nrm}[2]           {{{\left\lVert{{#2}}\right\lVert}_{{#1}}}}
\newcommand{\IND}[1]           {{\mathds{1}_{\{#1\}}}}    
\newcommand{\ind}[0]           {{\imath}}
\newcommand{\jnd}[0]           {{\jmath}}
\newcommand{\knd}[0]           {{\kappa}}
\newcommand{\tin}[0]           {{\cnst{t}}}
\newcommand{\blx}[0]           {{\cnst{n}}}
\newcommand{\tlx}[0]           {{\cnst{T}}}
\newcommand{\tinS}[0]          {{\set{T}}}
\newcommand{\trans}[1]         {{\oper{T}}_{{#1}}}

\newcommand{\domtr}[1]         {{\set{Q}}_{{#1}}}

\newcommand{\eav}[3]           {{\overline{{#1}}_{{#2}}^{{#3}}}}

\newcommand{\pint}[0]          {{\cnst{\zeta}}}

\newcommand{\EXS}[2]         {{\bf E}_{{#1}}\!\left[{#2}\right]}

\newcommand{\Lp}[2]          {{{\alg{L}}}^{{#1}}({#2})}   
\newcommand{\Lon}[1]         {{\Lp{1}{#1}}}
\newcommand{\fX}[0]          {{\cnst{f}}}

\newcommand{\gX}[0]          {{\cnst{g}}}   
\newcommand{\GX}[0]          {{\cnst{G}}}   
\newcommand{\hX}[0]          {{\cnst{h}}}
\newcommand{\fXS}[0]         {{\set{F}}}

\renewcommand{\div}[3]			{{\cnst{d}}_{{#1}}            \!\left({\left.           \! {#2}\right\Vert {#3}}                 \right)}

\newcommand{\KLD}[2]			{{\cnst{D}}                   \!\left({\left.           \! {#1}\right\Vert {#2}}                 \right)}

\newcommand{\FD}[3]				{{\cnsA{D}}_{{#1}}            \!\left(\left.            \! {#2}\right\Vert {#3}                  \right)}
\newcommand{\RD}[3]				{{\cnst{D}}_{{#1}}            \!\left(\left.            \! {#2}\right\Vert {#3}                  \right)}

\newcommand{\RDF}[4]			{{\cnst{D}}_{{#1}}^{{#2}}     \!\left(\left.            \! {#3}\right\Vert {#4}                  \right)}

\newcommand{\MI}[2]				{{\cnst{I}}               \!\left(           {#1};	  {#2}		\!\right)}

\newcommand{\RMI}[3]			{{\cnst{I}}_{{#1}}            \!\left(                  \! {#2};         \!{#3}                \!\right)} 
\newcommand{\FMI}[3]			{{\cnsA{I}}_{{#1}}            \!\left(                  \! {#2};         \!{#3}                \!\right)}

\newcommand{\RC}[2]				{{\cnst{C}}_{{#1},{#2}}}
\newcommand{\FC}[2]				{{\cnsA{C}}_{{#1},{#2}}}

\newcommand{\CRC}[3]			{{\cnst{C}}_{{#1},{#2},{#3}}}

\newcommand{\CR}[1]				{{\cnsA{S}}_{{#1}}}
\newcommand{\RCR}[2]			{{\cnsA{S}}_{{#1}}({#2})}
\newcommand{\FR}[2]				{{\cnsA{S}}_{{#1},{#2}}}
\newcommand{\RFR}[3]			{{\cnsA{S}}_{{#1},{#2}}({#3})}

\newcommand{\RR}[2]				{{\cnst{S}}_{{#1},{#2}}}
\newcommand{\RRR}[3]			{{\cnst{S}}_{{#1},{#2}}({#3})}

\newcommand{\costc}[0]			{{\cnst{\varrho}}}

\newcommand{\lbm}[0]			{{{\msr{l}}}}    
\newcommand{\rfm}[0]			{{{\msr{\nu}}}}

\newcommand{\cset}[0]			{{\set{A}}}

\newcommand{\rnb}[0]          {{\cnst{\beta}}}
\newcommand{\rnf}[0]          {{\cnst{\phi}}}

\newcommand{\rng}[0]          {{\cnst{\rho}}}
\newcommand{\rno}[0]          {{\cnst{\alpha}}}
\newcommand{\rnt}[0]          {{\cnst{\eta}}}

\newcommand{\brl}[0]           {{\alg{B}}}
\newcommand{\rborel}[1]        {{\brl}({#1})}
\newcommand{\borel}[2]         {{\brl}({#1},{#2})}
\newcommand{\oset}[0]          {{\set{N}}}

\newcommand{\oev}[0]           {{\set{E}}}
\newcommand{\opa}[0]           {{\alg{E}}}
\newcommand{\cha}[0]           {{\set{W}}}
\newcommand{\chu}[0]           {{\set{U}}}
\newcommand{\Scha}[1]          {{\cha}^{_{[#1]}}}
\newcommand{\Pcha}[1]          {{\varLambda}^{{{#1}}}}

\newcommand{\smea}[1]          {{{\alg{M}}({#1})}}
\newcommand{\fmea}[1]          {{{\alg{M}}^{^{+}}\!({#1})}}
\newcommand{\zmea}[1]          {{{\alg{M}}_{{0}}^{^{+}}\!({#1})}}
\newcommand{\pmea}[1]          {{{\alg{P}}({#1})}}

\newcommand{\pdis}[1]          {{{\set{P}}({#1})}}

\newcommand{\dinp}[0]          {{\cnst{x}}}

\newcommand{\inpS}[0]          {{\set{X}}}
\newcommand{\inpA}[0]          {{\alg{X}}}

\newcommand{\dout}[0]          {{\cnst{y}}}
\newcommand{\out}[0]           {{\rndv{Y}}}
\newcommand{\outS}[0]          {{\set{Y}}}
\newcommand{\outA}[0]          {{\alg{Y}}}

\newcommand{\dsta}[0]          {{\cnst{z}}}

\newcommand{\staS}[0]          {{\set{Z}}}

\newcommand{\nmn}[1]         {{{\msr{\pi}}_{{#1}}}}

\newcommand{\dnmn}[1]        {{{\msr{\pi}}_{{#1}}'}}    
    
\newcommand{\ddnmn}[1]       {{{\msr{\pi}}_{{#1}}''}}    
    
\newcommand{\mean}[0]        {{{\msr{\mu}}}}    
\newcommand{\mmn}[1]         {{{\mean}_{{#1}}}}

\newcommand{\dmn}[1]         {{{\msr{\mean}}_{{#1}}'}}    
    
\newcommand{\ddmn}[1]        {{{\msr{\mean}}_{{#1}}''}}

\newcommand{\tpn}[1]			{{{\mP}_{{[{#1}]}}}}

\newcommand{\mA}[0]				{{\msr{a}}}

\newcommand{\mB}[0]				{{\msr{b}}}

\newcommand{\mP}[0]				{{\msr{p}}}    
\newcommand{\pmn}[1]			{{{\mP}_{{#1}}}}
\newcommand{\pma}[2]			{{{\mP}_{{#1}}^{{#2}}}}

\newcommand{\mQ}[0]				{{\msr{q}}}    
\newcommand{\qmn}[1]			{{{\mQ}_{{#1}}}}

\newcommand{\mS}[0]				{{\msr{s}}}    
\newcommand{\smn}[1]			{{{\mS}_{{#1}}}}

\newcommand{\mU}[0]				{{\msr{u}}}

\newcommand{\mV}[0]				{{\msr{v}}}

\newcommand{\mW}[0]				{{\msr{w}}}    
\newcommand{\wmn}[1]			{{{\mW}_{{#1}}}}

\newcommand{\Wm}[0]				{{{\cnst{W}}}}

\newcommand{\csiszar}[0]							{Csisz\'{a}r~}

\newcommand{\harremoes}[0]							{Harremo\"{e}s~}
\newcommand{\korner}[0]								{K\"{o}rner~}
\newcommand{\renyi}[0]								{R\'{e}nyi~}
\newcommand{\topsoe}[0]								{Tops\o{e}~}

\makeatletter
\DeclareRobustCommand{\bigplus}{%
	\mathop{\vphantom{\sum}\mathpalette\@bigplus\relax}\slimits@
}
\newcommand{\@bigplus}[2]{\vcenter{\hbox{\make@bigplus{#1}}}}
\newcommand{\make@bigplus}[1]{%
	\sbox\z@{$\m@th#1\sum$}%
	\setlength{\unitlength}{\wd\z@}%
	\begin{picture}(1.4,1.4)
	\linethickness{.17ex}
	\Line(.7,.14)(.7,1.26)
	\Line(.14,.7)(1.26,.7)
	\end{picture}%
}
\DeclareRobustCommand{\bigtimes}{%
	\mathop{\vphantom{\sum}\mathpalette\@bigtimes\relax}\slimits@
}
\newcommand{\@bigtimes}[2]{\vcenter{\hbox{\make@bigtimes{#1}}}}
\newcommand{\make@bigtimes}[1]{%
	\sbox\z@{$\m@th#1\sum$}%
	\setlength{\unitlength}{\wd\z@}%
	\begin{picture}(1,1)
	\linethickness{.17ex}
	\Line(.1,.1)(.9,.9)
	\Line(.1,.9)(.9,.1)
	\end{picture}%
}
\makeatother
\begin{document}
\setcounter{tocdepth}{2}	
\pagestyle{plain}
\pagenumbering{arabic}
\hypersetup{hidelinks}
\maketitle 
\thispagestyle{empty}
\begin{abstract}
R\'{e}nyi's information measures
---the R\'{e}nyi information, mean, capacity, radius, and center---
are analyzed relying on the elementary properties of 
the R\'{e}nyi divergence and the power means. 
The van Erven-Harremo\"{e}s conjecture is proved for any positive order 
and for any set of probability measures on a given measurable space
and a generalization of it is established 
for the constrained variant of the problem.
The finiteness of the order $\alpha$ R\'{e}nyi capacity is shown to imply 
the continuity of the R\'{e}nyi capacity on $(0,\alpha]$ and the uniform 
equicontinuity of the R\'{e}nyi information, both as a family of functions 
of the order indexed by the priors and  as a family of functions of the 
prior indexed by the orders.
The R\'{e}nyi capacities and centers of various families of Poisson 
processes are derived as examples. 
\end{abstract}
\tableofcontents
\section{Introduction}\label{sec:introduction}
Information transmission problems are often posed on models with finite sample spaces
or on models with specific noise structures, such as Gaussian or Poisson models. 
As a result, certain fundamental observations such as the minimax theorem for the
Shannon capacity in terms of the Kullback-Leibler divergence or the existence of a 
unique ``capacity achieving output distribution'', i.e. the existence of a unique
Shannon center, are established either for models with finite sample spaces or 
for specific noise structures.
In \cite{kemperman74}, Kemperman proved these assertions far more generally by 
interpreting the channel as a set of probability measures on a given measurable 
space.

In a sense, Kemperman tacitly suggests a purely measure theoretic understanding of 
the Shannon capacity and center that is separated from their significance in the 
information transmission problems. 
Even without the generality afforded by the measure theoretic framework, such 
an understanding is appealing because Shannon capacity and center come up
in various information transmission problems, 
with very different operational meanings.
Consider for example a finite set \(\cha\) of probability mass functions on a 
finite output set \(\outS\).
\begin{itemize}
\item If we interpret \(\cha\) as a discrete channel that is to be used multiple times, 
then the Shannon capacity of \(\cha\) is the largest rate at which one can communicate 
reliably via the channel \(\cha\), \cite{shannon48}.
\item If we interpret \(\cha\) as a collection of sources that is to be encoded by a 
lossless variable length source code, then the Shannon capacity is a lower bound on the 
worst redundancy among the members of \(\cha\), which is off at most by one for
some lossless variable length source code, 
\cite{davissonL80}, \cite{gallager79}, \cite{ryabko79}.
\end{itemize}

In this paper we propose an analogous measure theoretic understanding for 
the \renyi capacity and center.
Our interest in these concepts stems from their operational significance 
in the channel coding problem; we elucidate that operational significance 
in our concurrent paper \cite{nakiboglu19B}.
Because of the generality of the measure theoretic model we adopt in this paper,  
we can discuss in \cite{nakiboglu19B} 
the operational significance of these concepts
for a diverse family of channels in a unified framework.
In the current paper our main aim is to present an analysis starting from the measure 
theoretic first principles and the elementary properties of the \renyi divergence.
We will first present a brief overview of the \renyi information, 
divergence, and mean. Then we proceed with the analysis of the \renyi capacity 
and center. 

Deriving the technical results employed in \cite{nakiboglu19B} is one of 
the main aims of the current paper;
however, the scope of our analysis is not restricted to the needs of 
the particular analysis we present in \cite{nakiboglu19B}.
We aim to build a more complete understanding of \renyi\!\!'s information measures that might 
lead us to new analysis techniques for the problems we investigate in \cite{nakiboglu19B} 
or for other information transmission problems involving \renyi\!\!'s information measures.
Our abstract and general framework is conducive to this purpose; in addition it 
allows us to observe certain phenomena that cannot be observed in simpler models.
For example,\footnote{This dichotomy is an immediate consequence of Lemma \ref{lem:capacityO}, 
	see page \pageref{dichotomy}.}  
the \renyi capacity is either a continuous function of the order on \((0,\infty)\) 
or a finite and continuous function of the order on  \((0,\rnf]\) that is infinite on 
\((\rnf,\infty)\) for some \(\rnf\in [1,\infty)\).
This dichotomy, however, cannot be observed with models with finite \(\cha\) or 
finite \(\outS\)
because the \renyi capacity is bounded if either \(\cha\) or \(\outS\) is finite.

In \cite{renyi61}, \renyi  provided an axiomatic characterization of a family of divergences 
for pairs of probability mass functions on a given finite sample space;
the resulting family of divergences, parametrized by positive real numbers, are named after him.
The definition of the \renyi divergence has been extended to pairs of probability measures.
Recently, van Erven and \harremoes provided a comprehensive investigation of various properties 
of the \renyi divergence in \cite{ervenH14}. 
For any \(\rno\) in \([0,\infty]\), the order \(\rno\) \renyi divergence between probability measures 
\(\mW\) and \(\mQ\), denoted by \(\RD{\rno}{\mW}{\mQ}\), is zero when \(\mW\) is equal to \(\mQ\) and 
non-negative when \(\mW\) is not equal to \(\mQ\).
Hence, given a measurable space \((\outS,\outA)\)  we can use the order \(\rno\) \renyi divergence 
to measure the spread of any set of probability measures \(\cha\) relative to any probability measure 
\(\mQ\) on \((\outS,\outA)\) as follows:
\begin{align}
\label{eq:def:relativeradius}
\RRR{\rno}{\cha}{\mQ} 
&\DEF \sup\nolimits_{\mW \in \cha}  \RD{\rno}{\mW}{\mQ}.
\end{align}
\(\RR{\rno}{\cha}(\mQ)\) is called 
\emph{the order \(\rno\) \renyi radius of \(\cha\) relative to \(\mQ\)}.
By taking the infimum of \(\RRR{\rno}{\cha}{\mQ}\) over all probability measures \(\mQ\) on \((\outS,\outA)\), 
we get an absolute measure of the spread of \(\cha\), called 
\emph{the order \(\rno\) \renyi radius of \(\cha\)},
\begin{align}
\label{eq:def:radius}
\RR{\rno}{\cha} 
&\DEF \inf\nolimits_{\mQ\in\pmea{\outA}} \sup\nolimits_{\mW \in \cha}  \RD{\rno}{\mW}{\mQ}.
\end{align}
Any probability measure \(\mQ\) on the measurable space \((\outS,\outA)\) satisfying 
\(\RRR{\rno}{\cha}{\mQ}=\RR{\rno}{\cha}\), is called 
\emph{an order \(\rno\) \renyi center of \(\cha\).}
The order one \renyi divergence is the Kullback-Leibler divergence;
hence
the order one \renyi radius and center
are
the Shannon radius and center referred to in \cite{kemperman74}. 
 
The Shannon capacity, defined as the supremum of the mutual information, 
is another measure 
of the spread of a set of probability measures on a given measurable space. 
In order to have a parametric generalization of the Shannon capacity, similar to the one 
provided by the \renyi radius to the Shannon radius, we need a parametric generalization 
of the mutual information. 
Sibson \cite{sibson69} proposed one such parametric generalization using 
the \renyi divergence, called the \renyi information, see Definition \ref{def:information}. 
For any set of probability measures \(\cha\) on a given measurable space \((\outS,\outA)\),
probability mass function \(\mP\) on \(\cha\), and positive real number \(\rno\),
\(\RMI{\rno}{\mP}{\cha}\) is \emph{the order \(\rno\) \renyi 
information}\footnote{Sibson defines ``the information radius of order \(\rno\)'' 
	through an infimum and then derives a closed form expression for it
in \cite[Thm. 2.2]{sibson69}. 
We take that closed form expression as the definition of 
the order \(\rno\) \renyi information.}
\emph{for prior \(\mP\)}.
The order one \renyi information equals to the mutual information.
For other positive real orders, the order \(\rno\) \renyi information 
can be described in terms of Gallager's function introduced 
in \cite{gallager65}:
\begin{align}
\label{eq:sibsongallager}
\RMI{\rno}{\mP}{\cha}
&=\left.\tfrac{E_{0}(\rng,\mP)}{\rng}\right\vert_{\rng=\frac{1-\rno}{\rno}}
&
&\forall \rno\in \reals{+}\setminus\{1\}
\end{align}
where Gallager's function \(E_{0}(\rho,\mP)\) is defined for
\(\rng>-1\) as 
\begin{align}
\label{eq:def:gallagersfunction}
E_{0}(\rng,\mP)
&\DEF-\ln \int \left(\sum\nolimits_{\mW}\mP(\mW) (\der{\mW}{\rfm})^{\frac{1}{1+\rng}}\right)^{1+\rng} \rfm(\dif{\dout}).
\end{align}
\emph{The order \(\rno\) \renyi capacity} \(\RC{\rno}{\cha}\) is defined as the supremum of 
the order \(\rno\) \renyi information \(\RMI{\rno}{\mP}{\cha}\) over all priors \(\mP\).

There are at least two other ways to define the \renyi information for which the order one 
\renyi information is equal to the mutual information: one by Arimoto \cite{arimoto77} 
and  another one by Augustin \cite{augustin78} and \csiszar \cite{csiszar95}.
A review of these three definitions of the \renyi information has recently been provided
by Verd\'{u} \cite{verdu15}.
Assuming \(\cha\) and \(\outS\) to be finite sets, \csiszar showed that 
the order \(\rno\) \renyi capacity for all three definitions of the \renyi information 
are equal to one another and to the order \(\rno\) \renyi radius, 
\cite[Prop. 1]{csiszar95}.

The extension of Kemperman's result \cite[Thm. 1]{kemperman74} about the Shannon 
capacity and center given in Theorem \ref{thm:minimax}, presented in the following, 
is among the most important observations about the \renyi capacity and center.
Theorem \ref{thm:minimax} 
establishes the equality of \(\RC{\rno}{\cha}\) and \(\RR{\rno}{\cha}\) 
for any positive order \(\rno\) and set of probability measures \(\cha\).
Furthermore, it asserts the existence of a unique order \(\rno\) \renyi center \(\qmn{\rno,\cha}\) 
whenever \(\RC{\rno}{\cha}\) is finite and characterizes the unique 
order \(\rno\) \renyi center in terms of the order \(\rno\) \renyi means.
These observations, however, have been reported in various forms before,
at least partially. 
In \cite{augustin69}, Augustin considered the orders in \((0,1)\),
proved a result equivalent to Theorem \ref{thm:minimax} for finite \(\cha\)'s 
and described how this result can be extended to arbitrary \(\cha\)'s.
Later, Augustin established a result, \cite[Thm. 26.6\ensuremath{'}]{augustin78}, that implies 
Theorem \ref{thm:minimax} for all orders in \(\rno\) in \((0,2)\).
\csiszar \cite[Prop. 1]{csiszar95} proved the equality \(\RC{\rno}{\cha}=\RR{\rno}{\cha}\)
for arbitrary positive order \(\rno\) assuming \(\cha\) and \(\outS\) are finite sets.

The equality of capacity to radius and the existence of a unique center, are phenomena that have 
been observed repeatedly in various contexts.
In order to clarify the standing of Theorem \ref{thm:minimax} among these results, we provide 
a more comprehensive discussion of the previous work on these fundamental observations in 
\S\ref{sec:kemperman}. 

The current paper and the concurrent paper \cite{nakiboglu19B} grew out of a desire to 
understand Augustin's proofs of the sphere packing bound given in \cite{augustin69}
and \cite{augustin78}  more intuitively.
Augustin's proofs are important because, among other things, they are the only proofs of the
sphere packing bound for non-stationary product channels, even for the case of discrete channels.
Concepts of \renyi capacity, radius, and center provide a way to express the principal novelty
 of Augustin's method in a succinct and intuitive way.
We discuss the novel observation underlying Augustin's method
and its promise briefly in \S\ref{sec:augustin}.

Similar to Theorem \ref{thm:minimax}, some of the observations that we discuss in the paper
have been reported before  either in terms of \renyi\!\!'s information measures 
\cite{csiszar95,sibson69} or in terms of other related quantities, 
such as Gallager's function, 
 \cite{augustin69,augustin78,gallager65,gallager}.  
But we also have a number of new observations that have not been reported before. 
We provide a tally of our most important contributions in \S\ref{sec:contributions}.

We conclude the current section with a summary of our notational conventions presented 
in \S\ref{sec:notation}. It is worth mentioning that only \S\ref{sec:notation} is
necessary to understand the rest of the paper; readers may bypass other parts 
of the current section depending on their interest and background.

The \renyi entropy \cite{renyi61} is another information measure, that is intimately 
related to the information measures discussed in this paper.
The \renyi entropy \cite{bobkovC15,ramS16} 
and its variants \cite{arimoto77,fehrS14,sasonV18A,teixeiraMA12}  
are of interest by themselves  \cite{chenMNF17,hansonD17,hoV15,rastegin11}; 
in addition they have been used to pose projection problems 
\cite{kumarS16,kumarS15A,kumarS15B} 
related to guessing \cite{arikan96,sasonV18B,sundaresan07}
and various questions about the information transmission problems 
\cite{bunteL14A,bunteL16,tanH18}.
Recently, there has been a revived interest in 
\renyi\!\!'s information measures and their operational significance  
\cite{bunteL14B,chengH16,chengH18,chengHT19,dalai13A,fongT16,shayevitz11,tomamichelH18,venkataramananJ18}, in general.

\subsection{Radius, Center, and Capacity}\label{sec:kemperman}
The concepts of radius and center, as we use them, are analogous to their 
counter parts in Euclidean geometry.
Let \(\cha\) be a set of points in the \(\blx\) dimensional Euclidean space
\(\reals{}^{\blx}\) and \(\mQ\) be a point in the same space.
Then one measure of the spread of \(\cha\) relative to \(\mQ\) is  the 
infimum of the radii of the \(\mQ\)-centered spheres including all points of \(\cha\), 
called \emph{the Chebyshev radius of \(\cha\) relative to \(\mQ\)}:
\begin{align}
\notag
\RCR{\cha}{\mQ}
&\DEF \sup\nolimits_{\mW\in\cha} \nrm{2}{\mW-\mQ} 
&
&\forall \cha \subset \reals{}^{\blx},\mQ\in\reals{}^{\blx}.
\end{align}
If we do not require the centers of the spheres to be at a given point \(\mQ\), 
then we get an absolute measure of the spread of \(\cha\), called  
\emph{the Chebyshev radius of \(\cha\)}:
\begin{align}
\notag
\CR{\cha}
&\DEF \inf\nolimits_{\mQ\in\reals{}^{\blx}}  \sup\nolimits_{\mW\in\cha} \nrm{2}{\mW-\mQ} 
&
&\forall \cha \subset \reals{}^{\blx}.
\end{align}
If \(\CR{\cha}\) is finite, 
then there exists\footnote{The existence follows from the extreme value theorem  for 
lower semicontinuous functions. The uniqueness is a result of the uniform convexity of finite dimensional Euclidean spaces.}
a unique Chebyshev center \(\qmn{\cha}\) satisfying \(\RCR{\cha}{\qmn{\cha}}=\CR{\cha}\).

For any set of points in a metric space \((\inpA,\cnst{d})\), one can define the  Chebyshev radius 
by replacing \(\reals{}^{\blx}\) with \(\inpA\) and \(\nrm{2}{\mW-\mQ}\) with \(\cnst{d}(\mW,\mQ)\) in the definition. 
However, neither the existence nor the uniqueness of the Chebyshev center is a foregone conclusion for such 
generalizations. Garkavi \cite[Thm. 1]{garkavi64} provides a three point set in a Banach space that does 
not have a Chebyshev center. In the Hamming space of length two binary strings, both \((0,0)\) and \((1,1)\) 
are Chebyshev centers of the set \(\cha=\{(0,1),(1,0)\}\).
See \cite[Ch. 15]{amir}, for a discussion of these concepts on the inner product spaces.

The Chebyshev radius is, in a sense, special because it is defined via the distance measure 
---the metric corresponding to the norm of the space for normed spaces and
the metric of the space for metric spaces---
that is a part of the description of the space. 
In principle, one can measure the relative and the absolute spread of the subsets of \(\inpS\) 
using any non-negative function \(\gX\) on \(\inpS\times\inpS\) satisfying \(\gX(\dinp,\dinp)=0\) 
for all \(\dinp\in\inpS\) and define a center accordingly. 
However, neither the existence nor the uniqueness of such a center is guaranteed.

When \(\inpS\) in the above formulation is the space of all probability measures \(\pmea{\outA}\)
on a measurable space \((\outS,\outA)\), one can measure the spread of a subset \(\cha\) of \(\pmea{\outA}\)
using the Kullback-Leibler divergence. 
The resulting radius is nothing but the Shannon radius of \(\cha\) and whenever the Shannon radius is finite 
the existence of a unique Shannon center follows from Kemperman's result 
\cite[Thm. 1]{kemperman74}.
The other assertion of Kemperman's result \cite[Thm. 1]{kemperman74} is the equality of 
the Shannon radius of \(\cha\) and the Shannon capacity of \(\cha\),
defined as the supremum of the mutual information 
\(\MI{\mP}{\cha}\) over all probability mass functions \(\mP\) on \(\cha\).
For the case where both \(\cha\) and \(\outS\) are finite sets, Kemperman's result  was already 
known at the time \cite[Thm. 4.5.1]{gallager}; in \cite{kemperman74} Kemperman attributes 
this special case to Shannon \cite{shannonW}. 
For the case when \(\outS\) is a finite set, first Gallager \cite[Thm. A]{gallager79} and 
then Davisson and Leon-Garcia \cite[Thm. 3]{davissonL80} proved results equivalent to Kemperman's. 
Later, Haussler \cite{haussler97} proved Kemperman's result assuming  \(\outS\) to be a complete 
separable metric space, i.e. Polish space, and \(\outA\) to be the associated Borel \(\sigma\)-algebra.

Theorem \ref{thm:minimax}, which we prove in the following,
extends Kemperman's result  to  the \renyi capacity and center 
of other orders.
The existence of a unique center under the finite capacity hypothesis 
and the equality of the capacity and the radius
have been confirmed in other contexts, as well.

\subsubsection{Radius for \(\fX\)-Divergence}\label{sec:fdivergence}
\csiszar \cite{csiszar63}, \cite{csiszar67A}, Morimoto \cite{morimoto63}, and Ali and Silvey \cite{aliS66} 
defined the \(\fX\)-divergence using convex functions, satisfying \(\fX(1)=0\).
The Kullback-Leibler divergence\footnote{For positive finite orders other 
	than one the \renyi divergence is not an \(\fX\)-divergence itself;
	but it can be written in terms of an \(\fX\)-divergence:
\(\RD{\rno}{\mW}{\mQ}\!=\!\tfrac{1}{\rno-1}\ln(1\!+\!(\rno\!-\!1) \FD{\fX}{\mW}{\mQ})\)
for \(\fX(\dinp)\!=\!\frac{\dinp^{\rno}-1}{\rno-1}\),
as previously pointed out in 
\cite[(14)]{csiszar66},
\cite[(1.10)]{csiszar67A},
\cite[(6)]{csiszar67B}, \cite[(1)]{sason16}, \cite[(80)]{sasonV16}.} 
is the \(\fX\)-divergence corresponding to the function \(\fX(\dinp)=\dinp \ln \dinp\). 
For any convex function \(\fX\) satisfying \(\fX(1)=0\), 
the absolute and relative \emph{\(\fX\)-radius} 
are defined in terms of the corresponding \(\fX\)-divergence as follows:
\begin{align}
\notag
\RFR{\fX}{\cha}{\mQ}
&\DEF \sup\nolimits_{\mW\in\cha} \FD{\fX}{\mW}{\mQ},
\\
\notag
\FR{\fX}{\cha}
&\DEF \inf\nolimits_{\mQ\in\pmea{\outA}} \sup\nolimits_{\mW\in\cha} \FD{\fX}{\mW}{\mQ}.
\end{align}
\emph{The \(\fX\)-information} and \emph{the \(\fX\)-capacity} are defined in terms 
of corresponding \(\fX\)-divergence
as follows
\begin{align}
\notag
\FMI{\fX}{\mP}{\cha}
&\DEF \inf\nolimits_{\mQ\in\pmea{\outA}}  \FD{\fX}{\mP \mtimes \cha}{\mP \otimes \mQ},
\\
\notag
\FC{\fX}{\cha}
&\DEF\sup\nolimits_{\mP\in\pdis{\cha}}\FMI{\fX}{\mP}{\cha}
\end{align}
where \(\mP \mtimes \cha\) is the probability measure whose marginal distribution 
on the support of \(\mP\) is \(\mP\) and whose conditional distribution is \(\mW\)
and \(\mP \otimes \mQ\) is the product measure.

The mutual information\footnote{For positive finite orders other than one
	the \renyi information can be written in terms of an \(\fX\)-information,
	using the analogous relation for divergences:
	\(\RMI{\rno}{\mP}{\cha}\!=\!\tfrac{1}{\rno-1}\ln(1\!+\!(\rno\!-\!1) \FMI{\fX}{\mP}{\cha})\)
	for \(\fX(\dinp)\!=\!\frac{\dinp^{\rno}-1}{\rno-1}\).} 
is the \(\fX\)-information corresponding to \(\fX(\dinp)=\dinp \ln \dinp\).
For \(\cha\)'s that are finite, \csiszar proved the following two assertions, see \cite[Thm. 3.2]{csiszar72}:
\begin{itemize}
\item \(\FC{\fX}{\cha}=\FR{\fX}{\cha}\) for any \(\fX\) that is strictly convex at \(1\).
\item There exists a unique \(\fX\)-center for any \(\fX\) 
that is strictly convex, provided that \(\FR{\fX}{\cha}\) is finite.
\end{itemize} 
For \(\fX\)'s that are strictly convex, it seems both assertions 
of \csiszar \cite[Thm. 3.2]{csiszar72} can be extended to arbitrary \(\cha\)'s 
using the technique employed by Kemperman, 
as Kemperman himself suggested in \cite{kemperman74}.
Gushchin and Zhdanov \cite{gushchinZ06} proved that \(\FC{\fX}{\cha}\)  equals to
\(\FR{\fX}{\cha}\) for any convex function \(\fX\) and any 
set of probability measures \(\cha\) 
provided that \(\outS\) is a complete separable metric space, i.e. Polish space, and \(\outA\) is the associated 
Borel \(\sigma\)-algebra.

\subsubsection[Radius in Quantum I.T.]{Radius in Quantum Information Theory}\label{sec:quantumit}
In this paper, we assume \(\cha\) to be a set of probability measures on a given measurable space.
This is a generalization of the case when \(\cha\) is a set of probability mass functions 
on a given finite set \(\outS\), i.e. the finite sample space case.
Another generalization of the finite sample space case is obtained by assuming \(\cha\) to be a
set of \(\abs{\outS}\)-by-\(\abs{\outS}\) positive semidefinite, trace one, Hermitian matrices.
In quantum information theory such matrices are called the density matrices; they represent
the states of a \(\abs{\outS}\) dimensional Hilbert space \({\cal H}\), \cite[\S1.2]{hayashi}.
The set of all such states is denoted by \({\cal S}({\cal H})\).
There is a one-to-one correspondence between the diagonal members of \({\cal S}({\cal H})\)
and the probability mass functions on \(\outS\).
As a result, statements about  subsets of \({\cal S}({\cal H})\) can be interpreted as 
generalizations of the corresponding statements about sets of probability mass functions on \(\outS\).

The definition of the Kullback-Leibler divergence has been extended to the members of \({\cal S}({\cal H})\);
it is, however, customarily called \emph{the quantum relative entropy} \cite[\S3.1.1]{hayashi}:
\begin{align}
\label{eq:def:QuantumKLD}
\KLD{\mW}{\mQ}
&\DEF\mbox{Tr} \mW(\ln \mW -\ln \mQ)
&
&\forall \mW, \mQ \in {\cal S}({\cal H}).
\end{align}
This definition can be interpreted as an extension because for 
the diagonal members  of \({\cal S}({\cal H})\), 
the quantum relative entropy as
defined in \eqref{eq:def:QuantumKLD}
is equal to the Kullback-Leibler divergence 
between the corresponding probability mass functions.
For any subset \(\cha\) of \({\cal S}({\cal H})\), the quantum Shannon radius is defined
as \(\inf\nolimits_{\mQ\in {\cal S}({\cal H})} \sup\nolimits_{\mW\in\cha}\KLD{\mW}{\mQ}\).

The definition of mutual information has been extended as well, but it is called 
\emph{the transmission information} \cite[\S4.1.1]{hayashi}:
\begin{align}
\label{eq:def:QuantumInformation}
\RMI{}{\mP}{\cha}
&\DEF \sum\nolimits_{\mW\in \cha} \mP(\mW)  \KLD{\mW}{\qmn{\mP}}
&
&\forall \mP\in \pdis{\cha}
\end{align} 
where \(\qmn{\mP}=\sum\nolimits_{\mW\in \cha} \mP(\mW) \mW\). 
Note that when \(\cha\) includes only diagonal members of \({\cal S}({\cal H})\), the above quantity 
equals to the mutual information for the prior \(\mP\) on the corresponding set of probability 
mass functions. 
The quantum Shannon capacity is defined as the supremum of \(\RMI{}{\mP}{\cha}\) over all
probability mass functions \(\mP\) on \(\cha\) with finite support.

The quantum Shannon capacity and radius are equal to one another for arbitrary 
\(\cha\subset {\cal S}({\cal H})\) provided that \({\cal H}\) is a finite dimensional Hilbert
space,\footnote{Results in \cite{ohyaPW97} and \cite{schumacherW01} were proved with additional assumptions.
In \cite{ohyaPW97}, Ohya, Petz, and Watanabe assumed \(\cha\) to be the image of an arbitrary Hilbert space 
under the channeling transformation.
In \cite{schumacherW01}, Shumacher and Westmoreland  assumed \(\cha\) to be a closed 
convex set. The existence of a unique quantum Shannon center is implicit in both \cite{ohyaPW97}
and \cite{schumacherW01}.} 
\cite[Thm. 4.1]{hayashi}, \cite[Thm. 3.5]{ohyaPW97}, \cite[(19)]{schumacherW01}.
This implies the equality of Shannon capacity and radius in the classical case provided that 
\(\outS\) is a finite set. 
However, neither Kemperman's result in \cite{kemperman74} nor the weaker result by Haussler in \cite{haussler97}
require \(\outS\) to be finite. Thus those results are not subsumed by the quantum Information
theoretic versions of Kemperman's result presented in \cite{hayashi}, \cite{ohyaPW97}, \cite{schumacherW01}.

The situation is similar for the quantum \renyi capacity, radius, and center.
All the results on the equality of the quantum \renyi capacity and radius that we are aware of 
\cite[Thm. 6]{dalai13A},
\cite[(4.74)]{hayashi}, \cite[Lemma I.3]{konigW09}, \cite[Thm. IV.8]{mosonyiH11}, \cite[Prop. 4.2]{mosonyiO17},
\cite[Lemma 14]{wildeWY14} 
assume \(\cha\) to be a subset of \({\cal S}({\cal H})\) for a finite dimensional Hilbert space \({\cal H}\).
Hence, to the best of our knowledge,
Theorem \ref{thm:minimax} is not subsumed by any of the known results in quantum information theory.

\subsection{Augustin's Method and the \renyi Center}\label{sec:augustin}
Augustin's proof of the sphere packing bound in \cite{augustin69} is one of the first 
few complete proofs of the sphere packing bound. 
Unlike its contemporaries by Shannon, Gallager and Berlekamp in \cite{shannonGB67A}
and by Haroutunian in \cite{haroutunian68},
Augustin's proof does not assume either 
the stationarity of the channel 
or the finiteness of the input set because it does not rely on a type based expurgation 
(i.e. a fixed composition argument). 
After decades, Augustin's proofs in \cite{augustin69} and \cite{augustin78} are still the 
only proofs of the sphere packing bound for non-stationary product channels, even in the 
finite input alphabet case.
Augustin's method has been applied to problems with feedback, as well.
Using a variant of his method, Augustin provides a proof sketch for 
the derivation of the sphere packing bound for codes on discrete stationary 
product channels with feedback in \cite{augustin78}; 
see \cite{nakiboglu19E} for a complete proof following this proof sketch. 
What we call the discrete stationary product channels with feedback are customarily 
called DMCs with feedback.

Despite their strength and generality, Augustin's derivations of the sphere packing 
bound is scarcely known to date, even among the specialists working on related problems.
In \cite[\S{\ref{B-sec:product-outerbound}}]{nakiboglu19B}, we derive sphere packing bounds using Augustin's method 
in a way that makes the roles of the \renyi capacity and center more salient and precise. 
Our bound for the product channels is sharper than the corresponding bounds in 
\cite{augustin69} and \cite{augustin78}.
In \cite[\S{\ref{B-sec:fproduct-outerbound}}]{nakiboglu19B}, we present a new proof of the sphere packing bound for 
the discrete product channels with feedback that facilitates the ideas of 
Haroutunian \cite{haroutunian77}  and Sheverdyaev \cite{sheverdyaev82}, 
as well as Augustin \cite{augustin69}, \cite{augustin78}.
Our new proof for the case with feedback holds for non-stationary channels satisfying 
certain stationarity hypothesis.  
In \cite[Appendix {\ref*{B-sec:operational}}]{nakiboglu19B}, 
we discuss other aspects of the operational significance 
of \renyi capacity and information for the channel coding problem.

The generality and strength of Augustin's results compel one to ask: 
What is the principle behind Augustin's proofs of the sphere packing bound?
A succinct answer exists for those who are already familiar with the 
concepts  of \renyi capacity, radius and center.\footnote{To be precise, Augustin 
does not work with \renyi\!\!'s information measures either in \cite{augustin69} or in 
\cite{augustin78}. It is, however, possible to  restate his observations in terms of 
\renyi\!\!'s information measures. 
His approach is eloquent and insightful, irrespective of the terms he chose 
to employ.}
In our judgment, the novel observation behind Augustin's proofs is the 
following:
\begin{align}
\notag
\lim\nolimits_{\rnf\to\rno} \RRR{\rno}{\cha}{\qmn{\rnf,\cha}}=\RC{\rno}{\cha}.
\end{align}
In words, by choosing \(\rnf\) close enough to \(\rno\), 
the order \(\rno\) \renyi radius relative to the order \(\rnf\) \renyi center 
can be made arbitrarily close to the order \(\rno\) \renyi capacity, 
which equals to the order \(\rno\) \renyi radius.
This observation seems benign enough to hold for other parametric families of divergences 
and corresponding capacities, radii, and centers.
Thus we believe that Augustin's method 
can probably be used to derive tight outer bounds in other information 
transmission problems.

\subsection{Main Contributions}\label{sec:contributions}
\begin{enumerate}[(1)]
\item\label{contribution:uniformequicontinuity} 
If \(\cha\) and \(\outA\) are finite sets, the continuity of the \renyi information 
is evident, both as a function of the order and as a function of the prior.
In their proof of the sphere packing bound \cite[p. 101]{shannonGB67A}, 
while proving the  continuity of the \renyi capacity in the order on 
\((0,1)\) ---for the finite \(\cha\) and \(\outS\) case---  
Shannon, Gallager, and Berlekamp asserted that the \renyi information is 
in fact equicontinuous as a family of functions of the order on \((0,1)\) 
indexed by the priors.
We strengthen their assertion by replacing the finiteness hypothesis on the sets \(\cha\) 
and \(\outS\) with a finiteness hypothesis for the \renyi capacity,
including orders greater than one, and
establishing uniformity of the equicontinuity,
see Lemma \ref{lem:finitecapacity}-(\ref{finitecapacity-uecO}).	
Furthermore, we show that the \renyi information is, also, uniformly 
equicontinuous when considered as a family of functions of the prior 
indexed by the orders, see Lemma \ref{lem:finitecapacity}-(\ref{finitecapacity-uecP}).

\item\label{contribution:ervenharremoes} 
Reflecting on \cite[Thm. 37]{ervenH14} for countable \(\outS\)'s at \(\rno\!=\!\infty\), 
van Erven and \harremoes  conjectured the following: 
\begin{conjecture*}[\!\!{\cite[Conjecture 1]{ervenH14}}]
If \(\RR{\rno}{\cha}<\infty\)  for 
an \(\rno\) in \((0,\infty]\)
and a \(\cha\!\subset\!\pmea{\outA}\)
then
there exists a unique \(\qmn{\rno,\cha}\!\in\!\pmea{\outA}\)
satisfying
\(\RR{\rno}{\cha}\!=\!\sup\nolimits_{\mW \in \cha} \RD{\rno}{\mW}{\qmn{\rno,\cha}}\).
Furthermore, for all \(\mQ\in\pmea{\outA}\) we have
\begin{align}
\notag
\sup\nolimits_{\mW \in \cha} \RD{\rno}{\mW}{\mQ}
&\geq  
\RR{\rno}{\cha}+\RD{\rno}{\qmn{\rno,\cha}}{\mQ}.
\end{align}
\end{conjecture*}
This conjecture is confirmed in Lemma \ref{lem:EHB} for the first 
time.\footnote{We were notified in \cite{harremoes16} that van Erven and \harremoes
	had a proof establishing their conjecture in \cite{ervenH14} 
	under some regularity conditions, at the time.}
Lemma \ref{lem:EHB} implicitly asserts the existence of a unique \(\qmn{\rno,\cha}\),
which is proved in Theorem \ref{thm:minimax}. 
This assertion, however, is not entirely new;
Augustin proved an equivalent assertion for orders in \((0,2)\) 
in \cite[Thm. 26.6\ensuremath{'}]{augustin78}
and gave a proof sketch for an equivalent assertion for orders in \((0,1)\)
in \cite{augustin69}.

In Appendix \ref{sec:constrainedcapacity}, we define \(\CRC{\rno}{\cha}{\cset}\) 
as the supremum of \(\RMI{\rno}{\mP}{\cha}\) over all 
priors \(\mP\) in \(\cset\)
and generalize the van Erven-\harremoes bound to the convex \(\cset\) case, 
see Definition \ref{def:constrainedRcapacity} and Lemma \ref{lem:CEHB}.

\item\label{contribution:infinitemodel}
Our framework allows us to pose and answer certain questions that are 
non-trivial only for infinite  \(\cha\)'s, i.e. infinite subsets of \(\pmea{\outA}\).
\begin{enumerate}[(a)]
	\item  There exists a countable subset  \(\cha'\) of \(\cha\) such that 
	\(\RC{\rno}{\cha'}=\RC{\rno}{\cha}\) for all 
	\(\rno\) in \([0,\infty]\), Lemma \ref{lem:capacityO}-(\ref{capacityO-countable}). 
	\item If \(\RC{\rnt}{\cha}\) is finite, then for all \(\epsilon>0\) there exists a 
	finite subset \(\cha'\) of \(\cha\) such that 
	\(\RC{\rno}{\cha'}>\RC{\rno}{\cha}-\epsilon\) for all \(\rno\) in \([\epsilon,\rnt]\), 
	Lemma \ref{lem:capacityO}-(\ref{capacityO-dini}).
	\item
	\(\RC{\rno}{\clos{\cha}}\!=\!\RC{\rno}{\cha}\) for all \(\rno\) in \((0,\infty]\) 
	where \(\clos{\cha}\) is the closure of 
	\(\cha\) in the topology of setwise convergence,
	Lemma \ref{lem:capacityEXT}-(\ref{capacityEXT-cl}).
	This has been pointed out by \csiszar and \korner 
	for \(\rno\) equals one case for finite \(\outS\)
	in \cite[Problem 8.10(b)]{csiszarkorner}.
\end{enumerate}	
\end{enumerate}

\subsection{Notational Conventions}\label{sec:notation}
For any set \(\outS\), we denote the set of all subsets of \(\outS\) by \(\sss{\outS}\)
and the set of all probability measures on finite subsets of \(\outS\) by \(\pdis{\outS}\).
For each \(\mP\in\pdis{\outS}\), 
i.e. for each probability mass function (\pmf\!\!),
we denote the set of all \(\dout\)'s in \(\outS\) for which  \(\mP(\dout)>0\), 
by \(\supp{\mP}\) and call it the support of \(\mP\).

We call the pair \((\outS,\outA)\) a measurable space iff \(\outA\) is a \(\sigma\)-algebra 
of the subsets of \(\outS\).
On a measurable space \((\outS,\outA)\), we denote 
the set of all finite signed measures by \(\smea{\outA}\),
the set of all finite measures by \(\zmea{\outA}\),
the set of all non-zero finite measures by \(\fmea{\outA}\),
and 
the set of all probability measures by  \(\pmea{\outA}\). 
A countable collection \(\opa\) of the subsets  of \(\outS\) is called a \(\outA\)-measurable partition
of \(\outS\) iff 
\(\cup_{\oev\in \opa}=\outS\), \(\emptyset\notin\opa\), \(\oev\cap\tilde{\oev}=\emptyset\) 
for all \(\oev,\tilde{\oev} \in \opa\), and \(\opa\subset\outA\),
\cite[Def. 10.8.1]{bogachev}.

A measure \(\mean\) on the measurable space \((\outS,\outA)\) is  absolutely continuous 
with respect to another measure \(\rfm\) on \((\outS,\outA)\), i.e. \(\mean\AC \rfm\), 
iff \(\mean(\oev)=0\) for any \(\oev \in \outA\) such that \(\rfm(\oev)=0\).
Measures \(\mean\) and \(\rfm\) are equivalent, i.e. \(\mean\sim\rfm\),
iff \(\mean\AC\rfm\)  and \(\rfm\AC\mean\). 
Measures \(\mean\) and \(\rfm\) are singular, i.e. \(\mean\perp\rfm\),
iff there exists an \(\oev \in \outA\) such that  \(\mean(\oev)=\rfm(\outS\setminus\oev)=0\).

A subset \(\cha\) of \(\fmea{\outA}\) is absolutely continuous with respect to 
a measure \(\rfm\), i.e. \(\cha\AC\rfm\), iff \(\mW\AC\rfm\) for all \(\mW\in \cha\).
A \(\sigma\)-finite measure \(\rfm\) is a reference measure for \(\cha\) iff \(\cha\AC\rfm\).
A subset \(\cha\) of \(\fmea{\outA}\)
is uniformly absolutely continuous with respect to \(\rfm\), i.e. \(\cha\UAC\rfm\), 
iff for every \(\epsilon>0\) there exists a 
\(\delta>0\) such that \(\mW(\oev)<\epsilon\)  for all \(\mW\in\cha\) provided that 
\(\rfm(\oev)<\delta\).
By \cite[p. 366 \& Thm. 2]{shiryaev}, \(\mean\AC\rfm\) iff \(\{\mean\}\UAC\rfm\).
Two subsets \(\cha\) and \(\chu\) of \(\pmea{\outA}\) are singular, i.e. \(\cha \perp \chu\),
iff there exists an \(\oev \in \outA\) such that \(\mW(\oev)=0\) for all \(\mW\in \cha\)
and \(\mU(\outS\setminus\oev)=0\) for all \(\mU\in\chu\).

We denote the Borel \(\sigma\)-algebra for the usual topology of 
the real numbers by \(\rborel{\reals{}}\).
We denote the essential supremum of a \(\outA\)-measurable,
i.e. \((\outA,\rborel{\reals{}})\)-measurable,
function \(\fX\) for the measure \(\rfm\)
on \((\outS,\outA)\) by \(\essup_{\rfm}\fX(\dout)\), i.e. 
\begin{align}
\notag
\essup\nolimits_{\rfm}\fX
\DEF\inf\{\gamma:\rfm(\{\dout:\fX(\dout)>\gamma\})=0\}.
\end{align}
We denote the integral of a measurable function \(\fX\) on \((\outS,\outA)\) 
with respect to the measure \(\rfm\) 
by \(\int \fX \rfm(\dif{\dout})\) or \(\int \fX(\dout) \rfm(\dif{\dout})\).
We denote the integral by \(\int\fX \dif{\dout}\) or \(\int \fX(\dout) \dif{\dout}\),
as well,  if it is on the real line and with respect to the Lebesgue measure.
If \(\rfm\) is a probability measure, then we also call the integral of \(\fX\) 
with respect to \(\rfm\) the expectation of \(\fX\) or the expected value of \(\fX\) 
and denote it by \(\EXS{\rfm}{\fX}\) or \(\EXS{\rfm}{\fX(\out)}\).

While discussing the continuity of measure valued functions and 
functions defined on sets of measures,
we use either the topology of setwise convergence or the total variation topology.
The topology of setwise convergence is the topology 
generated by the sets of the form \(\{\mean:\abs{\mean(\oev)-\tin}<\epsilon\}\)
for some \(\oev\in\outA\), \(\tin\in\reals{+}\), \(\epsilon\in\reals{+}\);
see \cite[\S4.7(v)]{bogachev} for a more detailed discussion. 
The total variation topology is the metric topology generated by the total variation norm.
For any \(\mean\) in \(\smea{\outA}\) the total variation norm of \(\mean\) is defined 
as
\begin{align}
\notag
\lon{\mean} 
&\DEF \sup\nolimits_{\oev \in \outA} \mean(\oev)-\mean(\outS\setminus\oev).
\end{align}
As a consequence of the Lebesgue decomposition theorem \cite[5.5.3]{dudley} and  
the Radon-Nikodym theorem \cite[5.5.4]{dudley} we have  
\begin{align}
\notag
\lon{\mean}&=\int \abs{\der{\mean}{\rfm}}  \rfm(\dif{\dout}) 
&
&\forall \mean,\rfm: \mean \AC \rfm.
\end{align}

Our notation will be overloaded for certain symbols; however, the relations represented 
by these symbols will be clear from the context.
We denote the products of topologies \cite[p. 38]{dudley}, 
\(\sigma\)-algebras \cite[p. 118]{dudley}, and measures \cite[Thm. 4.4.4]{dudley} by \(\otimes\).
We denote the Cartesian product  of sets \cite[p. 38]{dudley} by \(\times\). 
We denote the absolute value of real numbers and the size of sets
by \(\abs{\cdot}\). 
For extended real valued functions \(\fX\) and \(\gX\) on \(\outS\),
\(\fX\leq\gX\) iff \(\fX(\dout)\leq \gX(\dout)\) for all \(\dout\in \outS\).
For measures \(\mean\) and \(\rfm\) on \((\outS,\outA)\), 
\(\mean \leq \rfm\) iff \(\mean(\oev)\leq \rfm(\oev)\) for all \(\oev\in\outA\).

For \(\dinp,\dout\in\reals{}\), \(\dinp\wedge\dout\) is the minimum of \(\dinp\) and \(\dout\).
For extended real valued functions \(\fX\) and \(\gX\) on \(\outS\),
\(\fX\wedge\gX\) is the pointwise minimum of \(\fX\) and \(\gX\).
For \(\mean,\mW\in\smea{\outA}\), \(\mean\wedge\mW\) is the unique measure 
satisfying \(\der{\mean\wedge\mW}{\rfm}=\der{\mean}{\rfm}\wedge\der{\mW}{\rfm}\)
for any \(\rfm\) satisfying \(\mean\AC\rfm\) and \(\mW\AC\rfm\).
If \(\fXS\) is a set of real valued functions, then \(\wedge_{\fX\in\fXS}\fX\)
is the extended real valued function obtained by taking the pointwise infimum of 
\(\fX\)'s in \(\fXS\). 
For a \(\chu\subset\smea{\outA}\) satisfying \(\mW\leq\mU\)
for all \(\mU\in\chu\) for some \(\mW\in\smea{\outA}\), 
\(\wedge_{\mU\in\chu}\mU\) is the measure which is the infimum of
\(\chu\) with respect to the partial order \(\leq\).
The existence of a unique infimum is guaranteed by \cite[Thm. 4.7.5]{bogachev}. 
We use the symbol \(\vee\) analogously to \(\wedge\) but we represent maxima and 
suprema with it, rather than minima and infima.

\section{Preliminaries}\label{sec:preliminary} 
We commence our discussion by defining the mean measure and 
analyzing it, first as a function of the order for a given 
prior then as a function of the prior for a given order. 
After that we define the \renyi information using the mean 
measure and analyze it as a function of the order and the 
prior using the analysis of the mean measure.
Then we define the \renyi divergence and review those features 
of it that will be needed in our analysis.
We conclude the current section by defining the \renyi mean 
and deriving an alternative expression for the \renyi information 
in terms of the \renyi divergence using the \renyi mean. 

\subsection{The Mean Measure}\label{sec:powermean}
The weighted power means are generalizations of the weighted arithmetic mean.
For any positive real number \(\rno\) and \pmf \(\mP\) on non-negative real numbers, the order 
\(\rno\) mean for the prior \(\mP\) is 
\((\sum_{\dinp} \mP(\dinp) \dinp^{\rno})^{\sfrac{1}{\rno}}\). 
For any prior \(\mP\), the order \(\rno\) weighted mean
is a nondecreasing and continuously differentiable 
function of \(\rno\) on \(\reals{+}\).
Hence we can calculate its limit as \(\rno\) approaches zero, or infinity, 
using the L'Hospital's rule \cite[Thm. 5.13]{rudin}:
\begin{align}
\notag
\lim\nolimits_{\rno \downarrow 0} \left(\sum\nolimits_{\dinp} \mP(\dinp) \dinp^{\rno}\right)^{\sfrac{1}{\rno}}
&=\prod\nolimits_{\dinp}  \dinp^{\mP(\dinp)} 
\\
\notag
\lim\nolimits_{\rno \uparrow \infty} \left(\sum\nolimits_{\dinp} \mP(\dinp) \dinp^{\rno}\right)^{\sfrac{1}{\rno}}
&=\max\nolimits_{\dinp:\mP(\dinp)>0} \dinp.
\end{align}

The order \(\rno\) mean of measures for the prior \(\mP\) is defined via 
the pointwise order \(\rno\) mean of their Radon-Nikodym derivatives 
for the prior \(\mP\). 
In the following, we confine our discussion to the means of probability 
measure.
\begin{definition}\label{def:powermean}
Let \(\mP\) be a \pmf on \(\pmea{\outA}\) and \(\rfm\) 
be a reference measure for \(\mW\)'s with positive \(\mP(\mW)\). 
Then \emph{the order \(\rno\) mean of the Radon-Nikodym derivatives for the prior \(\mP\)} 
is\footnote{For each \(\mW\) with positive \(\mP(\mW)\), \(\der{\mW}{\rfm}\) exists for 
all \(\dout\) except for a \(\rfm\)-measure zero set by the Radon-Nikodym theorem \cite[5.5.4]{dudley}. 
Since there are only finite number of \(\mW\)'s with positive \(\mP(\mW)\),
\(\der{\mmn{\rno,\mP}}{\rfm}\) exists as a function of \(\rno\) from \([0,\infty]\) to \(\reals{\geq0}\)
for all \(\dout\) except for a \(\rfm\)-measure zero set.}  
\begin{align}
\label{eq:def:powermean-density}
 \der{\mmn{\rno,\mP}}{\rfm}
 &\DEF 
 \begin{cases}
 \prod\nolimits_{\mW:\mP(\mW)>0}    \left(\der{\mW}{\rfm}\right)^{\mP(\mW)}
 &\mbox{if~}\rno=0
 \\
 \left(\sum\nolimits_{\mW}  \mP(\mW)  \left(\der{\mW}{\rfm}\right)^{\rno} \right)^{\sfrac{1}{\rno}}
 &\mbox{if~}\rno \in \reals{+}
 \\
 \max\nolimits_{\mW:\mP(\mW)>0}\der{\mW}{\rfm}
 &\mbox{if~}\rno=\infty
 \end{cases}
 &
 &\rfm\mbox{-a.e.}
\end{align}
\emph{The order \(\rno\) mean measure for the prior \(\mP\)} is defined as
\begin{align}
\label{eq:def:powermean}
\mmn{\rno,\mP}(\oev)
&\DEF\int_{\oev} \der{\mmn{\rno,\mP}}{\rfm} \rfm(\dif{\dout})
&\forall \oev\in\outA. 
\end{align}
\end{definition}

In \eqref{eq:def:powermean-density} and throughout this section sums
of the form  \(\sum_{\mW}\) stands for sums of the form \(\sum_{\mW:\mP(\mW)>0}\).
In \eqref{eq:def:powermean-density}, \(\mW\) is a dummy variable used to express 
the elements of \(\pmea{\outA}\), i.e. probability measures on \((\outS,\outA)\). 
The probability mass assigned to each \(\mW\) by \(\mP\) is denoted by \(\mP(\mW)\). 
The reference measure \(\rfm\) is absent from the symbol for the mean measure 
because mean measure does not depend on the choice of the reference measure:
Let \(\widetilde{\mean}_{\rno,\mP}\) be the mean measure obtained using a reference measure 
\(\widetilde{\rfm}\) instead of \(\rfm\); then
\begin{align}
\notag
\mmn{\rno,\mP}(\oev)&=\widetilde{\mean}_{\rno,\mP}(\oev)
&
&\forall \rno \in [0,\infty] \mbox{~and~} \forall \oev \in \outA.
\end{align}
This follows from a standard application
of the Lebesgue decomposition theorem  and the Radon-Nikodym theorem. 

We are interested in the mean measure primarily as a tool to define and analyze 
the \renyi information. 
In \cite[\S26]{augustin78}, Augustin introduced the mean measure 
and derived some of the observations we present in 
Lemmas \ref{lem:powermeanequivalence}-\ref{lem:powermeanP}, albeit 
for different parametrizations of the order. 
Augustin, however, did not define or analyze the \renyi information 
in \cite{augustin78}.
Proofs of Lemmas \ref{lem:powermeanequivalence}-\ref{lem:powermeanP} 
are presented in Appendix \ref{sec:powermean-proofs}.

\begin{lemma}\label{lem:powermeanequivalence}
Let \(\mP\) be a \pmf on \(\pmea{\outA}\).
\begin{enumerate}[(a)]
\item\label{powermeanequivalence-a} 
\(\mmn{\rno,\mP} \sim \mmn{1,\mP}\) 
and 
\(\abs{\supp{\mP}}^{-\frac{1}{\rno}} \leq \lon{\mmn{\rno,\mP}}\leq \abs{\supp{\mP}}\) 
for any \(\rno\in (0,\infty]\). Furthermore, 
\(\lon{\mmn{1,\mP}}=1\).
\item\label{powermeanequivalence-b}
\(\mmn{0,\mP}\AC \mW\) for any \(\mW\in \supp{\mP}\) and \(\lon{\mmn{0,\mP}}\leq 1\). 
\end{enumerate}
\end{lemma}

The main consequence of Lemma \ref{lem:powermeanequivalence} is that \(\mmn{\rno,\mP}\AC\mmn{1,\mP}\) 
for all \(\rno\in[0,\infty]\). 
Hence, we can describe and analyze the mean measures via their Radon-Nikodym derivatives with 
respect to the order one mean measure.
We build our analysis of the mean measure as a function of the order  around 
this observation.
First, we analyze \(\der{\mmn{\rno,\mP}\!}{\mmn{1,\mP}\!}\) as a function the order \(\rno\) in 
Lemma \ref{lem:powermeandensityO};
then use the dominated convergence theorem to obtain the corresponding results for \(\mmn{\rno,\mP}\)
in Lemma \ref{lem:powermeanO}.

\begin{definition}
Let \(\mP\) be a \pmf on \(\pmea{\outA}\)
and \(\rno\) be in \([0,\infty]\).
Then the order \(\rno\) density for the prior \(\mP\) is 
\begin{align}
\label{eq:def:powermeandensity}
\nmn{\rno,\mP}
&\DEF \der{\mmn{\rno,\mP}}{\mmn{1,\mP}}.
\end{align}
\end{definition}
Note that for any \pmf \(\mP\) on  \(\pmea{\outA}\), the order \(\rno\) density for the prior \(\mP\) is a 
\(\outA\)-measurable function from \(\outS\) to \(\reals{}\) by the Radon-Nikodym 
theorem \cite[5.5.4]{dudley}.

The order \(\rno\) posteriors defined in the following provides us an alternative way to express
\(\nmn{\rno,\mP}\) and its derivatives.
\begin{definition}\label{def:posterior}
Let \(\mP\) be a \pmf on \(\pmea{\outA}\) and \(\rno\) be a positive real number. 
Then for each \(\dout\in\outS\)  \emph{the order \(\rno\) posterior} \(\tpn{\rno}\) is a 
\pmf\!\! on \(\pmea{\outA}\) given by 
\begin{align}
\label{eq:def:posteroir}
\tpn{\rno}(\mW|\dout)\DEF
\begin{cases}
\mP(\mW)\left(\der{\mW}{\mmn{\rno,\mP}}\right)^{\rno}
&\mbox{if~}\mP(\mW)>0 
\\
0
&\mbox{else}
\end{cases}.
\end{align}
\end{definition}
The order \(\rno\) posterior \pmf \(\tpn{\rno}\) is  a 
\(\outA\)-measurable function for each \(\mW\).
The order one posterior \pmf\!\! \(\tpn{1}\) is also called the posterior \pmf\!\!, in accordance with the usual terminology.
\begin{lemma}\label{lem:powermeandensityO}
For any \pmf \(\mP\) on \(\pmea{\outA}\)
the following statements hold for \(\mmn{1,\mP}\)-almost every \(\dout\).
\begin{enumerate}[(a)]
\item\label{powermeandensityO-a}
\(\delta^{\frac{1-\rno}{\rno}}\leq \nmn{\rno,\mP}\leq 1\) for \(\rno\in (0,1]\)
and 
\(1 \leq \nmn{\rno,\mP}\leq \delta^{\frac{1-\rno}{\rno}}\) for \(\rno\in [1,\infty)\)
where \(\delta=\min\nolimits_{\mW:\mP(\mW)>0} \mP(\mW)\).
Furthermore,
\begin{align}
\notag
\nmn{\rno,\mP}(\dout)
&\!=\!
\begin{cases}
\prod\nolimits_{\mW:\mP(\mW)>0} \left(\tfrac{\tpn{1}(\mW|\dout)}{\mP(\mW)}\right)^{\mP(\mW)}
&
\rno\!=\!0
\\
\left(\sum\limits_{\mW} \tpn{1}(\mW|\dout)^{\rno}  \mP(\mW)^{1-\rno}\right)^{\sfrac{1}{\rno}}\!
& \rno\!\in\!\reals{+}
\\
\max\nolimits_{\mW:\mP(\mW)>0} \tfrac{\tpn{1}(\mW|\dout)}{\mP(\mW)}
& \rno\!=\!\infty
\end{cases}.
\\
\notag
\tpn{\rno}(\mW|\dout)
&\!=\!\begin{cases}
\tfrac{\tpn{1}(\mW|\dout)^{\rno}\mP(\mW)^{1-\rno}}{\nmn{\rno,\mP}^{\rno}}
&\mbox{if~}\mP(\mW)>0
\\
0
&\mbox{else}
\end{cases}.
\end{align}
\item\label{powermeandensityO-b}
\(\nmn{\rno,\mP}\)  is a smooth function of \(\rno\) on \(\reals{+}\). 
Furthermore, the first two derivatives of \(\nmn{\rno,\mP}\) are given by 
\begin{align}
\notag
\der{}{\rno}\nmn{\rno,\mP}
&=\tfrac{\nmn{\rno,\mP}}{\rno^2}\sum\nolimits_{\mW}
\tpn{\rno}(\mW|\dout) \ln\tfrac{\tpn{\rno}(\mW|\dout)}{\mP(\mW)}.
\\
\notag
\der{^2}{\rno^2}\nmn{\rno,\mP} 
&=\tfrac{1-\rno}{\nmn{\rno,\mP}}
\left(\der{}{\rno}\nmn{\rno,\mP}\right)^{2}
-\tfrac{2}{\rno}\der{}{\rno}\nmn{\rno,\mP}
\\
\notag
&\qquad~\quad+\tfrac{\nmn{\rno,\mP}}{\rno^3}\sum\nolimits_{\mW}
\tpn{\rno}(\mW|\dout)
\left(\ln\tfrac{\tpn{\rno}(\mW|\dout)}{\mP(\mW)}\right)^2.
\end{align}
\item\label{powermeandensityO-c}
\((\nmn{\rno,\mP})^{\rno}\) is 
log-convex\footnote{Both of the following statements are equivalent to the log-convexity of 
\((\nmn{\rno,\mP})^{\rno}\) in \(\rno\): 
``\(\nmn{\frac{1}{1+\rng},\mP}\) is log-convex in \(\rng\)'' and 
``For any \(\beta \in [0,1]\) and \(\rno_0,\rno_{1} \in (0,\infty]\), 
\(\nmn{\rno_{\beta},\mP}\leq (\nmn{\rno_{0},\mP})^{1-\beta} (\nmn{\rno_{1},\mP})^{\beta}\)
where \(\rno_{\beta}\) is \(\rno_{\beta}=[(1-\beta) (\rno_0)^{-1}+\beta (\rno_{1})^{-1}]^{-1}\).''}
in \(\rno\) on \(\reals{+}\), i.e. for any \(\beta\in (0,1)\) and \(\rno_0,\rno_1 \in \reals{+}\)
\begin{align}
\notag
(\nmn{\rno_{\beta},\mP})^{\rno_{\beta}} 
&\leq (\nmn{\rno_{1},\mP})^{\beta \rno_{1}}  (\nmn{\rno_{0},\mP})^{(1-\beta)  \rno_{0}}
\end{align}
where \(\rno_{\beta}=\beta\rno_{1}+(1-\beta)\rno_0\).
Furthermore, for \(\rno_1\neq\rno_0\) the inequality is strict iff there exist \(\mW,\tilde{\mW}\in\supp{\mP}\)  
such that \(\tfrac{\tpn{1}(\mW|\dout)}{\mP(\mW)}>\tfrac{\mP(\tilde{\mW}|\dout)}{\mP(\tilde{\mW})}>0\).

\item\label{powermeandensityO-d} 
If there exists a \(\mW\) such that  \(\tpn{1}(\mW|\dout)>\mP(\mW)\), 
then \(\nmn{\rno,\mP}(\dout)\) 
is bounded, continuous, and monotone increasing in \(\rno\) on \([0,\infty]\), 
else \(\nmn{\rno,\mP}(\dout)=1\) for all \(\rno\) in \([0,\infty]\). 
\end{enumerate}
\end{lemma}

Lemma \ref{lem:powermeandensityO} establishes the density \(\nmn{\rno,\mP}\)
as a smooth function \(\mmn{1,\mP}\)-a.e. and provides expressions 
for its first two derivatives. 
These derivatives are \(\outA\)-measurable functions because 
\(\nmn{\rno,\mP}\) and \(\tpn{\rno}\) 
are \(\outA\)-measurable. Then using their \(\mmn{1,\mP}\)-integrals we can define 
two mappings: 
\begin{align}
\label{eq:def:dpowermean}
\dmn{\rno,\mP}(\oev)
&\DEF \int_{\oev}  (\dnmn{\rno,\mP})  \mmn{1,\mP}(\dif{\dout})
&
&\forall \oev\in \outA,
\\
\label{eq:def:ddpowermean}
\ddmn{\rno,\mP}(\oev)
&\DEF \int_{\oev}  (\ddnmn{\rno,\mP})  \mmn{1,\mP}(\dif{\dout})
&
&\forall \oev\in \outA
\end{align}
where 
\(\dnmn{\rno,\mP}\)
and
\(\ddnmn{\rno,\mP}\)
are shorthands  for 
\(\der{}{\rno}\nmn{\rno,\mP}\)
and
\(\der{^2}{\rno^2}\nmn{\rno,\mP}\).

Note that we have not claimed that either of these mappings is defining 
a measure for each \(\rno\).
Lemma \ref{lem:powermeanO} given in the following
establishes that fact and analyzes the mean measure 
\(\mmn{\rno,\mP}\) as a function of the order \(\rno\).

\begin{lemma}\label{lem:powermeanO}
For any \pmf \(\mP\) on \(\pmea{\outA}\).
\begin{enumerate}[(a)]
\item\label{powermeanO-a}
\(\mmn{\rno,\mP}\) is a continuous function of \(\rno\) from \([0,\infty]\) with its usual topology
to \(\zmea{\outA}\) with the total variation topology.
\item\label{powermeanO-b}
\(\dmn{\rno,\mP}\) is a continuous function of \(\rno\) from \((0,\infty)\) with its usual topology
to \(\zmea{\outA}\) with the total variation topology.
Furthermore, \(\der{}{\rno}\mmn{\rno,\mP}=\dmn{\rno,\mP}\) in the sense that 
\begin{align}
\notag
\left.\der{}{\rno}\mmn{\rno,\mP}(\oev)\right\vert_{\rno=\rnf} 
&=\dmn{\rnf,\mP}(\oev)
&
&\forall \oev \in \outA,~~\forall \rnf\in (0,\infty).
\end{align}
\item\label{powermeanO-c}
\(\ddmn{\rno,\mP}\) is a continuous function of \(\rno\) from \((0,\infty)\) with its usual topology
to \(\smea{\outA}\) with the total variation topology.
Furthermore, \(\der{}{\rno}\dmn{\rno,\mP}=\ddmn{\rno,\mP}\) in the sense that 
\begin{align}
\notag
\left.\der{}{\rno}\dmn{\rno,\mP}(\oev)\right\vert_{\rno=\rnf}
&=\ddmn{\rnf,\mP}(\oev)
&
&\forall \oev \in \outA,~~\forall \rnf\in (0,\infty).
\end{align}
\item\label{powermeanO-d} 
\(\lon{\mmn{\rno,\mP}}^{\rno}\) is a log-convex function of \(\rno\) on \((0,\infty)\) such that 
\begin{align}
\notag
\lim\nolimits_{\rno\downarrow0}\lon{\mmn{\rno,\mP}}^{\rno}=
\essup\nolimits_{\mmn{1,\mP}} \sum\nolimits_{\mW:\tpn{1}(\mW|\dout)>0} \mP(\mW).
\end{align}
The log-convexity  is strict everywhere on \((0,\infty)\),
unless  there exists a \(\gamma\geq 1\) satisfying \(\mmn{1,\mP}(\set{A}(\mP,\gamma))=1\) 
for 
\(\set{A}(\mP,\gamma)=\{\dout:\tfrac{\tpn{1}(\mW|\dout)}{\mP(\mW)}=\gamma,~\forall \mW:\tpn{1}(\mW|\dout)>0\}\).
If there exists such a \(\gamma\), then \(\lon{\mmn{\rno,\mP}}=\gamma^{\frac{\rno-1}{\rno}}\).
\item\label{powermeanO-e}
\(\lon{\mmn{\rno,\mP}}\) is a continuous and nondecreasing function of \(\rno\) from 
\([0,\infty]\) to \([0,\abs{\supp{\mP}}]\).
If there exist \(\mW\), \(\widetilde{\mW}\) in \(\supp{\mP}\) such that \(\mW\neq \widetilde{\mW}\),
 then \(\lon{\mmn{\rno,\mP}}\) is monotone increasing everywhere on \((0,\infty)\), else
\(\lon{\mmn{\rno,\mP}}=1\) for all \(\rno\) in \([0,\infty]\). 
\end{enumerate}
\end{lemma}

Lemma \ref{lem:powermeanO} described the properties of the mean measure
as a function of the order for a fixed prior.
Lemma \ref{lem:powermeanP}, given in the following, describes the properties of 
the mean measure as a function of the prior for a fixed order.
\begin{lemma}\label{lem:powermeanP}
Let \((\outS,\outA)\) be a measurable space.
\begin{enumerate}[(a)]
\item\label{powermeanP-a}
If \(\rno\in [0,1]\), then \(\mmn{\rno,\mP}\) and \(\lon{\mmn{\rno,\mP}}\)
are convex functions of \(\mP\) from \(\pdis{\pmea{\outA}}\) to 
\(\zmea{\outA}\) and \([0,1]\), respectively.
\item\label{powermeanP-b}
If \(\rno\!\in\![1,\!\infty]\), then 
\(\mmn{\rno,\mP}\) and \(\lon{\mmn{\rno,\mP}}\!\)
are concave functions of \(\mP\) from \(\pdis{\pmea{\outA}}\) to 
\(\fmea{\outA}\) and \([1,\infty)\), respectively.
\item\label{powermeanP-c} For any \(\pmn{1},\pmn{2} \in \pdis{\pmea{\outA}}\) such that \(\mP_{1}\neq\mP_{2}\),
let \(\smn{\wedge}\), \(\smn{1}\) and \(\smn{2}\) be
\(\smn{\wedge}\DEF 2\tfrac{\pmn{1}\wedge\pmn{2}}{2-\lon{\pmn{1}-\pmn{2}}}\),
\(\smn{1}\DEF 2\tfrac{\pmn{1}-\pmn{1}\wedge\pmn{2}}{\lon{\pmn{1}-\pmn{2}}}\),
and 
\(\smn{2}\DEF 2\tfrac{\pmn{2}-\pmn{1}\wedge\pmn{2}}{\lon{\pmn{1}-\pmn{2}}}\).
Then \(\smn{\wedge}, \smn{1}, \smn{2}  \in \pdis{\pmea{\outA}}\) and
\begin{align}
\notag
\pmn{1}
&=(1-\tfrac{\lon{\pmn{1}-\pmn{2}}}{2})\smn{\wedge} + \tfrac{\lon{\pmn{1}-\pmn{2}}}{2} \smn{1},
\\
\notag
\pmn{2}
&=(1-\tfrac{\lon{\pmn{1}-\pmn{2}}}{2})\smn{\wedge} + \tfrac{\lon{\pmn{1}-\pmn{2}}}{2} \smn{2},
\\
\notag
\smn{1} &\perp \smn{2}. 
\end{align} 
\item\label{powermeanP-d}If \(\rno\in (0,1]\), then 
for any \(\pmn{1},\pmn{2} \in \pdis{\pmea{\outA}}\) we have
\begin{align}
\notag
\lon{\mmn{\rno,\pmn{1}}-\mmn{\rno,\pmn{2}}}
&\leq \tfrac{1}{\rno}\lon{\pmn{1}-\pmn{2}} .
\end{align}
Hence \(\mmn{\rno,\mP}\) is a Lipschitz continuous function of \(\mP\) for the total variation topology
for \(\rno\in (0,1]\).
\item\label{powermeanP-e}
If \(\rno\in [1,\infty)\), then for any \(\pmn{1},\pmn{2} \in \pdis{\pmea{\outA}}\) we have
\begin{align}
\notag
\lon{\mmn{\rno,\pmn{1}}\!-\!\mmn{\rno,\pmn{2}}}
&\leq (\tfrac{1}{2}\lon{\pmn{1}\!-\!\pmn{2}})^{\frac{1}{\rno}} \lon{\mmn{\rno,\smn{1}}\!-\!\mmn{\rno,\smn{2}}}.
\end{align}
\end{enumerate}
\end{lemma}

\subsection{The \renyi Information}\label{sec:information}
\begin{definition}\label{def:information}
Let \(\cha\) be a subset of \(\pmea{\outA}\) 
and \(\mP\) be a \pmf\! on \(\cha\).
Then \emph{the order \(\rno\) \renyi information for the prior \(\mP\)} is
\begin{align}
\label{eq:def:information}
\hspace{-.3cm}\!\!\RMI{\rno}{\mP}{\cha} 
&\!\DEF\!\!\! 
\begin{cases}
\esinf\limits_{\mmn{1,\mP}}
\ln\tfrac{1}{\sum\limits_{\mW} \IND{\tpn{1}(\mW|\dout)>0}\mP(\mW)}
&\rno\!=\!0
\\
\tfrac{\rno}{\rno-1}\ln \lon{\mmn{\rno,\mP}} 
& \rno\!\in\!\reals{+}\!\!\setminus\!\!\{\!1\!\}
\\
\EXS{\mmn{1,\mP}\!}{\!\sum\limits_{\mW} \tpn{1}(\mW|\dout)\ln \tfrac{\tpn{1}(\mW|\dout)}{\mP(\mW)}\!}
& \rno\!=\!1
\\
\ln \lon{\mmn{\infty,\mP}}
&\rno\!=\!\infty
\end{cases}
\end{align} 
\end{definition} 
 
Sibson introduced this 
quantity\footnote{Sibson called 
	\(\inf_{\mQ\in\pmea{\outA}}\RD{\rno}{\mP \mtimes \cha}{\mP\otimes\mQ}\)
	``the information radius of order \(\rno\)'' 
	and proved that it equals to the expression given in  Definition \ref{def:information} 
in \cite[Thm. 2.2]{sibson69}.
Our presentation is different: Definition \ref{def:information} does not refer to any infimum; equivalence of the 
alternative definition is established in Lemma \ref{lem:information:def}. 
This is similar to the way things are, usually, handled for the mutual information: the mutual information is
defined without any reference to an infimum \cite[(2.28)]{coverthomas}, later it is shown to be equal to 
the infimum of certain Kullback-Leibler divergence \cite[Lemma 10.8.1]{coverthomas}.}
in \cite{sibson69} using works of \renyi \cite{renyi61} and \csiszar \cite{csiszar67A,csiszar67B}.
Prior to \cite{sibson69} in \cite{gallager65}, Gallager introduced \(E_{0}(\rng,\mP)\),
which is nothing but a scaled version of the \renyi information; 
see \eqref{eq:sibsongallager} and \eqref{eq:def:gallagersfunction}.
  
Note that \(\RMI{\rno}{\mP}{\cha}\) has the same value for all \(\cha\)'s for 
which \(\mP\) is in \(\pdis{\cha}\).
Hence, in principle, one can use \(\cnst{I}_{\rno}(\mP)\) rather than \(\RMI{\rno}{\mP}{\cha}\)
to denote the \renyi information.
Although this unconventional symbol would be more coherent with the one we use for the mean measure,
we refrain from using it for the fear of alienating readers who prefer the customary symbol. 
Another justification for using the conventional notation is the effect of the richness of \(\cha\)
--- as measured by \(\sup_{\mP\in\pdis{\cha}} \RMI{\rno}{\mP}{\cha}\)--- on the continuity of 
\(\RMI{\rno}{\mP}{\cha}\) as a function of \(\mP\), see Lemma  \ref{lem:finitecapacity}-(\ref{finitecapacity-uecP}). 

Properties of the \renyi information as a function of the order for fixed prior and as a function of the prior 
for fixed order are presented in Lemmas \ref{lem:informationO} and \ref{lem:informationP}, respectively. 
Proofs of Lemmas \ref{lem:informationO} and \ref{lem:informationP} are presented in Appendix \ref{sec:information-proofs}. 

\begin{lemma}\label{lem:informationO}
For any subset \(\cha\) of \(\pmea{\outA}\)
and \pmf \(\mP\) on \(\cha\),
\(\RMI{\infty}{\mP}{\cha}\leq\ln\abs{\supp{\mP}}\)  
and  
\(\RMI{\rno}{\mP}{\cha}\) is a non-negative continuously differentiable nondecreasing function of \(\rno\)
on \(\reals{+}\) such that
\begin{align}
\label{eq:lem:informationOlimz}
\RMI{0}{\mP}{\cha}
&=\lim\nolimits_{\rno\downarrow 0}\RMI{\rno}{\mP}{\cha},
\\
\label{eq:lem:informationOliminfty}
\RMI{\infty}{\mP}{\cha}
&=\lim\nolimits_{\rno\uparrow \infty}\RMI{\rno}{\mP}{\cha},
\\
\label{eq:lem:informationOder}
\der{}{\rno} \RMI{\rno}{\mP}{\cha}
&=
\begin{cases}
\tfrac{\rno}{\rno-1}\tfrac{\lon{\dmn{\rno,\mP}}}{\lon{\mmn{\rno,\mP}}}  
-\tfrac{\ln\lon{\mmn{\rno,\mP}}}{(\rno-1)^2} 
&\rno\!\in\!\reals{+}\!\!\setminus\!\{\!1\!\}
\\
\tfrac{\ddmn{1,\mP}(\outS)+2\lon{\dmn{1,\mP}}-\lon{\dmn{1,\mP}}^2}{2}
&\rno=1
\end{cases}.
\end{align}
If \(\mmn{1,\mP}(\set{A}(\mP,\gamma))=1\) for some \(\gamma\geq 1\),
then
\(\RMI{\rno}{\mP}{\cha}=\ln \gamma\)  for all \(\rno\in [0,\infty]\), else
\(\der{}{\rno} \RMI{\rno}{\mP}{\cha}>0\) for all \(\rno\in\reals{+}\),
where
\(\set{A}(\mP,\gamma)\DEF\{\dout:\tfrac{\tpn{1}(\mW|\dout)}{\mP(\mW)}=\gamma~~\forall\mW \mbox{~with positive ~}\tpn{1}(\mW|\dout)\}\).
\end{lemma}

Using the definitions of \(\dmn{\rno,\mP}\) and \(\ddmn{\rno,\mP}\), given in 
\eqref{eq:def:dpowermean} and \eqref{eq:def:ddpowermean},
together with  Lemma \ref{lem:powermeandensityO}-(\ref{powermeandensityO-b}), 
we get the following two alternative expressions for the derivative of 
\(\RMI{\rno}{\mP}{\cha}\) with respect to the order on \(\reals{+}\)
\begin{align}
\label{eq:information:der-A}
\!\!\!\der{}{\rno}\!\RMI{\rno}{\mP}{\cha}\!
&\!=\!
\begin{cases}
\!\tfrac{1}{(\rno-1)\rno}
\EXS{\varpi_{\rno}}{\!\ln\!\tfrac{\tpn{\rno}(\mW|\dout)}{\mP(\mW)}\!-\!\RMI{\rno}{\mP}{\cha}\!}
&\!\rno\!\neq\!1
\\
\!\tfrac{1}{2}
\EXS{\varpi_{1}\!}{\left(\!\ln\tfrac{\tpn{1}(\mW|\dout)}{\mP(\mW)}\!-\!\RMI{1}{\mP}{\cha}\!\right)^{2}}
&\!\rno\!=\!1
\end{cases}
\\
\label{eq:information:der-B}
&\!=\!
\begin{cases}
\!\tfrac{1}{(\rno-1)^{2}}
\EXS{\varpi_{\rno}}{\!\ln\tfrac{\tpn{\rno}(\mW|\dout) \nmn{\rno,\mP}}{\tpn{1}(\mW|\dout) \lon{\mmn{\rno,\mP}}}\!}
&\!\rno\!\neq\!1
\\
\!\tfrac{1}{2}
\EXS{\varpi_{1}\!}{\left(\!\ln\tfrac{\tpn{1}(\mW|\dout)}{\mP(\mW)}
	\!-\!\RMI{1}{\mP}{\cha}\!\right)^{2}\!}
&\!\rno\!=\!1
\end{cases}
\end{align}
where \(\varpi_{\rno}\) is a probability measure on \(\outA\otimes \sss{\supp{\mP}}\)
whose \(\outS\) marginal is \(\frac{\mmn{\rno,\mP}}{\lon{\mmn{\rno,\mP}}}\) 
and whose conditional distribution is \(\tpn{\rno}\).

The continuity and the convexity properties of the \renyi information
in the prior follow from the corresponding properties of the mean measure
described in Lemma \ref{lem:powermeanP}.
\begin{lemma}\label{lem:informationP}
Let \(\cha\) be a subset of \(\pmea{\outA}\). 
\begin{enumerate}[(a)]
\item\label{informationP-a} 
If \(\rno \in [0,1)\), then \(\RMI{\rno}{\mP}{\cha}\) is a non-negative  quasi-concave
function of \(\mP\) on \(\pdis{\cha}\) that is continuous for the total variation topology
on \(\pdis{\cha}\).
\item\label{informationP-b} 
If \(\rno \in [1,\infty]\), then \(\RMI{\rno}{\mP}{\cha}\) is a non-negative  concave 
function of \(\mP\) on \(\pdis{\cha}\).
\end{enumerate}
\end{lemma}

Gallager \cite[p. 18]{gallager65} and \csiszar \cite[Lemma 3.2]{csiszar72}
established the continuity of \(\RMI{\rno}{\mP}{\cha}\) in \(\mP\) on \(\pdis{\cha}\),
for finite \(\cha\)'s.
For arbitrary \(\cha\)'s, however, \(\RMI{\rno}{\mP}{\cha}\) is continuous only for orders 
in \((0,1)\);
for orders in \([1,\infty]\), \(\RMI{\rno}{\mP}{\cha}\) is continuous in \(\mP\) on \(\pdis{\cha}\) 
iff \(\sup_{\mP\in\pdis{\cha}}\RMI{\rno}{\mP}{\cha}\) is finite, 
see Lemma \ref{lem:finitecapacity}-(\ref{finitecapacity-d}).
The finiteness of \(\sup_{\mP\in\pdis{\cha}}\RMI{\rno}{\mP}{\cha}\)
also implies the uniform equicontinuity of the \renyi information, 
see Lemma \ref{lem:finitecapacity}-(\ref{finitecapacity-uecP},\ref{finitecapacity-uecO}).
The discontinuity of various Shannon information measures for 
countably infinite output sets have previously been pointed out 
by Ho and Yeung in \cite{hoY09}. 

\subsection{The \renyi Divergence}\label{sec:divergence}
\begin{definition}\label{def:divergence}
	Let \(\mW\) and \(\mQ\) be two 
	non-zero finite measures on the measurable space \((\outS,\outA)\); then
	\emph{the order \(\rno\) \renyi divergence between \(\mW\) and \(\mQ\)} is
	\begin{align}
	\label{eq:def:divergence}
	\!\!\RD{\rno}{\mW}{\mQ}
	&\!\DEF\!\! 
	\begin{cases}
	-\ln \mQ\left(\der{\mW}{\rfm}>0\right) 
	&
	\rno\!=\!0
	\\
	\tfrac{1}{\rno-1}\!\ln\int\!\left(\!\der{\mW}{\rfm}\!\right)^{\!\rno}
	\left(\!\der{\mQ}{\rfm}\!\right)^{\!1\!-\!\rno}\!\!\rfm(\!\dif{\dout}\!)
	&\rno\!\in\reals{+}\!\!\setminus\!\!\{\!1\!\}
	\\
	\int\!\der{\mW}{\rfm}\!\left(\!\ln\der{\mW}{\rfm}\!-\!\ln \der{\mQ}{\rfm}\!\right) \rfm(\!\dif{\dout}\!)
	&\rno\!=\!1
	\\
	\ln \essup_{\rfm} \der{\mW}{\rfm}/\der{\mQ}{\rfm}
	&\rno\!=\!\infty
	\end{cases}
	\end{align}
	where \(\rfm\) is any measure satisfying \(\mW\AC\rfm\) and \(\mQ\AC\rfm\).
\end{definition}

The \renyi divergence is usually defined for 
probability measures; the inclusion of finite measures allows us to express 
certain observations, such as Lemma \ref{lem:divergence-RM} given 
in the following, more 
succinctly.\footnote{It is also convenient while studying the concept of 
	the \renyi\!\!-Gallager information and capacity, see \cite{nakiboglu17}
	and \cite{nakiboglu18C}.}
Nonetheless, the propositions derived for the usual definition with probability 
measures suffice for our purposes most of the time.
We appropriate all the propositions we need for our analysis, 
except Lemma \ref{lem:divergence-RM},
from the recent paper of van Erven and \harremoes \cite{ervenH14}. 
The equivalence of Definition \ref{def:divergence} and the one used by 
van Erven and \harremoes in \cite{ervenH14} for probability measures 
follows from \cite[Thm. 4-6]{ervenH14}.

\begin{lemma}[\!\!{\cite[Thm. 3, Thm. 7]{ervenH14}}]\label{lem:divergence-order}
	For all \(\mW,\mQ\in\pmea{\outA}\),
	\(\RD{\rno}{\mW}{\mQ}\) is a 
	nondecreasing and lower semicontinuous function of \(\rno\) on \([0,\!\infty]\)
	that is continuous on \([0,(1\vee\chi_{\mW,\mQ})]\) where
	\(\chi_{\mW,\mQ}\DEF\sup\{\rno:\RD{\rno}{\mW}{\mQ}<\infty\}\).
\end{lemma}

Lemma \ref{lem:divergence-RM} is evident from the definition of \renyi divergence.
\begin{lemma}\label{lem:divergence-RM}
	Let \(\mW\), \(\mQ\), \(\mV\) be non-zero finite measures on \((\outS,\outA)\)
	and \(\rno\) be an order in \([0,\infty]\).
	\begin{itemize}
		\item If \(\mV\leq\mQ\), then \(\RD{\rno}{\mW}{\mQ}\leq \RD{\rno}{\mW}{\mV}\).
\item  If \(\mQ=\gamma\mV\) for some \(\gamma\in\reals{+}\)
and 
either \(\mW\) is a probability measure
or \(\rno\neq1\), then 
\(\RD{\rno}{\mW}{\mQ}=\RD{\rno}{\mW}{\mV}-\ln \gamma\).
	\end{itemize} 
\end{lemma}

Let \(\mW\) and \(\mQ\) be two probability measures on 
the measurable space \((\outS,\outA)\)
and \(\alg{G}\) be a sub-\(\sigma\)-algebra of \(\outA\).
Then the identities  \(\wmn{|\alg{G}}(\oev)=\mW(\oev)\) for all \(\oev\in\alg{G}\) 
and \(\qmn{|\alg{G}}(\oev)=\mQ(\oev)\) for all \(\oev\in\alg{G}\) uniquely define 
probability measures \(\wmn{|\alg{G}}\) and \(\qmn{|\alg{G}}\) on \((\outS,\alg{G})\).
In the following, we denote \(\RD{\rno}{\wmn{|\alg{G}}}{\qmn{|\alg{G}}}\) by 
\(\RDF{\rno}{\alg{G}}{\mW}{\mQ}\). 
\begin{lemma}[\!\!{\cite[Thm. 9]{ervenH14}}]\label{lem:divergence-DPI}
	For any \(\rno\in[0,\infty]\), probability measures \(\mW\) and \(\mQ\) on \((\outS,\outA)\)
	and sub-\(\sigma\)-algebra \(\alg{G}\subset\outA\) 
	\begin{align}
	\notag
	\RD{\rno}{\mW}{\mQ}
	&\geq\RDF{\rno}{\alg{G}}{\mW}{\mQ}. 
	\end{align}
\end{lemma}

\begin{lemma}[\!\!{\cite[Thm. 3, Thm. 31]{ervenH14}}]\label{lem:divergence-pinsker}
For any \(\rno\in[0,\infty]\), probability measures \(\mW\) and \(\mQ\) on \((\outS,\outA)\)
\begin{align}
\label{eq:lem:divergence-pinsker}
\RD{\rno}{\mW}{\mQ}
&\geq\tfrac{1\wedge\rno}{2} \lon{\mW-\mQ}^2.
\end{align}
\end{lemma}
For orders in \((0,1]\), the bound given in \eqref{eq:lem:divergence-pinsker} is called 
the Pinsker's inequality;
it has been proved by \csiszar \cite{csiszar67A} for \(\rno=1\) 
case and by Augustin\footnote{
	\(\tfrac{\lon{\mW\!-\!\mQ}^{2}}{2}\!\leq\!\tfrac{1-e^{(\rno-1)\RD{\rno}{\mW}{\mQ}}}{\rno\!(1\!-\!\rno)}\) 
	for all \(\mW,\mQ\!\in\!\pmea{\outA}\) and \(\rno\!\in\![-1,2]\)
	by \cite[Lemma 26.5a]{augustin78}.
	This implies \eqref{eq:lem:divergence-pinsker} for \(\rno\in(0,1)\) via \(e^{-\dinp}\geq 1-\dinp\).}
\cite{augustin78} and Gilardoni \cite{gilardoni10B} for \(\rno\in(0,1)\) 
case.
Furthermore the constant \(\sfrac{\rno}{2}\) 
is the best possible: for any \(\gamma<\sfrac{\rno}{2}\)
there are probability measures \(\mW\) and \(\mQ\) such that 
\(\gamma \lon{\mW-\mQ}^2>\RD{\rno}{\mW}{\mQ}\).
Determination of best lower bound on the \renyi divergence in 
terms of the total variation is an interesting and important 
problem but it is beyond the scope of the current 
manuscript.

\begin{remark}
	Kullback \cite{kullback67,kullback70} 
	bounded \(\RD{1}{\mW}{\mQ}\) from below by 
	\(\sfrac{\lon{\mW-\mQ}^2}{2}  +\sfrac{\lon{\mW-\mQ}^4}{36}\).
	Hence, Pinsker's inequality is tight only for \(\lon{\mW-\mQ}\approx0\).
	Vajda \cite{vajda70} established 
	\(\RD{1}{\mW}{\mQ}\geq \ln(\tfrac{2+\lon{\mW-\mQ}}{2-\lon{\mW-\mQ}})-\tfrac{2\lon{\mW-\mQ}}{2+\lon{\mW-\mQ}}\).
	Vajda's inequality is tight not only for \(\lon{\mW-\mQ}\approx0\) but also for \(\lon{\mW-\mQ}\approx2\).
	Fedotov, \harremoes\!\!, and \topsoe \cite{fedotovHT03} determined the tight lower bound on  
	\(\RD{1}{\mW}{\mQ}\) in terms of \(\lon{\mW-\mQ}\) in a parametric form.
	Gilardoni \cite{gilardoni06,gilardoni10A} proved an equivalent result for \(\fX\)-divergences for 
	twice differentiable \(\fX\)'s. 
	Gilardoni's result implies tight bounds for \renyi divergences, which are recently derived in a more explicit form 
	by Sason \cite[Prop. 1]{sason16}.
	The core observation in the derivation of tight Vajda's inequalities is the sufficiency of the probability
	measures on binary alphabets.  
	Guntuboyina, Saha,  and Schiebinger \cite{guntuboyinaSS14} have recently generalized this observation considerably
	and explained how one can determine tight bounds on an \(\fX\)-divergence when its arguments are constrained in 
	terms of other \(\fX\)-divergences. Recall that the total variation distance is the \(\fX\)-divergence for 
	\(\fX(\dinp)=\abs{\dinp-1}\).
\end{remark}

\begin{lemma}[\!\!{\cite[Thm.  12]{ervenH14}}]\label{lem:divergence-convexity}
	For any order \(\rno\in[0,\infty]\), 
	the order \(\rno\) \renyi divergence is convex in its second argument for 
	probability measures, i.e.
	for all \(\mW,\qmn{0},\qmn{1}\in\pmea{\outA}\) and \(\beta\in(0,1)\) we have
	\begin{align}
	\notag	
	\RD{\rno}{\mW}{\qmn{\beta}} 
	&\leq \beta\RD{\rno}{\mW}{\qmn{1}}+(1-\beta)\RD{\rno}{\mW}{\qmn{0}}
	\end{align}
	where \(\qmn{\beta}=\beta\qmn{1}+(1-\beta)\qmn{0}\).
\end{lemma}

\begin{lemma}[\!\!{\cite[Thm. 13]{ervenH14}}]\label{lem:divergence-quasiconvexity}
	For any order \(\rno\in[0,\infty]\),  
	the order \(\rno\) \renyi divergence is jointly quasi-convex in its arguments 
	for probability measures, i.e. 
	for all \(\wmn{0}\), \(\wmn{1}\), \(\qmn{0}\), \(\qmn{1}\) in \(\pmea{\outA}\) and 
	\(\beta\in(0,1)\) we have
	\begin{align}
	\notag	
	\RD{\rno}{\wmn{\beta}}{\qmn{\beta}} 
	&\leq \RD{\rno}{\wmn{1}}{\qmn{1}}\vee \RD{\rno}{\wmn{0}}{\qmn{0}}
	\end{align}
	where \(\wmn{\beta}=\beta\wmn{1}+(1-\beta)\wmn{0}\) and 
	\(\qmn{\beta}=\beta\qmn{1}+(1-\beta)\qmn{0}\).	
\end{lemma}

\begin{lemma}[\!\!{\cite[Thm 15]{ervenH14}}]\label{lem:divergence:lsc}
	For any order \(\rno\in(0,\infty]\), \(\RD{\rno}{\mW}{\mQ}\) is a 
	lower semicontinuous function of the pair of probability measures 
	\((\mW,\mQ)\) in the topology of setwise convergence.
\end{lemma}

The preceding lemmas discuss only the aspects of the \renyi divergence 
that are useful for our discussion. 
A more comprehensive discussion can be found in \cite{ervenH14}.
\vspace{-.1cm}

\subsection{The \renyi Mean}\label{sec:mean}
We have defined the \renyi information using a closed form expression. 
However, the original definition of the \renyi information by Sibson is  
in terms of an optimization of the
\renyi divergence over a set of probability measures.
These two definitions are equivalent, as it has already been shown by Sibson \cite[Thm. 2.2]{sibson69}.
In the following, we establish this equivalence and briefly discuss an alternative definition of
the \renyi information related to the aforementioned characterization in terms 
of the \renyi divergence.  

\begin{definition}\label{def:mean}
	Let \(\mP\) be a \pmf on \(\pmea{\outA}\); 
	then \emph{the order \(\rno\) \renyi mean for prior \(\mP\)} is 
	\begin{align}
	\label{eq:def:mean}
	\qmn{\rno,\mP}&\!\DEF\!\!
	\begin{cases}
	\tfrac{e^{-\RD{1}{\tpn{0}}{\tpn{1}}} \IND{\vartheta_{\mP}(\dout)=\bar{\vartheta}_{\mP}} \mmn{1,\mP}}{\int e^{-\RD{1}{\tpn{0}}{\tpn{1}}}\IND{\vartheta_{\mP}(\dout)=\bar{\vartheta}_{\mP}} \mmn{1,\mP}(\dif{\dout}) }
	&
	\rno\!=\!0
	\\
	\tfrac{\mmn{\rno,\mP}}{\lon{\mmn{\rno,\mP}}}
	& \rno\!\in\!(0,\!\infty]
	\end{cases}
	\end{align}
	where \(\vartheta_{\mP}(\dout)\!\DEF\!\sum_{\mW} \mP(\mW)\IND{\tpn{1}(\mW|\dout)>0}\),
	\(\bar{\vartheta}_{\mP}\!\DEF\!\essup_{\mmn{1,\mP}}\vartheta_{\mP}\),
	and 
	\(\tpn{0}(\mW|\dout)\!\DEF\!\tfrac{\mP(\mW) \IND{\tpn{1}(\mW|\dout)>0}}{\sum_{\mU} \mP(\mU) \IND{\mP(\mU|\dout)>0}}\).
\end{definition}
Then the following identity can be confirmed by substitution using \eqref{eq:def:divergence}:
For any \(\rno\) in \((0,\infty]\), \(\mP\) in \(\pdis{\cha}\), and 
\(\mQ\) in \(\pmea{\outA}\),
\begin{align}
\label{eq:sibson}
\!\RD{\rno}{\mP \mtimes \cha}{\mP\!\otimes\!\mQ} 
&\!=\!
\RD{\rno}{\mP \mtimes \cha}{\mP\!\otimes\!\qmn{\rno,\mP}} 
\!+\!\RD{\rno}{\qmn{\rno,\mP}}{\mQ}. 
\end{align}
This identity was first pointed out by Sibson in \cite[p. 153]{sibson69},
then by others \cite[(12)]{csiszar95}
\cite[(43)]{hayashiT17}
\cite[(38)]{polyanskiyV10}
\cite[Lemma 3]{sharmaW13}
\cite[(52)]{verdu15}.
For \(\rno=1\) case, it had been used by \topsoe in \cite{topsoe67},
even before Sibson \cite{sibson69},  and in \cite{topsoe72}.

On the other hand, one can also confirm by substitution that 
\(\RMI{\rno}{\mP}{\cha}\!=\!\RD{\rno}{\mP \mtimes \cha}{\mP\otimes \qmn{\rno,\mP}}\)
for all positive values of \(\rno\).
These two observations lead to the alternative characterization of 
the order \(\rno\) \renyi information in terms of the order \(\rno\) 
\renyi divergence presented in the following lemma, which is
valid for all non-negative orders.
\begin{lemma}\label{lem:information:def}
Let \(\cha\) be a subset of \(\pmea{\outA}\),
\(\mP\) be a \pmf  on \(\cha\), and
\(\rno\) be an order in \([0,\infty]\);
then
\begin{align}
\label{eq:lem:information:defA}
\RMI{\rno}{\mP}{\cha}
&\!=\!\RD{\rno}{\mP \mtimes \cha}{\mP\otimes\qmn{\rno,\mP}} 
\\
\label{eq:lem:information:defB}
&\!=\!\inf_{\!\!\mQ\in \pmea{\outA}}\!\!
\RD{\rno}{\mP \mtimes \cha}{\mP\otimes\mQ}
\\
\label{eq:lem:information:defC}
&\!=\!\inf_{\!\!\mQ\in \pmea{\outA}}\!\!
\RD{\rno}{\mmn{\rno,\mP}}{\mQ}
&
&\rno\!\in\!(0,\infty]\!\setminus\!\{1\}\!\! 
\end{align}
where \(\mP \mtimes \cha\) is the probability measure on \(\sss{\supp{\mP}} \otimes \outA\) 
whose marginal distribution on 
\(\supp{\mP}\) is \(\mP\) and whose conditional distribution is \(\mW\).
\end{lemma}
Proof of Lemma \ref{lem:information:def} is presented in Appendix \ref{sec:mean-proofs}. 
For any positive order \(\rno\) and prior \(\mP\), 
the only probability measure \(\mQ\) satisfying 
\(\RD{\rno}{\mP \mtimes \cha}{\mP\otimes \mQ}=\RMI{\rno}{\mP}{\cha}\)
is \(\qmn{\rno,\mP}\)
as a result of \eqref{eq:sibson}  and Lemmas \ref{lem:divergence-pinsker}, 
\ref{lem:information:def}. 
In other words, the order \(\rno\) \renyi mean for prior \(\mP\) is the unique 
minimizer for the infimum given in  
\eqref{eq:lem:information:defB} for positive orders \(\rno\). 
For \(\rno=0\), the order zero \renyi mean is still a minimizer by Lemma \ref{lem:information:def} but it is 
not necessarily the unique minimizer.
Any probability measure \(\mQ\) that is absolutely continuous in the \(\qmn{0,\mP}\) satisfies 
\(\RD{0}{\mP \mtimes \cha}{\mP  \otimes \mQ}=\RMI{0}{\mP}{\cha}\).

The definition of \renyi information we have adopted is not the only definition of \renyi information.
The following definition is first proposed by Augustin in \cite[\S 34]{augustin78} and later 
popularized by \csiszar \cite{csiszar95}
\begin{align}
\label{eq:def:augustininformation}
{\cnst{I}}_{{\rno}}^{{\scriptscriptstyle c}}\!\left(\!\mP;\!\cha\!\right)
&\DEF\inf_{\mQ\in\pmea{\outA}} \sum\nolimits_{\mW}\mP(\mW) \RD{\rno}{\mW}{\mQ}.
\end{align}
Unlike the definition we have adopted, the one given in  \eqref{eq:def:augustininformation} 
does not have an equivalent closed form expression. 
But for any finite positive order \(\rno\),  the infimum in  \eqref{eq:def:augustininformation} has a 
unique minimizer, which is a fixed point of an operator defined using 
\(\rno\) and \(\mP\), 
\cite{nakiboglu18C}. 
These properties were first proved by Augustin for orders between zero and one in \cite{augustin78}. 
Thus we have called the quantity defined in \eqref{eq:def:augustininformation},
the order \(\rno\) Augustin information in \cite{nakiboglu17}.
We present a more detailed discussion of the properties of the Augustin information and its relation
to the  \renyi information in \cite{nakiboglu18C}. 

Arimoto proposed a third definition for the \renyi information in \cite{arimoto77}.
Recently, Verd\'{u} has provided a discussion of the \renyi  entropy and these three definitions of 
the \renyi information in \cite{verdu15}.
\vspace{-.2cm}
\section{The \renyi Capacity}\label{sec:capacity}
\begin{definition}\label{def:capacity}
Let \(\rno\) be an order in \([0,\infty]\) and 
\(\cha\) be a subset of \(\pmea{\outA}\);
then \emph{the order \(\rno\) \renyi capacity of \(\cha\)} is 
\begin{align}
\label{eq:def:capacity}
\RC{\rno}{\cha} 
&\DEF \sup\nolimits_{\mP \in \pdis{\cha}}  \RMI{\rno}{\mP}{\cha}.
\end{align}
\end{definition}

Unlike the \renyi information, the  \renyi 
capacity is not a quantity that is introduced or discussed by Sibson in \cite{sibson69}.
In the spirit of his earlier work on \(\fX\)-divergences \cite{csiszar72}, 
\csiszar introduces it in \cite{csiszar95}. 
Prior to either work,
Shannon, Gallager, and Berlekamp had introduced a `capacity', i.e.
\(E_{0}(\rng,\cha)\), 
using \(E_{0}(\rng,\mP)\) in \cite{shannonGB67A}.
\(E_{0}(\rng,\cha)\) is a
scaled version of the \renyi capacity; in particular 
\(E_{0}(\rng,\cha)=\rng \RC{\frac{1}{1+\rng}}{\cha}\) 
for all non-zero \(\rng\) greater than minus one by \eqref{eq:sibsongallager}.

Using the alternative characterization of the \renyi information given in    
\eqref{eq:lem:information:defB}, 
we get the following expression for the order \(\rno\) \renyi capacity
for all \(\rno\) in \([0,\infty]\)
\begin{align}
\label{eq:capacity}
\RC{\rno}{\cha}
&=\sup\nolimits_{\mP \in \pdis{\cha}}\inf\nolimits_{\mQ\in\pmea{\outA}} \RD{\rno}{\mP \mtimes \cha}{\mP\otimes\mQ}.
\end{align}

For finite orders the \renyi capacity does not have a closed form expression.
The supremum given in the definition of the \renyi capacity need not to be 
finite, see Examples \ref{eg:infinitebutcontinuous} and \ref{eg:discontinuity}. 
Even when the supremum is finite it might not be achieved by any prior, 
i.e. there are \(\cha\)'s for which 
\(\RMI{\rno}{\mP}{\cha}<\RC{\rno}{\cha}\) for all \(\mP\in\pdis{\cha}\),
see Examples \ref{eg:singular-countable} and \ref{eg:erasure}.
When the supremum is achieved, the optimal prior might not be unique,
i.e. there are \(\cha\)'s for which 
\(\RMI{\rno}{\mP_{1}}{\cha}=\RMI{\rno}{\mP_{2}}{\cha}=\RC{\rno}{\cha}\) for 
\(\mP_{1}\neq\mP_{2}\) both of which are in \(\pdis{\cha}\), see
Example \ref{eg:extendedbsc}.
These subtleties, however, do not constitute a serious impediment for analyzing 
the \renyi capacity.

In \S\ref{sec:capacityorder}, we analyze the \renyi capacity as a function of the order. 
In \S\ref{sec:capacityfiniteness}, we determine necessary and sufficient conditions 
for the finiteness of the \renyi capacity and investigate the implications of the finiteness 
of the \renyi capacity on the continuity of the mean measure and the \renyi information.

\subsection[\(\RC{\rno}{\cha}\) as a Function of \(\rno\)]{The \renyi Capacity as a Function of the Order}\label{sec:capacityorder}

We are interested in characterizing the behavior of the \renyi capacity as a function of  
the order because the operational significance of the \renyi capacity
---at least for the channel coding problem and the sphere packing bound--- 
is not through its value at a specific order but through its behavior as a function of the order.
Parts (\ref{capacityO-ilsc},\ref{capacityO-zo},\ref{capacityO-convexity},\ref{capacityO-zofintieness},\ref{capacityO-continuity})
of Lemma \ref{lem:capacityO} characterize the behavior of the \renyi capacity for an
arbitrary \(\cha\) as a function of the order. 
In our analysis relying on the \renyi capacity 
some of our results might be valid only for countable or finite \(\cha\)'s
rather than arbitrary \(\cha\)'s.
Parts (\ref{capacityO-countable},\ref{capacityO-dini}) of Lemma \ref{lem:capacityO} are 
useful in such situations.\footnote{As pointwise statements, i.e. as statements for a given order, 
Lemma \ref{lem:capacityO}-(\ref{capacityO-countable},\ref{capacityO-dini}) 
follow trivially from the definition of the \renyi capacity. They are non-trivial only 
because their assertions hold for all orders for the same \(\cha'\).} 
See the proof of \cite[Corollary \ref*{D-cor:cost}]{nakiboglu18D} for such a 
situation for the Augustin capacity.

\begin{lemma}\label{lem:capacityO}
Let \(\cha\) be a subset of \(\pmea{\outA}\).
\begin{enumerate}[(a)]
\item\label{capacityO-ilsc} \(\RC{\rno}{\cha}\) is nondecreasing and lower semicontinuous  
in \(\rno\) on \([0,\infty]\).
\item\label{capacityO-countable} There  exists a countable subset  \(\cha'\) of \(\cha\) 
satisfying \(\RC{\rno}{\cha'}=\RC{\rno}{\cha}\) for all \(\rno\in[0,\infty]\).
\item\label{capacityO-zo} \(\tfrac{1-\rno}{\rno}\RC{\rno}{\cha}\) is  
nonincreasing and continuous in \(\rno\) on \((0,1)\)
and \(\RC{\rno}{\cha}\) is continuous in \(\rno\) on \((0,1]\).
\item\label{capacityO-convexity} \((\rno-1)\RC{\rno}{\cha}\) is 
convex in \(\rno\) on \((1,\infty)\).
\item\label{capacityO-zofintieness} If \(\RC{\rnt}{\cha}<\infty\) for an \(\rnt\in(0,1) \),
 then \(\RC{\rno}{\cha}\) is finite for all \(\rno\in [0,1)\).
\item\label{capacityO-continuity} If \(\RC{\rnt}{\cha}<\infty\) for an \(\rnt\in(0,\infty]\), 
then \(\RC{\rno}{\cha}\) is 
nondecreasing and continuous\footnote{We are unable to 
establish the continuity of \(\RC{\rno}{\cha}\) at \(\rno=0\) for arbitrary \(\cha\). For finite \(\cha\), Sion's minimax theorem implies the
continuity of \(\RC{\rno}{\cha}\) at \(\rno=0\), see Lemma \ref{lem:finitecapacity}-(\ref{finitecapacity-fc}).} 
in \(\rno\) on \((0,\rnt]\).
\item\label{capacityO-dini} If \(\RC{\rnt}{\cha}<\infty\) for an \(\rnt\in(0,\infty]\),
 then \(\forall\epsilon>0,\exists\) a finite subset  \(\cha'\) of \(\cha\) such that 
\(\RC{\rno}{\cha'}>\RC{\rno}{\cha}-\epsilon\) for all \(\rno\in[\epsilon,\rnt]\).
\end{enumerate}
\end{lemma}

The \renyi information \(\RMI{\rno}{\mP}{\cha}\) is continuous\label{dichotomy} 
in \(\rno\) for any \(\mP\) in \(\pdis{\cha}\) by Lemma \ref{lem:informationO},
however the \renyi capacity \(\RC{\rno}{\cha}\) is not necessarily continuous in \(\rno\). 
Yet, if the \renyi capacity \(\RC{\rno}{\cha}\) is not continuous in \(\rno\) on \((0,\infty]\), 
then it has a very specific shape as a result of Lemma \ref{lem:capacityO}:\label{para:dichotomy}
there exists a \(\rnf\in [1,\infty)\) such that \(\RC{\rno}{\cha}\) is bounded and continuous on \((0,\rnf]\) and 
infinite on \((\rnf,\infty]\).
In order to see why, first note that if \(\RC{1/2}{\cha}=\infty\), then \(\RC{\rno}{\cha}=\infty\) for all 
\(\rno\) in \((0,\infty]\) by 
Lemma \ref{lem:capacityO}-(\ref{capacityO-ilsc},\ref{capacityO-zofintieness}) and  \(\RC{\rno}{\cha}\) is continuous on \((0,\infty]\).
On the other hand, if \(\RC{\infty}{\cha}<\infty\), then \(\RC{\rno}{\cha}\) is continuous on \((0,\infty]\) by 
Lemma \ref{lem:capacityO}-(\ref{capacityO-continuity}).
Hence, \(\RC{\rno}{\cha}\) can fail to be continuous on \((0,\infty]\) only when \(\RC{1/2}{\cha}<\infty\) and \(\RC{\infty}{\cha}=\infty\).
Let \(\chi_{\cha}\) be the set of all orders \(\rno\) for which \(\RC{\rno}{\cha}\) is finite, i.e. 
\begin{align}
\notag
\chi_{\cha}
&\DEF\{\rno\in \reals{+} :\RC{\rno}{\cha}<\infty\}.
\end{align}
\(\chi_{\cha}\) is either of the form \((0,\rnf)\) for a \(\rnf\in [1,\infty]\) or of the form \((0,\rnf]\) 
for a \(\rnf\in[1,\infty)\) because \(\RC{\rno}{\cha}\) is nondecreasing by Lemma \ref{lem:capacityO}-(\ref{capacityO-ilsc}) 
and finite on \((0,1)\) by Lemma {\ref{lem:capacityO}}-(\ref{capacityO-zofintieness}). 
If \(\chi_{\cha}=(0,\rnf)\) for some \(\rnf\in [1,\infty]\), 
then \(\RC{\rno}{\cha}\) is continuous on \((0,\rnf]\)  
by Lemma \ref{lem:capacityO}-(\ref{capacityO-ilsc},\ref{capacityO-continuity}), 
\(\RC{\rno}{\cha}\) is infinite on \([\rnf,\infty]\) by the hypothesis,
and hence \(\RC{\rno}{\cha}\)  is continuous on \((0,\infty]\) by the pasting lemma  \cite[Thm. 18.3]{munkres}.
---Example \ref{eg:infinitebutcontinuous} provides a \(\cha\) for each \(\rnf\in(1,\infty)\) such that
\(\chi_{\cha}=(0,\rnf)\).---
Thus unless  \(\chi_{\cha}=(0,\rnf]\) for some \(\rnf\in [1,\infty)\), \(\RC{\rno}{\cha}\) is 
continuous on \((0,\infty]\).
If \(\chi_{\cha}=(0,\rnf]\), then \(\RC{\rno}{\cha}\) is bounded and continuous on \((0,\rnf]\)
and infinite on \((\rnf,\infty]\).
Hence the \renyi capacity has a unique discontinuity on \((0,\infty]\), which is at \(\rnf\).
---Example \ref{eg:discontinuity} provides a \(\cha\) for each \(\rnf\in[1,\infty)\) such that 
\(\RC{\rno}{\cha}\) has its unique discontinuity at \(\rnf\).\label{bookmarkA}---

\begin{proof}[Proof of Lemma \ref{lem:capacityO}]~
\begin{enumerate}
\item[(\ref{capacityO-ilsc})]
The pointwise supremum of a family of nondecreasing (lower semicontinuous) functions is 
nondecreasing (lower semicontinuous).
Then \(\RC{\rno}{\cha}\) is nondecreasing and lower semicontinuous in \(\rno\) on \([0,\infty]\) 
because \(\RC{\rno}{\cha}\) is the pointwise supremum of the family 
\(\{\!\RMI{\rno}{\mP}{\cha}\!\}_{\mP\in\pdis{\cha}}\)
and \(\RMI{\rno}{\mP}{\cha}\) is nondecreasing and continuous in \(\rno\) for each \(\mP\!\in\!\pdis{\cha}\) 
by Lemma \ref{lem:informationO}.

\item[(\ref{capacityO-countable})]
The \renyi capacity is a nondecreasing and lower semicontinuous function of the order by part (\ref{capacityO-ilsc}). Then 
\begin{align}
\notag
\RC{\rnt}{\cha}
&=\sup\nolimits_{\rno\in (0,\rnt)\cap \rationals{}} \RC{\rno}{\cha}
&
&\forall \rnt\in(0,\infty].
\end{align}
Consequently, \(\RC{\rno}{\cha'}=\RC{\rno}{\cha}\) for all \(\rno\) in \([0,\infty]\) if 
\(\RC{\rno}{\cha'}=\RC{\rno}{\cha}\) for all \(\rno\in\rationals{\geq0}\).
Choose a sequence of \pmf\!\!'s  
\(\{\pma{}{(\rno,\ind)}\}_{\ind\in\integers{+}}\) satisfying
\(\RMI{\rno}{\pma{}{(\rno,\ind)}}{\cha}\uparrow\RC{\rno}{\cha}\)
for each  \(\rno\in\rationals{\geq0}\). 
Let \(\cha'\) be \(\cup_{\rno\in\rationals{\geq0}} \cup_{\ind \in \integers{+}} \supp{\pma{}{(\rno,\ind)}}\).
Then \(\RC{\rno}{\cha'}=\RC{\rno}{\cha}\) for all \(\rno\in\rationals{\geq0}\); 
hence for all \(\rno\) in \([0,\infty]\). 
\(\cha'\) is countable because countable union of countable sets is countable.

\item[(\ref{capacityO-zo})]
The definitions  of \(\RMI{\rno}{\mP}{\cha}\) and \(\RC{\rno}{\cha}\) 
imply
\begin{align}
\notag
\tfrac{1-\rno}{\rno}\RC{\rno}{\cha}
&=\sup\nolimits_{\mP\in \pdis{\cha}}\ln\tfrac{1}{\lon{\mmn{\rno,\mP}}}
&
&\forall \rno\in (0,1).
\end{align} 
Furthermore, \(\lon{\mmn{\rno,\mP}}\) is nondecreasing and continuous in \(\rno\),
by Lemma \ref{lem:powermeanO}-(\ref{powermeanO-e}).
Then \(\tfrac{1-\rno}{\rno}\RC{\rno}{\cha}\) is nonincreasing and lower semicontinuous in \(\rno\) on \((0,1)\) 
because the pointwise supremum of a family of nonincreasing (lower semicontinuous) functions is nonincreasing
(lower semicontinuous). 
Thus \(\tfrac{1-\rno}{\rno}\RC{\rno}{\cha}\) and \(\RC{\rno}{\cha}\) are both continuous from the right on \((0,1)\). 
On the other hand \(\RC{\rno}{\cha}\) and \(\tfrac{1-\rno}{\rno}\RC{\rno}{\cha}\) are both continuous from 
the left on \((0,1)\) because \(\RC{\rno}{\cha}\) is nondecreasing and lower semicontinuous on \((0,1)\) 
by part (\ref{capacityO-ilsc}).
Consequently, \(\RC{\rno}{\cha}\) and \(\tfrac{1-\rno}{\rno}\RC{\rno}{\cha}\) are both continuous on \((0,1)\). 
Furthermore, \(\RC{\rno}{\cha}\) is  continuous on \((0,1]\) because 
\(\RC{\rno}{\cha}\) is nondecreasing and lower semicontinuous by part (\ref{capacityO-ilsc}).

\item[(\ref{capacityO-convexity})] 
\(\lon{\mmn{\rno,\mP}}^{\rno}\!\) is log-convex in \(\rno\) by Lemma  \ref{lem:powermeanO}-(\ref{powermeanO-d}).
On the other hand, the definitions  of \(\RMI{\rno}{\mP}{\cha}\) and \(\RC{\rno}{\cha}\) imply
\begin{align}
\notag
(\rno-1)\RC{\rno}{\cha}
&=\sup\nolimits_{\mP\in \pdis{\cha}} \rno\ln\lon{\mmn{\rno,\mP}}
&
&\forall \rno\in(1,\infty).
\end{align}
Then \((\rno\!-\!1)\!\RC{\rno}{\cha}\) is convex in \(\rno\) because
the pointwise supremum of a family of convex functions is convex.

\item[(\ref{capacityO-zofintieness})]
If \(\RC{\rnt}{\cha}\) is finite, then so is \(\RC{\rno}{\cha}\) 
for all \(\rno\) in \([0,\rnt]\) because \(\RC{\rno}{\cha}\) is nondecreasing 
in \(\rno\) by part (\ref{capacityO-ilsc}).
Furthermore, if \(\RC{\rnt}{\cha}\) is finite, then so is \(\RC{\rno}{\cha}\) 
for all \(\rno\) in \([\rnt,1)\) because \(\tfrac{1-\rno}{\rno}\RC{\rno}{\cha}\)
is nonincreasing in \(\rno\) on \((0,1)\) by part (\ref{capacityO-zo}).

\item[(\ref{capacityO-continuity})]
\(\RC{\rno}{\cha}\) is continuous in \(\rno\) on \((0,1]\) by part (\ref{capacityO-zo}). 
Thus we only need to prove the claim for the case when \(\rnt>1\) on \([1,\rnt]\). 
We prove the continuity of \(\RC{\rno}{\cha}\) in \(\rno\) 
first on \((1,\rnt]\), and then from the right at \(\rno=1\). 
If \(\RC{\rnt}{\cha}\) is finite for an \(\rnt\) in \((1,\infty)\),
then \((\rno-1)\RC{\rno}{\cha}\) 
is finite and convex in \(\rno\) on \([1,\rnt]\) by parts (\ref{capacityO-ilsc}) and 
(\ref{capacityO-convexity}). 
Then the continuity of \((\rno-1)\RC{\rno}{\cha}\),
and hence the continuity of \(\RC{\rno}{\cha}\),
in \(\rno\) on \((1,\rnt)\) follows from \cite[Thm. 6.3.3]{dudley}.
On the other hand \(\RC{\rno}{\cha}\) is continuous from the left because 
\(\RC{\rno}{\cha}\) is nondecreasing and lower semicontinuous
in \(\rno\) by part (\ref{capacityO-ilsc}).
Hence, \(\RC{\rno}{\cha}\) is continuous in \(\rno\) on \((1,\rnt]\).

If \(\RC{\infty}{\cha}\) is finite, then \(\RC{\rnt}{\cha}\) is finite for all \(\rnt\!\in\!\reals{+}\) by part (\ref{capacityO-ilsc})
and \(\RC{\rno}{\cha}\) is continuous in \(\rno\) on \(\reals{+} \) because the continuity of a function on a collection of open set implies 
its continuity on their union, \cite[Thm. 18.2]{munkres}.
This implies the continuity  on \((0,\infty]\) because 
\(\RC{\rno}{\cha}\) is nondecreasing and lower semicontinuous 
in \(\!\rno\!\) by  part (\ref{capacityO-ilsc}).

To prove the continuity of \(\RC{\rno}{\cha}\) from the right at one, we first prove that
\(\{\RMI{\rno}{\mP}{\cha}\}_{\mP\in\pdis{\cha}}\) is equicontinuous from the right at \(\rno\!=\!1\). 
The definitions of \(\dmn{\rno,\mP}\) and \(\RMI{\rno}{\mP}{\cha}\) 
given in  \eqref{eq:def:dpowermean} and \eqref{eq:def:information}  
and Lemma \ref{lem:powermeandensityO}-(\ref{powermeandensityO-b}) 
imply
\begin{align}
\notag
\RMI{\rno}{\mP}{\cha}-\RMI{1}{\mP}{\cha}
&=\tfrac{\rno \ln \lon{\mmn{\rno,\mP}}-(\rno-1)\lon{\dmn{1,\mP}}}{\rno-1}
\end{align}
for all \(\rno\) in \((1,\rnt]\) and \(\mP\) in \(\pdis{\cha}\).
The expression in the numerator is differentiable in \(\rno\) because \(\lon{\mmn{\rno,\mP}}\) 
is differentiable by Lemma \ref{lem:powermeanO}-(\ref{powermeanO-b}).
Furthermore, \(\der{}{\rno}\lon{\mmn{\rno,\mP}}=\lon{\dmn{\rno,\mP}}\) by 
Lemma \ref{lem:powermeanO}-(\ref{powermeanO-b}) 
and the numerator is zero at \(\rno=1\). 
Then by the mean value theorem  \cite[5.10]{rudin}, there exists a \(\rnf\in (1,\rno)\) such that
\begin{align}
\notag
\RMI{\rno}{\mP}{\cha}-\RMI{1}{\mP}{\cha}
&=\ln \lon{\mmn{\rnf,\mP}}+\rnf \tfrac{\lon{\dmn{\rnf,\mP}}}{\lon{\mmn{\rnf,\mP}}}-\lon{\dmn{1,\mP}}.
\end{align}
The expression on the right hand side is differentiable in \(\rnf\)  because  
 \(\der{}{\rnf}\lon{\mmn{\rnf,\mP}}=\lon{\dmn{\rnf,\mP}}\) and \(\der{}{\rnf}\lon{\dmn{\rnf,\mP}}=\ddmn{\rnf,\mP}(\outS)\) by 
Lemma \ref{lem:powermeanO}-(\ref{powermeanO-b},\ref{powermeanO-c}). 
On the other hand, \(\lon{\mmn{\rnf,\mP}}>0\) for \(\rnf\in\reals{+}\) and \(\lon{\mmn{1,\mP}}=1\) by 
Lemma \ref{lem:powermeanequivalence}-(\ref{powermeanequivalence-a}).
Then the expression on the right hand side is zero at \(\rnf=1\).
Hence, using the mean value theorem \cite[5.10]{rudin} once again we
can conclude that there exists a \(\rnb\in (1,\rnf)\) such that
\begin{align}
\label{eq:UECatone-A}
\tfrac{\RMI{\rno}{\mP}{\cha}-\RMI{1}{\mP}{\cha}}{\rnf-1}
&=2\tfrac{\lon{\dmn{\rnb,\mP}}}{\lon{\mmn{\rnb,\mP}}}
+\rnb\tfrac{\ddmn{\rnb,\mP}(\outS)}{\lon{\mmn{\rnb,\mP}}}
-\rnb  \tfrac{\lon{\dmn{\rnb,\mP}}^{2}}{\lon{\mmn{\rnb,\mP}}^{2}}.
\end{align}
On the other hand, using the definition of \(\ddmn{\rno,\mP}\)
given in  \eqref{eq:def:ddpowermean} together with  
Lemma \ref{lem:powermeandensityO}-(\ref{powermeandensityO-b}) 
and \(\rnb>1\) we get
\begin{align}
\notag
\!\tfrac{\ddmn{\rnb,\mP}(\outS)}{\lon{\mmn{\rnb,\mP}}}
&\leq \EXS{\qmn{\rnb,\mP}}{\sum\nolimits_{\mW}\!
	\tfrac{\tpn{\rnb}(\mW|\dout)}{\rnb^3}
	\ln^{2}\tfrac{\tpn{\rnb}(\mW|\dout)}{\mP(\mW)}}
-\tfrac{2\lon{\dmn{\rnb,\mP}}}{\rnb \lon{\mmn{\rnb,\mP}}}.
\end{align}
Then using Lemma \ref{lem:powermeandensityO}-(\ref{powermeandensityO-a})
and \eqref{eq:UECatone-A} we get
\begin{align}
\notag
\hspace{-.2cm}\tfrac{\RMI{\rno}{\mP}{\cha}-\RMI{1}{\mP}{\cha}}{\rnf-1}
&\!\leq\!\EXS{\qmn{\rnb,\mP}\!}{\!\sum\nolimits_{\mW}\!
	\tfrac{\tpn{\rnb}(\mW|\dout)}{\rnb^2}
	\ln^{2}\tfrac{\tpn{\rnb}(\mW|\dout)}{\mP(\mW)}\!}
\\
\notag
&\!=\!\EXS{\qmn{\rnb,\mP}\!}{\!\sum\nolimits_{\mW}\!\mP(\mW)
	\left[\!\tfrac{\tpn{1}(\mW|\dout)}{\mP(\mW)\nmn{\rnb,\mP}}\!\right]^{\rnb} 
	\!\ln^2\!\tfrac{\tpn{1}(\mW|\dout)}{\mP(\mW) \nmn{\rnb,\mP}}\!}.
\end{align}
Recall that \(\dinp^{\rnb}\ln^{2} \dinp\leq (\tfrac{2}{\rnb e})^{2}\)
for all \(\dinp\in [0,1]\) and \(\beta>0\)
and
\(\ln^{2} \dinp \leq (\tfrac{2}{\epsilon e})^{2} \dinp^{\epsilon}\)
for all \(\dinp\geq 1\) and \(\epsilon>0\).
Thus
\begin{align}
\notag
\tfrac{\RMI{\rno}{\mP}{\cha}-\RMI{1}{\mP}{\cha}}{\rnf-1}
&\leq
\EXS{\qmn{\rnb,\mP}}{(\tfrac{2}{\rnb e})^2+(\tfrac{2}{\epsilon e})^2
	(\tfrac{\nmn{\rnb+\epsilon,\mP}}{\nmn{\rnb,\mP}})^{\rnb+\epsilon}}.
\end{align}
Since \((\nmn{\rno,\mP})^{\rno}\) is log-convex in \(\rno\) by Lemma \ref{lem:powermeandensityO}-(\ref{powermeandensityO-c}),
\begin{align}
\notag
(\nmn{\rnb+\epsilon,\mP})^{\rnb+\epsilon} 
&\leq (\nmn{\rnb,\mP})^{\rnb+\epsilon-1} \nmn{\frac{\rnb}{1-\epsilon},\mP}
&
&\forall \epsilon\in (0,1), \rnb>1.
\end{align}
Then using the fact that \(\lon{\mmn{\rnb,\mP}}\geq \lon{\mmn{1,\mP}}=1\) we get
\begin{align}
\notag
\tfrac{\RMI{\rno}{\mP}{\cha}-\RMI{1}{\mP}{\cha}}{\rnf-1}
&\leq 
\left[(\tfrac{2}{\rnb e})^2+(\tfrac{2}{\epsilon e})^2  \lon{\mmn{\frac{\rnb}{1-\epsilon},\mP}}\right].
\end{align}
Note that \(\lon{\mmn{\frac{\rnb}{1-\epsilon},\mP}}\leq \lon{\mmn{\frac{\rno}{1-\epsilon},\mP}}\)
because \(\lon{\mmn{\rno,\mP}}\) is nondecreasing in \(\rno\) by Lemma \ref{lem:powermeanO}-(\ref{powermeanO-e}).
Then the definition of \renyi information, \(\rnb>1\), and \(\rnf\in (1,\rno) \) 
imply for any \(\epsilon\in (0,\tfrac{\rnt-1}{\rnt})\), 
\(\rno\in [1,(1-\epsilon)\rnt]\) and \(\mP\in\pdis{\cha}\) that
\begin{align}
\notag
\RMI{\rno}{\mP}{\cha}-\RMI{1}{\mP}{\cha}
&\leq \tfrac{8(\rno-1)}{\epsilon^{2} e^{2}} e^{\frac{\rno-1+\epsilon}{\rno}\RMI{\frac{\rno}{1-\epsilon}}{\mP}{\cha}}
\\
\label{eq:UECatone}
&\leq  \tfrac{8(\rno-1)}{\epsilon^{2}e^{2}} e^{\frac{\rnt-1}{\rnt}\RMI{\rnt}{\mP}{\cha}}.
\end{align}
Then for any \(\epsilon\in (0,\tfrac{\rnt-1}{\rnt})\) and \(\rno\in [1,(1-\epsilon)\rnt]\) we have 
\begin{align}
\notag
\RC{\rno}{\cha}
&\leq \sup\nolimits_{\mP \in \pdis{\cha}} \RMI{1}{\mP}{\cha}+\tfrac{8(\rno-1)}{\epsilon^{2}e^{2}}e^{\frac{\rnt-1}{\rnt}\RMI{\rnt}{\mP}{\cha}} 
\\
\notag
&\leq \RC{1}{\cha}+\tfrac{8(\rno-1)}{\epsilon^{2}e^{2}}e^{\frac{\rnt-1}{\rnt}\RC{\rnt}{\cha}}.
\end{align}
Hence, \(\RC{\rno}{\cha}\) is continuous from the right at \(\rno=1\) if \(\RC{\rnt}{\cha}<\infty\) for an \(\rnt>1\).

\item[(\ref{capacityO-dini})] 
Let us first consider  \(\rnt\in\reals{+} \) case and
construct a sequence \(\{\cha_{\ind}\}_{\ind\in\integers{+}}\) of finite subset of \(\cha\),  such that 
\(\RC{\rno}{\cha_{\ind}}\!\uparrow\!\RC{\rno}{\cha}\) for all \(\rno\in(0,\rnt]\).
Choose a \(\pma{}{(\ind,\jnd)}\) in \(\pdis{\cha}\)  such that  
\(\RMI{\jnd2^{-\ind}}{\pma{}{(\ind,\jnd)}}{\cha}\geq\RC{\jnd2^{-\ind}}{\cha}-\sfrac{1}{2^{\ind}}\)
for each \(\ind\in \integers{+}\) and non-negative integer \(\jnd\) not exceeding \(2^{\ind}\rnt\).
Let \(\cha_{0}\) be the empty set and \(\cha_{\ind}\) be
\(\cha_{\ind-1}\cup_{\jnd=0}^{\lfloor 2^{\ind}\rnt \rfloor} \supp{\pma{}{(\ind,\jnd)}}\)
for each \(\ind\in \integers{+}\). Then
\begin{align}
\notag
\RC{\rno}{\cha_{\ind}}
&\geq \RC{\rno}{\cha_{\ind-1}} 
&&\forall \rno \in[0,\infty],~\ind\in\integers{+} 
\\
\notag
\RC{\rno}{\cha_{\ind}}
&\geq \RC{\rno}{\cha}-\sfrac{1}{2^{\ind}} 
&&\forall \rno  \in \{\tfrac{0}{2^{\ind}},\ldots,\tfrac{\lfloor \rnt 2^{\ind} \rfloor}{2^{\ind}}\},
~\ind\in\integers{+}.
\end{align}
Then \(\RC{\rno}{\cha_{\ind}}\!\uparrow\!\RC{\rno}{\cha}\) for all dyadic rational numbers \(\rno\) less than \(\rnt\). 
Therefore \(\RC{\rno}{\cha_{\ind}}\!\uparrow\!\RC{\rno}{\cha}\) for all \(\rno\in [0,\rnt]\)
because the \renyi capacity is nondecreasing and lower semicontinuous.
Since \(\RC{\rnt}{\cha_{\ind}}\leq \RC{\rnt}{\cha}<\!\infty\),  \(\RC{\rno}{\cha_{\ind}}\)'s and 
\(\RC{\rno}{\cha}\) are continuous in \(\rno\) on \((0,\rnt]\) by part (\ref{capacityO-continuity}). 
Then as a result of Dini's theorem \cite[2.4.10]{dudley}, \(\{\RC{\rno}{\cha_{\ind}}\}_{\ind\in\integers{+}}\) 
converges to \(\RC{\rno}{\cha}\) uniformly on \([\epsilon,\rnt]\), i.e. for all \(\varepsilon>0\), there  exists 
an \(\ind\)  such that  
\(\sup_{\rno\in [\epsilon,\rnt]} \abs{\RC{\rno}{\cha}-\RC{\rno}{\cha_{\jnd}}}<\varepsilon\) for all \(\jnd>\ind\). 

For \(\rnt=\infty\) case,
let \(\kappa_{\ind}\) be the smallest integer satisfying \(\RC{\infty}{\cha}\leq \RC{\sfrac{\kappa_{\ind}}{2^{\ind}}}{\cha}+\sfrac{1}{2^{\ind}}\)
for each \(\ind\in\integers{+}\).
We employ the construction described above for  \(\jnd\)'s not exceeding \(\kappa_{\ind}\) rather than 
\(\jnd\)'s not exceeding \(\lfloor2^{\ind}\rnt\rfloor\).
\end{enumerate}
\vspace{-.4cm}
\end{proof}

\vspace{-.6cm}
\subsection[Finiteness \(\RC{\rno}{\cha}\)]{Finiteness of the \renyi Capacity}\label{sec:capacityfiniteness}
If \(\cha\) is a finite set, then \(\pdis{\cha}\) is compact for the total variation topology 
and various results relying on the compactness can be invoked while 
analyzing the \renyi information. 
For example if \(\cha\) is finite, 
then the compactness of \(\pdis{\cha}\) 
and Sion's minimax theorem imply the continuity of the \renyi capacity 
in the order on \([0,\infty]\), see 
Lemma \ref{lem:finitecapacity}-(\ref{finitecapacity-fc}).
When \(\cha\) is an infinite set, however, \(\pdis{\cha}\) is not compact. 
The finiteness of the \renyi capacity emerges as a shrewd substitute for the compactness of 
\(\pdis{\cha}\) that allows us to assert the continuity of the \renyi information, 
see Lemma \ref{lem:finitecapacity}-(\ref{finitecapacity-uecP},\ref{finitecapacity-uecO}).

Lemma \ref{lem:finitecapacity}-(\ref{finitecapacity-a}-\ref{finitecapacity-d}) 
characterize the finiteness of the order \(\rno\) \renyi capacity in terms of the 
properties of the order \(\rno\) mean measure or \renyi information. 
These equivalent conditions might be easier to confirm or reject for certain 
\(\cha\)'s. 
The equicontinuity results given in 
Lemma \ref{lem:finitecapacity}-(\ref{finitecapacity-uecP},\ref{finitecapacity-uecO}) 
imply that if \(\gamma_{1}\leq\RMI{\rno}{\mP}{\cha}\leq \gamma_{2}\) for all 
\(\mP\in\cset\) for some \(\rno\) in \((0,\rnt)\) and \(\gamma_{1}\) and \(\gamma_{2}\) 
in \([0,\RC{\rnt}{\cha}]\), then for any \(\epsilon>0\) there exists a \(\delta>0\) 
such that \(\gamma_{1}-\epsilon\leq\RMI{\rnf}{\mS}{\cha}\leq\gamma_{2}+\epsilon\) 
for all \(\rnf\) in \([\rno-\delta,\rno+\delta]\) and \(\mS\) in \(\pdis{\cha}\)
satisfying \(\inf_{\mP\in\cset} \lon{\mP-\mS}\leq \delta\).
This observation (or its variants, which can be obtained by employing either
part (\ref{finitecapacity-uecP}) or (\ref{finitecapacity-uecO}) on its own)
might be helpful when we are trying to bound the \renyi information 
or a related function uniformly over the orders and priors through 
a case by case analysis on a subset of \(\pdis{\cha}\) or on its
neighborhoods. 

\begin{lemma}\label{lem:finitecapacity}
Let \(\cha\) be a subset of \(\pmea{\outA}\).
\begin{enumerate}[(a)]
\item\label{finitecapacity-a} For \(\rno \in (0,1)\), \(\RC{\rno}{\cha}=\infty\)  iff 
there exists a sequence 
\(\{\pmn{\ind}\}_{\ind \in \integers{+}} \subset \pdis{\cha}\) such that \(\lim_{\ind \to \infty} \lon{\mmn{\rno,\pmn{\ind}}}=0\).
\item\label{finitecapacity-b} For \(\rno \in (1,\infty]\), \(\RC{\rno}{\cha}=\infty\) iff 
there exists a sequence \(\{\pmn{\ind}\}_{\ind \in \integers{+}} \subset \pdis{\cha}\) such that \(\lim_{\ind \to \infty} \lon{\mmn{\rno,\pmn{\ind}}}=\infty\).
\item\label{finitecapacity-c} 
For \(\rno\!\in\!(1,\infty)\), \(\RC{\rno}{\cha}\!<\!\infty\) iff  \(\mmn{\rno,\mP}\) 
is uniformly continuous in \(\mP\) for the total variation 
topology.\footnote{For \(\rno\in(0,1]\),
	\(\mmn{\rno,\mP}\) is uniformly continuous in \(\mP\), even when \(\RC{\rno}{\cha}=\infty\),
	because \(\mmn{\rno,\mP}\) is Lipschitz continuous on \(\pdis{\pmea{\outA}}\)
	by Lemma \ref{lem:powermeanP}-(\ref{powermeanP-d}).}
\item\label{finitecapacity-d} For \(\rno\!\in\![1,\infty]\), \(\RC{\rno}{\cha}\!<\!\infty\) iff  
\(\RMI{\rno}{\mP}{\cha}\) is continuous in \(\mP\) on \(\pdis{\cha}\)
for the total variation topology.
\item\label{finitecapacity-uecP} 
For  \(\rnt\in\reals{\geq0}\), if \(\RC{\rnt}{\cha}<\infty\),  
then \(\{\RMI{\rno}{\mP}{\cha}\}_{\rno\in[0,\rnt]}\) is uniformly 
equicontinuous,\footnote{For \(\rno \in (0,1)\), Lemma \ref{lem:informationP}-(\ref{informationP-a}) has established 
the continuity of  \(\RMI{\rno}{\mP}{\cha}\) in \(\mP\) without assuming \(\RC{\rno}{\cha}\) to be finite;
but the continuity is not uniform.}
in \(\mP\) on \(\pdis{\cha}\). 

\item\label{finitecapacity-uecO} 
For \(\rnt\in\reals{+}\), if \(\RC{\rnt}{\cha}<\infty\), then \(\{\RMI{\rno}{\mP}{\cha}\}_{\mP\in\pdis{\cha}}\)
is uniformly equicontinuous in \(\rno\) on every compact subset of 
\((0,\rnt)\).\footnote{In order to prove the uniform equicontinuity on compact subsets of \((0,\rnt)\), we prove 
the following stronger statement:  On every compact subset of \((0,\rnt)\), 
\(\{\RMI{\rno}{\mP}{\cha}\}_{\mP\in\pdis{\cha}}\) is a family of Lipschitz continuous functions of \(\rno\) with 
a common Lipschitz constant, see \eqref{eq:lem:uecO}.}
\item\label{finitecapacity-fc} If \(\abs{\cha}\!<\!\infty\), then \(\RC{\rno}{\cha}\) is  
nondecreasing and continuous in \(\rno\) on \([0,\infty]\).
\end{enumerate}
\end{lemma}
Proof of Lemma \ref{lem:finitecapacity} is deferred to Appendix \ref{sec:deferred-proofs}.
For \(\cha\)'s with infinite \(\RC{\rno}{\cha}\), the proof of part (\ref{finitecapacity-d}) 
establishes the discontinuity at every \(\mP\) in \(\pdis{\cha}\). 
For order one the discontinuity of \(\RMI{1}{\mP}{\cha}\) was observed by
Ho and Yeung \cite[Thm. 3]{hoY09} for a different topology for some \(\cha\).
For the same topology they established the continuity of \(\RMI{1}{\mP}{\cha}\)
whenever \(\outS\) is finite \cite[Corollary 8]{hoY09}.
They, however, did not characterize the conditions for 
the continuity of \(\RMI{1}{\mP}{\cha}\) in their framework.

\section{The \renyi Center}\label{sec:center}
The primary focus of this section is Theorem \ref{thm:minimax}, given in the following,
and its applications.
In \S\ref{sec:minimax} we prove Theorem \ref{thm:minimax} and  discuss 
alternative proofs based on Sion's minimax theorem. 
In \S\ref{sec:ervenharremoes} we first prove a lower bound on \(\RRR{\rno}{\cha}{\mQ}\),
i.e. the van Erven-\harremoes bound, 
then we use this bound to establish the continuity of the \renyi center as a function of the order.
\S\ref{sec:miscellaneous} is composed of various applications of Theorem \ref{thm:minimax} 
and the van Erven-\harremoes bound.
\begin{theorem}\label{thm:minimax}
For any  \(\rno \in  (0,\infty]\) and \(\cha \subset \pmea{\outA}\)  
\begin{align}
\RC{\rno}{\cha}
\label{eq:thm:minimax:capacity}
&=\sup\nolimits_{\mP \in \pdis{\cha}}\inf\nolimits_{\mQ \in \pmea{\outA}} \RD{\rno}{\mP \mtimes \cha}{\mP \otimes \mQ}
\\
\label{eq:thm:minimax}
&=
\inf\nolimits_{\mQ \in \pmea{\outA}} \sup\nolimits_{\mP \in \pdis{\cha}} \RD{\rno}{\mP \mtimes \cha}{\mP \otimes \mQ}
\\
\label{eq:thm:minimaxradius}
&=\inf\nolimits_{\mQ\in\pmea{\outA}}\sup\nolimits_{\mW \in \cha} \RD{\rno}{\mW}{\mQ}.
\end{align}
If \(\RC{\rno}{\cha}<\infty\), then there exists a unique \(\qmn{\rno,\cha}\) in 
\(\pmea{\outA}\), called the order \(\rno\) \renyi center, such that
\begin{align}
\label{eq:thm:minimaxcenter}
\RC{\rno}{\cha}
&=\sup\nolimits_{\mP \in \pdis{\cha}} \RD{\rno}{\mP \mtimes \cha}{\mP \otimes \qmn{\rno,\cha}}
\\
\label{eq:thm:minimaxradiuscenter}
&=\sup\nolimits_{\mW \in \cha} \RD{\rno}{\mW}{\qmn{\rno,\cha}}.
\end{align}
Furthermore, for every sequence of priors \(\{\pmn{\ind}\}_{\ind\in\integers{+}}\) satisfying  
\(\lim_{\ind \to \infty} \RMI{\rno}{\pmn{\ind}}{\cha}=\RC{\rno}{\cha}\),
the corresponding sequence of order \(\rno\) \renyi means \(\{\qmn{\rno,\pmn{\ind}}\}_{\ind\in\integers{+}}\)  
is a Cauchy sequence  for the total variation metric on \(\pmea{\outA}\) 
and \(\qmn{\rno,\cha}\) is the unique limit point of that Cauchy sequence.
\end{theorem}
Theorem \ref{thm:minimax} is stated for \(\mP\)'s that are probability mass functions on \(\cha\). 
However, the interpretation of the capacity as the radius implicit in \eqref{eq:thm:minimaxradius} and 
\eqref{eq:thm:minimaxradiuscenter} can be used to extend Theorem \ref{thm:minimax} to the case
when \(\mP\)'s are appropriately defined probability measures, see 
Theorem \ref{thm:Gminimax} in Appendix \ref{sec:generaldefinitions}.

For finite orders, neither the \renyi capacity nor the \renyi center has a closed form expression;
this, however, is not the case for order infinity. 
The following expressions can be confirmed using 
the observation described in \eqref{eq:necessaryandsufficientseq} by the interested reader.
\begin{align}
\label{eq:orderinfty-capacity}
\RC{\infty}{\cha}
&=\ln \lon{\bigvee\nolimits_{\mW\in \cha}\mW},
\\
\label{eq:orderinfty-center}
\qmn{\infty,\cha}
&=\left(\bigvee\nolimits_{\mW\in \cha}\mW \right) e^{-\RC{\infty}{\cha}}.
\end{align}
Before presenting the proof and applications of Theorem \ref{thm:minimax}, 
let us make a brief digression and discuss what is achieved by 
Theorem \ref{thm:minimax} itself. 

The expression in \eqref{eq:thm:minimaxradius} is nothing but 
the definition of the order \(\rno\) \renyi radius \(\RR{\rno}{\cha}\).
Hence, Theorem \ref{thm:minimax} establishes the equality of   
the order \(\rno\) \renyi capacity and 
the order \(\rno\) \renyi radius.
We prefer to express the equality of \(\RC{\rno}{\cha}\) 
and \(\RR{\rno}{\cha}\) as a minimax equality because 
unlike the equality of \(\RC{\rno}{\cha}\) and \(\RR{\rno}{\cha}\) 
itself, the minimax equality  continues to hold in the constrained variant 
of the problem, 
see Theorem \ref{thm:Cminimax} of  Appendix \ref{sec:constrainedcapacity}.

Theorem \ref{thm:minimax} strengthens this minimax equality by asserting the existence
of a unique \renyi center that is achieving the infimum in \eqref{eq:thm:minimax}. 
Recall that we have already established, in Lemma \ref{lem:information:def}, the existence 
of a unique \renyi mean \(\qmn{\rno,\mP}\) achieving the infimum in 
\eqref{eq:thm:minimax:capacity} for any \(\mP\) in \(\pdis{\cha}\).
The suprema in \eqref{eq:thm:minimax:capacity} and \eqref{eq:thm:minimax}, however, 
cannot be replaced by maxima in general.
Example \ref{eg:erasure} provides a \(\cha\) for which 
\(\inf\nolimits_{\mQ \in \pmea{\outA}} \RD{\rno}{\mP \mtimes \cha}{\mP \otimes \mQ}
<\RC{\rno}{\cha}\) 
and 
\(\RD{\rno}{\mP \mtimes \cha}{\mP \otimes \qmn{\rno,\cha}}<\RC{\rno}{\cha}\) 
for all \(\mP\) in \(\pdis{\cha}\).
Evidently, this subtlety exists only for infinite \(\cha\)'s; for finite \(\cha\)'s 
the compactness of \(\pdis{\cha}\) and the extreme value theorem guarantees the existence 
of a \(\mP\) achieving the supremum. 

The last assertion of Theorem \ref{thm:minimax}, relating the problem of determining 
the \renyi capacity to the problem of determining the \renyi center, is important 
because of its potential in simplifying the problem of determining the \renyi center  
---defined as the unique \(\qmn{\rno,\cha}\) satisfying \eqref{eq:thm:minimaxradiuscenter}.

In addition, Theorem \ref{thm:minimax}  provides a necessary and sufficient condition for
a prior \(\mP\) to satisfy \(\RMI{\rno}{\mP}{\cha}=\RC{\rno}{\cha}\). 
That is important because we do not have a closed form expression for the order \(\rno\) 
\renyi capacity, 
yet occasionally the symmetries of the elements of \(\cha\) or numerical calculations
suggest a prior \(\mP\) that might satisfy \(\RMI{\rno}{\mP}{\cha}=\RC{\rno}{\cha}\).
\begin{align}
\label{eq:necessaryandsufficientp}
\RMI{\rno}{\mP}{\cha}
&=\RC{\rno}{\cha}
&
&\mbox{iff}
&
\RRR{\rno}{\cha}{\qmn{\rno,\mP}}
&\leq\RMI{\rno}{\mP}{\cha}.
\end{align}
In order to see why \eqref{eq:necessaryandsufficientp} holds, note that if 
\(\RMI{\rno}{\mP}{\cha}=\RC{\rno}{\cha}\) then considering the sequence 
\(\{\pmn{\ind}\}_{\ind\in\integers{+}}\) where \(\pmn{\ind}=\mP\)
we can conclude that \(\qmn{\rno,\mP}=\qmn{\rno,\cha}\).
Then \(\RRR{\rno}{\cha}{\qmn{\rno,\mP}}\leq\RMI{\rno}{\mP}{\cha}\)
by \eqref{eq:thm:minimaxradiuscenter}.
On the other hand, if \(\RRR{\rno}{\cha}{\qmn{\rno,\mP}}\leq\RMI{\rno}{\mP}{\cha}\)
for some \(\mP\) in \(\pdis{\cha}\),
then \(\RMI{\rno}{\mP}{\cha}=\RC{\rno}{\cha}\) by \eqref{eq:thm:minimaxradius}
because 
\(\RMI{\rno}{\mP}{\cha}\leq\RC{\rno}{\cha}\) 
and \(\RR{\rno}{\cha}\leq \RRR{\rno}{\cha}{\qmn{\rno,\mP}}\)
by the definitions of \renyi capacity and center. 

Following a similar reasoning one can show that 
\(\{\pmn{\ind}\}_{\ind\in\integers{+}}\) is optimal 
iff \(\RRR{\rno}{\cha}{\lim_{\ind\to\infty} \qmn{\rno,\pmn{\ind}}}\leq \lim\nolimits_{\ind\to\infty}\RMI{\rno}{\pmn{\ind}}{\cha}\).
We chose the following less explicit characterization over the aforementioned one
in order to avoid ensuring the convergence of probability measures 
formally.\footnote{We only need \(\RC{\rno}{\cha}\leq\RR{\rno}{\cha}\), but not 
\(\RC{\rno}{\cha}=\RR{\rno}{\cha}\), in order to deduce \(\RMI{\rno}{\mP}{\cha}=\RC{\rno}{\cha}\) 
from \(\RRR{\rno}{\cha}{\qmn{\rno,\mP}}\leq\RMI{\rno}{\mP}{\cha}\).
The sufficiency of the conditions given in \eqref{eq:necessaryandsufficientp} 
and \eqref{eq:necessaryandsufficientseq} for the optimality follows from the max-min 
inequality and the definitions of radius and capacity without invoking Theorem \ref{thm:minimax}. 
We need Theorem \ref{thm:minimax} in order to assert their necessity.} 
\begin{align}
\label{eq:necessaryandsufficientseq}
\hspace{-.2cm}
\lim\limits_{\ind\to\infty} \RMI{\rno}{\pmn{\ind}}{\cha}\!=\!\RC{\rno}{\cha}
\mbox{~iff~}
\exists\mQ\!:\RRR{\rno}{\cha}{\mQ}\!\leq\!\lim\limits_{\ind\to\infty} \RMI{\rno}{\pmn{\ind}}{\cha}
\end{align}
where  \(\mQ\in\pmea{\outA}\) is implicit for the latter statement.
We determine the \renyi capacity in Examples
\ref{eg:singular-finite}, \ref{eg:extendedbsc}, \ref{eg:erasure}
using \eqref{eq:necessaryandsufficientp} 
and in Examples
\ref{eg:singular-countable}, \ref{eg:poissonchannel-mean}
and in Appendix \ref{sec:shiftchannel}
using  \eqref{eq:necessaryandsufficientseq}.

\eqref{eq:thm:minimaxcenter} of Theorem \ref{thm:minimax} and  \eqref{eq:sibson} imply that
\begin{align}
\label{eq:capacityLB}
\RD{\rno}{\qmn{\rno,\mP}}{\qmn{\rno,\cha}}
&\leq  \RC{\rno}{\cha}-\RMI{\rno}{\mP}{\cha}
&
&\forall \mP \in \pdis{\cha}.
\end{align}
Consequently, \(\RD{\rno}{\qmn{\rno,\mP}}{\qmn{\rno,\cha}}\) is close to zero 
whenever \(\RMI{\rno}{\mP}{\cha}\) is close to \(\RC{\rno}{\cha}\).

\subsection[Minimax Theorems \& Rel. Compactness]{Minimax Theorems and the Relative Compactness}\label{sec:minimax}
We start by proving Theorem \ref{thm:minimax} for finite \(\cha\) case.
In this case Theorem \ref{thm:minimax} can be strengthened slightly because 
the existence of an optimal prior is guaranteed.
The optimal prior, however, is not necessarily unique, 
see Example \ref{eg:extendedbsc};
even then, all such \(\mP\)'s have exactly the same \renyi mean. 
For finite \(\outS\) case, Lemma \ref{lem:capacityFLB} is well-known, though in a slightly different form,
see \cite[p. 128]{csiszarkorner}, \cite[Thm. 4.5.1]{gallager} for \(\rno=1\) case 
and \cite[p. 172]{csiszarkorner}, \cite[Thm. 5.6.5]{gallager} for \(\rno \in (0,1)\) case.
\cite[Thm. 3.2]{csiszar72} of \csiszar  implies Lemma \ref{lem:capacityFLB} for \(\rno\)'s
in \(\reals{+}\).

\begin{lemma}\label{lem:capacityFLB}
For any \(\rno\) in  \([0,\infty]\) and finite subset \(\cha\) of \(\pmea{\outA}\),
\(\exists \widetilde{\mP} \in \pdis{\cha}\) such that \(\RMI{\rno}{\widetilde{\mP}}{\cha}=\RC{\rno}{\cha}\).
If \(\rno\) is in \((0,\infty]\), 
then \(\exists!\qmn{\rno,\cha}\in\pmea{\outA}\) such that, 
\begin{align}
\label{eq:lem:capacityFLB}
\RD{\rno}{\qmn{\rno,\mP}}{\qmn{\rno,\cha}}
&\leq  \RC{\rno}{\cha}-\RMI{\rno}{\mP}{\cha}
&
&\forall \mP \in \pdis{\cha}.   
\end{align}
Hence, \(\qmn{\rno,\widetilde{\mP}}=\qmn{\rno,\cha}\) for all \(\widetilde{\mP}\) such that \(\RMI{\rno}{\widetilde{\mP}}{\cha}=\RC{\rno}{\cha}\).
\end{lemma}
\begin{proof}
\begin{enumerate}[(i)]
\item  {\it \(\exists \widetilde{\mP}\in\pdis{\cha}\) such that \(\RMI{\rno}{\widetilde{\mP}}{\cha}=\RC{\rno}{\cha}\):}
Since \(\abs{\supp{\mP}}\leq\abs{\cha}\) for all \(\mP\in \pdis{\cha}\), 
\(\RC{\rno}{\cha}\leq\ln \abs{\cha}\) by Lemma \ref{lem:informationO}. 
Then \(\RMI{\rno}{\mP}{\cha}\) is continuous on \(\pdis{\cha}\) by 
Lemmas \ref{lem:informationP}-(\ref{informationP-a}) and 
\ref{lem:finitecapacity}-(\ref{finitecapacity-d}).  
Then there exists a \(\widetilde{\mP}\) achieving the supremum
by the extreme value theorem, \cite[27.4]{munkres}
because \(\pdis{\cha}\) is compact for finite \(\cha\).

\item {\it If \(\RMI{\rno}{\widetilde{\mP}}{\cha}=\RC{\rno}{\cha}\) for an \(\rno\in(0,\infty] \),
	 then
\(\RD{\rno}{\qmn{\rno,\mP}}{\qmn{\rno,\widetilde{\mP}}}\leq\RC{\rno}{\cha}-\RMI{\rno}{\mP}{\cha}\)
for all \(\mP\in \pdis{\cha}\):}
Let \(\widetilde{\mP} \in \pdis{\cha}\) be such that 
\(\RMI{\rno}{\widetilde{\mP}}{\cha}=\RC{\rno}{\cha}\), \(\mP\) be any member 
of \(\pdis{\cha}\) and \(\pmn{\ind}\) be \(\pmn{\ind}=\tfrac{\ind-1}{\ind}\widetilde{\mP} +\tfrac{1}{\ind}{\mP}\) for 
\(\ind \in \integers{+}\). 

For \(\rno=\infty\) using Lemma \ref{lem:information:def} we get
\begin{align}
\notag
\RMI{\infty}{\pmn{\ind}}{\cha}
&=\left[\RMI{\infty}{\widetilde{\mP}}{\cha}+\RD{\infty}{\qmn{\infty,\widetilde{\mP}}}{\qmn{\infty,\pmn{\ind}}}\right]
\\
\notag
&\qquad~\qquad
\vee\left[\RMI{\infty}{{\mP}}{\cha}+\RD{\infty}{\qmn{\infty,{\mP}}}{\qmn{\infty,\pmn{\ind}}}\right].
\end{align}
Then \(\RD{\infty}{\qmn{\infty,\widetilde{\mP}}}{\qmn{\infty,\pmn{\ind}}}=0\)
because \(\RMI{\infty}{\pmn{\ind}}{\cha}\leq\RC{\infty}{\cha}\)
and \(\RMI{\infty}{\widetilde{\mP}}{\cha}=\RC{\infty}{\cha}\).
Consequently 
\(\qmn{\infty,\widetilde{\mP}}=\qmn{\infty,\pmn{\ind}}\)
and 
\(\RMI{\infty}{\pmn{\ind}}{\cha}=\RC{\infty}{\cha}\).
Thus
\begin{align}
\label{eq:capacityFLB-b-1}
\RMI{\infty}{{\mP}}{\cha}+\RD{\infty}{\qmn{\infty,{\mP}}}{\qmn{\infty,\widetilde{\mP}}}
&\leq \RC{\infty}{\cha}.
\end{align}
For \(\rno=1\) and \(\rno\in\reals{+}\setminus\{1\}\) we have
\begin{align}
\notag
\hspace{-.4cm}
\RMI{1}{\pmn{\ind}}{\cha}
&=\tfrac{\ind-1}{\ind}\left[\RMI{1}{\widetilde{\mP}}{\cha}+\RD{1}{\qmn{1,\widetilde{\mP}}}{\qmn{1,\pmn{\ind}}}\right]
\\
\notag
&\qquad~\qquad
+\tfrac{1}{\ind}\left[\RMI{1}{{\mP}}{\cha}+\RD{1}{\qmn{1,{\mP}}}{\qmn{1,\pmn{\ind}}}\right],
\\
\notag
\RMI{\rno}{\pmn{\ind}}{\cha}
&=\tfrac{1}{\rno-1}\!\ln\!\left[\!
\tfrac{\ind-1}{\ind} e^{(\rno-1)\left(\RMI{\rno}{\widetilde{\mP}}{\cha}+\RD{\rno}{\qmn{\rno,\widetilde{\mP}}}{\qmn{\rno,\pmn{\ind}}}\right)}\right.
\\
\notag
&\quad~\qquad~\qquad+\!\left.\tfrac{1}{\ind} e^{(\rno-1)\left(\RMI{\rno}{{\mP}}{\cha}+\RD{\rno}{\qmn{\rno,{\mP}}}{\qmn{\rno,\pmn{\ind}}}\right)}
\right].
\end{align}
Then using \(\RMI{\rno}{\pmn{\ind}}{\cha}\leq\RC{\rno}{\cha}\),
 \(\RMI{\rno}{\widetilde{\mP}}{\cha}=\RC{\rno}{\cha}\), and 
\(\RD{\rno}{\qmn{\rno,\widetilde{\mP}}}{\qmn{\rno,\pmn{\ind}}}\!\geq\!0\)
we get the following identity
\begin{align}
\notag
\RMI{\rno}{\mP}{\cha}+\RD{\rno}{\qmn{\rno,\mP}}{\qmn{\rno,\pmn{\ind}}}
&\leq \RC{\rno}{\cha}.
\end{align}

Similarly, using \(\RMI{\rno}{\pmn{\ind}}{\cha}\leq\RC{\rno}{\cha}\),
\(\RMI{\rno}{\widetilde{\mP}}{\cha}=\RC{\rno}{\cha}\), 
\(\RMI{\rno}{{\mP}}{\cha}\geq0\),
and
\(\RD{\rno}{\qmn{\rno,{\mP}}}{\qmn{\rno,\pmn{\ind}}}\geq0\)
 we get 
\begin{align}
\notag
\RD{\rno}{\qmn{\rno,\widetilde{\mP}}}{\qmn{\rno,\pmn{\ind}}}
&\!\leq\! 
\begin{cases}
\tfrac{1}{\rno-1}\ln \tfrac{\ind -e^{(1-\rno)\RC{\rno}{\cha}}}{\ind -1}
&\rno\!\in\!\reals{+}\!\setminus\!\{1\}
\\
\tfrac{\RC{\rno}{\cha}}{\ind-1}
&\rno\!=\!1
\end{cases}.
\end{align}
Then \(\qmn{\rno,\pmn{\ind}}\rightarrow\qmn{\rno,\widetilde{\mP}}\) 
in the total variation topology by Lemma \ref{lem:divergence-pinsker}.
Thus
\begin{align}
\notag
\RD{\rno}{\qmn{\rno,\mP}}{\qmn{\rno,\widetilde{\mP}}}
&\leq  \liminf\nolimits_{\ind \to \infty} 
\RD{\rno}{\qmn{\rno,\mP}}{\qmn{\rno,\pmn{\ind}}}
\end{align}
by Lemma \ref{lem:divergence:lsc}. Then 
\begin{align}
\label{eq:capacityFLB-b-2}
\RMI{\rno}{\mP}{\cha}+\RD{\rno}{\qmn{\rno,\mP}}{\qmn{\rno,\tilde{\mP}}}
&\leq \RC{\rno}{\cha}
&
&\forall \rno\in\reals{+}.
\end{align}
\item {\it If \(\rno\in(0,\infty]\), then \(\exists!\qmn{\rno,\cha}\in\pmea{\outA}\) satisfying \eqref{eq:lem:capacityFLB}
such that \(\qmn{\rno,\mP}=\qmn{\rno,\cha}\) for all \(\mP\) with 
\(\RMI{\rno}{\mP}{\cha}=\RC{\rno}{\cha}\):}
\eqref{eq:capacityFLB-b-1}, \eqref{eq:capacityFLB-b-2}
and Lemma \ref{lem:divergence-pinsker} implies that
\begin{align}
\notag
\RMI{\rno}{\mP}{\cha}+\tfrac{\rno\wedge 1}{2}\lon{\qmn{\rno,\mP}-\qmn{\rno,\tilde{\mP}}}^{2}
&\leq \RC{\rno}{\cha}.
\end{align}
Then \(\qmn{\rno,\tilde{\mP}}=\qmn{\rno,\mP}\) for any \(\mP\) satisfying
\(\RMI{\rno}{\mP}{\cha}=\RC{\rno}{\cha}\).
\end{enumerate}
\end{proof}

When \(\cha\) is not a finite but an arbitrary subset of \(\pmea{\outA}\), 
we cannot invoke the extreme 
value theorem to establish the existence an optimal prior \(\mP\)
satisfying \(\RMI{\rno}{\mP}{\cha}=\RC{\rno}{\cha}\) 
because \(\pdis{\cha}\) is not compact. 
Assuming \(\RC{\rno}{\cha}\) to be finite, Theorem \ref{thm:minimax} recovers all assertions 
of Lemma \ref{lem:capacityFLB}, but the existence of an optimal prior, 
albeit in a weaker form.

\begin{proof}[Proof of Theorem \ref{thm:minimax}]
For all \(\mP\!\in\!\pdis{\cha}\) and \(\mQ\!\in\!\pmea{\outA}\),
\eqref{eq:def:divergence} implies
\(\RD{\rno}{\mP \mtimes \cha}{\mP\otimes\mQ}\!\leq\!\max_{\mW\in \supp{\mP}} \RD{\rno}{\mW}{\mQ}\). 
Then considering \(\mP\)'s satisfying \(\mP(\mW)=1\) for a \(\mW\) in \(\cha\) 
we get
\begin{align}
\label{eq:minimax-1}
\sup\nolimits_{\mW \in \cha}  \RD{\rno}{\mW}{\mQ}
&=\sup\nolimits_{\mP \in \pdis{\cha}}  \RD{\rno}{\mP \mtimes \cha}{\mP\otimes\mQ}
\end{align}
for all \(\mQ\in\pmea{\outA}\).
Note that 
\eqref{eq:thm:minimax} implies \eqref{eq:thm:minimaxradius} and 
\eqref{eq:thm:minimaxcenter} implies \eqref{eq:thm:minimaxradiuscenter} because of \eqref{eq:minimax-1}.
Furthermore,
\eqref{eq:thm:minimax:capacity} is nothing but \eqref{eq:capacity}
and expression on the right hand side of  \eqref{eq:thm:minimax:capacity}
is bounded from above by the expression in \eqref{eq:thm:minimax}
as a result of max-min inequality.
Thus when \(\RC{\rno}{\cha}\) is infinite, \eqref{eq:thm:minimax} holds trivially.
When \(\RC{\rno}{\cha}\) is finite, the converse of max-min inequality, and hence \eqref{eq:thm:minimax}, 
follows from \eqref{eq:thm:minimaxcenter}.
Thus, we can assume \(\RC{\rno}{\cha}\) to be finite and prove the claims about \(\qmn{\rno,\cha}\)
in order to prove the theorem.
\begin{enumerate}[(i)]
\item {\it If \(\RC{\rno}{\cha}\!<\!\infty\) and \(\lim\nolimits_{\ind \to \infty}\!\RMI{\rno}{\pmn{\ind}}{\cha}\!=\!\RC{\rno}{\cha}\),
then \(\{\qmn{\rno,\pmn{\ind}}\}_{\ind\in\integers{+}}\) is a Cauchy sequence in \(\pmea{\outA}\) for the total variation metric:}
For any sequence \(\{\pmn{\ind}\}_{\ind\in \integers{+}}\) of members of \(\pdis{\cha}\) satisfying 
\(\lim\nolimits_{\ind \to \infty} \RMI{\rno}{\pmn{\ind}}{\cha}=\RC{\rno}{\cha}\), let \(\{\cha_{\ind}\}_{\ind\in\integers{+}}\) 
be  a nested sequence of finite subsets of \(\cha\) defined as follows, 
\begin{align}
\notag
\cha_{\ind} 
&\DEF \cup_{\jnd=1}^{\ind} \supp{\pmn{\jnd}}. 
\end{align} 
Then for any \(\ind\in\integers{+}\), there exists a unique \(\qmn{\rno,\cha_{\ind}}\) satisfying \eqref{eq:lem:capacityFLB} by 
Lemma \ref{lem:capacityFLB}.
Furthermore, \(\pdis{\cha_{\jnd}}\subset\pdis{\cha_{\ind}}\) for any \(\ind,\jnd\in\integers{+}\) such that \(\jnd\leq \ind\).
In order to bound
\(\lon{\qmn{\rno,\pmn{\jnd}}-\qmn{\rno,\pmn{\ind}}}\) for positive integers \(\jnd<\ind\), we use the triangle
inequality for \(\qmn{\rno,\pmn{\jnd}}\), \(\qmn{\rno,\pmn{\ind}}\), and \(\qmn{\rno,\cha_{\ind}}\):
\begin{align}
\label{eq:minimax-2}
\hspace{-.3cm}
\lon{\qmn{\rno,\pmn{\jnd}}\!-\!\qmn{\rno,\pmn{\ind}}}
\!\leq\!\lon{\qmn{\rno,\pmn{\jnd}}\!-\!\qmn{\rno,\cha_{\ind}}}
\!+\!\lon{\qmn{\rno,\pmn{\ind}}\!-\!\qmn{\rno,\cha_{\ind}}}.
\end{align} 
Let us proceed with bounding \(\lon{\qmn{\rno,\pmn{\jnd}}-\qmn{\rno,\cha_{\ind}}}\).
\begin{align}
\notag
\lon{\qmn{\rno,\pmn{\jnd}}-\qmn{\rno,\cha_{\ind}}}^2
&\mathop{\leq}^{(a)} \tfrac{2}{\rno\wedge 1}
\RD{\rno}{\qmn{\rno,\pmn{\jnd}}}{\qmn{\rno,\cha_{\ind}}}
\\
\notag
&\mathop{\leq}^{(b)} \tfrac{2}{\rno\wedge 1} 
\left[\RC{\rno}{\cha_{\ind}}-\RMI{\rno}{\pmn{\jnd}}{\cha_{\ind}}\right] 
\\
\notag
&\mathop{\leq}^{(c)} \tfrac{2}{\rno\wedge 1}
\left[\RC{\rno}{\cha}-\RMI{\rno}{\pmn{\jnd}}{\cha}\right].
\end{align}
where 
\((a)\) follows from Lemma \ref{lem:divergence-pinsker},
\((b)\) follows from Lemma \ref{lem:capacityFLB} because \(\widetilde{\mP}_{\jnd}\in\pdis{\cha_{\ind}}\),
and \((c)\)   follows from the identities \(\RMI{\rno}{\pmn{\jnd}}{\cha_{\ind}}=\RMI{\rno}{\pmn{\jnd}}{\cha}\) 
and \(\RC{\rno}{\cha_{\ind}}\leq\RC{\rno}{\cha}\).
We can obtain a similar bound on \(\lon{\qmn{\rno,\pmn{\ind}}-\qmn{\rno,\cha_{\ind}}}^2\).
Then \(\{\qmn{\rno,\pmn{\ind}}\}\) is a Cauchy sequence by  
\eqref{eq:minimax-2}.

\item {\it If \(\RC{\rno}{\cha}<\infty\), then \(\exists!~\qmn{\rno,\cha}\)
	in \(\pmea{\outA}\)
	satisfying
	\(\lim\nolimits_{\ind\to \infty}\lon{\qmn{\rno,\cha}-\qmn{\rno,\pmn{\ind}}}=0\) for all
	 \(\{\pmn{\ind}\}_{\ind\in\integers{+}}\) 
	satisfying 
\(\lim\nolimits_{\ind \to \infty} \RMI{\rno}{\pmn{\ind}}{\cha}=\RC{\rno}{\cha}\):}
Note that \(\smea{\outA}\) is a complete metric space for the total variation metric, i.e. every Cauchy sequence has 
a unique limit point in \(\smea{\outA}\),
because \(\smea{\outA}\) is a Banach space for the total variation topology \cite[Thm. 4.6.1]{bogachev}.
Then \(\{\qmn{\rno,\pmn{\ind}}\}_{\ind\in\integers{+}}\) has a unique limit point \(\qmn{\rno,\cha}\) in \(\smea{\outA}\). 
Since \(\pmea{\outA}\) is a closed set for the total variation topology and 
\(\qmn{\rno,\pmn{\ind}}\in \pmea{\outA}\) for all \(\ind\in\integers{+}\),
the limit point \(\qmn{\rno,\cha}\) is in \(\pmea{\outA}\) by \cite[Thm. 2.1.3]{munkres}.

We have established the existence of a unique limit point for any 
\(\{\pmn{\ind}\}_{\ind\in\integers{+}}\) satisfying 
\(\lim\nolimits_{\ind \to \infty}\RMI{\rno}{\pmn{\ind}}{\cha}=\RC{\rno}{\cha}\).
However, we have not ruled out the possibility of distinct limit points for different sequences 
satisfying the constraint. 
Let \(\{\pmn{\ind}\}_{\ind\in\integers{+}}\) and \(\{\tilde{\mP}_{\ind}\}_{\ind\in\integers{+}}\) 
be two sequences satisfying  
\(\lim\nolimits_{\ind \to \infty}\RMI{\rno}{\pmn{\ind}}{\cha}
=\lim\nolimits_{\ind \to \infty}\RMI{\rno}{\tilde{\mP}_{\ind}}{\cha}=\RC{\rno}{\cha}\),
with limit points \(\qmn{\rno,\cha}\) and \(\tilde{\mQ}_{\rno,\cha}\).
Let \(\{\hat{\mP}_{\ind}\}_{\ind\in \integers{+}}\) be a sequence whose elements for the odd indices  are the 
elements of \(\{\pmn{\ind}\}_{\ind\in \integers{+}}\) and whose elements for the even indices are the elements of 
\(\{\tilde{\mP}_{\ind}\}_{\ind\in \integers{+}}\).
Then \(\lim\nolimits_{\ind \to \infty} \RMI{\rno}{\hat{\mP}_{\ind}}{\cha}=\RC{\rno}{\cha}\); consequently the 
sequence \(\{\qmn{\rno,\hat{\mP}_{\ind}}\}_{\ind\in \integers{+}}\)  is Cauchy. Thus 
\(\{\qmn{\rno,\hat{\mP}_{\ind}}\}_{\ind\in \integers{+}}\)
and all of its subsequences has the same unique limit point  \(\hat{\mQ}_{\rno,\cha}\).
Then \(\qmn{\rno,\cha}=\hat{\mQ}_{\rno,\cha}=\tilde{\mQ}_{\rno,\cha}\).

\item {\it \(\qmn{\rno,\cha}\) satisfies the equality given in \eqref{eq:thm:minimaxcenter}:}
For any \(\mP\) in \(\pdis{\cha}\), let us consider a sequence \(\{\pmn{\ind}\}_{\ind\in\integers{+}}\) satisfying 
both \(\pmn{1}=\mP\) and \(\lim\nolimits_{\ind \to \infty} \RMI{\rno}{\pmn{\ind}}{\cha}=\RC{\rno}{\cha}\).
Then \(\mP\in \pdis{\cha_{\ind}}\) for all \(\ind\in\integers{+}\). 
Then using the inequality given in \eqref{eq:lem:capacityFLB} of Lemma \ref{lem:capacityFLB} together with
\eqref{eq:sibson} we get
\begin{align}
\label{eq:minimax-3}
\RD{\rno}{\mP \mtimes \cha}{\mP \otimes\qmn{\rno,\cha_{\ind}}}
&\leq  \RC{\rno}{\cha_{\ind}}
&
&\forall \ind.
\end{align}
Since \(\cha_{\ind}\) is a finite set, \(\exists\widetilde{\mP}_{\ind}\in \pdis{\cha_{\ind}}\) 
satisfying \(\RMI{\rno}{\widetilde{\mP}_{\ind}}{\cha_{\ind}}=\RC{\rno}{\cha_{\ind}}\)
and \(\qmn{\rno,\widetilde{\mP}_{\ind}}=\qmn{\rno,\cha_{\ind}}\)  by Lemma \ref{lem:capacityFLB}.
Then \(\RMI{\rno}{\widetilde{\mP}_{\ind}}{\cha_{\ind}}\geq \RMI{\rno}{\pmn{\ind}}{\cha_{\ind}}\)
because \(\pmn{\ind}\in \pdis{\cha_{\ind}}\) by construction.
Consequently \(\lim\nolimits_{\ind \to \infty} \RMI{\rno}{\widetilde{\mP}_{\ind}}{\cha}=\RC{\rno}{\cha}\). 
We have already established that for such a sequence
\(\qmn{\rno,\widetilde{\mP}_{\ind}}\rightarrow \qmn{\rno,\cha}\) in the total variation topology, 
and hence in the topology of setwise convergence. Then
the lower semicontinuity  of the \renyi divergence, i.e. Lemma \ref{lem:divergence:lsc}, 
the identity \(\RC{\rno}{\cha_{\ind}}\leq\RC{\rno}{\cha}\),
and  \eqref{eq:minimax-3} imply
\begin{align}
\notag
\RD{\rno}{\mP \mtimes \cha}{\mP \otimes\qmn{\rno,\cha}}\leq  \RC{\rno}{\cha}.
\end{align}
Thus using \eqref{eq:lem:information:defB} we get 
\begin{align}
\notag
\RMI{\rno}{\mP}{\cha}\leq \RD{\rno}{\mP \mtimes \cha}{\mP \otimes \qmn{\rno,\cha}}
&\leq \RC{\rno}{\cha}
&
&\forall \mP\in\pdis{\cha}.
\end{align}
Then  \eqref{eq:thm:minimaxcenter} follows the definition of \(\RC{\rno}{\cha}\).
\end{enumerate} 
\end{proof}

Theorem \ref{thm:minimax} is not just a minimax theorem, the assertions about the \renyi center are crucial.
But those assertions can be derived separately, if need be.
Leaving them aside, we discuss in the rest of this subsection when  \eqref{eq:thm:minimax} can be proved using 
Sion's minimax theorem \cite{komiya88}, \cite{sion58}.

Note that \(\pdis{\cha}\) is compact iff \(\cha\) is a finite set 
and \(\pmea{\outA}\) is compact iff \(\outA\) is a finite set. 
Consequently, when either \(\cha\) or \(\outA\) is finite,  \eqref{eq:thm:minimax} is an immediate 
consequence\footnote{Immediate after establishing  that \(\RD{\rno}{\mP \mtimes \cha}{\mP \otimes \mQ}\) is  
upper semicontinuous and quasi-concave in \(\mP\). 
The lower semicontinuity and the quasi-convexity of \(\RD{\rno}{\mP \mtimes \cha}{\mP \otimes \mQ}\) 
in \(\mQ\) follow from Lemmas \ref{lem:divergence-convexity} and \ref{lem:divergence:lsc}.} 
of Sion's minimax theorem \cite{komiya88}, \cite{sion58}. 
When \(\cha\) and \(\outA\) are both infinite sets, 
however, neither \(\pdis{\cha}\) nor \(\pmea{\outA}\) is compact ---for the total variation topology---
and we cannot directly apply Sion's minimax theorem.
Yet, it is possible to recover partial results using the concept of relative compactness.
Recall that a set of points in a topological space is called relatively compact if it has
a compact closure.

First note that as a result of Lemma \ref{lem:information:def} 
\begin{align}
\label{eq:sion}
\RMI{\rno}{\mP}{\cha}
&=\inf\nolimits_{\mQ \in \clos{\domtr{\rno,\cha}}}  \RD{\rno}{\mP \mtimes \cha}{\mP \otimes \mQ}
\end{align} 
for all \(\mP\)'s in \(\pdis{\cha}\) and \(\rno\)'s in \(\reals{+}\)
where  \(\domtr{\rno,\cha}\) is the convex hull of the set of all order 
\(\rno\) \renyi means:
\begin{align}
\notag
\domtr{\rno,\cha}
&\DEF \conv{\{\qmn{\rno,\mP}:\mP\in\pdis{\cha}\}}.
\end{align}
If \(\domtr{\rno,\cha}\) is  relatively compact in the topology of setwise convergence, 
Sion's minimax theorem imply that
\begin{align}
\notag
&\hspace{-.6cm}\sup\nolimits_{\mP\in \pdis{\cha}}
\inf\nolimits_{\mQ \in \clos{\domtr{\rno,\cha}}}  \RD{\rno}{\mP \mtimes \cha}{\mP \otimes \mQ}
\\
&\qquad=\inf\nolimits_{\mQ \in \clos{\domtr{\rno,\cha}}}
\sup\nolimits_{\mP\in \pdis{\cha}}  \RD{\rno}{\mP \mtimes \cha}{\mP \otimes \mQ}.
\end{align} 
We can replace \(\clos{\domtr{\rno,\cha}}\) by \(\pmea{\outA}\) in the expression on
the left hand side without changing its value as a result of \eqref{eq:sion}.
However, that operation can decrease the value of the right hand side because 
\(\clos{\domtr{\rno,\cha}}\subset \pmea{\outA}\). Thus we get,
\begin{align}
\notag
&\hspace{-.6cm}
\sup\nolimits_{\mP \in \pdis{\cha}} \inf\nolimits_{\mQ \in \pmea{\outA}} \RD{\rno}{\mP \mtimes \cha}{\mP \otimes \mQ}
\\
\notag
&\qquad\geq
\inf\nolimits_{\mQ \in \pmea{\outA}} \sup\nolimits_{\mP \in \pdis{\cha}} \RD{\rno}{\mP \mtimes \cha}{\mP \otimes \mQ}.
\end{align}
The reverse inequality is the max-min inequality, which is always valid. Thus \eqref{eq:thm:minimax} holds.

A set of finite measures \(\cha\) is relatively compact 
in the topology of setwise convergence
iff there exists a \(\rfm\in\pmea{\outA}\) such that \(\cha\UAC\rfm\) by a version of 
the Dunford-Pettis theorem \cite[4.7.25]{bogachev}. 
Using de la Vall\'{e}e Poussin's characterization of the uniform integrability \cite[Thm. 4.5.9]{bogachev} 
and monotonicity of the order \(\rno\) mean measure \(\mmn{\rno,\mP}\) in the order,
i.e. Lemma \ref{lem:powermeanO}-(\ref{powermeanO-b}), we can obtain sufficient conditions for the relative 
compactness of \(\domtr{\rno,\cha}\) in the topology of setwise convergence for any \(\rno\in\reals{+}\). 
As a result we get the following partial result:
\begin{lemma}\label{lem:minimax-alternative}
	Let \(\cha\) be subset of \(\pmea{\outA}\).
	\begin{enumerate}[(i)]
		\item If \(\exists \rfm\in\pmea{\outA}\) such that \(\cha \UAC \rfm\) 
		and \(\RR{\rnt}{\cha}<\infty\) for an \(\rnt\in(0,1)\),
		then \eqref{eq:thm:minimax} holds \(\forall\rno\in(0,1)\).
		\item  If \(\RR{\rnt}{\cha}\!<\!\infty\)  for an \(\rnt\!\in\![1,\infty]\),
		then \eqref{eq:thm:minimax} holds \(\forall\rno\!\in\!(0,\rnt]\).
	\end{enumerate}
\end{lemma}

\subsection[\(\qmn{\rno,\cha}\) as a Function of \(\rno\)]{The \renyi Center as a Function of the Order}\label{sec:ervenharremoes}
\(\RR{\rno}{\cha}\) is defined as the greatest lower bound of \(\RRR{\rno}{\cha}{\mQ}\). 
Then Theorem \ref{thm:minimax} implies, by  establishing  \(\RC{\rno}{\cha}=\RR{\rno}{\cha}\), 
that
\begin{align}
\notag
\RRR{\rno}{\cha}{\mQ}
&\geq \RC{\rno}{\cha}
&
&\forall \mQ\in\pmea{\outA}.
\end{align}
Van Erven and \harremoes have conjectured that a better lower bound on \(\RRR{\rno}{\cha}{\mQ}\) 
should hold, \cite[Conjecture 1]{ervenH14}.
Van Erven and \harremoes proved their claim for \(\rno=\infty\) case assuming 
that \(\outS\) is
countable,
\cite[Thm. 37]{ervenH14}. 
Lemma \ref{lem:EHB} establishes the van Erven-\harremoes bound for any positive order \(\rno\) and 
\(\cha\) satisfying \(\RC{\rno}{\cha}<\infty\), using Theorem \ref{thm:minimax}. 
A constrained generalization, i.e. Lemma \ref{lem:CEHB},
can be found in Appendix \ref{sec:constrainedcapacity}.
\begin{lemma}\label{lem:EHB} 
For any \(\rno \in  (0,\infty]\),  
\(\cha\subset\pmea{\outA}\) satisfying \(\RC{\rno}{\cha}<\infty\),
and \(\mQ \in \pmea{\outA}\), 
\begin{align}
\label{eq:lem:EHB}
\sup\nolimits_{\mW \in \cha} \RD{\rno}{\mW}{\mQ}
&\geq  \RC{\rno}{\cha}+\RD{\rno}{\qmn{\rno,\cha}}{\mQ}.
\end{align}
\end{lemma}
Lemma \ref{lem:EHB} quantifies how loose \(\RRR{\rno}{\cha}{\mQ}\)
---defined in \eqref{eq:def:relativeradius}--- 
is as an upper bound to \(\RC{\rno}{\cha}\),
as surmised by van Erven and \harremoes in \cite{ervenH14}.

\begin{proof}[Proof of Lemma \ref{lem:EHB}]
As a result of  \eqref{eq:sibson}  and \eqref{eq:lem:information:defA} we have,
\begin{align}
\notag
\sup\limits_{\tilde{\mP}\in\pdis{\cha}}  \RD{\rno}{\tilde{\mP} \mtimes \cha}{\tilde{\mP}  \otimes \mQ} 
&\geq \RD{\rno}{\mP \mtimes \cha}{\mP  \otimes \mQ} 
\\
\label{eq:radiusconjecture-B}
&=\RMI{\rno}{\mP}{\cha}+\RD{\rno}{\qmn{\rno,\mP}}{\mQ}
\end{align}
for all \(\mP \in \pdis{\cha}\).
Let \(\{\pmn{\ind}\}_{\ind\in\integers{+}}\) be a sequence of elements of \(\pdis{\cha}\) such that
\(\lim\nolimits_{\ind\to\infty}\RMI{\rno}{\pmn{\ind}}{\cha}=\RC{\rno}{\cha}\).
Then the sequence \(\{\qmn{\rno,\pmn{\ind}}\}_{\ind\in\integers{+}}\) is a Cauchy sequence with the unique 
limit point \(\qmn{\rno,\cha}\) by Theorem \ref{thm:minimax}.
Since \(\{\qmn{\rno,\pmn{\ind}}\}\to \qmn{\rno,\cha}\) in total variation topology, 
same convergence holds in the topology of setwise convergence 
because  every open neighborhood in the latter includes 
an open neighborhood in the former by the definitions of these topologies.
On the other hand, the order \(\rno\) \renyi divergence is lower 
semicontinuous for the topology of setwise convergence by Lemma \ref{lem:divergence:lsc}. Thus we have 
\begin{align}
\notag
\liminf\limits_{\ind \to \infty} 
\left[\RMI{\rno}{\pmn{\ind}}{\cha}+\RD{\rno}{\qmn{\rno,\pmn{\ind}}}{\mQ}\right]
&\geq \RC{\rno}{\cha}+\RD{\rno}{\qmn{\rno,\cha}}{\mQ}.
\end{align}
Then \eqref{eq:lem:EHB} follows from \eqref{eq:minimax-1} and \eqref{eq:radiusconjecture-B}.
\end{proof}

The van Erven-\harremoes bound allows us to use the continuity of \(\RC{\rno}{\cha}\) in \(\rno\) 
and  Pinsker's inequality to establish the continuity of \(\qmn{\rno,\cha}\) in 
\(\rno\) for the total variation topology.
\begin{lemma}\label{lem:centercontinuity} 
For any \(\cha \subset \pmea{\outA}\) and \(\rnt \in  (0,\infty]\)
such that \(\RC{\rnt}{\cha}<\infty\), 
\begin{align}
\label{eq:lem:centercontinuity}
\RC{\rnf}{\cha}-\RC{\rno}{\cha}
&\geq \RD{\rno}{\qmn{\rno,\cha}}{\qmn{\rnf,\cha}} 
\end{align}
for all \(\rno\) and \(\rnf\) satisfying
\(0<\rno<\rnf\leq\rnt\).
Furthermore, \(\qmn{\rno,\cha}\) is a continuous function of \(\rno\) on \((0,\rnt]\) 
for the total variation topology on \(\pmea{\outA}\).
\end{lemma}

The continuity of the \renyi center as a function of the order is important because 
it allows us to the interpret the \renyi centers as a transition probability from 
the interval on which the \renyi capacity is finite to \((\outS,\outA)\)
and apply Augustin's method, see \cite[\S\ref*{B-sec:averaging}]{nakiboglu19B} for a more 
detailed discussion.

\begin{proof}[Proof of Lemma \ref{lem:centercontinuity}]
For \(\mQ\!=\!\qmn{\rnf,\cha}\), Lemma \ref{lem:EHB} implies
\begin{align}
\label{eq:centercontinuity-A}
\sup\nolimits_{\mW\in\cha}\RD{\rno}{\mW}{\qmn{\rnf,\cha}}
&\geq \RC{\rno}{\cha}+\RD{\rno}{\qmn{\rno,\cha}}{\qmn{\rnf,\cha}}.
\end{align}
Since \(\RD{\rno}{\mW}{\qmn{\rnf,\cha}}\) is nondecreasing in \(\rno\)
by Lemma \ref{lem:divergence-order},
\begin{align}
\label{eq:centercontinuity-B}
\RD{\rnf}{\mW}{\qmn{\rnf,\cha}}
&\geq \RD{\rno}{\mW}{\qmn{\rnf,\cha}} 
&
&\forall \mW\in \cha,~\rnf\in[\rno,\rnt].
\end{align}
On the other hand by  \eqref{eq:thm:minimaxradiuscenter} of Theorem \ref{thm:minimax} we have
\begin{align}
\label{eq:centercontinuity-C}
\RC{\rnf}{\cha}
&= \sup\nolimits_{\mW\in \cha} \RD{\rnf}{\mW}{\qmn{\rnf,\cha}}
&
&\forall \rnf\in(0,\rnt].
\end{align}
\eqref{eq:lem:centercontinuity} follows from \eqref{eq:centercontinuity-A}, \eqref{eq:centercontinuity-B},
 and \eqref{eq:centercontinuity-C}.

Using  Lemma \ref{lem:divergence-pinsker} and  \eqref{eq:lem:centercontinuity} 
we get\footnote{For proving a similar continuity result in \cite{augustin78}, 
	instead of \eqref{eq:lem:centercontinuity}, 
	Augustin uses the inequality given in the following
	---which can be proved using 
	\eqref{eq:def:information}, \eqref{eq:def:divergence}, Lemma \ref{lem:powermeanO}-(\ref{powermeanO-b}),
	and Theorem \ref{thm:minimax}:
	If either  \(\rno\in[\rnf,\rnt]\) and \(\rnf<1\)
	or \(\rno\in(0,\rnf]\) and \(\rnf>1\) then
	\begin{align}
	\notag
	\RD{\rnf}{\qmn{\rno,\mP}}{\qmn{\rnf,\cha}}
	&\leq \RC{\rnf}{\cha}-\tfrac{\rnf}{\rnf-1}\ln \lon{\mmn{\rno,\mP}}  
	&
	&\forall\mP\in\pdis{\cha}.
	\end{align}
}
\begin{align}
\label{eq:centercontinuity-D}
\sqrt{\tfrac{2}{\rnf \wedge 1}(\RC{\rnf}{\cha}-\RC{\rno}{\cha})}
&\geq \lon{\qmn{\rnf,\cha}-\qmn{\rno,\cha}}. 
\end{align}  
Then, for the total  variation topology on \(\pmea{\outA}\), 
the continuity of \(\qmn{\rno,\cha}\) in \(\rno\) follows 
from the continuity \(\RC{\rno}{\cha}\) in \(\rno\) on \((0,\rnt]\), i.e. Lemma \ref{lem:capacityO}-(\ref{capacityO-continuity}).
\end{proof}

Lemma \ref{lem:centercontinuity} establishes the continuity of the \renyi center 
in the order for the total variation topology. 
We suspect a much stronger statement is true. 

\begin{conjecture}\label{con:renyicentermonotonicity}
For any \(\cha \subset \pmea{\outA}\) and \(\rnt \in  (0,\infty]\). 
satisfying \(\RC{\rnt}{\cha}<\infty\), 
\begin{align}
\label{eq:con:renyicentermonotonicity}
\mmn{\rnf,\cha} 
&\leq \mmn{\rnt,\cha}
&
&\forall \rnf\in (0,\rnt]
\end{align}
where  \(\mmn{\rnf,\cha}\DEF e^{\frac{\rnf-1}{\rnf} \RC{\rnf}{\cha}}\qmn{\rnf,\cha}\) for all \(\rnf\in (0,\rnt]\).
\end{conjecture}

For any \(\cha\) using the continuity of the \renyi center in the order, one can prove that 
there exists a \(\rfm\) in \(\pmea{\outA}\) 
such that \(\{\qmn{\rno,\cha}:\RC{\rno}{\cha}<\infty\}\AC\rfm\).
However, the continuity of the \renyi center as a function of the order for the total variation topology
does not imply the continuity of corresponding Radon-Nikodym derivative 
\(\der{\qmn{\rno,\cha}}{\rfm}\) as a function of \(\rno\) for \(\rfm\)-almost 
everywhere.
If Conjecture \ref{con:renyicentermonotonicity} is correct, then  it will imply 
the continuity of Radon-Nikodym derivative \(\der{\qmn{\rno,\cha}}{\rfm}\) as a function of 
\(\rno\) for \(\rfm\)-almost everywhere.

\begin{remark}\label{rem:continuityintvt} 
The continuity in the total variation topology does not 
	imply the continuity of the corresponding Radon-Nikodym derivative:
	Let the output space be the real numbers between \(-1\)  and \(2\), 
	and the Radon-Nikodym derivative of \(\qmn{\rno}\) with respect to the Lebesgue measure \(\lbm\) be
	\begin{align}
	\notag
	\der{\qmn{\rno}}{\lbm}
	&=\IND{
		\sin(\frac{1}{\tin-\rno})\leq
		\dout\leq \sin(\frac{1}{\tin-\rno})+\abs{\rno-\tin}}
	+\IND{0\leq\dout\leq 1}(1-\abs{\rno-\tin})
	\end{align}
	for a \(\tin\in(0,1)\).
	Evidently \(\lim_{\rnt\to\rno}\lon{\qmn{\rno}-\qmn{\rnt}}=0\) for all \(\rno\) in \((0,1)\). 
	But \(\der{\qmn{\rno}}{\lbm }\) is not continuous in \(\rno\) for any \(\dout\in (0,1)\) 
	at \(\tin\).
\end{remark}

\subsection[Unions, Cartesian Products, and More]{The Unions, Cartesian Products, Closures,  and More}\label{sec:miscellaneous}
This subsection is composed of applications of 
Theorem \ref{thm:minimax} and Lemma \ref{lem:EHB}.
Lemma \ref{lem:capacityUnion}, in the following, bounds from below and from 
above the \renyi capacity of a union of sets in terms of the \renyi capacities 
of the sets in the union.
Lemma \ref{lem:capacityProduct} establishes that the \renyi capacity of a 
Cartesian product is equal to the sum of the \renyi capacities of its components.
Lemma \ref{lem:capacityEps} shows that for any positive \(\epsilon\) the order 
\(\rno\) \renyi capacity of the set of \(\mW\)'s in \(\cha\) satisfying 
\(\RD{\rno}{\mW}{\qmn{\rno,\cha}}\geq \RC{\rno}{\cha}-\epsilon\) is equal 
to \(\RC{\rno}{\cha}\).
Lemma \ref{lem:capacityEXT} establishes the invariance of \(\RC{\rno}{\cha}\) 
under the closure and convexification operations on \(\cha\) and characterizes 
the relative compactness of \(\cha\) in terms of its \renyi capacity.
Proofs of these lemmas are presented in Appendix \ref{sec:deferred-proofs}.

\begin{lemma}\label{lem:capacityUnion}
For any \(\rno\in(0,\infty]\) and \(\cha\subset\pmea{\outA}\)
satisfying \(\cha=\cup_{\ind\in\tinS} \cha_{\ind}\) for some
\(\cha_{\ind}\subset\pmea{\outA}\) with finite \(\RC{\rno}{\cha_{\ind}}\)'s, 
\begin{align}
\label{eq:lem:capacityUnion}
\sup\nolimits_{\ind\in\tinS} \RC{\rno}{\cha_{\ind}}\leq 
\RC{\rno}{\cha}
&\leq \ln \sum\nolimits_{\ind\in\tinS} e^{\RC{\rno}{\cha_{\ind}}}.
\end{align}
Furthermore, 
\begin{itemize}
\item \(\RC{\rno}{\cha_{\ind}}=\RC{\rno}{\cha}\) iff 
\(\RRR{\rno}{\cha}{\qmn{\rno,\cha_{\ind}}}\leq \RC{\rno}{\cha_{\ind}}\).
\item If \(\RC{\rno}{\cha_{\ind}}=\RC{\rno}{\cha}\), then
\(\qmn{\rno,\cha}=\qmn{\rno,\cha_{\ind}}\). 
\item \(\RC{\rno}{\cha}=\ln \sum\nolimits_{\ind\in \tinS} e^{\RC{\rno}{\cha_{\ind}}}\) 
and \(\RC{\rno}{\cha}\) is finite
iff \(\tinS\) is finite and
\(\qmn{\rno,\cha_{\ind}}\perp\qmn{\rno,\cha_{\jnd}}\) for all \(\ind \neq \jnd\) in \(\tinS\).
\item If \(\tinS\) is finite and
\(\qmn{\rno,\cha_{\ind}}\perp\qmn{\rno,\cha_{\jnd}}\) for all \(\ind \neq \jnd\) in \(\tinS\),
then
\(\qmn{\rno,\cha}=(\sum\nolimits_{\jnd\in\tinS} e^{\RC{\rno}{\cha_{\jnd}}})^{-1}\sum\nolimits_{\ind\in\tinS} 
e^{\RC{\rno}{\cha_{\ind}}} \qmn{\rno,\cha_{\ind}}\).
\end{itemize}
\end{lemma}

One might think that \(\qmn{\rno,\cha_{\ind}}\perp\qmn{\rno,\cha_{\jnd}}\) iff \(\cha_{\ind}\perp\cha_{\jnd}\).
This, however, is true only for \(\rno\)'s in \([1,\infty]\).
For \(\rno\)'s in \((0,1)\), \(\cha_{\ind}\perp\cha_{\jnd}\) is a sufficient condition for
\(\qmn{\rno,\cha_{\ind}}\perp\qmn{\rno,\cha_{\jnd}}\), but it is not a necessary condition,
see Examples \ref{eg:singular-finite} and \ref{eg:singular-countable}.
Augustin \cite{augustin78} is the first one to point out this
subtlety and to present necessary and sufficient conditions for 
\(\RC{\rno}{\cha}=\ln \sum\nolimits_{\ind\in\tinS} e^{\RC{\rno}{\cha_{\ind}}}\),
to the best of our knowledge.
Bounds given in \eqref{eq:lem:capacityUnion} is well known \cite[p. 535, ex. 5.17]{gallager}. 
We use the van Erven-\harremoes bound in order to characterize the necessary and sufficient conditions for 
\(\sup\nolimits_{\ind\in\tinS}\RC{\rno}{\cha_{\ind}}=\RC{\rno}{\cha}\)
and \(\RC{\rno}{\cha}=\ln \sum\nolimits_{\ind\in\tinS} e^{\RC{\rno}{\cha_{\ind}}}\).\label{bookmarkB} 

Let \(\tinS\) be a finite set. For each \(\tin\in\tinS\), let \((\outS_{\tin},\outA_{\tin})\) be a measurable space
and \(\mW_{\tin}\) be a probability measure on \((\outS_{\tin},\outA_{\tin})\). 
Then there exists a unique product measure \(\bigotimes_{{\tin\in\tinS}}\mW_{\tin}\) on the measurable space 
\((\bigtimes_{{\tin\in\tinS}}\outS_{\tin},\bigotimes_{{\tin\in\tinS}}\outA_{\tin})\)
by \cite[Thm. 8.2.2]{dudley}.\footnote{The existence of a unique product measure is guaranteed 
for any finite collection of \(\sigma\)-finite measures by \cite[Thm. 4.4.4]{dudley} 
and 
for any countable collection of probability measures by \cite[Thm. 8.2.2]{dudley}.}
Let \(\cha_{\tin}\) be a subset of \(\pmea{\outA_{\tin}}\) for each \(\tin\in\tinS\).
Then using the existence of a unique product measure we can map the Cartesian product of the 
sets \(\cha_{\tin}\) uniquely to a subset of \(\pmea{\bigotimes_{{\tin\in\tinS}}\outA_{\tin}}\),
called the product of \(\cha_{\tin}\)'s.
Then the \renyi capacity of the product is equal to the sum of the \renyi capacities of its components
and the \renyi center of the product, whenever it exists, is equal to the product of the \renyi centers 
of its components. Lemma \ref{lem:capacityProduct} asserts these observations formally.

\begin{lemma}\label{lem:capacityProduct}
For any finite index set \(\tinS\), if
\(\outS=\bigtimes_{{\tin\in\tinS}}\outS_{\tin}\), \(\outA=\bigotimes_{{\tin\in\tinS}}\outA_{\tin}\),
and \(\cha=\left\{\mW: \mW=\bigotimes_{{\tin\in\tinS}}\mW_{\tin}:  \mW_{\tin}\in \cha_{\tin}\right\}\)
for some  \(\cha_{\tin}\subset \pmea{\outA_{\tin}}\), then
\begin{align}
\label{eq:lem:capacityProduct}
\RC{\rno}{\cha}
&=\sum\nolimits_{\tin\in\tinS}\RC{\rno}{\cha_{\tin}}
&
&
\forall \rno\in (0,\infty].
\end{align}
Furthermore, if \(\RC{\rno}{\cha}<\infty\), then \(\qmn{\rno,\cha}=\bigotimes_{\tin\in\tinS}\qmn{\rno,\cha_{\tin}}\).
\end{lemma}

Quite frequently,
the information transmission problems are analyzed on the product \(\cha\)'s.
Lemma \ref{lem:capacityProduct} is instrumental when that is the case. 
The derivation of the sphere packing bound presented in 
\cite[\S \ref*{B-sec:product-outerbound}]{nakiboglu19B} is a case in point.
The additivity of the \renyi capacity for products was first reported by Gallager
---in a slightly different form and for finite \(\cha\) and \(\outS\) case--- 
in his seminal paper \cite[Thm. 5]{gallager65}, see also
\cite[pp. 149-150, (5.6.59)]{gallager}.
Later, Augustin proved \cite[Lemma 26.7a]{augustin78}, which implies Lemma \ref{lem:capacityProduct};
see \cite[Lemma 3.6]{augustin69} for finite \(\cha\) case.

One curious question is whether or not one can give a class of priors for which the lower bound given
in \eqref{eq:capacityLB} is not too loose. 
Lemma \ref{lem:capacityEps} answers this question in the affirmative.
\begin{lemma}\label{lem:capacityEps}
For any \(\rno \in  (0,\infty]\),
\(\cha \subset \pmea{\outA}\) with finite \(\RC{\rno}{\cha}\),
and \(\epsilon\geq 0\), let \(\cha_{\rno,\epsilon}\) be
\begin{align}
\label{eq:def:kcha}
\cha_{\rno,\epsilon} 
&\DEF\left\{\mW\in\cha:\RD{\rno}{\mW}{\qmn{\rno,\cha}}\geq\RC{\rno}{\cha}-\epsilon\right\}.
\end{align}
Then for any \(\epsilon>0\), we have  \(\RC{\rno}{\cha_{\rno,\epsilon} }=\RC{\rno}{\cha}\)   
and\footnote{For \(\rno=\infty\),  \eqref{eq:lem:capacityEps} is valid for a broader class of \(\mP\)'s in particular
for all \(\mP\)'s such that \((\sum_{\mW \in \cha_{\rno,\epsilon} } \mP(\mW))>0\).}
\begin{align}
\label{eq:lem:capacityEps}
0\leq\RC{\rno}{\cha}-\RMI{\rno}{\mP}{\cha}-\RD{\rno}{\qmn{\rno,\mP}}{\qmn{\rno,\cha}} 
&\leq \epsilon
\end{align}
for all \(\mP\) in \(\pdis{\cha_{\rno,\epsilon} }\).
Furthermore, if \(\cha\) is a finite set, then \(\RC{\rno}{\cha_{\rno,0}}=\RC{\rno}{\cha}\) 
and \eqref{eq:lem:capacityEps} holds for \(\epsilon=0\).  
\end{lemma}

The main conclusion of Lemma \ref{lem:capacityEps} is the equality 
\(\RC{\rno}{\cha_{\rno,\epsilon}}=\RC{\rno}{\cha}\) for positive \(\epsilon\)'s. 
This is expected for a general \(\cha\) and evident, even for \(\epsilon=0\) case,  
for a finite \(\cha\)  because of the existence of an optimal \(\mP\) in \(\pdis{\cha}\)
for finite \(\cha\)'s.
One might be tempted to assume the validity of the assertions for \(\epsilon=0\) case for 
arbitrary \(\cha\)'s.
This, however, is not true; see Example \ref{eg:erasure} for a \(\cha\) for which \(\RC{\rno}{\cha}>0\) 
and \(\RC{\rno}{\cha_{\rno,0} }=0\).
Thus finiteness of \(\cha\) is not a superficial hypothesis 
for extending the claims to \(\epsilon=0\) case.

In order to apply certain technical tools, we occasionally need a given set 
to be closed, convex, or compact. 
The observations presented in Lemma \ref{lem:capacityEXT}, given in the following,
can be helpful in such situations. 
For example, 
if we can prove a statement about \renyi capacity assuming \(\cha\) to be
convex, then we can assert that statement for non-convex \(\cha\)'s using 
Lemma \ref{lem:capacityEXT}-(\ref{capacityEXT-ch}).
Furthermore, in certain situations, calculating the \renyi capacity might 
be easier for the convex hull or the closure of \(\cha\) when compared 
to \(\cha\) itself, see Example \ref{eg:erasure}.
Lemma \ref{lem:capacityEXT}-(\ref{capacityEXT-ch},\ref{capacityEXT-cl}) is
helpful in such situations. 
Note that
Lemma \ref{lem:capacityEXT}-(\ref{capacityEXT-ch},\ref{capacityEXT-cl})
when considered together with Lemma \ref{lem:capacityUnion} imply 
the equality of the \renyi centers of \(\cha\), \(\conv{\cha}\), and \(\clos{\cha}\)
whenever one of them exists.

\begin{lemma}\label{lem:capacityEXT}
Let \(\cha\)  be a subset of \(\pmea{\outA}\). 
\begin{enumerate}[(a)]
\item\label{capacityEXT-ch}
 \(\RC{\rno}{\conv{\cha}}=\RC{\rno}{\cha}\) for all \(\rno\in(0,\infty]\) where \(\conv{\cha}\) is the convex hull
of \(\cha\) given by \(\conv{\cha}\DEF \{\mmn{1,\mP}:\mP\in \pdis{\cha}\}\).

\item\label{capacityEXT-cl}
\(\RC{\rno}{\clos{\cha}}=\RC{\rno}{\cha}\) for all \(\rno\in(0,\infty]\) where \(\clos{\cha}\) is the closure of 
\(\cha\) in the topology of setwise convergence or a stronger topology on \(\pmea{\outA}\).

\item\label{capacityEXT-compact-S}
If \(\RC{\rnt}{\cha}\!<\!\infty\) for an \(\rnt\!\geq\!1\), then 
\(\{\!\mmn{\rno,\mP}\!:\!\rno\!\in\![0,\rnt],\mP\!\in\!\pdis{\!\cha\!}\!\}\) 
is uniformly absolutely continuous with respect to \(\qmn{\rnt,\cha}\) and relatively compact in both 
the topology of setwise convergence and the weak topology.

\item\label{capacityEXT-compact-N} The following four statements are 
equivalent:\footnote{Augustin proves the equivalence of 
	\(\lim_{\rno \uparrow 1} \tfrac{1-\rno}{\rno}\RC{\rno}{\cha}=0\)
	and \(\exists\mean\in\pmea{\outA}\) such that \(\cha\UAC\mean\),
	using Gallager's inner bound \cite[Thm. 1]{gallager65} and a different 
	characterization of the relative compactness he derives in \cite{augustin78}.
	Our proof is measure theoretic and self-contained.}
\begin{enumerate}[(i)]
\item \(\lim_{\rno \uparrow 1} \tfrac{1-\rno}{\rno}\RC{\rno}{\cha}=0\).
\item \(\exists\mean\in\pmea{\outA}\) such that \(\cha\UAC\mean\).
\item \(\cha\) has compact closure in the topology of setwise convergence.
\item \(\cha\) has compact closure in the weak topology.
\end{enumerate}
\end{enumerate}
\end{lemma}

Each assertion of Lemma \ref{lem:capacityEXT} 
is proved using Theorem \ref{thm:minimax} together with some
other observations.
The invariance of \(\RC{\rno}{\cha}\) under the closure and 
the convexification operations on \(\cha\), presented in 
Lemma \ref{lem:capacityEXT}-(\ref{capacityEXT-ch},\ref{capacityEXT-cl}),
follow from the lower semicontinuity and 
the quasi-convexity of the \renyi divergence in its first  argument. 
Lemma \ref{lem:capacityEXT}-(\ref{capacityEXT-compact-S}) follows
from the monotonicity of \(\mmn{\rno,\mP}\) in \(\rno\) and 
de la Vall\'{e}e Poussin's characterization of the 
uniform integrability, i.e.  \cite[Thm. 4.5.9]{bogachev}.

Arguably, the most interesting observation of Lemma \ref{lem:capacityEXT} is the 
following:
\(\lim_{\rno \uparrow1} \tfrac{\rno-1}{\rno}\RC{\rno}{\cha}=0\) iff there exists a 
\(\mean\) in \(\pmea{\outA}\) satisfying \(\cha\UAC\mean\). This 
characterization is important because \(\cha\) 
is relatively compact, i.e. has a compact closure, in the topology of setwise convergence 
iff there exists a \(\mean\) in \(\pmea{\outA}\) satisfying \(\cha\UAC\mean\) by 
\cite[Thm. 4.7.25]{bogachev}. 
Since the topology of set wise convergence and the weak 
topology have exactly the same family of sets as their compact sets by \cite[Thm. 4.7.25]{bogachev}, 
the uniform absolute continuity also characterizes the relative compactness in the weak topology.

\begin{remark}\label{rem:weaktopology} 
	The weak topology on \(\smea{\outA}\) is 
	the topology generated by all continuous linear functions from \(\smea{\outA}\) 
	with the total variation topology to \(\reals{}\) with its usual topology. 
	Then the weak topology is weaker than the total variation topology, 
	i.e. the initial topology. 
	On the other hand, the topology of setwise convergence on \(\smea{\outA}\) is 
	the topology generated by the functions 
	\(\{\fX_{\oev}:\oev \in \outA\}\) where \(\fX_{\oev}(\mean)=\int_{\oev} \mean(\dif{\dout})\)
	for \(\oev\) in \(\outA\) and \(\mean\) in \(\smea{\outA}\). 
	Since \(\fX_{\oev}:\smea{\outA}\to\reals{}\) is a continuous linear function 
	for any \(\oev\in \outA\),
	the weak topology is stronger than the topology of setwise convergence.
	Nevertheless, the weak topology and the topology of setwise 
	convergence have exactly the same class of compact sets, 
	\cite[Thm. 4.7.25]{bogachev}.
	
	Our use of the term weak topology is consistent with the convention used in functional analysis, 
	see \cite[pp. 281,291]{bogachev}.
	While discussing the convergence of measures, however, the term weak topology is commonly used 
	to describe another topology. 
	If there is a topology on \(\outS\) and \(\outA\) is the resulting 
	Baire \(\sigma-\)algebra \cite[p. 12]{bogachev} of the subsets of \(\outS\), 
	then one can interpret 
	the space of measures as a space of linear functionals on the space of continuous and bounded 
	functions on \(\outS\). The weak* topology on the space of measures in this setting is often 
	called the weak topology \cite[Def. 8.1.2]{bogachev}. Although it is a very important 
	and useful concept in general, the weak topology in this second sense is not relevant in our 
	discussion because we have not assumed any topological structure on \(\outS\).
	\end{remark}  
\section{Examples}\label{sec:examples}
The order \(\rno\) \renyi entropy of a binary random variable, denoted \(\hX_{\rno}(\delta)\),
allows us to write certain expressions succinctly in some of the examples. 
For any \(\delta\in[0,1]\) it is defined as 
\begin{align}
\label{eq:def:binaryentropy}
\hX_{\rno}(\delta)
&\DEF 
\begin{cases}
\tfrac{1}{1-\rno}\ln (\delta^{\rno}+(1-\delta)^{\rno})
&\rno \neq 1 
\\
\delta\ln \tfrac{1}{\delta}
+(1-\delta)\ln \tfrac{1}{1-\delta}
&\rno=1
\end{cases}.
\end{align}
\subsection[Probabilities on Discrete Output Spaces]{Probabilities on Discrete Output Spaces}\label{sec:examples-discrete}
For \(\rno\geq1\), \(\qmn{\rno,\cha}\perp\qmn{\rno,\chu}\) iff \(\cha\perp\chu\). 
For \(\rno\in(0,1)\), \(\cha\perp\chu\) implies \(\qmn{\rno,\cha}\perp\qmn{\rno,\chu}\) 
but the converse is not true, i.e. 
\(\qmn{\rno,\cha}\perp\qmn{\rno,\chu}\) can hold even when \(\cha\) and \(\chu\)
are nonsingular.
Examples \ref{eg:singular-finite} and \ref{eg:singular-countable} provide 
such sets of probability measures. 
\begin{example}\label{eg:singular-finite}
For a \(\delta\in (0,\tfrac{1}{5})\), let \(\cha\) and \(\chu\) 
be\footnote{When \(\cha\) and \(\outS\) are finite sets and \(\outA\!=\!\sss{\outS}\), it is customary to describe \(\cha\)
	using a matrix. Each row corresponds to an element of \(\cha\), each column corresponds to an element of \(\outS\)
	and the element on the row \(\mW\) and  the column \(\dout\) is equal to \(\mW(\dout)\). 
	With a slight abuse of notation we denote the resulting matrix by \(\cha\), as well.}
\begin{align}
\notag
\cha
&= \left[
\begin{matrix}
1-\delta             & \delta               & 0                    & 0\\
\delta               & 1-\delta             & 0                    & 0\\
\tfrac{1-\delta}{2}  & \tfrac{1-\delta}{2}  & \tfrac{\delta}{2}  & \tfrac{\delta}{2}
\end{matrix} 
\right]
&
\chu
&= \left[
\begin{matrix}
0                      & 0                      & 1-\delta           & \delta               \\
0                      & 0                      & \delta             & 1-\delta              
\end{matrix} 
\right]
\end{align}
The third member of \(\cha\) is not singular with the members of \(\chu\); 
thus \(\cha\) is not singular with \(\chu\). 
We show in the following that \(\qmn{\rno,\cha}\perp\qmn{\rno,\chu}\) for all 
\(\rno \in (0,\fX^{-1}(\tfrac{\delta}{1-\delta})]\) where \(\fX^{-1}:[0,\tfrac{1}{4}]\to[0,1]\) 
is the inverse of the bijective decreasing function \(\fX(\dinp)\DEF(2^{1-\dinp}-1)^{\sfrac{1}{\dinp}}\).

For \(\chu\) and \(\mP=[\begin{matrix} \sfrac{1}{2} & \sfrac{1}{2} \end{matrix}]\) we have
\begin{align}
\notag
\RMI{\rno}{\mP}{\chu}
&= \ln 2-\hX_{\rno}(\delta)
&
&
&
\qmn{\rno,\mP}
&=\left[\begin{matrix} 0 & 0 & \sfrac{1}{2} & \sfrac{1}{2} \end{matrix} \right]
\end{align}
where \(\hX_{\rno}(\delta)\) is  defined in  \eqref{eq:def:binaryentropy}.
On the other hand, both \(\mU\)'s in \(\chu\) satisfy
\(\RD{\rno}{\mU}{\qmn{\rno,\mP}}=\ln 2-\hX_{\rno}(\delta)\) . 
Then \eqref{eq:necessaryandsufficientp} implies that
\(\RC{\rno}{\chu}=\ln 2-\hX_{\rno}(\delta)\) and 
\(\qmn{\rno,\chu}=\left[\begin{matrix} 0 & 0 & \sfrac{1}{2} & \sfrac{1}{2} \end{matrix} \right]\).

For \(\cha\) and \(\widetilde{\mP}=[\begin{matrix} \sfrac{1}{2} & \sfrac{1}{2} & 0 \end{matrix}]\) we have
\begin{align}
\notag
\RMI{\rno}{\widetilde{\mP}}{\cha}
&= \ln 2-\hX_{\rno}(\delta)
&
\qmn{\rno,\widetilde{\mP}}
&=\left[\begin{matrix} \tfrac{1}{2} & \tfrac{1}{2} & 0 & 0 \end{matrix} \right].
\end{align}
The first two \(\mW\)'s in \(\!\cha\!\) satisfy
\(\RD{\rno}{\mW}{\qmn{\rno,\widetilde{\mP}}}\!=\!\ln 2\!\!-\!\hX_{\rno}(\delta)\).
The third one satisfy \(\RD{\rno}{\mW}{\qmn{\rno,\widetilde{\mP}}}\!\leq\!\ln 2\!\!-\!\hX_{\rno}(\delta)\)
if and only if  \(\rno\!\leq\!\fX^{-1}(\frac{\delta}{1-\delta})\).
Consequently, \eqref{eq:necessaryandsufficientp} implies 
\(\RC{\rno}{\cha}\!=\!\ln 2\!-\!\hX_{\rno}(\delta)\) and 
\(\qmn{\rno,\cha}=\left[\begin{matrix}\sfrac{1}{2} & \sfrac{1}{2} &  0 & 0 \end{matrix} \right]\)
for all \(\rno \in (0,\fX^{-1}(\frac{\delta}{1-\delta})]\).
\end{example}

Example \ref{eg:singular-countable} provides sets of 
probability measures that are not even pairwise disjoint but 
they have singular \renyi centers for all orders in \((0,1)\). 
Example \ref{eg:singular-countable} also demonstrates the 
possible absence of an optimal prior for infinite sets of 
probability measures.

\begin{example}\label{eg:singular-countable}
	Let \((\outS,\outA)\) be \((\integers{},\sss{\integers{}})\) and let \(\cha_{\ind}\) be
	\begin{align}
	\notag
	\cha_{\ind}
	&=\{\mW^{\ind,\jnd}: \jnd \in \integers{}\}.
	&
	&\forall \ind \in \integers{} 
	\end{align}
	where \(\mW^{\ind,\jnd}(\dout)=(\IND{\dout=\ind}+\IND{\dout=\jnd})/2\).
	
	For any \(\rno\) in \((0,1)\) and sequence \(\{\pmn{\knd}\}_{\knd\in\integers{+}}\subset \pdis{\cha_{\ind}}\) of 
	uniform distributions with strictly increasing support
	\(\qmn{\rno,\pmn{\knd}}\) converges to \(\IND{\cdot=\ind}\) 
	in the total variation topology and 
	\(\lim\nolimits_{\knd \to \infty} \RMI{\rno}{\pmn{\knd}}{\cha_{\ind}}=\tfrac{\rno \ln 2}{1-\rno}\). 
	Furthermore, if \(\mQ(\cdot)=\IND{\cdot=\ind}\) then
	\(\RD{\rno}{\mW}{\mQ}\leq\tfrac{\rno \ln 2}{1-\rno}\)
	for all \(\mW\in \cha_{\ind}\). 
	Thus \(\RC{\rno}{\cha_{\ind}}=\tfrac{\rno \ln 2}{1-\rno}\) 
	and \(\qmn{\rno,\cha_{\ind}}(\cdot)=\IND{\cdot=\ind}\) 
	for all \(\rno\!\in\!(0,1)\)
	by \eqref{eq:necessaryandsufficientseq}.
	
	Note that \(\cha_{\ind}\)'s are not singular with one another, 
	in fact \(\cha_{\ind}\cap\cha_{\jnd}=\{\mW^{\ind,\jnd}\}\).
	Nonetheless, 
	\(\qmn{\rno,\cha_{\ind}} \perp \qmn{\rno,\cha_{\jnd}}\) for all \(\rno\)
	in \((0,1)\) whenever \(\ind\neq \jnd\)
	and we can use Lemma \ref{lem:capacityUnion} to calculate the
	\renyi capacity of any finite union of \(\cha_{\ind}\)'s. 
	For any finite set of integers \(\cset\), let \(\cha_{\cset} \) be 
	\(\cha_{\cset}=\cup_{\ind\in \cset} \cha_{\ind}\); then 
	\begin{align}
	\notag
	\RC{\rno}{\cha_{\cset}}
	&=\tfrac{\rno \ln 2}{1-\rno}+\ln \abs{\cset}
	&
	&\mbox{and}
	&
	\qmn{\rno,\cha_{\cset}}(\dout)
	&=\abs{\cset}^{-1}\IND{\dout\in\cset}.
	\end{align} 
	Furthermore, for any \(\mP\in \pdis{\cha_{\cset}}\) using \eqref{eq:sibson} 
	and \eqref{eq:lem:information:defA} we get
	\begin{align}
	\notag
	\RD{\rno}{\mP \mtimes \cha_{\cset}}{\mP\otimes\qmn{\rno,\cha_{\cset}}}
	&=\RMI{\rno}{\mP}{\cha_{\cset}}+\RD{\rno}{\qmn{\rno,\mP}}{\qmn{\rno,\cha_{\cset}}}.
	\end{align}
	Recall that \(\RD{\rno}{\mP \mtimes \cha_{\cset}}{\mP\otimes\qmn{\rno,\cha_{\cset}}}\leq\RC{\rno}{\cha_{\cset}}\)
	by Theorem \ref{thm:minimax} 
	and \(\RD{\rno}{\qmn{\rno,\mP}}{\qmn{\rno,\cha_{\cset}}}\geq0\) by
	Lemma \ref{lem:divergence-pinsker}.
	In addition
	\begin{itemize}
		\item \(\RD{\rno}{\mP \mtimes \cha_{\cset}}{\mP\otimes\qmn{\rno,\cha_{\cset}}}<\RC{\rno}{\cha_{\cset}}\)
		for any \(\mP\in\pdis{\cha_{\cset}}\) satisfying \(\mP(\mW)>0\) a \(\mW\) such that \(\supp{\mW}\subset\cset\).
		\item \(\RD{\rno}{\qmn{\rno,\mP}}{\qmn{\rno,\cha_{\cset}}}>0\)  for any \(\mP\in\pdis{\cha_{\cset}}\) 
		satisfying \(\mP(\mW)>0\) a \(\mW\) such that \(\supp{\mW}\nsubseteq\cset\).
	\end{itemize}
	Thus \(\RMI{\rno}{\mP}{\cha_{\cset}}<\RC{\rno}{\cha_{\cset}}\) for any \(\mP\in\pdis{\cha_{\cset}}\)
	and finite \(\cset\).
\end{example}

In Example \ref{eg:singular-finite} the optimal \(\mP\) satisfying \(\RMI{\rno}{\mP}{\cha}=\RC{\rno}{\cha}\)
was unique. However, this is not the case in general as demonstrated by Example \ref{eg:extendedbsc},
given in the following.
\begin{example}\label{eg:extendedbsc}
For a \(\delta\in[0,\sfrac{1}{2}]\), let \(\cha\) be
\begin{align}
\notag
\cha
&= \left[
\begin{matrix}
\delta					&\delta					&\sfrac{1}{2}-\delta	&\sfrac{1}{2}-\delta\\
\sfrac{1}{2}-\delta		&\sfrac{1}{2}-\delta	&\delta					&\delta\\
\delta					&\sfrac{1}{2}-\delta	&\sfrac{1}{2}-\delta	&\delta\\
\sfrac{1}{2}-\delta		&\delta					&\delta					&\sfrac{1}{2}-\delta
\end{matrix} 
\right].
\end{align}
Let \(\pmn{\rnb}\) be 
\([\begin{matrix} \sfrac{\rnb}{2} & \sfrac{\rnb}{2} & \sfrac{(1-\rnb)}{2}& \sfrac{(1-\rnb)}{2} \end{matrix}]\)
for any \(\rnb\in[0,1]\).
Then for all \(\rno\) in \(\reals{+}\) and \(\rnb\) in \([0,1]\) we have
\begin{align}
\notag
\RMI{\rno}{\pmn{\rnb}}{\cha}
&=\ln 2-\hX_{\rno}(2\delta)
&
&
&
\qmn{\rno,\pmn{\rnb}}
&=\left[\begin{matrix} \sfrac{1}{4} & \sfrac{1}{4} & \sfrac{1}{4} & \sfrac{1}{4}\end{matrix} \right].
\end{align} 
Furthermore, \(\RD{\rno}{\mW}{\qmn{\rno,\pmn{\rnb}}}=\RMI{\rno}{\pmn{\rnb}}{\cha}\) 
for all \(\mW\) in \(\cha\). 
Thus \(\RMI{\rno}{\pmn{\rnb}}{\cha}=\RC{\rno}{\cha}\) and \(\qmn{\rno,\cha}=\qmn{\rno,\pmn{\rnb}}\)
for all \(\beta\) in \([0,1]\)  and \(\rno\) in \(\reals{+}\)
by \eqref{eq:necessaryandsufficientp}.
\end{example}

We have demonstrated in Example \ref{eg:singular-countable} that for certain infinite
\(\cha\)'s \(\RMI{\rno}{\mP}{\cha}<\RC{\rno}{\cha}\) for all \(\mP\) in \(\pdis{\cha}\).
Example \ref{eg:erasure}, given in the following, demonstrates that a stronger assertion
``\(\RD{\rno}{\mW}{\qmn{\rno,\cha}}<\RC{\rno}{\cha}\) for all \(\mW\) in \(\cha\)'' 
is true for certain infinite \(\cha\)'s.
Hence, the claims of Lemma \ref{lem:capacityEps}  about \(\cha_{\rno,\epsilon}\) 
cannot be extended to \(\epsilon=0\) case for infinite \(\cha\)'s, 
because for the \(\cha\) given in Example \ref{eg:erasure} 
\(\RC{\rno}{\cha}>0\) and \(\cha_{\rno,0}=\emptyset\).
\begin{example}\label{eg:erasure}
Let us assume \(\gamma\in(0,1)\) and  \(\blx\in\integers{+}\).
Let \(\outS\) be \(\{0,\ldots,\blx\}\),
\(\outA\) be \(\sss{\outS}\), \(\chu\) and \(\cha\) be
\begin{align}
\notag
\chu
&=\{\mW^{\delta,\jnd}: \jnd \in\{1,\ldots,\blx\}, \delta\in[\gamma,1]\},
\\
\notag
\cha
&=\{\mW^{\delta,\jnd}: \jnd \in\{1,\ldots,\blx\}, \delta\in(\gamma,1]\}
\end{align}
where \(\mW^{\delta,\jnd}(\dout)=\IND{\dout=\jnd}(1-\delta)+\IND{\dout=0}\delta\).

Let \(\mP\in\pdis{\chu}\) be \(\mP(\mW^{\delta,\jnd})=\tfrac{1}{\blx}\IND{\delta=\gamma}\).
Then
\begin{align}
\notag
\RMI{\rno}{\mP}{\chu}
&=\begin{cases}
\tfrac{\rno}{\rno-1}\ln \left[\gamma+(1-\gamma)\blx^{\frac{\rno-1}{\rno}}\right]
&\rno\in \reals{+}\setminus\{1\}
\\
(1-\gamma) \ln \blx
&\rno=1
\end{cases},
\\
\notag
\qmn{\rno,\mP}(\dout)
&=\tfrac{\gamma\IND{\dout=0}}{\gamma+(1-\gamma)\blx^{1-\sfrac{1}{\rno}}}
+\sum\nolimits_{\jnd=1}^{\blx}
\tfrac{(1-\gamma)\blx^{-\sfrac{1}{\rno}}\IND{\dout=\jnd}}{\gamma+(1-\gamma)\blx^{1-\sfrac{1}{\rno}}}.
\end{align}
Furthermore, one can confirm by substitution that
\begin{align}
\notag
\RD{\rno}{\!\mW^{\delta,\jnd\!}}{\!\qmn{\rno,\mP}\!}
&\!=\!\begin{cases}
\!\tfrac{1}{1-\rno}\!\ln\!\tfrac{(1+((\sfrac{1}{\gamma})-1)\blx^{1-\sfrac{1}{\rno}})^{1-\rno}}{\delta^{\rno}
	+(1-\delta)^{\rno}((\sfrac{1}{\gamma})-1)^{1-\rno}\blx^{1-\sfrac{1}{\rno}}}
&\rno\neq 1\\
\delta\ln\tfrac{\delta}{\gamma}+ (1-\delta)\ln \tfrac{(1-\delta)\blx}{1-\gamma}
&\rno=1
\end{cases}.
\end{align}
Then \(\RD{\rno}{\mW}{\qmn{\rno,\mP}}\leq \RMI{\rno}{\mP}{\chu}\)
for all \(\mW\in \chu\) and consequently,
\(\RC{\rno}{\chu}=\RMI{\rno}{\mP}{\chu}\) and \(\qmn{\rno,\chu}=\qmn{\rno,\mP}\) by \eqref{eq:necessaryandsufficientp}.

Since \(\chu\) is the closure of \(\cha\) in the topology of setwise convergence,
\(\RC{\rno}{\cha}=\RC{\rno}{\chu}\)
for all \(\rno\in\reals{+}\) by Lemma \ref{lem:capacityEXT}-(\ref{capacityEXT-cl}).
Consequently,  \(\qmn{\rno,\cha}=\qmn{\rno,\chu}\) by Lemma \ref{lem:capacityUnion}
because \(\cha\subset\chu\). 

Then
\(\RD{\rno}{\mW}{\qmn{\rno,\cha}}\!<\!\RC{\rno}{\cha}\) for all \(\mW\) in \(\cha\)
and \(\rno\) in \(\reals{+}\).
Hence
\(\RMI{\rno}{\mP}{\cha}<\RC{\rno}{\cha}\) for all \(\mP\) in \(\pdis{\cha}\) and \(\rno\) in \(\reals{+}\)
by Lemma \ref{lem:information:def}
and
\(\cha_{(\rno,0)}\!=\!\emptyset\) for all \(\rno\) in \(\reals{+}\) by definition.
\end{example}

\subsection[Shift Invariant Families of Probabilities]{Shift Invariant Families of Probabilities}\label{sec:examples-invariant}
The shift invariant sets of probability measures on the unit interval
are relatively easy to analyze.
Nevertheless, when considered as a function of the order the \renyi capacities
of these sets form a diverse collection
and it is relatively easy to construct examples and counterexamples for 
the behavior of \renyi capacity as function of the order using this family.

First we consider the set of modular shifts of a probability 
measure on the unit interval, which  is called 
``channel with additive noise on the unit circle'' by 
Agustin in   \cite{augustin78}.

\begin{example}\label{eg:shiftchannel}
Let \(\outS\) be \([0,1)\),  \(\outA\) be \(\rborel{[0,1)}\), and \(\fX\) be a non-negative 
Lebesgue measurable function  such that  \(\int_{0}^{1} \fX(\dout) \dif{\dout}=1\).
Then \(\Scha{\fX}\) is the set of all probability measures whose Radon-Nikodym derivatives
with respect to the Lebesgue measure \(\lbm\) is a mod one shift of \(\fX\):
\begin{align}
\label{eq:def:shiftchannel}
\Scha{\fX}
&\DEF \left\{\mW: \der{\mW}{\lbm}=\fX \circ \trans{\dinp}~\mbox{for some~}\dinp\in [0,1)\right\}
\end{align}
where \(\trans{\dinp}(\dout)\DEF \dout-\dinp-\lfloor\dout-\dinp\rfloor\).

Let us denote the measure whose Radon-Nikodym derivative is \(\fX\) by \(\wmn{\fX}\).
Note that \(\RD{\rno}{\mW}{\lbm}=\RD{\rno}{\wmn{\fX}}{\!\lbm}\) for any \(\mW\) in \(\Scha{\fX}\)
and \(\rno\in (0,\infty]\). 
Thus \(\sup_{\mW\in \Scha{\fX}} \RD{\rno}{\mW}{\lbm}=\RD{\rno}{\wmn{\fX}}{\!\lbm}\)
for any \(\rno\in (0,\infty]\). 

If \(\RC{\rno}{\Scha{\fX}}\) is finite for an \(\rno\in (0,\infty]\), 
then \(\exists!\qmn{\rno,\Scha{\fX}}\) in \(\pmea{\outA}\) such that 
\begin{align}
\notag
\RD{\rno}{\mW}{\qmn{\rno,\Scha{\fX}}}
&\leq \RC{\rno}{\Scha{\fX}}
&
&\forall \mW \in \Scha{\fX}
\end{align}
by Theorem \ref{thm:minimax}.
On the other hand \(\qmn{\rno,\Scha{\fX}}=\qmn{s}+\qmn{ac}\) 
where \(\qmn{s}\perp\lbm\) and \(\qmn{ac} \AC \lbm\),
by the Lebesgue decomposition theorem \cite[5.5.3]{dudley}. 
Then 
\(\RD{\rno}{\mW}{\qmn{\rno,\Scha{\fX}}}=\RD{\rno}{\mW}{\qmn{ac}}\) for all 
\(\mW\) in \(\Scha{\fX}\) by \eqref{eq:def:divergence}
because \(\mW\AC \lbm\) for all \(\mW\) in \(\Scha{\fX}\). Thus 
\begin{align}
\notag
\RD{\rno}{\mW}{\sfrac{\qmn{ac}}{{\lon{{\qmn{ac}}}}}}
&=\RC{\rno}{\Scha{\fX}}+\ln \lon{{\qmn{ac}}}
&
&\forall \mW \in \Scha{\fX}.
\end{align}
If \(\lon{\qmn{ac}}<1\), then \(\sup_{\mW\in  \Scha{\fX}} 
\RD{\rno}{\mW}{\sfrac{\qmn{ac}}{{\lon{{\qmn{ac}}}}}}<\RC{\rno}{\Scha{\fX}}\).
This, however, is impossible because of Theorem \ref{thm:minimax}. Thus \(\lon{{\qmn{ac}}}=1\),
\(\lon{{\qmn{s}}}=0\) and \(\qmn{\rno,\Scha{\fX}} \AC \lbm\). 

Since \(\qmn{\rno,\Scha{\fX}} \AC \lbm\), the Radon-Nikodym derivative \(\der{\qmn{\rno,\Scha{\fX}}}{\lbm}\)
exists by the Radon-Nikodym theorem \cite[5.5.4]{dudley}. 
Since \(\Scha{\fX}\) is invariant under mod one shifts 
by construction, its \renyi centers 
need to be invariant under mods one shift, as well. 
Furthermore, \(\lbm\) is invariant under mod one shifts. Hence,
\begin{align}
\notag
\der{\qmn{\rno,\Scha{\fX}}}{\lbm}
&=\der{\qmn{\rno,\Scha{\fX}}}{\lbm} \circ \trans{\dinp}
&
& \forall \dinp \in [0,1)
\end{align}
Thus \(\der{\qmn{\rno,\Scha{\fX}}}{\lbm}\) needs to be a constant.
That constant is one because \(\qmn{\rno,\Scha{\fX}}\) is a probability measure. 
Therefore \(\qmn{\rno,\Scha{\fX}}\!=\!\lbm\) and \(\RC{\rno}{\Scha{\fX}}\!=\!\RD{\rno}{\wmn{\fX}}{\!\lbm}\) 
whenever \(\RC{\rno}{\Scha{\fX}}\) is finite.
When it is infinite 
so is \(\RD{\rno}{\wmn{\fX}}{\!\lbm}\) by Theorem \ref{thm:minimax}
because \(\RD{\rno}{\!\mW}{\!\lbm}\) equals \(\RD{\rno}{\!\wmn{\fX}}{\!\lbm}\) 
for all \(\mW\!\) in \(\!\Scha{\fX}\).
Hence, \(\RC{\rno}{\Scha{\fX}}\!=\!\RD{\rno}{\wmn{\fX}}{\!\lbm}\), i.e. 
\begin{align}
\label{eq:shiftchannelcapacity}
\RC{\rno}{\Scha{\fX}}
&=
\begin{cases}
\tfrac{1}{\rno-1}\ln \int \fX^{\rno}(\dout) \dif{\dout} 
&\rno\in \reals{+} \neq 1 
\\
\int \fX(\dout) \ln \fX(\dout) \dif{\dout} 
&\rno=1
\\
\ln  \essup_{\lbm}\fX(\dout) 
&\rno=\infty 
\end{cases}.
\end{align}
\eqref{eq:shiftchannelcapacity} is derived using the Ergodic theorem 
in Appendix \ref{sec:shiftchannel}.
\end{example}

As a result of Lemma \ref{lem:capacityO}, 
\(\RC{\rno}{\cha}\) is either continuous in \(\rno\) on \((0,\infty]\) or 
continuous and bounded on \((0,\rnf]\) and infinite on \((\rnf,\infty]\) for an \(\rnf\in [1,\infty)\).
The following two examples are special cases of Example \ref{eg:shiftchannel} which demonstrate that 
the \renyi capacity can become
infinite for some orders larger than one while being continuous on \((0,\infty]\) and 
the \renyi capacity can have a discontinuity
at any order in \([1,\infty)\).
\begin{example}\label{eg:infinitebutcontinuous} 
\(\fX_{\beta}(\dout)=(1-\beta)\dout^{-\beta}\) and \(\beta\in (0,1)\).
\begin{align}
\notag
\RC{\rno}{\Scha{\fX_{\beta}}}
&=
\begin{cases}
\tfrac{\rno\ln (1-\beta) -\ln(1-\rno\beta)}{\rno-1}
&\rno\in [0,1)\cup(1,\beta^{-1})
\\
\tfrac{\beta}{1-\beta}+\ln(1-\beta)
&\rno=1 
\\
\infty
&\rno\in [\beta^{-1},\infty]
\end{cases}
\end{align}
\(\RC{\rno}{\Scha{\fX_{\beta}}}\) is continuous 
on \((0,\infty]\) and monotone increasing and finite on 
\((0,\beta^{-1})\).   
\end{example}

\begin{example}\label{eg:discontinuity}
The existence of the discontinuity is related to the integrability of \(\fX \ln \fX\) and \(\fX^{\rno}\)
because \(\RC{\rno}{\Scha{\fX}}=\RD{\rno}{\wmn{\fX}}{\!\lbm}\).
\begin{itemize}
\item If \(\fX(\dout)\!=\!2\tfrac{\IND{0<\dout<\sfrac{1}{e}}}{\dout (\ln \frac{1}{\dout})^{3}}\),
 then 
\(\RC{1}{\Scha{\fX}}\!=\!\ln 2\sqrt{e}\) and \(\RC{\rno}{\Scha{\fX}}=\infty\) for all \(\rno\) in \((1,\infty]\). 
\item If \(\fX(\dout)\!=\!\tfrac{\dout^{-\frac{1}{\rnf}}\IND{0<\dout<\sfrac{1}{e}}}{(\ln\frac{1}{\dout})\int_{1-\frac{1}{\rnf}}^{\infty} 
\frac{e^{-\dsta}}{\dsta}\dif{\dsta}}\) 
for a \(\rnf\) in \((1,\infty)\),
 then 
\(\RC{\rnf}{\Scha{\fX}}\!=\!\tfrac{\ln(\rnf-1)}{1-\rnf}
-\tfrac{\rnf}{\rnf-1}\ln \int_{1-\frac{1}{\rnf}}^{\infty} \tfrac{e^{-\dsta}}{\dsta} \dif{\dsta}\)
and \(\RC{\rno}{\Scha{\fX}}\!=\!\infty\) for all \(\rno\) in \((\rnf,\infty]\). 
\end{itemize}
\end{example}

In all of the examples we have considered thus far the \renyi capacity is not only 
continuous but also differentiable in the order. 
This, however, is not the case in general.
\begin{example}\label{eg:shiftinvariant}
Let \(\fXS\) be a family of non-negative Lebesgue measurable functions such that  \(\int \fX \dif{\dout}=1\) for 
all \(\fX\in \fXS\). Then \(\Scha{\fXS}\) is the set of all probability measures whose Radon-Nikodym derivative  is a mod 
one shift of an \(\fX\) in \(\fXS\):
\begin{align}
\label{eq:def:shiftinvariant}
\Scha{\fXS}
&\!\DEF\!\left\{\mW: \der{\mW}{\lbm}\!=\!\fX \circ \trans{\dinp}~\mbox{for some~}\dinp\in [0,1),\mbox{~}\fX\in \fXS\right\}
\end{align}
where \(\trans{\dinp}(\dout)\DEF \dout-\dinp-\lfloor\dout-\dinp\rfloor\).

Note that \(\sup_{\mW \in \Scha{\fXS}}\RD{\rno}{\mW}{\lbm}\!=\!\sup_{\fX\in \fXS} \RD{\rno}{\wmn{\fX}}{\!\lbm}\)
because \(\Scha{\fXS}\!=\!\cup_{\fX\in\fXS}\Scha{\fX}\) 
and \(\RD{\rno}{\mW}{\lbm}\!=\!\RD{\rno}{\wmn{\fX}}{\!\lbm}\) for all \(\mW\!\) in \(\Scha{\fX}\).
Thus \(\RC{\rno}{\Scha{\fXS}}\!\leq\!\sup_{\fX\in \fXS} \RD{\rno}{\wmn{\fX}}{\!\lbm}\) by
Theorem \ref{thm:minimax}.
On the other hand, the reverse inequality follows from 
\eqref{eq:shiftchannelcapacity} and Lemma \ref{lem:capacityUnion}.
Thus,
\(\RC{\rno}{\Scha{\fXS}}\!=\!\sup_{\fX\in \fXS} \RD{\rno}{\wmn{\fX}}{\!\lbm}\), i.e.
\begin{align}
\label{eq:shiftinvariantcapacity}
\RC{\rno}{\Scha{\fXS}}
&=
\begin{cases}
\sup\limits_{\fX\in \fXS}
\tfrac{1}{\rno-1}\ln \int \fX^{\rno}(\dout) \dif{\dout} 
&\rno\in \reals{+} \neq 1 
\\
\sup\limits_{\fX\in \fXS}
\int \fX(\dout) \ln \fX(\dout) \dif{\dout} 
&\rno=1
\\
\sup\limits_{\fX\in \fXS}
\ln  \essup_{\lbm}\fX(\dout) 
&\rno=\infty 
\end{cases}.
\end{align}
If \(\fXS=\{2\dout,\frac{1}{2\sqrt{\dout}}\}\), 
then \(\RC{\rno}{\Scha{\fXS}}\) is not differentiable at \(\rno=\frac{1}{2}\).
\end{example}

\subsection[Families of Poisson Point Processes]{Certain Families of Poisson Point Processes}\label{sec:examples-poisson}
The following examples demonstrate the generality of our framework by determining
the \renyi capacity of various families of Poisson point processes with integrable intensity functions,
on real line.\footnote{The analysis we  present in the following can be applied to the spatial Poisson 
processes defined on appropriately chosen subsets of the Euclidean space without any major modification.
We restrict our analysis to the one dimensional case, because even the one dimensional case has 
a structure that is rich enough to demonstrate the generality of our framework.}  
Some of these families have been considered before 
in the context of channel coding problems, 
such as the ones in \eqref{eq:def:poissonchannel-constrained-A} and  
\eqref{eq:def:poissonchannel-bounded} in the following
(see \cite{burnashevK99}, \cite{davis80}, \cite{kabanov78}, \cite{wyner88-a}, \cite{wyner88-b}),
others have not been considered before,
such as the ones in 
\eqref{eq:def:poissonchannel-mean}, \eqref{eq:def:poissonchannel-constrained-B},
and  \eqref{eq:def:poissonchannel-product}.

The Poisson point processes are, sometimes, formulated and analyzed via the characterization of the interarrival times
without even mentioning the Radon-Nikodym derivatives, see \cite[Ch. 2]{gallagerSP}. For many applications such 
an approach turns out to be sufficient; as a result, the Radon-Nikodym derivatives of Poisson point processes are
not as well-known as one would expect. Considering this fact,  we follow the approach of Burnashev and Kutoyants in 
\cite{burnashevK99} and start our discussion  with a brief refresher on the Radon-Nikodym derivatives of 
the Poisson processes.

For any \(\tlx \in \reals{+}\), let \(\inpS_{\tlx}\) be the set of all nondecreasing, right-continuous, integer valued  
functions on \((0,\tlx]\). 
The sample paths of Poisson point processes are members of \(\inpS_{\tlx}\). 
Furthermore, any Poisson point process with deterministic intensity function \(\fX\) can be represented 
by a unique probability measure on the measurable space \((\outS,\outA)\) for \(\outS=\inpS_{\tlx}\) when \(\outA\) is an 
appropriately chosen 
\(\sigma-\)algebra.\footnote{One choice of \(\outA\) that works is the Borel \(\sigma-\)algebra 
	for the topology generated by the Skorokhod metric \(\mS\) on \(\inpS_{\tlx}\), denoted by 
	\(\borel{\inpS_{\tlx}}{\mS}\). In fact, \(\borel{\inpS_{\tlx}}{\mS}\) is rich enough to express 
the Poisson point processes whose intensity functions are not deterministic but Markovian, i.e. the intensity at any 
\(\tin\in(0,\tlx]\) depends on the previous arrivals. Kabanov's original work \cite{kabanov78}  considers such 
Poisson point processes, as well.} 

For any sample path \(\dout\in\outS\), we denote 
the \(\jnd^{th}\) arrival time  by \(\tau_{\jnd}(\dout)\)
and 
the number of arrivals up to and including time \(\tin\) by \(N_{\tin}(\dout)\).
The probability measure associated with a Poisson process with the intensity function \(\fX\) is denoted by  \(\wmn{\fX}\).
The probability measure of the Poisson point process with constant intensity \(\gamma\) is denoted by \(\rfm_{\gamma}\).
If \(\gamma=1\), we also use \(\rfm\) to denote \(\rfm_{\gamma}\), i.e. \(\rfm=\rfm_{1}\).

For any non-negative integrable function \(\fX\) on \((0,\tlx]\) the
associated probability measures \(\wmn{\fX}\) is absolutely continuous with respect to \(\rfm\) and 
the Radon-Nikodym derivative \(\der{\wmn{\fX}}{\rfm}\) is given 
by,\footnote{\(\left(\prod\nolimits_{\tau_{\jnd}(\dout)\leq\tlx} \fX(\tau_{\jnd}(\dout))\right)\)  
stands for \(1\) for \(\dout\)'s that do not have any arrivals.}
\cite[(2.1)]{burnashevK99}, \cite[VI.6.T12, p187]{bremaud},
\begin{align}
\label{eq:poissonRND}
\der{\wmn{\fX}}{\rfm}(\dout)
&=\left(\prod\nolimits_{\tau_{\jnd}(\dout)\leq\tlx} \fX(\tau_{\jnd}(\dout))\right)
e^{\int_{0}^{\tlx} (1-\fX(\tin))\dif{\tin}}.
\end{align}

For any non-negative measurable function \(\gX\), 
the following expression for the 
expectation\footnote{In \cite{burnashevK99}, Burnashev and Kutoyants express the identities given in 
	\eqref{eq:poissonRND} and \eqref{eq:poissonexpectation} more succinctly and elegantly, as follows: 
	\begin{align}
	\notag
	\der{\wmn{\fX}}{\rfm} (\dout)
	&=e^{\int_{0}^{\tlx} (\ln \fX(\tin)) \dout(\dif{\tin}) +\int_{0}^{\tlx} (1-\fX(\tin))\dif{\tin}},
	\\
	\notag
	\int e^{\int_{0}^{\tlx} (\ln \gX(\tin)) \dout(\dif{\tin})} \wmn{\fX}(\dif{\dout})
	&=e^{\int_{0}^{\tlx} (\gX(\tin)-1)\fX(\tin)\dif{\tin}}.
	\end{align}
	In the expressions \(\int_{0}^{\tlx}(\ln\fX(\tin))\dout(\dif{\tin})\) and \(\int_{0}^{\tlx}(\ln \gX(\tin))\dout(\dif{\tin})\),
	the sample path \(\dout\) is interpreted as a measure that is equal to the sum of Dirac delta functions located at 
	the arrival times of the sample path \(\dout\).}
follows from \eqref{eq:poissonRND}, \cite[(2.2)]{burnashevK99}: 
\begin{align}
\label{eq:poissonexpectation}
\int \left(\prod\nolimits_{\tau_{\jnd}(\dout)\leq \tlx}  \gX(\tau_{\jnd}(\dout))\right) \wmn{\fX}(\dif{\dout})
&=e^{\int_{0}^{\tlx} (\gX(\tin)-1)\fX(\tin) \dif{\tin}}.
\end{align}
An immediate consequence of  \eqref{eq:poissonRND}  and  \eqref{eq:poissonexpectation} is the following expression 
for the \renyi divergence between \(\wmn{\fX}\) and \(\wmn{\gX}\) for integrable intensity functions \(\fX\) and \(\gX\)
and positive real orders:
\begin{align}
\label{eq:poissondivergence}
\RD{\rno}{\wmn{\fX}}{\wmn{\gX}}
&\!=\!\begin{cases}
\int_{0}^{\tlx} \left(\tfrac{\fX^{\rno}\gX^{1-\rno}-\fX}{\rno-1}-\fX+\gX \right) \dif{\tin}
&\rno\neq 1
\\
\int_{0}^{\tlx} \left(\fX \ln \tfrac{\fX}{\gX} -\fX+\gX\right) \dif{\tin}
&\rno=1
\end{cases}.
\end{align}
For positive real orders other than one \eqref{eq:poissondivergence} follows from \eqref{eq:poissonRND} and \eqref{eq:poissonexpectation} by substitution, via 
the definition of the \renyi divergence. 
On the other hand, \(\RD{1}{\wmn{\fX}}{\wmn{\gX}}=\lim_{\rno \uparrow 1}\RD{\rno}{\wmn{\fX}}{\wmn{\gX}}\)
because the \renyi divergence is continuous in order on \([0,1]\) by  Lemma \ref{lem:divergence-order}.
Then the expression for \(\RD{1}{\wmn{\fX}}{\wmn{\gX}}\) follows from the dominated convergence theorem \cite[2.8.1]{bogachev} 
and the expression for \(\RD{\rno}{\wmn{\fX}}{\wmn{\gX}}\) for \(\rno\in (0,1)\) 
because \(\tfrac{\dinp^{\rno}-\dinp}{\rno-1} \uparrow \dinp \ln \dinp\)  as \(\rno \uparrow 1\) for any \(\dinp \geq 0\).

Let us proceed with defining the set of Poisson point processes we will be investigating.
\begin{definition}\label{def:poisson} 
For any \(\tlx\in\reals{+}\) and  intensity levels \(\mA\), \( \costc\), \(\mB\) satisfying 
\(0\leq\mA\leq\costc\leq\mB\leq\infty\), 
let 
\(\Pcha{\tlx,\mA,\mB,\costc}\), \(\Pcha{\tlx,\mA,\mB,\leq\costc}\), \(\Pcha{\tlx,\mA,\mB,\geq\costc}\), 
and \(\Pcha{\tlx,\mA,\mB}\)
be the set of all Poisson point processes  with \([\mA,\mB]\) valued deterministic integrable intensity functions on \((0,\tlx]\) 
with an average equal to \(\costc\),
less than or equal to \(\costc\),
greater than or equal to  \(\costc\),
and 
in \([\mA,\mB]\), respectively:
\begin{align}
\label{eq:def:poissonchannel-mean}
\Pcha{\tlx,\mA,\mB,\costc}
&\!\DEF\!\left\{\wmn{\fX}:\mA\leq \fX\leq\mB\mbox{~and~}\mbox{\(\int_{0}^{\tlx}\)}\fX(\tin) \dif{\tin}=\tlx\costc\right\},
\\
\label{eq:def:poissonchannel-constrained-A}
\Pcha{\tlx,\mA,\mB,\leq\costc}
&\!\DEF\!\left\{\wmn{\fX}:\mA\leq \fX\leq\mB\mbox{~and~}\mbox{\(\int_{0}^{\tlx}\)}\fX(\tin) \dif{\tin}\leq\tlx\costc\right\},
\\
\label{eq:def:poissonchannel-constrained-B}
\Pcha{\tlx,\mA,\mB,\geq\costc}
&\!\DEF\!\left\{\wmn{\fX}:\mA\leq \fX\leq\mB\mbox{~and~}\mbox{\(\int_{0}^{\tlx}\)}\fX(\tin) \dif{\tin}\geq\tlx\costc\right\},
\\
\label{eq:def:poissonchannel-bounded}
\Pcha{\tlx,\mA,\mB}
&\!\DEF\!\left\{\wmn{\fX}:\mA\leq \fX\leq\mB\right\}.
\end{align}
\end{definition}

The convention proposed in Definition \ref{def:poisson} allows us to refer to various families of 
Poisson point processes without confusion. 
However, explicitly stating the dependence on \(\tlx\), \(\mA\), and \(\mB\) is not necessary 
whenever the values of \(\tlx\), \(\mA\), and \(\mB\) are unambiguous.
When this is the case we use 
\(\Pcha{\costc}\) for \(\Pcha{\tlx,\mA,\mB,\costc}\),
\(\Pcha{\leq\costc}\) for \(\Pcha{\tlx,\mA,\mB,\leq\costc}\),
\(\Pcha{\geq\costc}\) for \(\Pcha{\tlx,\mA,\mB,\geq\costc}\),
and
\(\Pcha{}\) for \(\Pcha{\tlx,\mA,\mB}\).

In the following, we first determine the \renyi capacity and center of \(\Pcha{\tlx,\mA,\mB,\costc}\),
and then use these expressions to calculate the \renyi capacity and center of 
families described in Definition \ref{def:poisson} and in \eqref{eq:def:poissonchannel-product}.

\begin{example}\label{eg:poissonchannel-mean}
For any \(\tlx\!\in\!\reals{+}\), \(\mA,\mB\!\in\!\reals{\geq0}\) such that 
\(\mA\leq\mB\),  and \(\costc\in[\mA,\mB]\), 
\begin{align}
\label{eq:poissonchannel-mean-capacity}
\RC{\rno}{\Pcha{\costc}}
&=\begin{cases}
\tfrac{\rno}{\rno-1}(\pint_{\rno,\costc}-\costc)\tlx
&\rno\neq 1
\\
\left(
 \tfrac{\costc-\mA}{\mB-\mA} \mB\ln \tfrac{\mB}{\costc}
+\tfrac{\mB-\costc}{\mB-\mA} \mA\ln \tfrac{\mA}{\costc}
\right)\tlx
&\rno=1
\end{cases},
\\
\label{eq:poissonchannel-mean-center}
\qmn{\rno,\Pcha{\costc}}
&=\rfm_{\pint_{\rno,\costc}},
\\
\label{eq:poissonchannel-mean-center-intensity}
\pint_{\rno,\costc}
&\DEF \left(\tfrac{\costc-\mA}{\mB-\mA} \mB^{\rno}+\tfrac{\mB-\costc}{\mB-\mA} \mA^{\rno}\right)^{\sfrac{1}{\rno}}.
\end{align}
An alternative expression for \(\RC{\rno}{\Pcha{\costc}}\) is the following:
\begin{align}
\label{eq:poissonchannel-mean-capacity-alt}
\RC{\rno}{\Pcha{\costc}}
&=\tfrac{\costc-\mA}{\mB-\mA} \RD{\rno}{\rfm_{\mB}}{\rfm_{\pint_{\rno,\costc}}}
 +\tfrac{\mB-\costc}{\mB-\mA} \RD{\rno}{\rfm_{\mA}}{\rfm_{\pint_{\rno,\costc}}}.
\end{align}

If \(\costc\) is equal to \(\mA\) or \(\mB\),
 then \(\Pcha{\costc}\) has just one element;
consequently \(\RC{\rno}{\Pcha{\costc}}\) is zero and the only element of \(\Pcha{\costc}\) is also the \renyi center. 
For \(\costc\)'s in \((\mA,\mB)\), we first determine the \renyi capacity and center assuming that \(\tfrac{\costc-\mA}{\mB-\mA}\) 
is a rational number by giving a sequence of priors \(\{\pmn{\ind}\}\) and a
probability measure  \(\mQ\) satisfying
\(\lim_{\ind\to\infty}\RMI{\rno}{\pmn{\ind}}{\Pcha{\costc}}=\RRR{\rno}{\Pcha{\costc}}{\mQ}\).
Then we determine the \renyi capacity of \(\Pcha{\costc}\) with irrational \(\tfrac{\costc-\mA}{\mB-\mA}\) using
the continuity of the resulting expression in \(\mB\) and the monotonicity of \(\RC{\rno}{\cha}\) in \(\cha\).

There exists positive integers \(\ell\) and \(\blx\) such that \(\tfrac{\costc-\mA}{\mB-\mA}=\tfrac{\ell}{\blx}\) because
\(\tfrac{\costc-\mA}{\mB-\mA}\) is a rational number and \(\mB>\costc>\mA\). Then there are \(\binom{\blx}{\ell}\)
length \(\blx\) sequences of \(\mA\)'s and \(\mB\)'s with \(\ell\) \(\mB\)'s and \((\blx-\ell)\) \(\mA\)'s. These 
sequences will be the building blocks for \(\fX\)'s with positive \(\pmn{\ind}(\wmn{\fX})\).

For each positive integer \(\ind\) let us divide the interval \((0,\tlx]\) into \(2^{\ind}\blx\) 
half open intervals of the form \((\tfrac{\jnd-1}{2^{\ind}\blx}\tlx,\tfrac{\jnd}{2^{\ind}\blx}\tlx]\) 
for \(\jnd \in \{1,\ldots,2^{\ind}\blx\}\). Now consider \(\fX\)'s such that:
\begin{itemize}
\item \(\fX\) is  \(\{\mA,\mB\}\) valued function that is constant in all intervals of the form  
\((\tfrac{\jnd-1}{2^{\ind}\blx} \tlx,\tfrac{\jnd}{2^{\ind}\blx}\tlx]\) for \(\jnd\in\{1,\ldots,2^{\ind}\blx\}\).
\item \(\ell=\sum_{\knd=0}^{\blx-1}\IND{\fX(\frac{\blx \jnd-\knd}{2^{\ind}\blx}\tlx)=\mB}\) for all \(\jnd\in\{1,\ldots,2^{\ind}\}\).
\end{itemize}
For every such \(\fX\)  corresponding \(\wmn{\fX}\) is in \(\Pcha{\costc}\).
Furthermore, there are \(\binom{\blx}{\ell}^{(2^{\ind})}\) distinct \(\fX\)'s.
The prior \(\pmn{\ind}\) has equal probability mass on all \(\wmn{\fX}\)'s with 
the above described \(\fX\)'s.
Then using   \eqref{eq:poissonRND} we can calculate the Radon-Nikodym derivative of \(\mmn{\rno,\pmn{\ind}}\), 
\begin{align}
\notag
\der{\mmn{\rno,\pmn{\ind}}}{\rfm}(\dout)
&=e^{(1-\costc)\tlx}\left(
\sum\nolimits_{\knd=1}^{\binom{\blx}{\ell}^{(2^{\ind})}} 
\tfrac{\left(\prod\nolimits_{\tau_{\jnd}(\dout) \leq \tlx}  \fX_{\knd}(\tau_{\jnd}(\dout))\right)^{\rno}}{\binom{\blx}{\ell}^{(2^{\ind})}}
\right)^{\sfrac{1}{\rno}}.
\end{align}
For the sample paths, i.e. \(\dout\)'s, that do not have more than one arrival in any of the intervals  of the form 
\((\tfrac{\jnd-1}{2^{\ind}}\tlx,\tfrac{\jnd}{2^{\ind}}\tlx]\), one can simplify the expression for the  
Radon-Nikodym derivative significantly. In particular,
\begin{align}
\notag
\der{\mmn{\rno,\pmn{\ind}}}{\rfm}(\dout)
&=e^{(1-\costc)\tlx}\left(\tfrac{\ell \mB^{\rno}+(\blx-\ell) \mA^{\rno}}{\blx}\right)^{\frac{N_{\tlx}(\dout)-N_{0}(\dout)}{\rno}}
&
&\forall \dout \in \oev_{\ind}
\end{align}
where \(N_{\tin}(\dout)\) is the number of arrivals on \((0,\tin]\) for
the sample path \(\dout\) and  \(\oev_{\ind}\in \outA\) is defined as
\begin{align}
\notag
\oev_{\ind}
&\!\DEF\!\left\{\dout: \abs{N_{\frac{\jnd}{2^{\ind}}\tlx}(\dout)-N_{\frac{\jnd-1}{2^{\ind}}\tlx}(\dout)}\leq 1~\forall\jnd\in\{1,\ldots,2^{\ind}\}\right\}. 
\end{align}
Since \(\oev_{\ind}\subset \oev_{\ind+1}\) the following holds 
\(\forall\dout \in \cup_{\ind\in\integers{+}} \oev_{\ind}\)
\begin{align}
\notag
\lim\nolimits_{\ind \to \infty} \der{\mmn{\rno,\pmn{\ind}}}{\rfm}(\dout)
&=e^{(1-\costc)\tlx}\left(\tfrac{\ell \mB^{\rno}+(\blx-\ell) \mA^{\rno}}{\blx}\right)^{\frac{N_{\tlx}(\dout)-N_{0}(\dout)}{\rno}}.
\end{align}
Using the complete independence of the Poisson processes on disjoint intervals and the probability mass function of the counting 
process, \cite[Thm. 2.2.10]{gallagerSP}, \cite[II.1.(1.9), p. 22]{bremaud}, we can calculate the probability \(\rfm(\oev_{\ind})\):
\begin{align}
\notag
\rfm(\oev_{\ind})
&=(e^{-\frac{\tlx}{2^{\ind}}}+\tfrac{\tlx}{2^{\ind}} e^{-\frac{\tlx}{2^{\ind}}})^{(2^{\ind})}
\\
\notag
&=e^{-\tlx}(1+\tfrac{\tlx}{2^{\ind}})^{(2^{\ind})}.
\end{align}
Then \(\lim_{\ind \to \infty} \rfm(\oev_{\ind})=1\) and consequently \(\rfm(\cup_{\ind\in\integers{+}} \oev_{\ind})=1\). 
Thus convergence on \((\cup_{\ind\in\integers{+}} \oev_{\ind})\) implies \(\rfm-\)a.e. convergence:
\begin{align}
\notag
\der{\mmn{\rno,\pmn{\ind}}}{\rfm}(\dout)
&\xrightarrow{\rfm-a.e}e^{(1-\costc)\tlx}\left(\tfrac{\ell \mB^{\rno}+(\blx-\ell) \mA^{\rno}}{\blx}\right)^{\frac{N_{\tlx}(\dout)-N_{0}(\dout)}{\rno}}.
\end{align}
On the other hand \(\der{\mmn{\rno,\pmn{\ind}}}{\rfm}(\dout)\!\leq\! e^{(1-\costc)\tlx}\tfrac{\mB^{N_{\tlx}(\dout)}}{\mB^{N_{0}(\dout)}}\) because 
\(\fX(\tin)\!\leq\!\mB\).
Furthermore 
\(\int e^{(1-\costc)\tlx}\tfrac{\mB^{N_{\tlx}}}{\mB^{N_{0}}}\rfm(\dif{\dout})\!=\!e^{(\mB-\costc)\tlx}\).
Thus the dominated convergence theorem \cite[2.8.1]{bogachev} implies that
\begin{align}
\notag
\lim\limits_{\ind \to \infty} \lon{\mmn{\rno,\pmn{\ind}}} 
&=e^{(1-\costc)\tlx}\int \left(\tfrac{\ell \mB^{\rno}+(\blx-\ell) \mA^{\rno}}{\blx}\right)^{\frac{N_{\tlx}(\dout)-N_{0}(\dout)}{\rno}} \rfm(\dif{\dout}) 
\\
\notag
&=e^{\left(\left(\frac{\ell \mB^{\rno}+(\blx-\ell) \mA^{\rno}}{\blx}\right)^{\sfrac{1}{\rno}}-\costc\right)\tlx}.
\end{align}
Thus using \eqref{eq:def:information}
and the fact that \(\tfrac{\costc-\mA}{\mB-\mA}=\tfrac{\ell}{\blx}\) we get
\begin{align}
\label{eq:poissonchannel-mean-capacity-A}
\lim\nolimits_{\ind \to \infty} \RMI{\rno}{\pmn{\ind}}{\Pcha{\costc}}
&=\tfrac{\rno}{\rno-1}\left(\pint_{\rno,\costc}-\costc\right)\tlx
&
&\forall \rno\neq 1.
\end{align}
On the other hand for any \(\gamma\in \reals{+}\) and \(\fX\!:\!(0,\tlx]\!\to\![\mA,\mB]\) 
satisfying \(\int_{0}^{\tlx}\fX(\tin)\dif{\tin}=\tlx\costc\), as a result of 
\eqref{eq:poissondivergence}
\begin{align}
\notag
\RD{\rno}{\wmn{\fX}}{\rfm_{\gamma}}\!
&\!=\!\int_{0}^{\tlx} \left[\tfrac{\gamma^{1-\rno}}{\rno-1}\fX^{\rno}(\tin)-\tfrac{\rno}{\rno-1} \fX(\tin)+\gamma\right] \dif{\tin}
\\
\notag
&\!\leq\!\int_{0}^{\tlx}\!\!\tfrac{\gamma^{1-\rno}}{\rno-1}\!\left[\tfrac{\fX(\tin)-\mA}{\mB-\mA}\mB^{\rno}
\!+\! \tfrac{\mB-\fX(\tin)}{\mB-\mA}\mA^{\rno}\right]\dif{\tin}
\!-\!\tfrac{\rno \costc \tlx}{\rno-1}\!+\!\gamma\!\tlx
\\
\notag
&=\left[\tfrac{\gamma^{1-\rno}}{\rno-1} \left(\tfrac{\costc-\mA}{\mB-\mA}\mB^{\rno}+ 
\tfrac{\mB-\costc}{\mB-\mA}\mA^{\rno}\right) -\tfrac{\rno}{\rno-1} \costc+\gamma\right] \tlx 
\\
\label{eq:poissonchannel-mean-capacity-C}
&=\tfrac{\rno}{\rno-1}(\pint_{\rno,\costc}-\costc)\tlx +\RD{\rno}{\rfm_{\pint_{\rno,\costc}}}{\rfm_{\gamma}}
\end{align}
where the inequality follows from the convexity of the function \(\tfrac{\dinp^{\rno}}{\rno-1}\) 
in \(\dinp\) and the Jensen's inequality.

Using 
\eqref{eq:poissonchannel-mean-capacity-A} 
and
\eqref{eq:poissonchannel-mean-capacity-C}
for \(\gamma=\pint_{\rno,\costc}\)
we can conclude that
\(\lim_{\ind\to\infty}\RMI{\rno}{\pmn{\ind}}{\Pcha{\costc}}=
\RRR{\rno}{\Pcha{\costc}}{\rfm_{\pint_{\rno,\costc}}}\).
Then for \(\rno\)'s other than one \eqref{eq:poissonchannel-mean-capacity}  follows 
from \eqref{eq:necessaryandsufficientseq}
for values of \(\costc\) making \(\tfrac{\costc-\mA}{\mB-\mA}\) a rational number.
For values of \(\costc\) making \(\tfrac{\costc-\mA}{\mB-\mA}\) a rational number,
\eqref{eq:poissonchannel-mean-capacity} for \(\rno=1\) case follows from 
the expression for \(\rno\neq1\) case 
via  L'Hospital's rule \cite[Thm. 5.13]{rudin}
because  the \renyi capacity is a continuous function of the order 
on \((0,1]\) by Lemma \ref{lem:capacityO}-(\ref{capacityO-zo}).

We now prove that \eqref{eq:poissonchannel-mean-capacity} holds for values of \(\costc\) for which
\(\tfrac{\costc-\mA}{\mB-\mA}\) is irrational.
First note that 
\(\Pcha{\tlx,\mA,\mB_{1},\costc} \subset \Pcha{\tlx,\mA,\mB_{2},\costc}\)
for any \(\tlx\), \(\mA\), \(\costc\), \(\mB_{1}\), \(\mB_{2}\) satisfying \(\mB_{1}\leq \mB_{2}\), 
by the definition of \(\Pcha{\tlx,\mA,\mB,\costc}\) given in \eqref{eq:def:poissonchannel-mean}. 
Then \(\RC{\rno}{\Pcha{\tlx,\mA,\mB_{1},\costc}}\leq \RC{\rno}{\Pcha{\tlx,\mA,\mB_{2},\costc}}\)
by definition. 
Then \eqref{eq:poissonchannel-mean-capacity} holds for the case when \(\tfrac{\costc-\mA}{\mB-\mA}\) is irrational 
as a result of the continuity of the expression on the right hand side of \eqref{eq:poissonchannel-mean-capacity} as
a function of \(\mB\)  for each \(\rno\in\reals{+}\).

For orders other than one \eqref{eq:poissonchannel-mean-center} follows 
from Theorem \ref{thm:minimax} because 
\(\RRR{\rno}{\Pcha{\costc}}{\rfm_{\pint_{\rno,\costc}}}=\RC{\rno}{\Pcha{\costc}}\)
by \eqref{eq:poissonchannel-mean-capacity} and \eqref{eq:poissonchannel-mean-capacity-C}.
In order extend \eqref{eq:poissonchannel-mean-center}  to \(\rno=1\) case we invoke
the continuity of \renyi center established Lemma \ref{lem:centercontinuity}.
\end{example}

\begin{example}\label{eg:poissonchannel-constrained}
For any \(\tlx\!\in\!\reals{+}\), \(\mA,\mB\!\in\!\reals{\geq0}\) such that 
\(\mA\leq\mB\),  and \(\costc\in[\mA,\mB]\), 
\begin{align}
\label{eq:poissonchannel-constrained-capacity}
\RC{\rno}{\Pcha{\leq\costc}}
&=\RC{\rno}{\Pcha{\costc \wedge \costc_{\rno}}}
&
\RC{\rno}{\Pcha{\geq\costc}}
&=\RC{\rno}{\Pcha{\costc \vee \costc_{\rno}}}
\\
\label{eq:poissonchannel-constrained-center}
\qmn{\rno,\Pcha{\leq\costc}}
&=\qmn{\rno,\Pcha{\costc \wedge \costc_{\rno}}}
&
\qmn{\rno,\Pcha{\geq\costc}}
&=\qmn{\rno,\Pcha{\costc \vee \costc_{\rno}}}
\end{align}
where \(\RC{\rno}{\Pcha{\costc}}\) and \(\qmn{\rno,\Pcha{\costc}}\) are given in  
\eqref{eq:poissonchannel-mean-capacity} and \eqref{eq:poissonchannel-mean-center}
and  \(\costc_{\rno}\) is defined as follows:
\begin{align}
\label{eq:poissonchannel-bounded-optimal}
\costc_{\rno}
&\DEF\begin{cases}
\rno^{\frac{\rno}{1-\rno}}(\tfrac{\mB-\mA}{\mB^{\rno}-\mA^{\rno}})^{\frac{1}{1-\rno}}
+\tfrac{\mA\mB^{\rno}-\mB \mA^{\rno}}{\mB^{\rno}-\mA^{\rno}}
&
\rno\neq 1
\\
e^{-1}\mB^{\frac{\mB}{\mB-\mA}}\mA^{-\frac{\mA}{\mB-\mA}}
&
\rno=1
\end{cases}.
\end{align}
Since \(\Pcha{\leq\costc}\) is the union of \(\Pcha{\gamma}\) for 
\(\gamma\) in \([\mA,\costc]\),
\(\RC{\rno}{\Pcha{\leq\costc}}\) equals \(\RC{\rno}{\Pcha{\costc \wedge \costc_{\rno}}}\)
iff
\(\RRR{\rno}{{\Pcha{\leq\costc}}}{\qmn{\rno,\Pcha{\costc \wedge \costc_{\rno}}}}
\leq \RC{\rno}{\Pcha{\costc \wedge \costc_{\rno}}}\)
by Lemma \ref{lem:capacityUnion}.

On the other hand when considered together with 
the convexity of \(\tfrac{\dinp^{\rno}-\dinp}{\rno-1}\) in \(\dinp\) for \(\rno\neq1\) case and 
the convexity of \(\dinp\ln\dinp\) in  \(\dinp\) for \(\rno=1\) case,
\eqref{eq:poissondivergence} implies 
 \begin{align}
\label{eq:poissonchannel-constrained-A}
\RD{\rno}{\mW}{\rfm_{\mS}}
&\leq \tfrac{\mB-\gamma}{\mB-\mA} \RD{\rno}{\rfm_{\mA}}{\rfm_{\mS}}+\tfrac{\gamma-\mA}{\mB-\mA} \RD{\rno}{\rfm_{\mB}}{\rfm_{\mS}}
\end{align}
for all \(\mW\) in \(\Pcha{\gamma}\). Furthermore, 
the definitions of \(\pint_{\rno,\gamma}\) and \(\costc_{\rno}\)
given in \eqref{eq:poissonchannel-mean-center-intensity} and \eqref{eq:poissonchannel-bounded-optimal} 
imply that
\begin{align}
\label{eq:poissonchannel-constrained-B}
\RD{\rno}{\rfm_{\mA}}{\rfm_{\pint_{\rno,\gamma}}}
&\leq\RD{\rno}{\rfm_{\mB}}{\rfm_{\pint_{\rno,\gamma}}}
&
&\forall \gamma\in[\mA,\costc_{\rno}].
\end{align}
Using \eqref{eq:poissonchannel-constrained-A} and \eqref{eq:poissonchannel-constrained-B}
together with the alternative expression for \(\RC{\rno}{\Pcha{\costc}}\) 
given in \eqref{eq:poissonchannel-mean-capacity-alt} we get
\begin{align}
\notag
\RRR{\rno}{\Pcha{\leq\costc}}{\rfm_{\pint_{\rno,\costc\wedge\costc_{\rno}}}}
&\leq \tfrac{\mB-\costc\wedge\costc_{\rno}}{\mB-\mA} \RD{\rno}{\rfm_{\mA}}{\rfm_{\pint_{\rno,\costc\wedge\costc_{\rno}}}}\\
\notag
&\qquad~\qquad +\tfrac{\costc\wedge\costc_{\rno}-\mA}{\mB-\mA} \RD{\rno}{\rfm_{\mB}}{\rfm_{\pint_{\rno,\costc\wedge\costc_{\rno}}}}
\\
\notag
&=\RC{\rno}{\Pcha{\costc\wedge\costc_{\rno}}}.
\end{align} 
Thus \(\RRR{\rno}{{\Pcha{\leq\costc}}}{\qmn{\rno,\Pcha{\costc \wedge \costc_{\rno}}}}
\leq \RC{\rno}{\Pcha{\costc \wedge \costc_{\rno}}}\)
follows from \eqref{eq:poissonchannel-mean-center}. Hence 
\(\RC{\rno}{\Pcha{\leq\costc}}=\RC{\rno}{\Pcha{\costc \wedge \costc_{\rno}}}\)
and
\(\qmn{\rno,\Pcha{\leq\costc}}=\qmn{\rno,\Pcha{\costc \wedge \costc_{\rno}}}\)
by Lemma \ref{lem:capacityUnion}.

Assertions about \(\Pcha{\geq\costc}\) derived similarly using the following observations:
\(\Pcha{\geq\costc}\)
is the union of \(\Pcha{\gamma}\) for \(\gamma\) in \([\costc,\mB]\) and
\begin{align}
\label{eq:poissonchannel-constrained-C}
\RD{\rno}{\rfm_{\mA}}{\rfm_{\pint_{\rno,\gamma}}}
&\geq\RD{\rno}{\rfm_{\mB}}{\rfm_{\pint_{\rno,\gamma}}}
&
&\forall \gamma\in[\costc_{\rno},\mB].
\end{align}
\end{example}

\begin{example}\label{eg:poissonchannel-bounded}
For any \(\tlx\!\in\!\reals{+}\) and \(\mA,\mB\!\in\!\reals{\geq0}\) such that 
\(\mA\leq\mB\),
\begin{align}
\label{eq:poissonchannel-bounded-capacity-original}
\RC{\rno}{\Pcha{}}
&=\RC{\rno}{\Pcha{\costc_{\rno}}},
\\
\label{eq:poissonchannel-bounded-center-original}
\qmn{\rno,\Pcha{}}
&=\qmn{\rno,\Pcha{\costc_{\rno}}},
\end{align}
where \(\RC{\rno}{\Pcha{\costc}}\), \(\qmn{\rno,\Pcha{\costc}}\), \(\costc_{\rno}\) 
are described in 
\eqref{eq:poissonchannel-mean-capacity},
\eqref{eq:poissonchannel-mean-center},
\eqref{eq:poissonchannel-bounded-optimal}
because \(\Pcha{}=\Pcha{\leq \mB}\).
By substitution we get the following more explicitly expressions:
\begin{align}
\label{eq:poissonchannel-bounded-capacity}
\RC{\rno}{\Pcha{}}
&=\begin{cases}
\left((\tfrac{\rno(\mB-\mA)}{\mB^{\rno}-\mA^{\rno}})^{\frac{1}{1-\rno}}
-\tfrac{\rno}{\rno-1} \tfrac{\mA\mB^{\rno}-\mB \mA^{\rno}}{\mB^{\rno}-\mA^{\rno}}
\right)\tlx
&\rno\neq 1
\\
\left(e^{-1}\mB^{\frac{\mB}{\mB-\mA}}\mA^{-\frac{\mA}{\mB-\mA}}
-\tfrac{\mA \mB}{\mB-\mA}\ln \tfrac{\mB}{\mA}
\right) \tlx
&\rno=1 
\end{cases},
\\
\label{eq:poissonchannel-bounded-center}
\qmn{\rno,\Pcha{}}
&=\rfm_{\pint_{\rno}},
\\
\label{eq:poissonchannel-bounded-center-intensity}
\pint_{\rno}
&\DEF
\begin{cases}
\rno^{\frac{1}{1-\rno}}(\tfrac{\mB-\mA}{\mB^{\rno}-\mA^{\rno}})^{\frac{1}{1-\rno}}
&
\rno\in\reals{+} \setminus\{1\}
\\
e^{-1}\mB^{\frac{\mB}{\mB-\mA}}\mA^{-\frac{\mA}{\mB-\mA}}
&
\rno=1
\end{cases}.
\end{align}
One can also confirm \(\pint_{\rno}=\pint_{\rno,\costc_{\rno}}\) using 
\eqref{eq:poissonchannel-mean-center-intensity} and \eqref{eq:poissonchannel-bounded-optimal}.
The following expression for \(\RC{\rno}{\Pcha{}}\) is equivalent to \eqref{eq:poissonchannel-bounded-capacity}:
\begin{align}
\label{eq:poissonchannel-bounded-capacity-alt}
\RC{\rno}{\Pcha{}}
&=\tfrac{\costc_{\rno}-\mA}{\mB-\mA} \RD{\rno}{\rfm_{\mB}}{\rfm_{\pint_{\rno}}}
 +\tfrac{\mB-\costc_{\rno}}{\mB-\mA} \RD{\rno}{\rfm_{\mA}}{\rfm_{\pint_{\rno}}}.
\end{align}
\end{example}

In the preceding examples, we have assumed the intensity functions 
are bounded above by a constant;
we replace this constant with an integrable function 
in Example \ref{eg:poissonchannel-product} given in the following.
Let us first give a formal definition.
\begin{definition}\label{def:poissonproduct} 
For any \(\tlx\) in \(\reals{+}\), \(\mA\) in \(\reals{\geq0}\),
and  Lebesgue integrable function \(\gX\) on \((0,\tlx]\) satisfying
\(\gX\geq \mA\),
\(\Pcha{\tlx,\mA,\gX(\cdot)}\) is 
the set of all Poisson point processes with 
deterministic intensity functions \(\fX\) 
satisfying \(\mA\leq\fX\leq\gX\):
\begin{align}
\label{eq:def:poissonchannel-product}
\Pcha{\tlx,\mA,\gX(\cdot)}
&\DEF \left\{\wmn{\fX}:\mA\leq\fX(\tin)\leq\gX(\tin)~~\forall \tin \in(0,\tlx] \right\}.
\end{align}
\end{definition}

\begin{example}\label{eg:poissonchannel-product}
For any \(\tlx\in\reals{+}\), \(\mA\in\reals{\geq0}\), and \(\gX\in\Lon{\lbm}\)
satisfying \(\gX(\tin)\geq \mA\) for all \(\tin\) in \((0,\tlx]\) we have
\begin{align}
\label{eq:poissonchannel-product-capacity}
\RC{\rno}{\Pcha{}}
&\!=\!\begin{cases}
\!\int_{0}^{\tlx}\!\left[\!(\tfrac{\rno(\gX-\mA)}{\gX^{\rno}-\mA^{\rno}})^{\frac{1}{1-\rno}}
\!-\!\tfrac{\rno}{\rno-1}\!\tfrac{\mA\gX^{\rno}-\gX \mA^{\rno}}{\gX^{\rno}-\mA^{\rno}}
\!\right]\!\dif{\tin}\!
&\rno\!\neq\!1
\\
\!\int_{0}^{\tlx}\!\left[\!
e^{-1}\gX^{\frac{\gX}{\gX-\mA}}\mA^{-\frac{\mA}{\gX-\mA}}
\!-\!\tfrac{\mA \gX}{\gX-\mA}\ln \tfrac{\gX}{\mA}
\!\right]\!\dif{\tin}\!
&\rno\!=\!1 
\end{cases},
\\
\label{eq:poissonchannel-product-center}
\qmn{\rno,\Pcha{}}
&\!=\!\wmn{\pint_{\rno}},
\\
\label{eq:poissonchannel-product-center-intensity}
\pint_{\rno}(\tin)
&\!\DEF\!
\begin{cases}
\rno^{\frac{1}{1-\rno}}(\tfrac{\gX(\tin)-\mA}{\gX^{\rno}(\tin)-\mA^{\rno}})^{\frac{1}{1-\rno}}
&
\rno\neq 1
\\
e^{-1}[\gX(\tin)]^{\frac{\gX(\tin)}{\gX(\tin)-\mA}}\mA^{-\frac{\mA}{\gX(\tin)-\mA}}
&
\rno=1
\end{cases}.
\end{align}

If \(\gX\) is a simple function, 
then we can apply
\eqref{eq:poissonchannel-bounded-capacity} and \eqref{eq:poissonchannel-bounded-center}
for each possible value of \(\gX\), separately. 
Then \eqref{eq:poissonchannel-product-capacity} and \eqref{eq:poissonchannel-product-center} 
follow from Lemma \ref{lem:capacityProduct} because simple functions can only take finite 
number of distinct values.\footnote{We are not overlooking the issue of contiguity for the 
inverse of the image of \(\gX\) because Example \ref{eg:poissonchannel-bounded} holds
as is for Poisson processes defined on any measurable set of Lebesgue measure \(\tlx\), not just the 
interval \((0,\tlx]\).}
On the other hand, there exists a sequence of simple functions 
\(\{\gX^{(\ind)}\}_{\ind\in\integers{+}}\) satisfying \(\mA\leq \gX^{(\ind)}\) and \(\gX^{(\ind)}\uparrow\gX\)
for any measurable \(\gX\).
Evidently \(\RC{\rno}{\Pcha{\tlx,\mA,\gX^{(\ind)}(\cdot)}}\leq\RC{\rno}{\Pcha{\tlx,\mA,\gX(\cdot)}}\)
because \(\Pcha{\tlx,\mA,\gX^{(\ind)}(\cdot)}\subset \Pcha{\tlx,\mA,\gX(\cdot)}\). 
Furthermore, if \(\fX\) satisfies \(\mA\leq\fX\leq\gX\), then
\begin{align}
\notag
\abs{\rno^{\frac{1}{1-\rno}}(\tfrac{\fX(\tin)-\mA}{\fX^{\rno}(\tin)-\mA^{\rno}})^{\frac{1}{1-\rno}}
-\tfrac{\rno}{\rno-1} \tfrac{\mA\fX^{\rno}(\tin)-\fX(\tin) \mA^{\rno}}{\fX^{\rno}(\tin)-\mA^{\rno}}}  
&\leq \rno^{\frac{1}{1-\rno}} \gX(\tin). 
\end{align}
for all \(\tin\in(0,\tlx]\).

Then the integral on the right hand side of \eqref{eq:poissonchannel-product-capacity}
equals to \(\lim_{\ind \to \infty} \RC{\rno}{\Pcha{\tlx,\mA,\gX^{(\ind)}(\cdot)}}\) by the 
dominated convergence theorem \cite[2.8.1]{bogachev}.
Hence, it is a lower bound on \(\RC{\rno}{\Pcha{\tlx,\mA,\gX(\cdot)}}\).
It is, also, an upper bound on \(\RC{\rno}{\Pcha{\tlx,\mA,\gX(\cdot)}}\)
by \eqref{eq:thm:minimaxradius} 
because \(\RD{\rno}{\mW}{\wmn{\pint_{\rno}}}\) is bounded from above by it for all 
\(\mW\in\Pcha{\tlx,\mA,\gX(\cdot)}\).
Then \eqref{eq:poissonchannel-product-center} follows from the uniqueness of the 
\renyi center.
\end{example}

\section{Discussion}\label{sec:conclusion}
In this paper, we define and analyze the order \(\rno\) \renyi capacity \(\RC{\rno}{\cha}\) 
and the order \(\rno\) \renyi radius \(\RR{\rno}{\cha}\) for an arbitrary 
set of probability measures \(\cha\) on an arbitrary measurable space.
Our most important contributions are proving the van Erven-\harremoes conjecture, i.e. Lemma \ref{lem:EHB}, 
and two uniform equicontinuity results on the \renyi information, i.e.
Lemma \ref{lem:finitecapacity}-(\ref{finitecapacity-uecP},\ref{finitecapacity-uecO}).
We also prove a minimax theorem, i.e. Theorem \ref{thm:minimax},  which has been 
previously reported by Augustin in \cite{augustin78} in a different form 
and for orders between zero and two. 
Theorem \ref{thm:minimax}  establishes not only the equality of \(\RC{\rno}{\cha}\) and \(\RR{\rno}{\cha}\)
for any \(\rno\) and \(\cha\) but also  the existence of a unique order \(\rno\) \renyi center 
whenever \(\RC{\rno}{\cha}\) is finite.  
Our analysis leads to certain immediate consequences for two generalizations: 
\(\CRC{\rno}{\cha}{\cset}\) defined for \(\cset\subset\pdis{\cha}\) 
and \(\RC{\rno}{\Wm}\) defined for transition probability \(\Wm\).
We introduce those generalizations formally 
and discuss the implications of our analysis on them 
in Appendices \ref{sec:constrainedcapacity} and \ref{sec:generaldefinitions}.

Results of our analysis, also, encourage one to consider certain related problems:
\begin{itemize}
\item We do not assume any topological structure on the output space \(\outS\). 
Although this is a strength because of the generality of our results, 
it is also a weakness because of the obliviousness of our analysis towards the 
interactions between \renyi\!\!'s information measures and the topological structure of
the output space.
In almost all of the applications, \(\outA\) is a 
Borel or Baire \(\sigma\)-algebra of the topological space \((\outS,\tau)\); 
usually there is an even more specific structure.
In most of the applications, \(\outA\) is the Borel \(\sigma\)-algebra of
a complete separable metric space \((\outS,\cnst{d})\). 
Thus one can define metrics other than 
the total variation metric on \(\cha\) and \(\pdis{\cha}\) using the metric \(\cnst{d}\) and
analyze the behavior of \renyi\!\!'s information measures on the resulting 
topologies.
Such models have already been considered in the context of the arbitrarily varying 
channels \cite{csiszar92,thanh92} and the typicality 
\cite{jeon14,mitran15,raginsky13}.

\item 
It is easy to confirm that continuity of the order \(\rno\) \renyi capacity as a function of the order \(\rno\)
implies the continuity of the corresponding \(\fX\)-capacity \(\FC{\fX_{\rno}}{\cha}\) as a function of \(\rno\) 
where \(\fX_{\rno}(\dinp)=\frac{\dinp^{\rno}-1}{\rno-1}\).
The existence of similar, but more general, continuity results for richer classes of \(\fX\)-divergences
with appropriate topologies is expected.
What is plausible, but not evident, to us is the existence of a topology on the set of all convex \(\fX\)'s that 
ensures the continuity of the corresponding \(\fX\)-capacities in \(\fX\) for all \(\cha\) on the region that 
\(\fX\)-capacities are finite.
The interaction of topologies on the space of convex functions and corresponding 
\(\fX\)-capacities seems to be a fertile subject of inquiry. 

\item 
We use the definition of the \renyi information proposed by Sibson \cite{sibson69}.
In \eqref{eq:def:augustininformation} we provide the expression for the alternative definition 
of the \renyi information proposed by Augustin \cite{augustin78} and \csiszar \cite{csiszar95}.
We call this quantity the Augustin information.
Theorems \ref{thm:minimax}, \ref{thm:Cminimax}, \ref{thm:Gminimax}, 
and many of the other propositions have their 
analogues for the Augustin information, see \cite{nakiboglu18C,nakiboglu17}.
The Augustin capacity and center are of interest to us because they are 
better suited than the \renyi capacity and center for deriving 
the sphere packing bound for memoryless channels, see  \cite{nakiboglu18D,nakiboglu17}. 
\end{itemize}
We have avoided using information theoretic concepts such as  code, channel, or rate in our discussion 
because we believe \renyi\!\!'s information measures can and should be defined and understood on their own 
as measure theoretic concepts first. 
\renyi\!\!'s information measures, however, do have operational meaning in various information transmission problems. 
We discuss the case of channel coding problem in \cite{nakiboglu19B}.
\numberwithin{equation}{section}
\begin{appendix}
\subsection{The Constrained \renyi Capacity}\label{sec:constrainedcapacity}
\begin{definition}\label{def:constrainedRcapacity}
	For any \(\rno\!\in\![0,\infty]\), \(\cha\!\subset\!\pmea{\outA}\),
	\(\cset\!\subset\!\pdis{\cha}\),
	\emph{the order \(\rno\) \renyi capacity of \(\cha\) for constraint set \(\cset\)} is 
	\begin{align}
	\label{eq:def:constrainedcapacity}
	\CRC{\rno}{\cha}{\cset} 
	&\DEF \sup\nolimits_{\mP \in \cset}  \RMI{\rno}{\mP}{\cha}.
	\end{align}
\end{definition}
Note that \(\CRC{\rno}{\cha}{\pdis{\cha}}\!=\!\RC{\rno}{\cha}\) and
\(\CRC{\rno}{\cha}{\{\mP\}}\!=\!\RMI{\rno}{\mP}{\cha}\) for any \(\cha\) 
and \(\mP\!\in\!\pdis{\cha}\). 
%
Furthermore, the  proof of Theorem \ref{thm:minimax} works as is for any convex \(\cset\) subset
of \(\pdis{\cha}\), not just \(\pdis{\cha}\) itself.
Thus the minimax theorem continues to hold for \(\CRC{\rno}{\cha}{\cset}\);
the alternative expression for \(\CRC{\rno}{\cha}{\cset} \) is, however,
no longer (guaranteed to be) equal to the \renyi radius. 
\begin{theorem}\label{thm:Cminimax}
	For any \(\rno\!\in\!(0,\infty]\), \(\cha\!\subset\!\pmea{\outA}\), 
	and convex \(\cset\!\subset\!\pdis{\cha}\),
	\begin{align}
	\label{eq:thm:Cminimax:capacity}
	\CRC{\rno}{\cha}{\cset}
	&=\sup\nolimits_{\mP \in \cset} \inf\nolimits_{\mQ \in \pmea{\outA}} \RD{\rno}{\mP \mtimes \cha}{\mP \otimes \mQ}
	\\
	\label{eq:Cminimax}
	&=
	\inf\nolimits_{\mQ \in \pmea{\outA}} \sup\nolimits_{\mP \in \cset} \RD{\rno}{\mP \mtimes \cha}{\mP \otimes \mQ}.
	\end{align}
	If \(\CRC{\rno}{\cha}{\cset} <\infty\), then there exists a unique 
	\(\qmn{\rno,\cha,\cset}\) in \(\pmea{\outA}\),
	called the order \(\rno\) \renyi center for constraint set \(\cset\), such that
	\begin{align}
	\label{eq:Cminimaxcenter}
	\CRC{\rno}{\cha}{\cset} 
	&=\sup\nolimits_{\mP \in \cset} \RD{\rno}{\mP \mtimes \cha}{\mP \otimes \qmn{\rno,\cha,\cset}}.
	\end{align}
	Furthermore, for every sequence of priors \(\{\pmn{\ind}\}_{\ind\in\integers{+}}\subset \cset\) satisfying   
	\(\lim_{\ind \to \infty} \RMI{\rno}{\pmn{\ind}}{\cha}=\CRC{\rno}{\cha}{\cset} \),
	the corresponding sequence of order \(\rno\) \renyi means \(\{\qmn{\rno,\pmn{\ind}}\}_{\ind\in\integers{+}}\)  
	is a Cauchy sequence  for the total variation metric on \(\pmea{\outA}\) 
	and \(\qmn{\rno,\cha,\cset}\) is the unique limit point of that Cauchy sequence.
\end{theorem}
A similar modification is needed for the van Erven-\harremoes bound, i.e. for Lemma \ref{lem:EHB}, as well.
\begin{lemma}\label{lem:CEHB}  
	For any \(\rno\!\in\!(0,\infty]\), \(\cha\!\subset\!\pmea{\outA}\),
	convex  \(\cset\!\subset\!\pdis{\cha}\) satisfying \(\CRC{\rno}{\cha}{\cset} <\infty\),
	and \(\mQ \in \pmea{\outA}\)
	\begin{align}
	\notag
	\CRC{\rno}{\cha}{\cset} +\RD{\rno}{\qmn{\rno,\cha,\cset}}{\mQ}
	&\leq  \sup\nolimits_{\mP \in \cset} \RD{\rno}{\mP \mtimes \cha}{\mP \otimes\mQ}.
	\end{align}
\end{lemma}

Lemma \ref{lem:centercontinuity} establishing the continuity of the \renyi centers in the order holds for the
constrained \renyi centers. 
We prove it using Theorem \ref{thm:Cminimax} and Lemma \ref{lem:CEHB} instead of Theorem \ref{thm:minimax} and
Lemma \ref{lem:EHB}.

\subsection[Transition Probabilities and \(\RC{\rno}{\Wm}\)]{The \renyi Capacity of Transition Probabilities}\label{sec:generaldefinitions}
We have defined the order \(\rno\) \renyi information \(\RMI{\rno}{\mP}{\cha}\) for any \pmf \(\mP\) on a set of 
probability measures \(\cha\). 
We show in the following ---using the concept of transition probability and the expression for \(\RMI{\rno}{\mP}{\cha}\) 
given in  \eqref{eq:lem:information:defB}--- that for appropriately chosen \(\sigma\)-algebra \(\alg{W}\),
one can extend 
the definition of \(\RMI{\rno}{\mP}{\cha}\) to \(\mP\)'s that are probability measures on \((\cha,\alg{W})\).
Furthermore, we show that if \(\alg{W}\) is countably 
separated,\footnote{A \(\sigma\)-algebra \(\inpA\) on \(\inpS\) is countably separated, \cite[Def. 6.5.1]{bogachev}, 
	if there exists an at most countable collection sets \(\{\oev_{\ind}\}\subset\inpA\) separating the points of \(\inpS\). 
	A collection \(\{\oev_{\ind}\}\) of subsets of \(\inpS\) is said to be separating the points of \(\inpS\), if for
	every pair of distinct points \(\dsta\) and \(\dinp\) in \(\inpS\) there exists an \(\oev_{\ind}\) which includes
	only one of  \(\dsta\) and \(\dinp\). The Borel \(\sigma\)-algebra of any separable metric space is countably separated.
	The Borel \(\sigma\)-algebra of any separable metric space is also countably generated, i.e. it is the minimum 
	\(\sigma\)-algebra of a countable collection of sets.}
then Theorem \ref{thm:minimax} holds for this more general case, see Theorem \ref{thm:Gminimax}. 
\begin{definition}
	Let \((\inpS,\inpA)\) and \((\outS,\outA)\) be measurable  spaces. Then a function \(\Wm:\inpS\times\outA\to[0,1]\)
	is called a transition probability (a stochastic kernel / a Markov kernel) from \((\inpS,\inpA)\) to \((\outS,\outA)\) 
	if it satisfies the following two conditions: 
	\begin{enumerate}[(i)]
		\item For all \(\dinp\in\inpS\), the function \(\Wm(\cdot|\dinp):\outA\to[0,1]\) is a probability measure on \((\outS,\outA)\).
		\item For all \(\oev\in\outA\), the function \(\Wm(\oev|\cdot):\inpS\to[0,1]\) is a \(\inpA\)-measurable function. 
	\end{enumerate}
\end{definition} 

By \cite[Thm. 10.7.2.]{bogachev}, 
for any transition probability \(\Wm\) and probability measure \(\mP\) on \((\inpS,\inpA)\) 
there exists a unique probability measure \(\mP \mtimes \Wm\) on \((\inpS\times\outS,\inpA\otimes\outA)\) 
satisfying
\begin{align}
\notag
\mP \mtimes \Wm(\oev_{\dinp}\times\oev_{\dout})
&=\int_{\oev_{\dinp}} \Wm(\oev_{\dout}|\dinp) \mP(\dif{\dinp})
\end{align}
for all \(\oev_{\dinp}\in \inpA\) and \(\oev_{\dout}\in \outA\).
Now, we can define the order \(\rno\) \renyi information for \(\mP\) on the transition probability \(\Wm\).
\begin{definition}\label{def:Ginformation}
	For any \(\rno\in[0,\infty]\), 
	transition probability \(\Wm\) from \((\inpS,\inpA)\) to \((\outS,\outA)\),
	and \(\mP\in\pmea{\inpA}\),
	\emph{the order \(\rno\) \renyi information for prior \(\mP\)} is defined as
	\begin{align}
	\label{eq:def:Ginformation}
	\RMI{\rno}{\mP}{\Wm}
	&\DEF\inf\nolimits_{\mQ\in \pmea{\outA}}\RD{\rno}{\mP \mtimes \Wm}{\mP\otimes\mQ}.
	\end{align}
\end{definition}

Definitions \ref{def:information} and \ref{def:Ginformation} 
are equivalent because of  Lemma \ref{lem:information:def}. 
Using the definition of \(\RMI{\rno}{\mP}{\Wm}\) we can define the order \(\rno\) \renyi capacity of 
a transition probability \(\Wm\).

\begin{definition}\label{def:Gcapacity}
	For any \(\rno\in[0,\infty]\) and 
	transition probability \(\Wm\) from \((\inpS,\inpA)\) to \((\outS,\outA)\), 
	\emph{the order \(\rno\) \renyi capacity} is
	\begin{align}
	\label{eq:def:Gradius}
	\RC{\rno}{\Wm} 
	&\DEF \sup\nolimits_{\mP \in\pmea{\inpA}}  \RMI{\rno}{\mP}{\Wm}.
	\end{align}
\end{definition}

The analysis of the \renyi capacity for an arbitrary transition probability \(\Wm\) is beyond 
the scope of this paper. 
However, if the \(\sigma\)-algebra \(\inpA\) is countably separated, then we can use 
Theorem \ref{thm:minimax} to show that \(\RC{\rno}{\Wm}=\RC{\rno}{\cha}\) for a
\(\cha\subset\pmea{\outA}\).

\begin{theorem}\label{thm:Gminimax}
	For any  \(\rno\in(0,\infty]\) and transition probability
	\(\Wm\) from \((\inpS,\inpA)\) to \((\outS,\outA)\) for a 
	countably separated \(\sigma\)-algebra \(\inpA\)
	\begin{align}
	\label{eq:thm:Gminimax:capacity}
	\RC{\rno}{\Wm}
	&=\sup\nolimits_{\mP \in \pmea{\inpA}}\inf\nolimits_{\mQ \in \pmea{\outA}} \RD{\rno}{\mP \mtimes \Wm}{\mP \otimes \mQ}
	\\
	\label{eq:thm:Gminimax}	
	&=
	\inf\nolimits_{\mQ \in \pmea{\outA}} \sup\nolimits_{\mP \in \pmea{\inpA}} \RD{\rno}{\mP \mtimes \Wm}{\mP \otimes \mQ}
	\\
	\label{eq:thm:Gminimaxradius}
	&=\inf\nolimits_{\mQ\in\pmea{\outA}}\sup\nolimits_{\mW \in \cha} \RD{\rno}{\mW}{\mQ}
	\end{align}
	where \(\cha\DEF\{\Wm(\cdot|\dinp):\dinp\in\inpS\}\).
	If \(\RC{\rno}{\Wm}<\infty\), then there exists a unique \(\qmn{\rno,\Wm}\) in \(\pmea{\outA}\),
	called the order \(\rno\) \renyi center, such that
	\begin{align}
	\label{eq:Gminimaxcenter}
	\RC{\rno}{\Wm}
	&=\sup\nolimits_{\mP \in \pmea{\inpA}} \RD{\rno}{\mP \mtimes \Wm}{\mP \otimes \qmn{\rno,\Wm}}
	\\
	\label{eq:Gminimaxradiuscenter}
	&=\sup\nolimits_{\mW \in \cha} \RD{\rno}{\mW}{\qmn{\rno,\Wm}}.
	\end{align}
\end{theorem}

\begin{proof}[Proof of Theorem \ref{thm:Gminimax}]
	Since \(\inpA\) is countably separated, all singletons are in \(\inpA\) by \cite[Thm. 6.5.7]{bogachev} 
	and \(\pdis{\inpS}\subset\pmea{\inpA}\). 
	Consequently, using max-min inequality we get
	\begin{align}
	\notag
	&\sup\nolimits_{\mP \in \pdis{\cha}}\inf\nolimits_{\mQ \in \pmea{\outA}} \RD{\rno}{\mP \mtimes \cha}{\mP \otimes \mQ}
	\\
	\notag
	&\qquad\leq 
	\sup\nolimits_{\mP\in\pmea{\inpA}}\inf\nolimits_{\mQ\in\pmea{\outA}} \RD{\rno}{\mP \mtimes \Wm}{\mP \otimes \mQ}
	\\
	\label{eq:Gminimax-A}
	&\qquad\leq 
	\inf\nolimits_{\mQ\in\pmea{\outA}}\sup\nolimits_{\mP\in\pmea{\inpA}}\RD{\rno}{\mP \mtimes \Wm}{\mP \otimes \mQ}
	\end{align} 
	On the other hand, for any \(\rno\in(0,\infty]\) as a result of Tonelli-Fubini theorem \cite[4.4.5]{dudley}
	and the definition of the \renyi divergence given in \eqref{eq:def:divergence}
	we have
	\begin{align}
	\notag
	\RD{\rno}{\mP \mtimes \Wm}{\mP \otimes \mQ}
	&\leq \sup\nolimits_{\dinp\in\inpS} \RD{\rno}{\Wm(\cdot|\dinp)}{\mQ} 
	\\
	\label{eq:Gminimax-B}
	&=\sup\nolimits_{\mW\in\cha} \RD{\rno}{\mW}{\mQ}. 
	\end{align}
	Hence,
	\begin{align}
	\notag
	\inf\nolimits_{\mQ \in \pmea{\outA}}
	& \sup\nolimits_{\mP \in \pmea{\inpA}}  \RD{\rno}{\mP \mtimes \Wm}{\mP \otimes \mQ}\\
	\label{eq:Gminimax-C}
	&\qquad\leq \inf\nolimits_{\mQ \in \pmea{\outA}} \sup\nolimits_{\mW\in\cha} \RD{\rno}{\mW}{\mQ}.
	\end{align}
	Theorem \ref{thm:minimax} and the inequalities given in 
	\eqref{eq:Gminimax-A}, \eqref{eq:Gminimax-B}, and \eqref{eq:Gminimax-C}
	imply \(\RC{\rno}{\Wm}\!=\!\RC{\rno}{\cha}\) and 
	Theorem \ref{thm:Gminimax} for \(\qmn{\rno,\Wm}\!=\!\qmn{\rno,\cha}\).
\end{proof}

Theorem \ref{thm:minimax} and  \eqref{eq:Gminimax-B} imply that \(\RC{\rno}{\Wm}\leq\RC{\rno}{\cha}\)
even when \(\inpA\) is not countably separated.

\subsection{Deferred Proofs}\label{sec:deferred-proofs}
The following parametric function allows us to write certain expressions succinctly in the proofs:  
\begin{align}
\label{eq:def:binarydivergence}
\div{}{\dinp}{\dsta}
&\!\DEF\!\dinp \ln \tfrac{\dinp}{\dsta}+(1-\dinp)\ln \tfrac{1-\dinp}{1-\dsta}
&
&\forall \dinp,\dsta\in[0,1].
\end{align}

\begin{proof}[Proof of Lemma \ref{lem:finitecapacity}]~
	\begin{enumerate}
		\item[(\ref{finitecapacity-a})]
		For any \(\rno \in (0,1)\) the definitions of \(\RMI{\rno}{\mP}{\cha}\) and \(\RC{\rno}{\cha}\)
		imply 
		\(\inf\nolimits_{\mP\in \pdis{\cha}}\lon{\mmn{\rno,\mP}}=e^{\frac{\rno}{\rno-1}\RC{\rno}{\cha}}\). 
		
		\item[(\ref{finitecapacity-b})]
		\(\sup\nolimits_{\mP\in \pdis{\cha}}\lon{\mmn{\rno,\mP}}=e^{\frac{\rno}{\rno-1}\RC{\rno}{\cha}}\) for any  \(\rno \in (1,\infty)\)
		and \(\sup\nolimits_{\mP\in \pdis{\cha}}\lon{\mmn{\infty,\mP}}=e^{\RC{\infty}{\cha}}\)
		by the definitions of \(\RMI{\rno}{\mP}{\cha}\) and \(\RC{\rno}{\cha}\).
		
		\item[(\ref{finitecapacity-c})]   
		Let us first prove that if \(\RC{\rno}{\cha}<\infty\), 
		then \(\mmn{\rno,\mP}\) is uniformly continuous in \(\mP\).
		Lemma \ref{lem:powermeanP}-(\ref{powermeanP-e}) and the triangle inequality 
		imply 
		\begin{align}
		\notag
		\lon{\mmn{\rno,\pmn{1}}\!-\!\mmn{\rno,\pmn{2}}}
		&\leq \lon{\pmn{1}\!-\!\pmn{2}}^{\frac{1}{\rno}} 2^{\frac{\rno-1}{\rno}}(\lon{\mmn{\rno,\smn{1}}}\vee \lon{\mmn{\rno,\smn{2}}}).
		\end{align}
		for all \(\rno\) in \([1,\infty)\)  and \(\pmn{1}\), \(\pmn{2}\) in \(\pdis{\cha}\)
		where \(\smn{1}\) and \(\smn{2}\) are members of \(\pdis{\cha}\) determined by 
		the decomposition given in Lemma \ref{lem:powermeanP}-(\ref{powermeanP-c}).
		
		On the other hand \(\lon{\mmn{\rno,\mS}}\leq e^{\frac{\rno-1}{\rno}\RC{\rno}{\cha}}\)
		for any \(\mS\) in \(\pdis{\cha}\) by the proof of part (\ref{finitecapacity-b}). Thus
		\begin{align}
		\notag
		\lon{\mmn{\rno,\pmn{1}}-\mmn{\rno,\pmn{2}}}
		&\leq \lon{\pmn{1}-\pmn{2}}^{\frac{1}{\rno}}  e^{\frac{\rno-1}{\rno} (\RC{\rno}{\cha}+\ln 2)}
		\end{align}
		for all \(\rno\) in \([1,\infty)\)  and \(\pmn{1}\), \(\pmn{2}\) in \(\pdis{\cha}\) 
		Consequently, if \(\RC{\rno}{\cha}\) is finite, then \(\mmn{\rno,\mP}\) is uniformly 
		continuous in \(\mP\).
		
		We are left with proving that  \(\RC{\rno}{\cha}=\infty\) implies the absence of 
		uniformly continuity in \(\mP\) for \(\mmn{\rno,\mP}\). 
		For any \(\mS\) in \(\pdis{\cha}\) let \(\smn{\ind}\) be
		\begin{align}
		\notag
		\smn{\ind}
		&=(1-\tfrac{1}{\ind})\mS+\tfrac{1}{\ind}\pmn{\ind} 
		&
		&\forall \ind \in \integers{+}
		\end{align}
		where  \(\pmn{\ind}\)'s are such that \(\lon{\mmn{\rno,\pmn{\ind}}} \geq \ind\). 
		The existence of such \(\pmn{\ind}\)'s follows from part (\ref{finitecapacity-b}). 
		Then \(\tfrac{\mmn{\rno,\pmn{\ind}}}{\ind^{1/\rno}}\leq \mmn{\rno,\smn{\ind}}\)
		by the definition of mean measure.
		Thus \(\lon{\mmn{\rno,\smn{\ind}}}\geq \ind^{\frac{\rno-1}{\rno}}\)
		and using the triangle inequality we get 
		\begin{align}
		\notag
		\lon{\mmn{\rno,\smn{\ind}}-\mmn{\rno,\mS}} 
		&\geq \ind^{\frac{\rno-1}{\rno}}-\lon{\mmn{\rno,\mS}}.
		\end{align}
		On the other hand, \(\lon{\mS-\smn{\ind}}\leq \sfrac{2}{\ind}\)
		by the triangle inequality, as well.
		Thus \(\lon{\mmn{\rno,\mP}\!-\!\mmn{\rno,\mS}}\) is an unbounded function of \(\mP\) on every neighborhood of \(\mS\),
		i.e.  \(\mmn{\rno,\mP}\) is not continuous at \(\mP=\mS\).
		
		\item[(\ref{finitecapacity-d})]
		If \(\RC{\rno}{\cha}\) is infinite, there is a sequence of
		\(\{\pmn{\ind}\}_{\ind\in\integers{+}}\) such that 
		\(\lim_{\ind\uparrow\infty}\RMI{\rno}{\pmn{\ind}}{\cha}=\infty\).
		Let \(\pmn{\rnb,\ind}=(1-\rnb)\mP+\rnb\pmn{\ind}\),
		for any \(\mP\).
		Then the concavity of the order \(\rno\) \renyi  information 
		in the prior for \(\rno\)'s in \([1,\infty]\), established in 
		Lemma \ref{lem:informationP}-(\ref{informationP-b}),
		and the non-negativity of the \renyi information imply
		\begin{align}
		\notag
		\RMI{\rno}{\pmn{\rnb,\ind}}{\cha}-\RMI{\rno}{\mP}{\cha}
		&\geq \rnb(\RMI{\rno}{\pmn{\ind}}{\cha}-\RMI{\rno}{\mP}{\cha})
		\end{align}
		for all \(\rnb\in(0,1)\) and \(\ind\in\integers{+}\).
		On the other hand \(\lon{\mP-\pmn{\rnb,\ind}}\leq 2\rnb\). 
		Thus \(\RMI{\rno}{\mP}{\cha}\) is not continuous in \(\mP\),
		whenever \(\RC{\rno}{\cha}\) is infinite. 
		The continuity of  \(\RMI{\rno}{\mP}{\cha}\) in
		\(\mP\) for the case when \(\RC{\rno}{\cha}\) is finite follows from
		part (\ref{finitecapacity-uecP}).
		
		\item[(\ref{finitecapacity-uecP})]
		We establish the uniform equicontinuity by proving establishing the following bound 
		\begin{align}
		\notag
		&\hspace{-.7cm}\sup_{\rno\in [0,\rnt]}  \abs{\RMI{\rno}{\pmn{2}}{\cha}\!-\!\RMI{\rno}{\pmn{1}}{\cha}}
		\\
		&\!\leq\! 
		\begin{cases}
		\ln(\!\tfrac{1}{1-\delta}\!\wedge\!\tfrac{e^{\RC{0}{\cha}}}{\delta}\!)
		\!+\!\ln(1\!-\!\delta\!+\!\delta e^{\RC{0}{\cha}})\!\!
		&\rnt\!=\!0
		\\
		\ln\tfrac{1-\delta+\delta e^{\RC{\rnt}{\cha}}}{\left[ (1-\delta)^{\frac{1}{\rnt}}
			+\delta^{\frac{1}{\rnt}} e^{\frac{\rnt-1}{\rnt}\RC{\rnt}{\cha}}
			\right]^{\frac{\rnt}{1-\rnt}}}
		&\rnt\!\in\!\reals{+}\!\!\setminus\!\{\!1\!\} 
		\\
		\hX_{1}(\delta)\!+\!\delta\RC{1}{\cha}\!+\!\ln(1\!-\!\delta\!+\!\delta e^{\RC{1}{\cha}})\!\!
		&
		\rnt\!=\!1
		\end{cases}
		\label{eq:lem:uecP}
		\end{align}
		where \(\delta=\tfrac{\lon{\pmn{1}-\pmn{2}}}{2}\) and \(\hX_{\rno}(\cdot)\) is defined in 
		\eqref{eq:def:binaryentropy}.
		
		As a result of the decomposition given Lemma \ref{lem:powermeanP}-(\ref{powermeanP-c}) 
		we can write \(\pmn{1}\) as \(\pmn{1}=(1-\delta)\smn{\wedge}+\delta\smn{1}\)
		for some \(\smn{\wedge}\) and \(\smn{1}\) in \(\pdis{\cha}\).
		Using \eqref{eq:def:divergence}, \eqref{eq:sibson}, and \eqref{eq:lem:information:defA} 
		we get 
		\begin{align}
		\notag
		\hspace{-.6cm}
		\RMI{1}{\pmn{1}}{\cha}
		&\!=\!(1-\delta)\RMI{1}{\smn{\wedge}}{\cha}+(1-\delta)\RD{1}{\qmn{1,\smn{\wedge}}}{\qmn{1,\pmn{1}}}
		\\
		&~\qquad~+\delta \RMI{1}{\smn{1}}{\cha}     +   \delta \RD{1}{\qmn{1,\smn{1}}}{\qmn{1,\pmn{1}}}.
		\label{eq:informationdecomposition-one}
		\end{align}
		Similarly for positive orders other than one we have,
		\begin{align}
		\notag
		\hspace{-.6cm}
		\RMI{\rno}{\pmn{1}}{\cha}
		\!=\!\tfrac{1}{\rno-1}\!\!\ln&\!
		\left[\!(1\!-\!\delta\!)e^{(\rno-1)\left[\RMI{\rno}{\smn{\wedge}}{\cha}+\RD{\rno}{\qmn{\rno,\smn{\wedge}}}{\qmn{\rno,\pmn{1}}}\right]}
		\right.
		\\
		&~\left. 
		+\delta e^{(\rno-1)\left[\RMI{\rno}{\smn{1}}{\cha}     +\RD{\rno}{\qmn{\rno,\smn{1}}}{\qmn{\rno,\pmn{1}}}\right]}\!\right]\!\!
		\label{eq:informationdecomposition-oto}.
		\end{align}
		\hspace{-.1cm}Since the \renyi divergence is non-negative by Lemma \ref{lem:divergence-pinsker},
		\begin{align}
		\notag 
		\RMI{\rno}{\pmn{1}}{\cha}
		&\!\geq\!
		\begin{cases}
		(1-\delta)\RMI{1}{\smn{\wedge}}{\cha}+\delta \RMI{1}{\smn{1}}{\cha}
		&\!
		\rno\!=\!1
		\\
		\tfrac{\ln\left[(1-\delta) e^{(\rno-1)\RMI{\rno}{\smn{\wedge}}{\cha}} +\delta e^{(\rno-1)\RMI{\rno}{\smn{1}}{\cha}} \right]}{\rno-1}
		&\!
		\rno\!\neq\!1
		\end{cases}
		\\
		\notag
		&\!\geq\!\RMI{\rno}{\smn{\wedge}}{\cha}-\gX(\delta,\rno,\RMI{\rno}{\smn{\wedge}}{\cha}-\RMI{\rno}{\smn{1}}{\cha})
		\end{align}		
		where the function \(\gX(\delta,\rno,\gamma)\) is defined for any \(\delta\in [0,1]\), \(\rno\in\reals{+}\), and \(\gamma \in\reals{}\) as follows
		\begin{align}
		\notag
		\gX(\delta,\rno,\gamma)
		&\!\DEF\!\begin{cases}
		\delta \gamma
		&
		\rno\!=\!1
		\\
		\tfrac{1}{1-\rno}\ln\left[(1-\delta) +\delta e^{(1-\rno) \gamma }  \right]
		&
		\rno\!\neq\!1
		\end{cases}.
		\end{align}
		Given \(\delta\) and \(\gamma\), 
		\(\gX(\delta,\rno,\gamma)\) is 
		nonincreasing\footnote{For any fixed \((\delta,\gamma)\) pair, 
			\(\gX(\delta,\rno,\gamma)\) is a continuous and differentiable function of \(\rno\)
			satisfying \(\pder{}{\rno}\gX(\delta,\rno,\gamma)\leq 0\). In particular
			\begin{align}
			\notag
			\pder{}{\rno}\gX(\delta,\rno,\gamma)
			&=\tfrac{-1}{(1-\rno)^2}\div{}{\tfrac{\delta e^{(1-\rno) \gamma}}{(1-\delta)+\delta e^{(1-\rno) \gamma}}}{\delta}.
			\end{align}} 
		in \(\rno\). Then 
		\begin{align}
		\notag
		\RMI{\rno}{\pmn{1}}{\cha}
		&\geq
		\RMI{\rno}{\smn{\wedge}}{\cha}-\gX(\delta,0,\RMI{\rno}{\smn{\wedge}}{\cha}-\RMI{\rno}{\smn{1}}{\cha}) 
		\end{align}
		for all \(\rno\) in \((0,\rnt]\). 
		Furthermore, given \(\delta\) and \(\rno\), \(\gX(\delta,\rno,\gamma)\) is 
		nondecreasing in \(\gamma\). 
		Then using \(\RMI{\rno}{\smn{1}}{\cha}\geq 0\),
		\(\RMI{\rno}{\smn{\wedge}}{\cha}\leq \RMI{\rnt}{\smn{\wedge}}{\cha}\),
		and \(\RMI{\rnt}{\smn{\wedge}}{\cha}\leq \RC{\rnt}{\cha}\)
		we get
		\begin{align}
		\hspace{-.3cm}
		\label{eq:divergencelowerboundF}
		\RMI{\rno}{\pmn{1}}{\cha}
		&\geq
		\RMI{\rno}{\smn{\wedge}}{\cha}\!-\!\gX(\delta,0,\RC{\rnt}{\cha})
		\end{align}
		for all \(\rno\!\in\!(0,\rnt]\).
		On the other hand, \(\pmn{2}=(1-\delta)\smn{\wedge}+\delta\smn{2}\) 
		by the decomposition 
		given in Lemma \ref{lem:powermeanP}-(\ref{powermeanP-c}).
		Then
		\begin{align}
		\notag
		\!(1\!-\!\delta)^{\frac{1}{\rno}} \mmn{\rno,\smn{\wedge}}\leq \mmn{\rno,\pmn{2}}
		\end{align}
		as a result of the definition of the mean measure.
		Thus
		\begin{align}
		\notag
		e^{\frac{\rno-1}{\rno}(\RMI{\rno}{\smn{\wedge}}{\cha}-\RMI{\rno}{\pmn{2}}{\cha})}\!(1\!-\!\delta)^{\frac{1}{\rno}} \qmn{\rno,\smn{\wedge}}\leq \qmn{\rno,\pmn{2}}
		\end{align}
		by \eqref{eq:def:information} and \eqref{eq:def:mean}. Applying Lemma \ref{lem:divergence-RM} we get
		\begin{align}
		\notag
		\RD{\rno\!}{\qmn{\rno,\smn{\wedge}}}{\qmn{\rno,\pmn{2}}}
		&\leq  \RD{\rno\!}{\qmn{\rno,\smn{\wedge}}}{\!(1\!-\!\delta)^{\frac{1}{\rno}} \mmn{\rno,\smn{\wedge}}}
		\\
		\notag
		&=\tfrac{(1-\rno)(\RMI{\rno}{\smn{\wedge}}{\cha}-\RMI{\rno}{\pmn{2}}{\cha})-\ln(1-\delta)}{\rno}
		\end{align}
		for all \(\rno\)'s in \(\reals{+}\).
		Using the corresponding upper bound on 
		\(\RD{\rno}{\qmn{\rno,\smn{2}}}{\qmn{\rno,\pmn{2}}}\) together 
		with \eqref{eq:informationdecomposition-one} and 
		\eqref{eq:informationdecomposition-oto} we get the following 
		bound for all positive real orders
		\begin{align}
		\notag
		\hspace*{-.2cm}
		\RMI{\rno}{\pmn{2}}{\cha}
		&\!\leq\!
		\begin{cases}
		(1-\delta)\RMI{1}{\smn{\wedge}}{\cha}\!+\!\delta \RMI{1}{\smn{2}}{\cha}
		+\hX_{1}(\delta) 
		&
		\rno=1
		\\
		\tfrac{\rno\ln \left[\!(1-\delta)^{\frac{1}{\rno}} e^{\frac{\rno-1}{\rno}\RMI{\rno}{\smn{\wedge}}{\cha}}  
			\!+\!\delta^{\frac{1}{\rno}} e^{\frac{\rno-1}{\rno}\RMI{\rno}{\smn{2}}{\cha}}  
			\!\right]}{\rno-1}
		&
		\rno\neq 1
		\end{cases}
		\\
		\notag
		&=\RMI{\rno}{\smn{\wedge}}{\cha}+\fX(\delta,\rno,\RMI{\rno}{\smn{2}}{\cha}-\RMI{\rno}{\smn{\wedge}}{\cha})
		\end{align}
		where the function \(\fX(\delta,\rno,\gamma)\) is defined 	
		for any \(\delta\in [0,1]\), \(\rno\in\reals{+}\), and \(\gamma \in\reals{}\) 
		as follows
		\begin{align}
		\notag
		\fX(\delta,\rno,\gamma)
		&\!\DEF\!\begin{cases}
		\delta \gamma +\hX_{1}(\delta)
		&\rno=1
		\\
		\tfrac{\rno}{\rno-1}\ln \left[(1-\delta)^{\frac{1}{\rno}}+\delta^{\frac{1}{\rno}}e^{\frac{\rno-1}{\rno}\gamma}\right]
		&
		\rno\neq 1
		\end{cases}.
		\end{align}
		For any fixed \((\delta,\gamma)\) pair,  \(\fX(\delta,\rno,\gamma)\) is 
		nondecreasing\footnote{For any fixed \((\delta,\gamma)\) pair, \(\fX(\delta,\rno,\gamma)\) 
			is a continuous and differentiable function of \(\rno\) satisfying
			\(\pder{}{\rno}\fX(\delta,\rno,\gamma)\geq 0\). In particular
			\begin{align}
			\notag
			\pder{}{\rno}\fX(\delta,\rno,\gamma)
			&=\tfrac{1}{(1-\rno)^2}\div{}{\tfrac{(1-\delta)^{\sfrac{1}{\rno}}}{(1-\delta)^{\sfrac{1}{\rno}}+\delta^{\sfrac{1}{\rno}}e^{(1-\sfrac{1}{\rno})\gamma}}}{1-\delta}.
			\end{align}}
		in \(\rno\). Then for any \(\rno\) in \((0,\rnt]\) we  have
		\begin{align}
		\notag
		\RMI{\rno}{\pmn{2}}{\cha}
		&\leq
		\RMI{\rno}{\smn{\wedge}}{\cha}+\fX(\delta,\rnt,\RMI{\rno}{\smn{2}}{\cha}-\RMI{\rno}{\smn{\wedge}}{\cha}). 
		\end{align}
		Furthermore, given \(\delta\) and \(\rno\),  \(\fX(\delta,\rno,\gamma)\) is 
		nondecreasing in \(\gamma\). Then using 
		\(\RMI{\rno}{\smn{\wedge}}{\cha}\geq 0\),
		\(\RMI{\rno}{\smn{2}}{\cha}\leq \RMI{\rnt}{\smn{2}}{\cha}\),
		and 
		\(\RMI{\rnt}{\smn{2}}{\cha}\leq \RC{\rnt}{\cha}\)
		we get
		\begin{align}
		\label{eq:divergenceupperboundH}
		\RMI{\rno}{\pmn{2}}{\cha}
		&\leq
		\RMI{\rno}{\smn{\wedge}}{\cha}+\fX(\delta,\rnt,\RC{\rnt}{\cha})
		\end{align}
		for all \(\rno\) in \((0,\rnt]\). 
		Using \eqref{eq:divergencelowerboundF} and \eqref{eq:divergenceupperboundH} 
		together with the definition of 
		the \renyi capacity given in \eqref{eq:def:capacity} we get
		\begin{align}
		\notag
		\RMI{\rno}{\pmn{2}}{\cha}\!-\!\RMI{\rno}{\pmn{1}}{\cha}
		&\!\leq\!\fX(\delta,\rnt,\RC{\rnt}{\cha})\!+\!\gX(\delta,0,\RC{\rnt}{\cha}).
		\end{align}
		A lower bound on \(\RMI{\rno}{\pmn{2}}{\cha}\!-\!\RMI{\rno}{\pmn{1}}{\cha}\) can be obtained using the same arguments with the roles 
		of \(\pmn{1}\) and \(\pmn{2}\) reversed. This establishes \eqref{eq:lem:uecP} 
		for \(\rnt>0\) and \(\rno\in(0,\rnt]\).

		In order to establish \eqref{eq:lem:uecP} for \(\rno=0\), 
		recall the definition of the order zero \renyi information given in \eqref{eq:def:information}.
		\begin{align}
		\notag
		\RMI{0}{\pmn{1}}{\cha}
		&=-\ln\essup_{\mmn{1,\pmn{1}}} 
		\left(
		(1-\delta)\!\sum\nolimits_{\mW:\smn{\wedge}(\mW|\dout)>0}\!\!\!\smn{\wedge}(\mW)
		\right.
		\\
		\notag
		&\qquad~\qquad~\qquad~\qquad\left. 
		+\delta\!\sum\nolimits_{\mW:\smn{1}(\mW|\dout)>0}\!\!\!\smn{1}(\mW) \right)
		\\
		\notag
		&\geq -\ln\left((1-\delta)e^{-\RMI{0}{\smn{\wedge}}{\cha}}+\delta \right)
		\\
		\notag
		&=  \RMI{0}{\smn{\wedge}}{\cha} -\ln\left( 1-\delta+\delta e^{\RMI{0}{\smn{\wedge}}{\cha}} \right).
		\end{align}
		Note that \(\RMI{0}{\smn{\wedge}}{\cha}\leq \RMI{\rnt}{\smn{\wedge}}{\cha}\) 
		by Lemma \ref{lem:informationO}
		and \(\RMI{\rnt}{\smn{\wedge}}{\cha}\leq \RC{\rnt}{\cha}\) by definition. Then
		\begin{align}
		\label{eq:divergencelowerboundzero}
		\RMI{0}{\pmn{1}}{\cha}
		&\geq   \RMI{0}{\smn{\wedge}}{\cha}
		-\ln\left( 1-\delta+\delta e^{\RC{\rnt}{\cha}} \right).
		\end{align}
		On the other hand,
		\begin{align}
		\notag
		\RMI{0}{\pmn{2}}{\cha}
		&=\esinf\limits_{\mmn{1,\pmn{2}}}
		\ln\tfrac{1}{\sum\nolimits_{\mW:\pmn{2}(\mW|\dout)>0}\left((1-\delta)\smn{\wedge}(\mW)+\delta\smn{2}(\mW) \right)} 
		\\
		\notag
		&\leq 
		\left(\esinf\nolimits_{\mmn{1,\smn{\wedge}}}\ln \tfrac{1}{(1-\delta)\sum\nolimits_{\mW:\smn{\wedge}(\mW|\dout)>0} \smn{\wedge}(\mW)}\right)
		\\
		\notag
		&~\qquad~\wedge
		\left(\esinf\nolimits_{\mmn{1,\smn{2}}}\ln \tfrac{1}{\delta\sum\nolimits_{\mW:\smn{2}(\mW|\dout)>0} \smn{2}(\mW)}\right)
		\\
		\notag
		&=
		(\RMI{0}{\smn{\wedge}}{\cha}+\ln\tfrac{1}{1-\delta})
		\wedge
		(\RMI{0}{\smn{2}}{\cha}+\ln\tfrac{1}{\delta})
		\end{align}
		Then \(\RMI{0}{\smn{\wedge}}{\cha}\!\geq\!0\) and \(\RMI{0}{\smn{2}}{\cha}\!\leq\!\RMI{\rnt}{\smn{2}}{\cha}\!\leq\!\RC{\rnt}{\cha}\) imply
		\begin{align}
		\label{eq:divergenceupperboundzero} 
		\RMI{0}{\pmn{2}}{\cha}
		&\leq\RMI{0}{\smn{\wedge}}{\cha}
		+\left(\ln\tfrac{1}{1-\delta} 
		\wedge\ln\tfrac{e^{\RC{\rnt}{\cha}}}{\delta}
		\right).
		\end{align}
		Thus using  \eqref{eq:divergencelowerboundzero} and \eqref{eq:divergenceupperboundzero} we get
		\begin{align}
		\notag
		\hspace{-.2cm}
		\RMI{0}{\pmn{2}}{\cha}\!-\!\RMI{0}{\pmn{1}}{\cha}
		&\leq \ln\left(1-\delta+\delta e^{\RC{\rnt}{\cha}} \right)\\
		\label{eq:informationdecompositionboundzero}
		&~\qquad
		+\ln\left(\tfrac{1}{1-\delta} 
		\wedge\tfrac{e^{\RC{\rnt}{\cha}}}{\delta}
		\right).
		\end{align}
		
		A lower bound on \(\RMI{0}{\pmn{2}}{\cha}-\RMI{0}{\pmn{1}}{\cha}\) can be obtained using the same arguments with the roles 
		of \(\pmn{1}\) and \(\pmn{2}\) reversed.
		Consequently,  \eqref{eq:lem:uecP} holds for \(\rnt=0\), \(\rno=0\) case. 
		In order to  establish \eqref{eq:lem:uecP} for \(\rnt>0\), \(\rno=0\) case,  note that
		\begin{align}
		\notag
		\tfrac{\rnt \ln \left[ (1-\delta)^{\frac{1}{\rnt}}   
			+\delta^{\frac{1}{\rnt}} e^{\frac{\rnt-1}{\rnt}\RC{\rnt}{\cha}}
			\right]}{\rnt-1}
		&\geq 
		\begin{cases}
		\ln\tfrac{1}{1-\delta}
		&
		\tfrac{1-\delta}{\delta e^{-\RC{\rnt}{\cha}}}\geq 1
		\\
		\ln\tfrac{e^{\RC{\rnt}{\cha}}}{\delta}
		&
		\tfrac{1-\delta}{\delta e^{-\RC{\rnt}{\cha}}}\leq 1
		\end{cases}
		\\
		\notag
		&\geq
		\left[
		\ln\tfrac{1}{1-\delta}
		\wedge
		\ln\tfrac{e^{\RC{\rnt}{\cha}}}{\delta }
		\right]
		\end{align}
		Thus \eqref{eq:lem:uecP} holds for \(\rnt>0\), \(\rno=0\) case, as well. 
		
		\item[(\ref{finitecapacity-uecO})]
		In order to establish the uniform equicontinuity 
		we prove the Lipschitz continuity of \(\{\RMI{\rno}{\mP}{\cha}\}_{\mP\in\pdis{\cha}}\)
		in \(\rno\) on compact subsets of \((0,\rnt)\) with a common Lipschitz constant:
		If \(\rno\) and \(\rnf\) in 
		\([\epsilon,\rnt\!-\!\epsilon]\)
		for an \(\epsilon\in (0,\epsilon_{\rnt}]\),
		then 
		\begin{align}
		\label{eq:lem:uecO}		
		\abs{\RMI{\rno}{\mP}{\cha}\!-\!\RMI{\rnf}{\mP}{\cha}} 
		&\leq \tfrac{\gamma_{\rnt}}{\epsilon^2}\abs{\rno-\rnf}
		\end{align}
		for all \(\mP\) in \(\pdis{\cha}\)
		where \(\epsilon_{\rnt}\) and \(\gamma_{\rnt}\) are defined as follows
		\begin{align}
		\notag
		\epsilon_{\rnt}
		&\!\DEF\!\begin{cases}
		\tfrac{\rnt}{2}
		&\rnt\in(0,1] 
		\\
		\tfrac{\rnt-1}{8\rnt}
		&\rnt\in(1,\infty)
		\end{cases},
		\\
		\notag
		\gamma_{\rnt}
		&\!\DEF\!\begin{cases}
		\RC{\rnt}{\cha}
		&\rnt\in(0,1] 
		\\
		\rnt\RC{\rnt}{\cha}	+\tfrac{5e^{2\RC{\rnt}{\cha}}}{2 e^{2}}
		&\rnt\in(1,\infty)
		\end{cases}.
		\end{align} 
		Since \(\lon{\mmn{\rno,\mP}}^{\rno}\) is a log-convex in \(\rno\) 
		by Lemma \ref{lem:powermeanO}-(\ref{powermeanO-d}),
		\begin{align}
		\notag
		\lon{\mmn{\rno,\mP}}^{\rno}
		&\leq  
		\lon{\mmn{\rnb,\mP}}^{\rnb \frac{\rno-\rnf}{\rnb-\rnf}}
		\lon{\mmn{\rnf,\mP}}^{\rnf\frac{\rnb-\rno}{\rnb-\rnf}}.
		\end{align}
		for any \(\rnf\), \(\rno\), \(\rnb\) satisfying \(0<\rnf<\rno<\rnb\) 
		and \(\mP\in\pdis{\cha}\).

		Let us start with \(\rnt\in(0,1]\) and \(\epsilon\in (0,\tfrac{\rnt}{2}]\) case. 
		Then for any \(\rnf\), \(\rno\), \(\rnb\) satisfying
		\(0<\rnf<\epsilon\leq \rno<\rnb\leq\rnt-\epsilon\), 
		\begin{align}
		\notag
		\RMI{\rnb}{\mP}{\cha}-\RMI{\rno}{\mP}{\cha}
		&=\tfrac{1}{1-\rno}\ln \tfrac{\lon{\mmn{\rno,\mP}}^{\rno}}{\lon{\mmn{\rnb,\mP}}^{\frac{\rnb(1-\rno)}{1-\rnb}}}
		\\
		\notag
		&\leq \tfrac{1}{1-\rno}
		\ln \tfrac{\lon{\mmn{\rnb,\mP}}^{\rnb \frac{\rno-\rnf}{\rnb-\rnf}}\lon{\mmn{\rnf,\mP}}^{\rnf\frac{\rnb-\rno}{\rnb-\rnf}}}{\lon{\mmn{\rnb,\mP}}^{\frac{\rnb(1-\rno)}{1-\rnb}}}
		\\
		\notag
		&=\tfrac{(\rnb-\rno)(1-\rnf) }{(1-\rno)(\rnb-\rnf)} 
		(\RMI{\rnb}{\mP}{\cha}-\RMI{\rnf}{\mP}{\cha})
		\\
		\notag
		&\leq \tfrac{\rnb-\rno}{\epsilon(\epsilon-\rnf)} 
		\RMI{\rnb}{\mP}{\cha}.
		\end{align}
		The above bound holds for any \(\rnf\) in \((0,\epsilon)\). Furthermore, 
		the \renyi information is a nondecreasing function of the order by 
		Lemma \ref{lem:informationO}. Then 
		\begin{align}
		\label{eq:finitecapacity-uecO-A}
		0\leq\RMI{\rnb}{\mP}{\cha}-\RMI{\rno}{\mP}{\cha}
		&\leq \tfrac{\RC{\rnt}{\cha}}{\epsilon^{2}}(\rnb-\rno)
		\end{align}
		for any \(\mP\) in \(\pdis{\cha}\) and \(\rnb\),  \(\rno\) satisfying \(\epsilon\leq\rno\leq \rnb\leq\rnt-\epsilon\). 
		
		\eqref{eq:finitecapacity-uecO-A} establishes \eqref{eq:lem:uecO} for 
		\(\rnt\in(0,1]\) and \(\epsilon\in (0,\tfrac{\rnt}{2}]\) case.
		
		We proceed with \(\rnt\in (1,\infty)\) and \(\epsilon\in (0,\tfrac{\rnt-1}{8\rnt}]\) case.
		For any \(\rnf\), \(\rno\), \(\rnb\) such that \(1+\epsilon\leq \rnf<\rno\leq\rnb-\epsilon\) 
		and \(\mP\) in \(\pdis{\cha}\) we have
		\begin{align}
		\notag
		\RMI{\rno}{\mP}{\cha}-\RMI{\rnf}{\mP}{\cha}
		&=\tfrac{1}{\rno-1}\ln \tfrac{\lon{\mmn{\rno,\mP}}^{\rno}}{\lon{\mmn{\rnf,\mP}}^{\frac{\rnf(\rno-1)}{\rnf-1}}}
		\\
		\notag
		&\leq \tfrac{1}{\rno-1}
		\ln \tfrac{\lon{\mmn{\rnb,\mP}}^{\rnb \frac{\rno-\rnf}{\rnb-\rnf}}\lon{\mmn{\rnf,\mP}}^{\rnf\frac{\rnb-\rno}{\rnb-\rnf}}}{{\lon{\mmn{\rnf,\mP}}^{\frac{\rnf(\rno-1)}{\rnf-1}}}}
		\\
		\notag
		&=\tfrac{(\rno-\rnf)(\rnb-1)}{(\rno-1)(\rnb-\rnf)} 
		(\RMI{\rnb}{\mP}{\cha}-\RMI{\rnf}{\mP}{\cha})
		\\
		\label{eq:finitecapacity-uecO-B}
		&\leq \tfrac{(\rno-\rnf)}{\epsilon^{2}} 
		\rnb\RC{\rnb}{\cha}.
		\end{align}
		If \(0<\rno-\rnf<\epsilon\), then at least one of the three closed intervals \([\epsilon,1-\epsilon]\), \([\tfrac{1}{2},\tfrac{5\rnt-1}{4\rnt}]\), \([1+\epsilon,\rnt-\epsilon]\)
		includes both \(\rno\) and \(\rnf\). 
		When  \(\rno\) and \(\rnf\) are in \([\epsilon,1-\epsilon]\) we use \eqref{eq:finitecapacity-uecO-A}.
		When  \(\rno\) and \(\rnf\) are in \([1+\epsilon,\rnt-\epsilon]\) we use \eqref{eq:finitecapacity-uecO-B}.
		Derivation of the bound for the second interval takes some effort. 
		Let us first finish the proof of \eqref{eq:lem:uecO} 
		assuming that the bound given in   \eqref{eq:finitecapacity-uecO-C} holds for the second interval. 
		Then for any \(\rnf\), \(\rno\) such that \(\epsilon\leq \rnf \leq \rno\leq (\rnf+\epsilon)\wedge (\rnt-\epsilon)\) we have
		\begin{align}
		\notag
		\hspace{-.4cm}
		0&\!\leq\!\tfrac{\RMI{\rno}{\mP}{\cha}-\RMI{\rnf}{\mP}{\cha}}{\rno-\rnf}
		\\
		\notag
		&\!\leq\!
		\begin{cases}
		\!\tfrac{\RC{\rnt}{\cha}}{\epsilon^{2}} 
		&\epsilon\!\leq\!\rnf\!\leq\!\rno\!\leq\!(\rnf\!+\!\epsilon)\!\wedge\!(1\!-\!\epsilon)
		\\
		\!4\RC{\rnt}{\cha}\!+\!\tfrac{160\rnt^{2}e^{2\RC{\rnt}{\cha}}}{e^{2}(\rnt-1)^{2}}  
		&\tfrac{1}{2}\!\leq\!\rnf\!\leq\!\rno\!\leq\!(\rnf\!+\!\epsilon)\!\wedge\!\tfrac{5\rnt-1}{4\rnt}
		\\
		\!\tfrac{\rnt \RC{\rnt}{\cha}}{\epsilon^{2}} 
		&1\!+\!\epsilon\!\leq\!\rnf\!\leq\!\rno\!\leq\!(\rnf\!+\!\epsilon)\!\wedge\!(\rnt\!-\!\epsilon)
		\end{cases}
		\end{align}
		Thus for any \(\rno\) and \(\rnf\)  satisfying 
		\(\epsilon\!\leq\!\rnf\!\leq\!\rno\!\leq(\rnf\!+\!\epsilon)\!\wedge\! (\rnt\!-\!\epsilon)\)
		and \(\mP\) in \(\pdis{\cha}\) we have
		\begin{align}
		\notag
		\RMI{\rno}{\mP}{\cha}-\RMI{\rnf}{\mP}{\cha}
		&\leq \tfrac{\rno-\rnf}{\epsilon^2} \left[
		\rnt \RC{\rnt}{\cha}+\tfrac{5e^{2\RC{\rnt}{\cha}}}{2 e^{2}}
		\right].
		\end{align}
		Note that the preceding bound is linear with a uniform constant, 
		thus the hypothesis \(\rnf\leq\rno\leq\rnf+\epsilon\) can be removed
		without loss of generality.
		Thus \eqref{eq:lem:uecO} holds for \(\rnt\in (1,\infty)\) case 
		for any \(\epsilon\in(0,\epsilon_{\rnt}]\), as well.
		
		We are left with establishing the bound given in \eqref{eq:finitecapacity-uecO-C}. 
		For orders other than one,  (\ref{def:information}) and \eqref{eq:lem:informationOder} imply that
		\begin{align}
		\notag
		\der{}{\rno} \RMI{\rno}{\mP}{\cha}
		&=\tfrac{1}{\rno-1}\left[
		\rno\tfrac{\lon{\dmn{\rno,\mP}}}{\lon{\mmn{\rno,\mP}}}-\tfrac{\RMI{\rno}{\mP}{\cha}}{\rno} \right].
		\end{align}
		The expression in the brackets is differentiable in \(\rno\) on \(\reals{+} \) 
		because \(\lon{\mmn{\rno,\mP}}\) is positive and
		\(\lon{\mmn{\rno,\mP}}\), \(\lon{\dmn{\rno,\mP}}\), and \(\RMI{\rno}{\mP}{\cha}\) 
		are differentiable by 
		Lemmas \ref{lem:powermeanequivalence}-(\ref{powermeanequivalence-a}),
		\ref{lem:powermeanO}-(\ref{powermeanO-b},\ref{powermeanO-c}), 
		and \ref{lem:informationO}. 
		Furthermore, the expression in the brackets is equal to zero at \(\rno=1\).
		Then as a result of the mean value theorem \cite[5.10]{rudin}  for each 
		\(\rno\in[\sfrac{1}{2},1)\) there exists a \(\rnf\in(\rno,1)\) and such that 
		\begin{align}
		\label{eq:finitecapacity-uecO-C-1}
		\der{}{\rno} \RMI{\rno}{\mP}{\cha}
		&=\left.\der{}{\rno}\left[
		\tfrac{\rno\lon{\dmn{\rno,\mP}}}{\lon{\mmn{\rno,\mP}}}-\tfrac{\RMI{\rno}{\mP}{\cha}}{\rno} \right]
		\right\vert_{\rno=\rnf}.
		\end{align}
		Using the expressions for derivatives given in 
		Lemmas \ref{lem:powermeandensityO}-(\ref{powermeandensityO-b})
		and \ref{lem:powermeanO}-(\ref{powermeanO-b},\ref{powermeanO-c}) 
		we get
		\begin{align}
		\notag
		\hspace{-.1cm}
		\der{}{\rno} \tfrac{\rno\lon{\dmn{\rno,\mP}}}{\lon{\mmn{\rno,\mP}}}
		&=\tfrac{\rno\ddmn{\rno,\mP}(\outS)}{\lon{\mmn{\rno,\mP}}}
		+\tfrac{\lon{\dmn{\rno,\mP}}}{\lon{\mmn{\rno,\mP}}}
		-\tfrac{\rno\lon{\dmn{\rno,\mP}}^{2}}{\lon{\mmn{\rno,\mP}}^{2}}
		\\
		\notag
		&=\EXS{\qmn{\rno,\mP}}{\sum\nolimits_{\mW} \tfrac{\tpn{\rno}(\mW|\dout)}{\rno^{2}}\left(\ln\tfrac{\tpn{\rno}(\mW|\dout)}{\mP(\mW)}\right)^2}
		\\
		\notag
		&\quad+
		\EXS{\qmn{\rno,\mP}}{\tfrac{\rno(1-\rno)(\dnmn{\rno,\mP})^2}{(\nmn{\rno,\mP})^2}}
		\!-\!\tfrac{\lon{\dmn{\rno,\mP}}}{\lon{\mmn{\rno,\mP}}}
		\!-\!\tfrac{\rno\lon{\dmn{\rno,\mP}}^{2}}{\lon{\mmn{\rno,\mP}}^{2}}
		\\
		\notag
		&\leq \EXS{\qmn{\rno,\mP}}{\!\tfrac{\dnmn{\rno,\mP}^{2}}{4\nmn{\rno,\mP}^{2}}\!+\!\sum\nolimits_{\mW} \!\tfrac{\tpn{\rno}(\mW|\dout)}{\rno^{2}}\left(\ln\tfrac{\tpn{\rno}(\mW|\dout)}{\mP(\mW)}\right)^2}
		\end{align} 		
		Then using  
		\(\sum\nolimits_{\mW} \tpn{\rno}(\mW|\dout)\tfrac{1}{\rno}\ln\tfrac{\tpn{\rno}(\mW|\dout)}{\mP(\mW)}
		=\tfrac{\rno \dnmn{\rno,\mP}}{\nmn{\rno,\mP}}\), which follows from Lemma \ref{lem:powermeandensityO}-(\ref{powermeandensityO-b}), we get
		\begin{align}
		\label{eq:finitecapacity-uecO-C-2}
		\der{}{\rno} \tfrac{\rno\lon{\dmn{\rno,\mP}}}{\lon{\mmn{\rno,\mP}}}
		&\!\leq\!\tfrac{4\rno^{2}+1}{4\rno^{2}}\!\EXS{\qmn{\rno,\mP}}{\sum\limits_{\mW} \!\tfrac{\tpn{\rno}(\mW|\dout)}{\rno^{2}}\left(\ln\tfrac{\tpn{\rno}(\mW|\dout)}{\mP(\mW)}\right)^2}
		\end{align}
		Since \(\RMI{\rno}{\mP}{\cha}\) is differentiable and nondecreasing in \(\rno\)
		\begin{align}
		\notag
		\der{}{\rno}\tfrac{\RMI{\rno}{\mP}{\cha}}{\rno}
		&=-\tfrac{\RMI{\rno}{\mP}{\cha}}{\rno^2}+\tfrac{1}{\rno}\der{}{\rno}\RMI{\rno}{\mP}{\cha}
		\\
		\label{eq:finitecapacity-uecO-C-3}
		&\geq -\tfrac{\RMI{\rno}{\mP}{\cha}}{\rno^2}
		\end{align}
		Using \eqref{eq:finitecapacity-uecO-C-1}, \eqref{eq:finitecapacity-uecO-C-2}, and \eqref{eq:finitecapacity-uecO-C-3} we can conclude that
		there exists a \(\rnf\in(\rno,1)\) such that
		\begin{align}
		\notag
		\hspace{-.2cm}
		\der{}{\rno} \RMI{\rno}{\mP}{\cha}
		&\!\leq\!2\EXS{\qmn{\rnf,\mP}}{
			\sum\limits_{\mW} \tfrac{\tpn{\rnf}(\mW|\dout)}{\rnf^{2}}
			\left[\!\ln\tfrac{\tpn{\rnf}(\mW|\dout)}{\mP(\mW)}\!\right]^2}
		\!+\!4\RMI{\rnf}{\mP}{\cha}
		\end{align}
		Similarly for all \(\rno\in (1,\infty)\) there exists a \(\rnf\in (1,\rno)\) satisfying the same identity. 
		Furthermore, one can confirm by substitution for the expression given in  \eqref{eq:information:der-A} that 
		\begin{align}
		\notag
		\left. \der{}{\rno} \RMI{\rno}{\mP}{\cha}\right\vert_{\rno=1}
		&\leq\tfrac{1}{2}
		\EXS{\qmn{1,\mP}}{\sum\limits_{\mW}\tpn{1}(\mW|\dout)\left(\ln\tfrac{\tpn{1}(\mW|\dout)}{\mP(\mW)}\right)^2}.
		\end{align}
		Thus there exist an \(\rnf\in (\tfrac{1}{2},\tfrac{5\rnt -1}{4\rnt})\) such that
		\begin{align}
		\notag
		&\hspace{-.3cm}
		\sup\nolimits_{\rno\in [\frac{1}{2},\frac{5\rnt -1}{4\rnt}]} \der{}{\rno} \RMI{\rno}{\mP}{\cha}
		\\
		\label{eq:finitecapacity-uecO-C-4}
		&\!\leq\!4\RMI{\rnf}{\mP}{\cha}\!+\!2 
		\EXS{\qmn{\rnf,\mP}}{\!\sum\limits_{\mW}\!\tfrac{\tpn{\rnf}(\mW|\dout)}{\rnf^{2}}\!
			\left[\!\ln\tfrac{\tpn{\rnf}(\mW|\dout)}{\mP(\mW)}\!\right]^2\!}\!.
		\end{align}
		Note that 
		\(\dinp^{\rnf}\ln^{2} \dinp\leq \tfrac{4}{e^{2}\rnf^{2}}\IND{\dinp\in [0,1)}+\tfrac{4\dinp^{\rnb}}{e^{2}(\rnb-\rnf)^{2}}\IND{\dinp>1}\)	for all \(\rnb>\rnf\). 
		Then using Lemma \ref{lem:powermeandensityO}-(\ref{powermeandensityO-a}) we get 
		the following bound for all \(\rnf\) in \([\tfrac{1}{2},\tfrac{5\rnt-1}{4\rnt}]\)
		\begin{align}
		\notag
		&\hspace{-.2cm}\sum\nolimits_{\mW}\tpn{\rnf}(\mW|\dout)\left( \tfrac{1}{\rnf} \ln\tfrac{\tpn{\rnf}(\mW|\dout)}{\mP(\mW)}\right)^2
		\\
		\notag
		&=\sum\nolimits_{\mW}\mP(\mW) \left(\tfrac{\tpn{1}(\mW|\dout)}{\mP(\mW) \nmn{\rnf,\mP}}\right)^{\rnf} 
		\left(\ln\tfrac{\tpn{1}(\mW|\dout)}{\mP(\mW) \nmn{\rnf,\mP}}\right)^{2}
		\\
		\notag
		&\leq \tfrac{4}{e^{2}\rnf^{2}}+\tfrac{4}{e^{2}}\left(\tfrac{3\rnt-1}{2\rnt}-\rnf\right)^{-2}
		\sum\nolimits_{\mW}\mP(\mW) \left(\tfrac{\tpn{1}(\mW|\dout)}{\mP(\mW) \nmn{\rnf,\mP}}\right)^{\frac{3\rnt-1}{2\rnt}} 
		\\
		\notag
		&= \tfrac{4}{e^{2}\rnf^{2}}+\tfrac{4}{e^{2}}\left(\tfrac{3\rnt-1}{2\rnt}-\rnf\right)^{-2}
		\left(\sfrac{\nmn{\frac{3\rnt-1}{2\rnt},\mP}}{\nmn{\rnf,\mP}}\right)^{\frac{3\rnt-1}{2\rnt}}
		\\
		\label{eq:finitecapacity-uecO-C-5}
		&\leq \tfrac{16}{e^{2}}+\tfrac{64}{e^{2}}\left(\tfrac{\rnt}{\rnt-1}\right)^{2}
		\left(\sfrac{\nmn{\frac{3\rnt-1}{2\rnt},\mP}}{\nmn{\rnf,\mP}}\right)^{\frac{3\rnt-1}{2\rnt}}.
		\end{align}
		On the other hand \((\nmn{\rno,\mP})^{\rno}\) is log-convex in \(\rno\) 
		by Lemma \ref{lem:powermeandensityO}-(\ref{powermeandensityO-c})
		and \(\nmn{\rno,\mP}\) is nondecreasing in \(\rno\) by Lemma \ref{lem:powermeandensityO}-(\ref{powermeandensityO-d}).
		Thus for all  \(\rnf\) in \([\tfrac{1}{2},\tfrac{5\rnt-1}{4\rnt}]\) we have
		\begin{align}
		\notag
		\left(\nmn{\frac{3\rnt-1}{2\rnt},\mP}\right)^{\frac{3\rnt-1}{2\rnt}}
		&\leq (\nmn{\frac{2\rnt\rnf}{2\rnt\rnf-\rnt+1},\mP}) (\nmn{\rnf,\mP})^{\rnf \frac{\rnt-1}{2\rnt \rnf}} 
		\\
		\label{eq:finitecapacity-uecO-C-6}
		&\leq (\nmn{\rnt,\mP}) (\nmn{\rnf,\mP})^{\frac{\rnt-1}{2\rnt}}.
		\end{align}
		Using equations \eqref{eq:finitecapacity-uecO-C-4}, \eqref{eq:finitecapacity-uecO-C-5}, and \eqref{eq:finitecapacity-uecO-C-6}
		we get
		\begin{align}
		\notag
		\hspace{-.2cm}
		\der{}{\rno} \RMI{\rno}{\mP}{\cha}
		&\!\leq\!4\RMI{\rnt}{\mP}{\cha}+\tfrac{32}{e^{2}}+
		\tfrac{128}{e^{2}}\left(\tfrac{\rnt}{\rnt-1}\right)^{2}\tfrac{\lon{\mmn{\rnt,\mP}}}{\lon{\mmn{\sfrac{1}{2},\mP}}}
		\\
		\notag
		&\!=\!4\RMI{\rnt}{\mP}{\cha}\!+\!\tfrac{32}{e^{2}}\!+\!
		\tfrac{128\rnt^{2}}{e^{2}(\rnt-1)^{2}}
		e^{\frac{\rnt-1}{\rnt}\RMI{\rnt}{\mP}{\cha}+\RMI{\frac{1}{2}}{\mP}{\cha}}
		\end{align}
		for all \(\rno\) in \([\tfrac{1}{2},\tfrac{5\rnt -1}{4\rnt}]\).
		Since \(\RMI{\rno}{\mP}{\cha}\) is nondecreasing in \(\rno\) by Lemma \ref{lem:informationO},
		the definition of \renyi capacity implies
		\begin{align}
		\notag
		\der{}{\rno} \RMI{\rno}{\mP}{\cha}
		&\leq 4\RC{\rnt}{\cha}+\tfrac{32}{e^{2}}+
		\tfrac{128}{e^{2}}\left(\tfrac{\rnt}{\rnt-1}\right)^{2} e^{2\RC{\rno}{\cha}}
		\end{align}
		for all \(\rno\) in \([\tfrac{1}{2},\tfrac{5\rnt -1}{4\rnt}]\) and
		\(\mP\) in \(\pdis{\cha}\). Hence, 
		\begin{align}
		\label{eq:finitecapacity-uecO-C}
		\hspace{-.4cm}
		\RMI{\rno}{\mP}{\cha}\!-\!\RMI{\rnf}{\mP}{\cha}
		&\!\leq\!(\rno\!-\!\rnf)
		\left[\!4\RC{\rnt}{\cha}\!+\!\tfrac{160 \rnt^2 e^{2\RC{\rnt}{\cha}}}{e^{2} (\rnt-1)^{2}}\!\right]
		\end{align}
		for all \(\rnf\), \(\rno\) in \([\tfrac{1}{2},\tfrac{5\rnt -1}{4\rnt}]\)
		satisfying	\(\rnf\leq\rno\) and \(\mP\) in \(\pdis{\cha}\).
		\item[(\ref{finitecapacity-fc})]
		For any \(\mP\in\pdis{\cha}\), \(\RMI{\rno}{\mP}{\cha}\) is nondecreasing and continuous in 
		\(\rno\) on \([0,\infty]\) by Lemma \ref{lem:informationO}. 
		Then \(\RMI{\rno}{\mP}{\cha}\) is a quasi-convex continuous function of \(\rno\)  
		satisfying 
		\(\RMI{\rno}{\mP}{\cha}=\inf\nolimits_{\rnt\in(\rno,\infty)} \RMI{\rnt}{\mP}{\cha}\)
		for any \(\mP\) in \(\pdis{\cha}\).
		Using the definition of \(\RC{\rno}{\cha}\) we get
		\begin{align}
		\label{eq:finitecapacity-fc-a}
		\hspace{-.5cm}
		\RC{\rno}{\cha}
		&\!=\!\sup\nolimits_{\mP\in \pdis{\cha}}\inf\nolimits_{\rnt\in(\rno,\infty)}\!\RMI{\rnt}{\mP}{\cha}
		&
		&\forall \rno\!\in\!\reals{\geq0} .
		\end{align}
		Since \(\RMI{\rno}{\mP}{\cha}\leq\ln\abs{\cha}\) by Lemma \ref{lem:informationO},
		if \(\cha\) is finite, then \(\RC{\rno}{\cha}\) is finite for all \(\rno\in\reals{\geq0} \) and \(\RMI{\rno}{\mP}{\cha}\) is continuous in \(\mP\) on \(\pdis{\cha}\) 
		for all \(\rno\in\reals{+} \) by part (\ref{finitecapacity-uecP}). 
		Furthermore \(\RMI{\rno}{\mP}{\cha}\) is quasi-concave in \(\mP\) for all \(\rno\in\reals{+} \) by Lemma \ref{lem:informationP}. 
		Then we can change the order of the supremum and the infimum in \eqref{eq:finitecapacity-fc-a}
		using Sion's minimax theorem, \cite[Cor. 3.3]{sion58}, \cite{komiya88}
		because \(\pdis{\cha}\) is compact.
		\begin{align}
		\notag
		\RC{\rno}{\cha}
		&=\inf\nolimits_{\rnt\in(\rno,\infty)}\sup\nolimits_{\mP\in \pdis{\cha}} \RMI{\rnt}{\mP}{\cha}
		\\
		\notag
		&=\inf\nolimits_{\rnt\in(\rno,\infty)}\RC{\rnt}{\cha}
		&
		&\forall \rno \in \reals{\geq0} .
		\end{align}
		Then \(\RC{\rno}{\cha}\) is continuous from the right.
		On the other hand \(\RC{\rno}{\cha}\) is continuous from the left
		because it is nondecreasing and lower semicontinuous on
		\([0,\infty]\) by  Lemma \ref{lem:capacityO}-(\ref{capacityO-ilsc}).
		\vspace{-.5cm} 
	\end{enumerate}
\end{proof}

\begin{proof}[Proof of Lemma \ref{lem:capacityUnion}]
	\underline{\(\sup\nolimits_{\ind\in\tinS} \RC{\rno}{\cha_{\ind}}\leq \RC{\rno}{\cha}\):}
	\vspace{-.2cm}
	\begin{align}
	\notag
	\RC{\rno}{\cha}
	&\mathop{\geq}^{(a)} \RRR{\rno}{\cha_{\ind}}{\qmn{\rno,\cha}}
	\\
	\notag
	&\mathop{\geq}^{(b)} \RC{\rno}{\cha_{\ind}}+\RD{\rno}{\qmn{\rno,\cha_{\ind}}}{\qmn{\rno,\cha}} 
	\\
	\notag
	&\mathop{\geq}^{(c)} \RC{\rno}{\cha_{\ind}}+\tfrac{\rno\wedge 1}{2}\lon{\qmn{\rno,\cha_{\ind}}-\qmn{\rno,\cha}}^{2}
	\end{align}
	where \((a)\) follows from 
	\eqref{eq:def:relativeradius}, Theorem \ref{thm:minimax}, and \(\cha_{\ind}\subset\cha \),
	\((b)\) follows from Lemma \ref{lem:EHB},
	\((c)\) follows from Lemma \ref{lem:divergence-pinsker}.
	Consequently, \(\RC{\rno}{\cha}\geq\sup\nolimits_{\ind\in\tinS} \RC{\rno}{\cha_{\ind}}\) 
	and if \(\RC{\rno}{\cha_{\ind}}= \RC{\rno}{\cha}\), then \(\qmn{\rno,\cha}=\qmn{\rno,\cha_{\ind}}\).
	\begin{itemize}
		\item If \(\RC{\rno}{\cha_{\ind}}= \RC{\rno}{\cha}\) and \(\qmn{\rno,\cha}=\qmn{\rno,\cha_{\ind}}\),
		then \(\RRR{\rno}{\cha}{\qmn{\rno,\cha_{\ind}}}\leq \RC{\rno}{\cha_{\ind}}\)
		by Theorem \ref{thm:minimax}.
		\item If \(\RRR{\rno}{\cha}{\qmn{\rno,\cha_{\ind}}}\leq \RC{\rno}{\cha_{\ind}}\), then 
		\(\RC{\rno}{\cha}\leq \RC{\rno}{\cha_{\ind}}\) because
		\(\RC{\rno}{\cha}=\RR{\rno}{\cha}\) by Theorem \ref{thm:minimax}
		and 
		\(\RR{\rno}{\cha}\leq \RRR{\rno}{\cha}{\qmn{\rno,\cha_{\ind}}}\) by definition.
		Then \(\RRR{\rno}{\cha}{\qmn{\rno,\cha_{\ind}}}\leq \RC{\rno}{\cha_{\ind}}\)
		implies \(\RC{\rno}{\cha}=\RC{\rno}{\cha_{\ind}}\) because
		\(\RC{\rno}{\cha}\geq \RC{\rno}{\cha_{\ind}}\) by definition.
	\end{itemize}
	
	\underline{\(\RC{\rno}{\cha}\leq \ln \sum\nolimits_{\ind\in\tinS} e^{\RC{\rno}{\cha_{\ind}}}\):}
	If \(\tinS\) is  infinite, then the inequality holds trivially because the right 
	hand side is infinite. Thus, we will establish the inequality
	assuming \(\tinS\) is finite. 
	Let \(\mV\) be \(\mV\DEF\bigvee_{\ind\in\tinS} e^{\RC{\rno}{\cha_{\ind}}}\qmn{\rno,\cha_{\ind}}\).
	Then
	\begin{align}
	\notag
	\RRR{\rno}{\cha}{\sfrac{\mV}{\lon{\mV}}}
	&\mathop{=}^{(a)}\max\nolimits_{\ind \in\tinS} \RRR{\rno}{\cha_{\ind}}{\sfrac{\mV}{\lon{\mV}}}
	\\
	\notag
	&\mathop{\leq}^{(b)} \max\nolimits_{\ind \in\tinS} \RRR{\rno}{\cha_{\ind}}{\qmn{\rno,\cha_{\ind}}}-\ln e^{\RC{\rno}{\cha_{\ind}}} +\ln \lon{\mV} 
	\\
	\label{eq:capacityUnion-a}
	&\mathop{=}^{(c)}\ln \lon{\mV}
	\end{align}
	where \((a)\) follows from \eqref{eq:def:relativeradius} and \eqref{eq:def:radius},
	\((b)\) follows from Lemma \ref{lem:divergence-RM} because
	\(e^{\RC{\rno}{\cha_{\ind}}}\qmn{\rno,\cha_{\ind}}\!\leq\!\mV\), 
	and \((c)\) follows from Theorem \ref{thm:minimax}. 
	On the other hand, \(\lon{\mV}\!\leq\!\sum\nolimits_{\ind\in\tinS} e^{\RC{\rno}{\cha_{\ind}}}\) by 
	the definition of \(\mV\).
	Then \(\RC{\rno}{\cha}\leq \ln \sum\nolimits_{\ind\in\tinS} e^{\RC{\rno}{\cha_{\ind}}}\)
	by Theorem \ref{thm:minimax}.
	\begin{itemize}
		\item If \(\tinS\) is infinite, then 
		\(\sum\nolimits_{\ind\in\tinS} e^{\RC{\rno}{\cha_{\ind}}}\) is infinite.
		If \(\tinS\) is finite but \(\qmn{\rno,\cha_{\ind}}\) and \(\qmn{\rno,\cha_{\jnd}}\) are not singular 
		for some distinct \(\ind\) and \(\jnd\), then 
		\(\RC{\rno}{\cha}\!<\!\ln \sum\nolimits_{\ind\in\tinS} e^{\RC{\rno}{\cha_{\ind}}}\)
		by  \eqref{eq:capacityUnion-a} and Theorem \ref{thm:minimax} because
		\(\lon{\mV}\!<\!\sum\nolimits_{\ind\in\tinS} e^{\RC{\rno}{\cha_{\ind}}}\).
		Consequently, if \(\RC{\rno}{\cha}=\ln \sum\nolimits_{\ind\in\tinS} e^{\RC{\rno}{\cha_{\ind}}}<\infty\), 
		then \(\tinS\) is finite and 
		\(\qmn{\rno,\cha_{\ind}}\perp\qmn{\rno,\cha_{\jnd}}\) for all \(\ind \neq \jnd\)
		in \(\tinS\).

		\item If \(\tinS\) is finite and \(\qmn{\rno,\cha_{\ind}}\perp\qmn{\rno,\cha_{\jnd}}\) 
		for all \(\ind \neq \jnd\), 
		then any \(\mU\in \pmea{\outA}\) can be written as 
		\(\mU=\sum_{\ind=0}^{\abs{\tinS}}\mU_{\ind}\) where \(\mU_{\ind}\) 
		are finite measures  such that \(\mU_{\ind}\AC \qmn{\rno,\cha_{\ind}}\) for \(\ind\in\{1,\ldots,\abs{\tinS}\} \)  and 
		\(\mU_{0}\perp  (\sum_{\ind\in\tinS} \qmn{\rno,\cha_{\ind}})\)
		by the Lebesgue decomposition theorem \cite[5.5.3]{dudley}.
		Then using Lemmas \ref{lem:divergence-RM} and \ref{lem:divergence-DPI}, we get
		\begin{align}
		\notag
		\RD{\rno}{\qmn{\rno,\cha_{\ind}}}{\mU}
		&\geq -\ln \lon{\mU_{\ind}}.
		\end{align}
		Thus Lemma \ref{lem:EHB} implies
		\begin{align}
		\notag
		\RRR{\rno}{\cha_{\ind}}{\mU}
		&\geq \RC{\rno}{\cha_{\ind}}-\ln \lon{\mU_{\ind}}
		&
		&\forall \mU\in\pmea{\outA}.
		\end{align}
		Since \(\RRR{\rno}{\cha}{\mU}= \max\nolimits_{\ind\in\tinS}\RRR{\rno}{\cha_{\ind}}{\mU}\)
		for all \(\mU\) in \(\pmea{\outA}\) and 
		\(\RC{\rno}{\cha}=\inf\nolimits_{\mU\in\pmea{\outA}}\RRR{\rno}{\cha}{\mU}\) 
		by Theorem \ref{thm:minimax}, we get
		\begin{align}
		\notag
		\RC{\rno}{\cha}
		&\geq \inf\nolimits_{\mU\in\pmea{\outA}}
		\max\nolimits_{\ind\in\tinS}  \ln \tfrac{e^{\RC{\rno}{\cha_{\ind}}}}{\lon{\mU_{\ind}}} 
		\\
		\notag
		&\geq \inf\nolimits_{\mU\in\pmea{\outA}}\ln \tfrac{\sum\nolimits_{\ind\in\tinS} e^{\RC{\rno}{\cha_{\ind}}}}{\sum_{\ind\in\tinS} \lon{\mU_{\ind}}}
		\\
		\notag
		&\geq \ln \sum\nolimits_{\ind\in\tinS} e^{\RC{\rno}{\cha_{\ind}}}.
		\end{align}
		Then \(\RC{\rno}{\cha}=\ln \sum\nolimits_{\ind\in\tinS} e^{\RC{\rno}{\cha_{\ind}}}\)
		because we have already proved the reverse inequality.
		Furthermore, \(\qmn{\rno,\cha}=\sfrac{\tilde{\mU}}{\lon{\tilde{\mU}}}\) for 
		\(\tilde{\mU}=\sum_{\ind\in\tinS} e^{\RC{\rno}{\cha_{\ind}}}\qmn{\rno,\cha_{\ind}}\)
		by Theorem \ref{thm:minimax}.
		because
		\(\RRR{\rno}{\cha_{\ind}}{\sfrac{\tilde{\mU}}{\lon{\tilde{\mU}}}}=\RC{\rno}{\cha}\).
	\end{itemize}
\end{proof}

\begin{proof}[Proof of Lemma \ref{lem:capacityProduct}]
	By the definition of \(\RMI{\rno}{\mP}{\cha}\) for all \(\mP\) satisfying 
	\(\mP=\bigotimes_{\tin\in \tinS} \mP_{\tin}\) for some \(\mP_{\tin}\in \pdis{\cha_{\tin}}\) we have
	\begin{align}
	\label{eq:capacityProduct:A}
	\RMI{\rno}{\mP}{\cha}
	&=\sum\nolimits_{\tin\in\tinS} \RMI{\rno}{\mP_{\tin}}{\cha_{\tin}}
	&\forall \rno \in [0,\infty].
	\end{align}
	Furthermore,
	\(\{\mP:\mP=\bigotimes_{\tin\in\tinS}\mP_{\tin},\mP_{\tin}\in\pdis{\cha_{\tin}},\forall\tin\in\tinS\}\) is 
	a subset of \(\pdis{\cha}\).
	Then
	\begin{align}
	\notag
	\RC{\rno}{\cha}
	&\geq \sup_{\mP_{1},\mP_{2},\ldots,\mP_{\abs{\tinS}}}\sum\nolimits_{\tin\in\tinS} \RMI{\rno}{\mP_{\tin}}{\cha_{\tin}}
	\\
	\label{eq:capacityProduct:B}
	&=\sum\nolimits_{\tin\in\tinS} \RC{\rno}{\cha_{\tin}} 
	&
	&\forall \rno\in[0,\infty].
	\end{align}
	
	Let us proceed with proving \(\RC{\rno}{\cha}\leq\sum_{\tin\in \tinS}\RC{\rno}{\cha_{\tin}}\). 
	If there exists a \(\tin\in\tinS\) such that \(\RC{\rno}{\cha_{\tin}}=\infty\),
	then the inequality holds trivially. 
	Else, \(\RC{\rno}{\cha_{\tin}}<\infty\) for all \(\tin\in\tinS\) and by Theorem \ref{thm:minimax}  
	there exists a \(\qmn{\rno,\cha_{\tin}}\) for each \(\tin\in\tinS\) such that
	\begin{align}
	\notag
	\RD{\rno}{\mW_{\tin}}{\qmn{\rno,\cha_{\tin}}}
	&\leq \RC{\rno}{\cha_{\tin}}
	&
	&\forall \mW_{\tin}\in \cha_{\tin}.
	\end{align}
	Recall that all \(\mW\)'s in \(\cha\) can be written as \(\mW=\bigotimes_{\tin\in\tinS}\mW_{\tin}\) for some 
	\(\mW_{\tin}\in\cha_{\tin}\) by the hypothesis. 
	Then for \(\mQ\DEF\bigotimes_{\tin\in\tinS}\qmn{\rno,\cha_{\tin}}\) by the definition of 
	the \renyi divergence given \eqref{eq:def:divergence}
	and Tonelli-Fubini theorem \cite[4.4.5]{dudley} we have 
	\begin{align}
	\notag
	\RD{\rno}{\mW}{\mQ}
	&=\sum\nolimits_{\tin \in \tinS}  \RD{\rno}{\mW_{\tin}}{\qmn{\rno,\cha_{\tin}}}
	\\
	\label{eq:producradius:C}
	&\leq\sum\nolimits_{\tin\in\tinS} \RC{\rno}{\cha_{\tin}} 
	&
	&\forall \mW\in\cha.
	\end{align}
	Then \(\RC{\rno}{\cha}\leq\sum\nolimits_{\tin\in\tinS} \RC{\rno}{\cha_{\tin}}\) by 
	\eqref{eq:thm:minimaxradius}
	Thus \eqref{eq:lem:capacityProduct}	holds and 
	\(\qmn{\rno,\cha}=\mQ\) follows from  	\eqref{eq:producradius:C} and Theorem \ref{thm:minimax}
	for the case when \(\RC{\rno}{\cha}<\infty\). 
\end{proof}

\begin{proof}[Proof of Lemma \ref{lem:capacityEps}]
	\eqref{eq:def:divergence} implies that
	\begin{align}
	\notag
	\RD{\rno}{\mP \mtimes \cha}{\mP  \otimes \qmn{\rno,\cha}}
	&\leq \sup\nolimits_{\mW \in \supp{\mP}} \RD{\rno}{\mW}{\qmn{\rno,\cha}},
	\\
	\notag
	\RD{\rno}{\mP \mtimes \cha}{\mP  \otimes \qmn{\rno,\cha}}
	&\geq \inf\nolimits_{\mW \in \supp{\mP}} \RD{\rno}{\mW}{\qmn{\rno,\cha}}.
	\end{align}
	On the other hand,  \eqref{eq:sibson} and \eqref{eq:lem:information:defA} imply that
	\begin{align}
	\notag
	\RD{\rno}{\mP \mtimes \cha}{\mP  \otimes \qmn{\rno,\cha}}
	&=\RMI{\rno}{\mP}{\cha}+
	\RD{\rno}{\qmn{\rno,\mP}}{\qmn{\rno,\cha}}.
	\end{align}
	Then  \eqref{eq:lem:capacityEps} follows from \eqref{eq:def:kcha}
	and Theorem \ref{thm:minimax}.
	
	Note that \(\RC{\rno}{\cha_{\rno,\epsilon}}\) is bounded above by \(\RC{\rno}{\cha}\)
	and hence finite by definition.
	Thus \(\cha_{\rno,\epsilon}\) has a unique order \(\rno\) \renyi center 
	by Theorem \ref{thm:minimax}.
	If \(\qmn{\rno,\cha_{\rno,\epsilon}}\!=\!\qmn{\rno,\cha}\), then 
	\(\RC{\rno}{\cha_{\rno,\epsilon}}\!=\!\RC{\rno}{\cha}\)  by 
	the definition of \(\cha_{\rno,\epsilon} \) and Theorem \ref{thm:minimax}.
	
	We show in the following by contradiction that 
	\(\qmn{\rno,\cha_{\rno,\epsilon}}\)  equals to \(\qmn{\rno,\cha}\).
	Let \(\mQ=(1-e^{-\sfrac{\epsilon}{2}})\qmn{\rno,\cha_{\rno,\epsilon} }+e^{-\sfrac{\epsilon}{2}}\qmn{\rno,\cha}\). 
	Then using Lemma \ref{lem:divergence-RM} and \eqref{eq:def:kcha} we get
	\begin{align}
	\label{eq:capacityEps-A}
	\sup\nolimits_{\mW\in \cha\setminus \cha_{\rno,\epsilon} }
	\RD{\rno}{\mW}{\mQ}
	&\leq \RC{\rno}{\cha}-\sfrac{\epsilon}{2}.
	\end{align}
	The convexity of the \renyi divergence in its second argument, i.e. 
	Lemma \ref{lem:divergence-convexity}, 
	and Theorem \ref{thm:minimax} imply that
	\begin{align}
	\notag
	&\hspace{-.2cm}\sup\nolimits_{\mW\in \cha_{\rno,\epsilon} }
	\RD{\rno}{\mW}{\mQ}
	\\
	\notag
	&\leq \sup\limits_{\mW\in \cha_{\rno,\epsilon} } 
	\left[(1-e^{-\sfrac{\epsilon}{2}})\RD{\rno}{\mW}{\qmn{\rno,\cha_{\rno,\epsilon} }}+e^{-\sfrac{\epsilon}{2}}\RD{\rno}{\mW}{\qmn{\rno,\cha}}\right]
	\\
	\notag
	&\leq 
	(1-e^{-\sfrac{\epsilon}{2}})\RRR{\rno}{\cha_{\rno,\epsilon}}{\qmn{\rno,\cha_{\rno,\epsilon} }}
	+e^{-\sfrac{\epsilon}{2}}  \RRR{\rno}{\cha_{\rno,\epsilon}}{\qmn{\rno,\cha}}
	\\
	\label{eq:capacityEps-B}
	&= (1-e^{-\sfrac{\epsilon}{2}})\RC{\rno}{\cha_{\rno,\epsilon} }+e^{-\sfrac{\epsilon}{2}}\RC{\rno}{\cha}.
	\end{align}
	If \(\qmn{\rno,\cha_{\rno,\epsilon}}\!\neq\!\qmn{\rno,\cha}\), then 
	\(\RC{\rno}{\cha_{\rno,\epsilon}}<\RC{\rno}{\cha}\) by Lemma \ref{lem:capacityUnion}.
	Then 
	\(\RRR{\rno}{\cha}{\mQ}< \RC{\rno}{\cha}\) by \eqref{eq:capacityEps-A} and \eqref{eq:capacityEps-B}. 
	However, this is a contradiction by Theorem \ref{thm:minimax}. Thus
	\(\qmn{\rno,\cha_{\rno,\epsilon}}\!=\!\qmn{\rno,\cha}\) holds.
	
	As a result of the definition of \(\cha_{\rno,\epsilon} \), an element of \(\cha\) is in \(\cha_{\rno,0} \) iff it is in 
	\(\cha_{\rno,\epsilon} \) for all \(\epsilon>0\), i.e.  \(\bigcap_{\epsilon>0} \cha_{\rno,\epsilon} =\cha_{\rno,0} \). 
	Consequently, if \(\cha\) is a finite set, 
	then \(\cha_{\rno,\epsilon} =\cha_{\rno,0} \) for 
	small enough \(\epsilon\). Then \(\RC{\rno}{\cha_{\rno,0} }=\RC{\rno}{\cha_{\rno,\epsilon} }=\RC{\rno}{\cha}\). 
	Furthermore,  \eqref{eq:lem:capacityEps} holds for \(\epsilon=0\) because \eqref{eq:lem:capacityEps} holds for all \(\epsilon>0\).
	
	For arbitrary \(\cha\)'s, 
	identity \(\bigcap_{\epsilon>0}\!\cha_{\rno,\epsilon}\!=\!\cha_{\rno,0} \) does not imply that 
	\(\cha_{\rno,\epsilon} =\cha_{\rno,0} \) for some \(\epsilon>0\).
	\(\cha_{\rno,0} \) can be the empty set or a non-empty set such that \(\RC{\rno}{\cha_{\rno,0} }<\RC{\rno}{\cha}\),
	see Example \ref{eg:erasure}.
\end{proof}

\begin{proof}[Proof of Lemma \ref{lem:capacityEXT}]~
	\begin{enumerate}
		\item[(\ref{capacityEXT-ch})] 
		\(\RC{\rno}{\cha}\!\leq\!\RC{\rno}{\conv{\cha}}\) by definition because
		\(\cha\!\subset\!\conv{\cha}\).
		If \(\RC{\rno}{\cha}\!=\!\infty\), then the reverse inequality
		\( \RC{\rno}{\conv{\cha}}\!\leq\!\RC{\rno}{\cha}\) holds trivially. 
		If \(\RC{\rno}{\cha}<\infty\),
		then \(\exists!\qmn{\rno,\cha}\) satisfying  
		\(\sup\nolimits_{\mW\in\cha} \RD{\rno}{\mW}{\qmn{\rno,\cha}}=\RC{\rno}{\cha}\)
		by Theorem \ref{thm:minimax}.
		Then as a result of the quasi-convexity of the \renyi divergence in its first argument, 
		i.e. Lemma \ref{lem:divergence-quasiconvexity}, we have
		\begin{align}
		\notag 
		\RD{\rno}{\mmn{1,\mP}}{\qmn{\rno,\cha}}
		&\leq \max\nolimits_{\mW\in \supp{\mP}}\RD{\rno}{\mW}{\qmn{\rno,\cha}}
		\\
		\notag
		&\leq \RC{\rno}{\cha}
		\end{align}
		for all \(\mP\in\pdis{\cha}\). 
		Then \(\RC{\rno}{\conv{\cha}}\leq \RC{\rno}{\cha}\) by 
		Theorem \ref{thm:minimax}. 
		
		\item[(\ref{capacityEXT-cl})]
		\(\RC{\rno}{\cha}\leq \RC{\rno}{\clos{\cha}}\) by definition because \(\cha \subset \clos{\cha}\).
		If \(\RC{\rno}{\cha}=\infty\), then the reverse inequality 
		\( \RC{\rno}{\clos{\cha}} \leq \RC{\rno}{\cha}\) holds trivially. 
		If \(\RC{\rno}{\cha}<\infty\), then \(\exists!\qmn{\rno,\cha}\) satisfying 
		\(\sup\nolimits_{\mW\in\cha}\RD{\rno}{\mW}{\qmn{\rno,\cha}}=\RC{\rno}{\cha}\) by Theorem \ref{thm:minimax}.
		Furthermore,
		for all \(\mV\in\pmea{\outA}\) and \(\epsilon>0\)  there exists an open set \(\oset\) containing \(\mV\), i.e. a 
		neighborhood of \(\mV\), such that 
		\begin{align}
		\notag
		\RD{\rno}{\mV}{\qmn{\rno,\cha}}-\epsilon
		&<\RD{\rno}{\mS}{\qmn{\rno,\cha}}
		&
		&\forall \mS\in \oset
		\end{align}
		by the lower semicontinuity, 
		i.e. Lemma \ref{lem:divergence:lsc}.
		If \(\mV\in \clos{\cha}\), then every open set containing \(\mV\) contains a member of \(\cha\).
		Thus \(\RD{\rno}{\mV}{\qmn{\rno,\cha}}-\epsilon<\RC{\rno}{\cha}\) 
		for every \(\mV\) in \(\clos{\cha}\) and positive \(\epsilon\).
		Then\footnote{This observation is nothing but the definition of the continuity:
			A function \(\fX:\inpS\to\staS\) is continuous iff for any  \(\set{A} \subset \inpS\), 
			\(\fX(\clos{\set{A}})\subset\clos{\fX(\set{A})}\)
			by \cite[Thm. 18.1]{munkres}. 
			If we chose \(\inpS\) to be \(\pmea{\outA}\) with the topology 
			of setwise convergence, 
			\(\staS\) to be \((-\infty,\infty]\) with the topology generated by the sets of the form 
			\((\dsta,\infty]\) for \(\dsta\in\reals{}\), and 
			\(\fX\) to be \(\fX(\cdot)=\RD{\rno}{\cdot}{\qmn{\rno,\cha}}\),
			then the lower semicontinuity of the \renyi divergence in its first argument
			is equivalent to the continuity of \(\fX\). 
			On the other hand, \(\fX(\cha) \subset (-\infty,\RC{\rno}{\cha}]\) by Theorem \ref{thm:minimax}
			and  \((-\infty,\RC{\rno}{\cha}]\) is a closed set for the topology we have chosen for \((-\infty,\infty]\).
			Thus
			\(\fX(\clos{\cha})\subset \clos{\fX(\cha)} \subset (-\infty,\RC{\rno}{\cha}]\), i.e. 
			\(\RD{\rno}{\mV}{\qmn{\rno,\cha}}\leq \RC{\rno}{\cha}\) for all \(\mV\in\clos{\cha}\).} 
		\(\RD{\rno}{\mV}{\qmn{\rno,\cha}}\!\leq\!\RC{\rno}{\cha}\) for every \(\mV\in \clos{\cha}\)
		and 
		\(\RC{\rno}{\clos{\cha}}\leq \RC{\rno}{\cha}\) by 
		Theorem \ref{thm:minimax}.
		
		The closure of \(\cha\) for a topology stronger than 
		the topology of setwise convergence is  a subset of the closure of \(\cha\) for 
		the topology of setwise convergence and a superset of \(\cha\). Thus its \renyi capacity is bounded from below and from 
		above  by \(\RC{\rno}{\cha}\).  
		
		\item[(\ref{capacityEXT-compact-S})]
		If \(\RC{\rnt}{\cha}<\infty\), then 
		as a result Theorem \ref{thm:minimax}, Lemma \ref{lem:information:def}, and \eqref{eq:sibson} 
		there exists a unique \(\qmn{\rnt,\cha}\) satisfying 
		\begin{align}
		\notag
		\RD{\rnt}{\qmn{\rnt,\mP}}{\qmn{\rnt,\cha}}
		&\leq \RC{\rnt}{\cha}-\RMI{\rnt}{\mP}{\cha}
		&
		&\forall \mP\in\pdis{\cha}.
		\end{align}
		
		If \(\rnt>1\), then using the definitions of \renyi information and divergence
		given in \eqref{eq:def:information} and \eqref{eq:def:divergence} we get
		\begin{align}
		\notag
		\int (\der{\mmn{\rnt,\mP}}{\rfm})^{\rnt} (\der{\qmn{\rnt,\cha}}{\rfm})^{1-\rnt} \rfm(\dif{\dout})
		&\leq e^{(\rnt-1)\RC{\rnt}{\cha}}
		&
		&\forall \mP\in\pdis{\cha}.
		\end{align}
		Since \(\RC{\rnt}{\cha}\) is finite this implies that \(\mmn{\rnt,\mP}\AC\qmn{\rnt,\cha}\).
		On the other hand \(\mmn{\rno,\mP}\) is nondecreasing in \(\rno\), 
		in the sense that if \(\rno<\rnt\) then \(\mmn{\rno,\mP}\leq \mmn{\rnt,\mP}\),
		by Lemma \ref{lem:powermeanO}-(\ref{powermeanO-a},\ref{powermeanO-b}). 
		Hence,
		\begin{align}
		\notag
		\int (\der{\mmn{\rno,\mP}}{\qmn{\rnt,\cha}})^{\rnt}  \qmn{\rnt,\cha}(\dif{\dout})
		&\leq e^{(\rnt-1)\RC{\rnt}{\cha}}
		\end{align}
		for all \(\mP\) in \(\pdis{\cha}\) and \(\rno\) in \([0,\rnt]\). 
		Then \(\der{\mmn{\rno,\mP}}{\qmn{\rnt,\cha}}\)'s are \(\qmn{\rnt,\cha}\)-integrable 
		and the set \(\{\der{\mmn{\rno,\mP}}{\qmn{\rnt,\cha}}:\mP\in\pdis{\cha},~\rno\in[0,\rnt]\}\)
		satisfies the necessary and sufficient condition for the uniform 
		integrability\footnote{A set of \(\qmn{\rnt,\cha}\)-integrable functions is uniformly integrable 
			iff it has compact closure in the weak topology of \(\Lon{\qmn{\rnt,\cha}}\) by 
			Dunford-Pettis theorem \cite[4.7.18]{bogachev}. 
			Thus \(\{\der{\mmn{\rno,\mP}}{\qmn{\rnt,\cha}}:\mP\in\pdis{\cha},~\rno\in[0,\rnt]\}\) has 
			compact closure in the weak topology of \(\Lon{\qmn{\rnt,\cha}}\). Since we have chosen to 
			work with the space of measures rather than the space of integrable functions we have stated 
			our result in terms of relative compactness in the space of measures rather than integrable 
			functions.} 
		determined by 
		de la Vall\'{e}e Poussin \cite[Thm. 4.5.9]{bogachev}, for the growth function \(G(\dinp)=\dinp^{\rnt}\).
		But when the reference measure is finite, the uniform integrability is equivalent to 
		the uniform absolute continuity of the integrals and boundedness in \(\Lon{\qmn{\rnt,\cha}}\) 
		by \cite[Thm. 4.5.3]{bogachev}, which in our case is nothing but the uniform absolute 
		continuity with respect to \(\qmn{\rnt,\cha}\) and boundedness in total variation norm 
		for the set of all mean measures. 
		Thus \({\{\mmn{\rno,\mP}:\mP\in \pdis{\cha},~\rno\in[0,\rnt]\}}\UAC \qmn{\rnt,\cha}.\)
		
		On the other hand by \cite[Thm. 4.7.25]{bogachev}, a set of measures 
		is uniformly absolutely continuous with respect to a finite measure and bounded in variation norm
		iff it has compact closure in the topology of setwise convergence.
		A set of measures has compact closure in the topology of setwise convergence 
		iff it has compact closure in the weak topology by  \cite[Thm. 4.7.25]{bogachev},
		as well.
		
If \(\rnt=1\), then using \(\dinp \ln \dinp \geq -\sfrac{1}{e}\),  
\(\lon{\mmn{1,\mP}}\!=\!1\),
and the definition of the \renyi divergence given in \eqref{eq:def:divergence} 
we get,
\begin{align}
\notag
\int \GX\left(\der{\mmn{1,\mP}}{\qmn{1,\cha}}\right)
\qmn{1,\cha}(\dif{\dout})
\leq 
\RC{1}{\cha}-\RMI{1}{\mP}{\cha}+\tfrac{1}{e}+1
\end{align}
for all \(\mP\!\in\!\pdis{\cha}\)
where  \(\GX(\dinp)\!=\!\dinp\IND{0\leq\dinp<e}\!+\!\dinp\ln\dinp\!\IND{\dinp\geq e}\).		
		
Since mean measure is a nondecreasing function of the order 
by Lemma \ref{lem:powermeanO}-(\ref{powermeanO-a},\ref{powermeanO-b})
and \(\GX(\dinp)\) is an increasing function of \(\dinp\),  we have
\begin{align}
\notag
\int \GX\left(\der{\mmn{\rno,\mP}}{\qmn{1,\cha}}\right)
\qmn{1,\cha}(\dif{\dout})
&\leq \RC{1}{\cha}+\tfrac{1}{e}+1 
\end{align}
for all \(\mP\) in \(\pdis{\cha}\) and 
\(\rno\) in \([0,1]\).
The rest of the proof for \(\rnt=1\) case is identical to that of \(\rnt>1\) case.
		
		\item[(\ref{capacityEXT-compact-N})] 
		The equivalence of the last three statements to one another is  a version of Dunford-Pettis theorem \cite[4.7.25]{bogachev}. 
		Thus  we will only prove the equivalence of the first two statements.
		
		Let us first prove the direct part: if there exists a \(\mean\) in
		\(\pmea{\outA}\) satisfying \(\cha\UAC\mean\), then  
		\(\lim_{\rno \uparrow 1} \tfrac{1-\rno}{\rno}\RC{\rno}{\cha}=0\).
		Note that \(\RC{\rno}{\cha}\leq \sup_{\mW\in \cha} \RD{\rno}{\mW}{\mQ}\) for all \(\rno\in (0,1)\) and \(\mQ\in\pmea{\outA}\) by 
		Theorem \ref{thm:minimax}. Thus	using \eqref{eq:def:divergence} we get
		\begin{align}
		\notag
		\hspace{-.4cm}
		\limsup_{\rno \uparrow 1}\!\tfrac{1-\rno}{\rno}\RC{\rno}{\cha}
		&\!\leq\!\limsup_{\rno \uparrow 1}\!\sup_{\mW\in \cha}
		\RD{1-\rno}{\mean}{\mW}
		\\
		\label{eq:compact-N-a}
		&\!\leq\!\limsup_{\rno \uparrow 1}\!\sup_{\mW\in \cha}\tfrac{-1}{\rno}\!\ln\!\EXS{\mean}{(\der{\mW}{\mean})^{\rno}}\!.
		\end{align}
		Since \(\cha\UAC\mean\), for all \(\epsilon>0\) there exists a \(\delta>0\) such that 
		if \(\mean(\oev)\leq \delta\) for an \(\oev \in \outA\), 
		then \(\mW(\oev)\leq \epsilon\) for all \(\mW\in\cha\).
		On the other hand \(\mean(\der{\mW}{\mean}>\tfrac{1}{\delta})\leq \delta\) by Markov inequality.
		Hence 
		\begin{align}
		\label{eq:compact-N-b}
		\mW(\der{\mW}{\mean}>\tfrac{1}{\delta})\leq \epsilon.
		\end{align}
		On the other hand using \eqref{eq:def:divergence} we get
		\begin{align}
		\notag
		\EXS{\mean}{(\der{\mW}{\mean})^{\rno}}
		&\!\geq\!\EXS{\mW}{(\der{\mW}{\mean})^{\rno-1} \IND{\der{\mW}{\mean}\in (0,\tfrac{1}{\delta}]}} 
		\\
		\notag
		&\!\geq\!\delta^{1-\rno} (1- \mW(\der{\mW}{\mean}>\tfrac{1}{\delta}))
		&
		&\forall \mW\in \cha.
		\end{align}
		Then as a result of \eqref{eq:compact-N-a} and \eqref{eq:compact-N-b} we have
		\begin{align}
		\notag
		\limsup\nolimits_{\rno \uparrow 1}\!\tfrac{1-\rno}{\rno}\RC{\rno}{\cha}
		&\leq \tfrac{1}{1-\epsilon}
		&
		&\forall \epsilon>0.
		\end{align}
		Then \(\lim\nolimits_{\rno \uparrow 1} \tfrac{1-\rno}{\rno}\RC{\rno}{\cha}=0\)
		because \(\RC{\rno}{\cha}\geq 0\).

		We are left with proving the converse statement: if 
		\(\lim_{\rno \uparrow 1} \tfrac{1-\rno}{\rno}\RC{\rno}{\cha}=0\), 
		then there exists a \(\mean\in\pmea{\outA}\) such that \(\cha\UAC\mean\). 
		We start with proving the following statement 
		about the \renyi centers: For every \(\epsilon>0\) there exists a \((\rnf,\delta)\) pair 
		such that \(\rnf\in(0,1)\), \(\delta\in (0,\epsilon)\), and if 
		\(\qmn{\rnf,\cha}(\oev)\leq \delta\), then \(\mW(\oev)<\epsilon\) for all \(\mW\in\cha\).
		
		For any \(\epsilon\!>\!0\) there exists a \(\rnf\!\in\!(0,1)\) such that 
		\(e^{\frac{\rnf-1}{\rnf}\RC{\rnf}{\cha}}\!>\!1\!-\!\tfrac{\epsilon}{2}\) because 
		\(\lim_{\rno\uparrow1}\!\tfrac{1-\rno}{\rno}\RC{\rno}{\cha}\!=\!0\). On the other hand, 
		\(\RDF{\rno}{\sigma(\{\oev\})}{\mW}{\qmn{\rno,\cha}}\leq \RC{\rno}{\cha}\) for any \(\mW\in\cha\) and 
		\(\oev\in\outA\), as a result of Lemma \ref{lem:divergence-DPI} and
		Theorem \ref{thm:minimax}. Then the above described \(\rnf\) satisfies
		\begin{align}
		\label{eq:compact-N-d}
		\fX(\mW(\oev),\qmn{\rnf,\cha}(\oev))		
		&\geq (1-\tfrac{\epsilon}{2})^{\rnf}
		\end{align}
		for all \(\oev\in\outA\) and \(\mW\in\cha\)
		where the function \(\fX(\dinp,\dsta)\) is defined 
		for all \(\dinp\in[0,1]\) and \(\dsta\in[0,1]\) as
		\begin{align}
		\notag
		\fX(\dinp,\dsta)
		&\DEF\dinp^{\rnf} \dsta^{1-\rnf}+(1-\dinp)^{\rnf}(1-\dsta)^{1-\rnf}.
		\end{align}
		Given \(\epsilon\in (0,0.5)\) and the corresponding \(\rnf\in (0,1)\),
		let \(\delta\) be the unique \(\dsta\) in \((0,\epsilon)\) 
		satisfying \(\fX(\epsilon,\dsta)=(1-\sfrac{\epsilon}{2})^{\rnf}\).
		Such a \(\dsta\) exists 
		because \(\fX(\epsilon,0)=(1-\epsilon)^{\rnf}\), \(\fX(\epsilon,\epsilon)=1\)
		and \(\fX(\epsilon,\dsta)\) is 
		monotone increasing and continuous in \(\dsta\) on  \([0,\epsilon]\). 
		On the other hand \(\fX(\dinp,\dsta)<\fX(\epsilon,\delta)=(1-\sfrac{\epsilon}{2})^{\rnf}\)
		for any \(\dsta\in[0,\delta)\) and \(\dinp\in [\epsilon,1]\)
		because
		\(\fX(\dinp,\dsta)\) is monotone increasing in \(\dsta\) on  \([0,\dinp]\) for any \(\dinp\in(0,1]\)
		and monotone decreasing in \(\dinp\) on  \([\dsta,1]\) for any \(\dsta\in[0,1)\).
		Hence, using  \eqref{eq:compact-N-d} we can conclude that 
		if \(\qmn{\rnf,\cha}(\oev)<\delta\) for a \(\oev\in \outA\), 
		then \(\mW(\oev)<\epsilon\) for all \(\mW\in \cha\).
		In the following we use this property to construct a \(\mean\) such that \(\cha \UAC \mean\).
		
		Let \(\mean\) be \(\sum\nolimits_{\ind\in\integers{+}} \tfrac{\qmn{\rnf_{\ind},\cha}}{2^{\ind}}\)
		where \((\rnf_{\ind},\delta_{\ind})\) is the pair associate with \(\epsilon\!=\!\tfrac{1}{\ind}\).
		Then for any \(\oev\!\in\!\outA\) and \(\ind\!\in\!\integers{+}\), 
		if \(\mean(\oev)\!\leq\!\tfrac{\delta_{\ind}}{2^{\ind}}\),
		then \(\qmn{\rnf_{\ind},\cha}(\oev)\!\leq\!\delta_{\ind}\) 
		and consequently \(\mW(\oev)\leq \sfrac{1}{\ind}\) for all \(\mW\) in \(\cha\). 
		Thus for any \(\epsilon>0\) 
		if \(\mean(\oev)\leq\tfrac{\delta_{\lceil\sfrac{1}{\epsilon}\rceil}}{2^{\lceil\sfrac{1}{\epsilon}\rceil}}\) 
		for an \(\oev\in\outA\), then \(\mW(\oev)<\epsilon\) for all \(\mW\) in \(\cha\). 
	\end{enumerate}
\end{proof}
\end{appendix}
\section*{Acknowledgment}
The author would like to thank Fatma Nakibo\u{g}lu and Mehmet Nakibo\u{g}lu for their hospitality; 
this work simply would not have been possible without it.  
The author would like to thank 
Imre Csisz\'{a}r for pointing out Agustin's work at Austin in 2010 ISIT,
Harikrishna R. Palaiyanur for sending him Augustin's manuscript \cite{augustin78},
Reviewer I for pointing out \cite{hoY09,konigW09,mosonyiH11,mosonyiO17,ohyaPW97,sharmaW13,wildeWY14},
and
G\"{u}ne\c{s} Nakibo\u{g}lu, Robert G. Gallager, Hao-Chung Cheng,
and the reviewers for their suggestions on the manuscript.
\bibliographystyle{plain} 
\bibliography{main}
\onecolumn
\addcontentsline{toc}{subsection}{Proofs Omitted From IT Transactions Submission}
\addtocontents{toc}{\setcounter{tocdepth}{-1}}

\section*{Proofs Omitted From IT Transactions Submission}
In the following, unless specified explicitly to be otherwise all \(\sum_{\mW}\), 
\(\prod_{\mW}\), \(\vee_{\mW}\), \(\max_{\mW}\), \(\min_{\mW}\) 
stand for the corresponding expression with the subscript 
``\(\mW:\mP(\mW)>0\).'' 
\subsection[The Proofs About \(\mmn{\rno,\mP}\)]{Proofs of the Lemmas on the Mean Measure}\label{sec:powermean-proofs}
\begin{proof}[Proof of Lemma \ref{lem:powermeanequivalence}]~
	\begin{enumerate}
		\item[(\ref{lem:powermeanequivalence}-\ref{powermeanequivalence-a})]
		For any  \(\tilde{\mW}\) such that \(\mP(\tilde{\mW})>0\) and \(\rno\in \reals{+}\), the following inequalities hold \(\rfm\)-a.e.
		\begin{align}
		\notag
		(\mP(\tilde{\mW}))^{\sfrac{1}{\rno}} \der{\tilde{\mW}}{\rfm}
		\leq
		\left(\sum\nolimits_{\mW}  \mP(\mW)  \left(\der{\mW}{\rfm}\right)^{\rno} \right)^{\sfrac{1}{\rno}}
		&\leq 
		\bigvee\nolimits_{\mW}\der{\mW}{\rfm}
		\leq
		\sum\nolimits_{\mW} \der{\mW}{\rfm}.
		\end{align} 
		Then for any  \(\tilde{\mW}\) such that \(\mP(\tilde{\mW})>0\), \(\rno \in (0,\infty]\), and \(\oev \in \outA\),
		\begin{align}
		\notag
		(\mP(\tilde{\mW}))^{\frac{1}{\rno}} \tilde{\mW}(\oev)
		\leq
		\mmn{\rno,\mP}(\oev)
		&\leq 
		\sum\nolimits_{\mW} \mW(\oev).
		\end{align} 
		Thus for any \(\rno\in (0,\infty]\),  \(\mmn{\rno,\mP}(\oev)=0\) iff \(\mW(\oev)=0\) for all \(\mW\) such that 
		\(\mP(\mW)>0\).
		Then \(\mmn{1,\mP}\sim \mmn{\rno,\mP}\) for all \(\rno\in (0,\infty]\). 
		
		Note that \(\lon{\mW}=1\) for all \(\mW\) in \(\pmea{\outA}\), then \(\lon{\mmn{1,\mP}}=1\) for all \(\mP\).
		Furthermore, there exists a \(\tilde{\mW}\) such that \(\mP(\tilde{\mW})\geq \tfrac{1}{\abs{\supp{\mP}}}\)
		for all \(\mP\),
		then \(\abs{\supp{\mP}} ^{-\sfrac{1}{\rno}}\leq\lon{\mmn{\rno,\mP}}\leq \abs{\supp{\mP}}\).
		
		\item[(\ref{lem:powermeanequivalence}-\ref{powermeanequivalence-b})]
		As a result of the H\"{o}lder's inequality,
		\begin{align}
		\notag
		\mmn{0,\mP}(\oev)
		&=\int_{\oev} \prod\nolimits_{\mW}\left(\der{\mW}{\rfm}\right)^{\mP(\mW)} \rfm(\dif{\dout})
		\\
		\notag
		&\leq \prod\nolimits_{\mW}
		\left( \int_{\oev} \der{\mW}{\rfm}\rfm(\dif{\dout})\right)^{\mP(\mW)} 
		\\
		\notag
		&= \prod\nolimits_{\mW}\left(\mW(\oev)\right)^{\mP(\mW)}. 
		\end{align}
		Then \(\mmn{0,\mP}(\oev)=0\) whenever \(\mW(\oev)=0\) and \(\mmn{0,\mP} \AC \mW\)
		for all \(\mW\) such that \(\mP(\mW)>0\).
		Since \(\mW(\outS)=1\) for all \(\mW\) in \(\pmea{\outA}\), \(\lon{\mmn{0,\mP}}=\mmn{0,\mP}(\outS)\leq 1\). 
	\end{enumerate}
\end{proof}

\begin{proof}[Proof of Lemma \ref{lem:powermeandensityO}]~
	\begin{enumerate}
		\item[(\ref{lem:powermeandensityO}-\ref{powermeandensityO-a})]
		Let us establish the expressions for \(\nmn{\rno,\mP}(\dout)\) and \(\tpn{\rno}(\mW|\dout)\), first.
		Note that
		\(\tpn{1}(\mW|\dout)=\mP(\mW)\der{\mW}{\mmn{1,\mP}}\) for all \(\mW\) such that 
		\(\mP(\mW)>0\) by the definition of \(\tpn{\rno}(\mW|\dout)\) given in  \eqref{eq:def:posteroir}.
		Then  the expressions for \(\nmn{\rno,\mP}(\dout)\) follows from the definitions of 
		\(\der{\mmn{\rno,\mP}}{\rfm}\) and \(\nmn{\rno,\mP}\)
		given in  \eqref{eq:def:powermean-density} and \eqref{eq:def:powermeandensity}, respectively.
		
		On the other hand, \(\mW\AC \mmn{1,\mP}\) for all \(\mW\)  such that \(\mP(\mW)>0\) by definition and \(\mmn{\rno,\mP}\sim \mmn{1,\mP}\) by Lemma \ref{lem:powermeanequivalence}. 
		Thus,
		\begin{align}
		\notag
		\der{\mW}{\mmn{\rno,\mP}}
		&=\der{\mmn{1,\mP}}{\mmn{\rno,\mP}}\der{\mW}{\mmn{1,\mP}}
		\\
		\notag
		&=\tfrac{1}{\nmn{\rno,\mP}}\tfrac{\tpn{1}(\mW|\dout)}{\mP(\mW)}
		&&\forall \mW:\mP(\mW)>0.
		\end{align}
		Then the expression for \(\tpn{\rno}(\mW|\dout)\) follows from its definition
		given in \eqref{eq:def:posteroir}.
		
		In order to bound \(\nmn{\rno,\mP}\) from below and from above \(\mmn{1,\mP}\)-a.e. we use the expression 
		for \(\nmn{\rno,\mP}\) we have just derived.
		Note that \(\left(\sum\nolimits_{\mW} \tpn{1}(\mW|\dout)^{\rno}\right)^{\sfrac{1}{\rno}}\geq 1\) 
		for \(\rno\in(0,1]\). Then
		\begin{align}
		\notag 
		\left(\sum\nolimits_{\mW} \tpn{1}(\mW|\dout)^{\rno}  \mP(\mW)^{1-\rno}\right)^{\sfrac{1}{\rno}}
		&\geq \left(\sum\nolimits_{\mW} \tpn{1}(\mW|\dout)^{\rno} \delta^{1-\rno} \right)^{\sfrac{1}{\rno}}
		\\
		\notag
		&\geq \delta^{\frac{1-\rno}{\rno}}.
		\end{align}
		On the other hand, as a result of the H\"{o}lder's inequality we have
		\begin{align}
		\notag
		\left(\sum\nolimits_{\mW} \tpn{1}(\mW|\dout)^{\rno}  \mP(\mW)^{1-\rno}\right)^{\sfrac{1}{\rno}}
		&\leq 
		\left(\sum\nolimits_{\mW} \tpn{1}(\mW|\dout)\right)
		\left(\sum\nolimits_{\mW} \mP(\mW)\right)^{\frac{1-\rno}{\rno}}
		\\
		\notag
		&= 1.
		\end{align}
		Thus \(\delta^{\frac{1-\rno}{\rno}}\leq \nmn{\rno,\mP}\leq 1\) for \(\rno\in (0,1]\).
		
		In order to obtain the bound for \(\rno\) in \([1,\infty)\), we use 
		the identity
		\(\left(\sum\nolimits_{\mW} \tpn{1}(\mW|\dout)^{\rno}\right)^{\sfrac{1}{\rno}}\leq 1\), 
		which is valid for all \(\rno\) in \([1,\infty)\),
		together with the reverse H\"{o}lder's inequality.

		\item[(\ref{lem:powermeandensityO}-\ref{powermeandensityO-b})]  
		\(\tfrac{\tpn{1}(\mW|\dout)}{\mP(\mW)}\) is a non-negative real number for all \(\mW\) such that \(\mP(\mW)>0\)
		and \(\tfrac{\tpn{1}(\mW|\dout)}{\mP(\mW)}\) is positive at least for one such  \(\mW\). 
		Then expression for \(\nmn{\rno,\mP}\) given  in part (\ref{powermeandensityO-a}) is a smooth 
		function\footnote{For any positive integer \(K\), non-negative real numbers \(\mA_{\ind}\) and \(\mB_{\ind}\) for 
			\(\ind\) in \(\{1,2,\ldots,K\}\), the function \((\sum_{\ind=1}^{K} \mA_{\ind} \mB_{\ind}^{\rno})^{\sfrac{1}{\rno}}\) 
			is a smooth function of \(\rno\) on \(\reals{+} \), because
			the exponential function and the logarithm are smooth functions
			and
			composition, sum, and product of smooth functions are also smooth.}
		of \(\rno\) on \(\reals{+} \). Identities for the derivatives of \(\nmn{\rno,\mP}\)
		follow from the chain rule and elementary rules of differentiation.
		
		\item[(\ref{lem:powermeandensityO}-\ref{powermeandensityO-c})]  
		As a result of the H\"{o}lder's inequality we have,
		\begin{align}
		\notag
		\left(\nmn{\rno_{\beta},\mP}\right)^{\rno_{\beta}}
		&=\sum\nolimits_{\mW}  \mP(\mW)  
		\left(\tfrac{\tpn{1}(\mW|\dout)}{\mP(\mW)}\right)^{\rno_{1} \beta}
		\left(\tfrac{\tpn{1}(\mW|\dout)}{\mP(\mW)}\right)^{\rno_{0}(1-\beta)}
		\\
		\notag
		&\leq 
		\left(\sum\nolimits_{\mW}  \mP(\mW)  
		\left(\tfrac{\tpn{1}(\mW|\dout)}{\mP(\mW)}\right)^{\rno_{1}}
		\right)^{\beta}
		\left(\sum\nolimits_{\mW}  \mP(\mW)  
		\left(\tfrac{\tpn{1}(\mW|\dout)}{\mP(\mW)}\right)^{\rno_{0}}
		\right)^{(1-\beta)}
		\\
		\notag
		&=(\nmn{\rno_{1},\mP})^{\rno_1\beta}(\nmn{\rno_{0},\mP})^{\rno_0(1-\beta)}
		\end{align}
		Furthermore, the inequality is strict unless there exists a \(\gamma\) such that
		\(\mP(\mW)(\tfrac{\tpn{1}(\mW|\dout)}{\mP(\mW)})^{\rno_1}=\gamma \mP(\mW)(\tfrac{\tpn{1}(\mW|\dout)}{\mP(\mW)})^{\rno_0}\)
		for all \(\mW\) such that \(\tpn{1}(\mW|\dout)>0\). 
		Thus inequality is strict iff there exist \(\mW, \tilde{\mW}\in\supp{\mP}\)  
		such that \(\tfrac{\tpn{1}(\mW|\dout)}{\mP(\mW)}>\tfrac{\tpn{1}(\tilde{\mW}|\dout)}{\mP(\tilde{\mW})}>0\).

		\item[(\ref{lem:powermeandensityO}-\ref{powermeandensityO-d})]  
		The continuity of \(\nmn{\rno,\mP}\) in \(\rno\) on \(\reals{+} \) follows from 
		the smoothness of \(\nmn{\rno,\mP}\) established in part (\ref{powermeandensityO-b}). 
		In order to show the continuity on \([0,\infty]\) we need to establish the continuity at zero and at infinity. 
		Note that \(\dinp^{\rno}\) is a smooth function of \(\rno\) for any \(\dinp\in\reals{+} \) 
		and weighted sums of smooth functions are also smooth.
		Thus \((\nmn{\rno,\mP})^{\rno}(\dout)\) is a smooth function of \(\rno\) 
		and we can use L'Hospital's rule \cite[Thm. 5.13]{rudin} 
		for calculating the limits of \(\nmn{\rno,\mP}(\dout)\) at zero and infinity:  
		\begin{align}
		\label{eq:nmn-limzero}
		\lim_{\rno \to 0} \left(\sum\nolimits_{\mW} \mP(\mW) \left(\tfrac{\tpn{1}(\mW|\dout)}{\mP(\mW)}\right)^{\rno} \right)^{\sfrac{1}{\rno}}
		&=\prod\nolimits_{\mW}    \left(\tfrac{\tpn{1}(\mW|\dout)}{\mP(\mW)}\right)^{\mP(\mW)}
		\\
		\label{eq:nmn-liminfty}
		\lim_{\rno \to \infty} \left(\sum\nolimits_{\mW} \mP(\mW) \left(\tfrac{\tpn{1}(\mW|\dout)}{\mP(\mW)}\right)^{\rno} \right)^{\sfrac{1}{\rno}}
		&=\max\nolimits_{\mW}  \tfrac{\tpn{1}(\mW|\dout)}{\mP(\mW)}
		\end{align}
		Thus \(\lim\nolimits_{\rno \to 0} \nmn{\rno,\mP}(\dout)=\nmn{0,\mP}(\dout)\)
		and \(\lim\nolimits_{\rno \to \infty} \nmn{\infty,\mP}(\dout)=\nmn{\infty,\mP}(\dout)\)
		hold \(\mmn{1,\mP}\) almost everywhere.
		Thus \(\nmn{\rno,\mP}(\dout)\) is continuous on \([0,\infty]\). 
		
		On the other hand, using the Jensen's inequality and the convexity of the function \(\ln \sfrac{1}{\dinp}\) we get,
		\begin{align}
		\notag
		\dnmn{\rno,\mP} 
		&\geq - \tfrac{\nmn{\rno,\mP}}{\rno^2} \ln \sum\nolimits_{\mW:\tpn{1}(\mW|\dout)>0}\mP(\mW) 
		\\
		\label{eq:dnmn-positivity}
		&\geq  0.
		\end{align}
		Since the function \(\ln\sfrac{1}{\dinp}\) is strictly convex, the first inequality is strict 
		and \(\dnmn{\rno,\mP}(\dout)\) is positive unless \(\mP(\mW)=\tpn{1}(\mW|\dout)\) 
		for all \(\mW\) such that \(\mP(\mW)>0\). 
		Thus \(\nmn{\rno,\mP}(\dout)\) is monotone increasing in \(\rno\) unless \(\mP(\mW)=\tpn{1}(\mW|\dout)\) 
		for all \(\mW\) such that \(\mP(\mW)>0\). 
		Boundedness is already established in part (\ref{powermeandensityO-a}).
	\end{enumerate}
\end{proof}

\begin{proof}[Proof of Lemma \ref{lem:powermeanO}]~
	\begin{enumerate}
		\item[(\ref{lem:powermeanO}-\ref{powermeanO-a})] 
		For all \(\dout\in\outS\) ---except for a \(\mmn{1,\mP}\)-measure zero set---
		density \(\nmn{\rno,\mP}\) is a non-negative function of \(\rno\) continuous on \([0,\infty]\)
		by Lemma \ref{lem:powermeandensityO}-(\ref{powermeandensityO-d}).
		Thus for any sequence \(\{\rno_{\ind}\}\) such that \(\lim_{\ind\to\infty} \rno_{\ind}= \rno\) 
		we have \(\lim_{\ind\to\infty}\nmn{\rno_{\ind},\mP}=\nmn{\rno,\mP}\) \(\mmn{1,\mP}\)-a.e. 
		Since \(\nmn{\rno_{\ind},\mP}\leq \nmn{\infty,\mP}\) by Lemma \ref{lem:powermeandensityO}-(\ref{powermeandensityO-d})
		and \(\nmn{\infty,\mP}\leq \tfrac{1}{\min_{\mW} \mP(\mW)}\) by Lemma \ref{lem:powermeandensityO}-(\ref{powermeanO-a}),
		we can apply the dominated convergence theorem \cite[2.8.1]{bogachev}.
		Thus \(\{\nmn{\rno_{\ind},\mP}\}\xrightarrow{\Lon{\mmn{1,\mP}}}\nmn{\rno,\mP}\), i.e.  
		\begin{align}
		\notag
		\lim_{\ind \to \infty}
		\int \abs{\nmn{\rno_{\ind},\mP}-\nmn{\rno,\mP}}
		\mmn{1,\mP}(\dif{\dout})
		&=0.
		\end{align}
		Then \(\{\mmn{\rno_{\ind},\mP}\}\) converges  to \(\mmn{\rno,\mP}\) in the total variation topology,
		for any sequence \(\{\rno_{\ind}\}\) such that \(\lim_{\ind\to\infty}\rno_{\ind}=\rno\). 
		Then \(\mmn{\rno,\mP}\) is a continuous function of \(\rno\) from \([0,\infty]\) with  its usual topology to \(\zmea{\outA}\) 
		with the total variation topology because  \([0,\infty]\) with its usual topology is a metrizable space, see \cite[Thm. 21.3]{munkres}.
		
		\item[(\ref{lem:powermeanO}-\ref{powermeanO-b})] 
		For \(\dmn{\rno,\mP}\) defined in  \eqref{eq:def:dpowermean} to be a finite measure, \(\dnmn{\rno,\mP}\) should be a non-negative 
		\(\mmn{1,\mP}\)-integrable function.
		The density \(\dnmn{\rno,\mP}\) is  non-negative by  \eqref{eq:dnmn-positivity}.
		By the expression for \(\dnmn{\rno,\mP}\) given Lemma \ref{lem:powermeandensityO}-(\ref{powermeandensityO-b})
		and 
		the bound for \(\dnmn{\rno,\mP}\) given in  Lemma \ref{lem:powermeandensityO}-(\ref{powermeandensityO-a})
		we have 
		\begin{align}
		\notag
		\dnmn{\rno,\mP}
		&=\tfrac{\nmn{\rno,\mP}}{\rno^2} \sum\nolimits_{\mW} \tpn{\rno}(\mW|\dout) \ln\tfrac{\tpn{\rno}(\mW|\dout)}{\mP(\mW)}
		\\
		\notag
		&\leq \tfrac{\nmn{\rno,\mP}}{\rno^2} \ln \tfrac{1}{\min_{\mW} \mP(\mW)}.
		\\
		\notag
		&=\tfrac{1}{\rno^2}\tfrac{1}{\min_{\mW} \mP(\mW)}\ln \tfrac{1}{\min_{\mW} \mP(\mW)}
		\end{align}
		Thus \(\dnmn{\rno,\mP}\) is bounded and \(\dmn{\rno,\mP}\) is a finite measure,  
		i.e. \(\dmn{\rno,\mP}\in\zmea{\outA}\).
		We can apply the dominated convergence theorem \cite[2.8.1]{bogachev}  
		for \(\dmn{\rno,\mP}\) as we did for \(\mmn{\rno,\mP}\) 
		in part (\ref{powermeanO-a}) in order to establish the continuity of \(\dmn{\rno,\mP}\) as a function of \(\rno\).
		Furthermore, \(\left.\der{}{\rno}\mmn{\rno,\mP}(\oev)\right\vert_{\rno=\rnf}=\dmn{\rnf,\mP}(\oev)\)
		follows from the boundedness of \(\dnmn{\rno,\mP}\)
		and the definitions of \(\dmn{\rno,\mP}\) and \(\dnmn{\rno,\mP}\) 
		by \cite[Cor. 2.8.7.(ii)]{bogachev} for \(X=\oev\).
		One can apply the Tonelli-Fubini theorem \cite[4.4.5]{dudley} to obtain an equivalent result,
		instead of invoking \cite[Cor. 2.8.7.(ii)]{bogachev}.

		\item[(\ref{lem:powermeanO}-\ref{powermeanO-c})] 
		For \(\ddmn{\rno,\mP}\) defined in \eqref{eq:def:ddpowermean} to be a finite signed measure, \(\ddnmn{\rno,\mP}\) should be 
		a \(\mmn{1,\mP}\)-integrable function. 
		By the expression for \(\ddnmn{\rno,\mP}\) given in Lemma \ref{lem:powermeandensityO}-(\ref{powermeandensityO-b})
		we have 
		\begin{align}
		\label{eq:ddnmn-boundedness}
		-\tfrac{2}{\rno^3} \left(\ln \tfrac{1}{\min_{\mW} \mP(\mW)}\right)  \nmn{\rno,\mP}
		\leq 
		\ddnmn{\rno,\mP}
		&\leq \left[\tfrac{1+\rno}{\rno^4} \left(\ln \tfrac{1}{\min_{\mW} \mP(\mW)}\right)^2
		+\tfrac{4}{e^{2}\rno^{3}} \right]\nmn{\rno,\mP}
		\end{align}
		The proof of the continuity is similar to the corresponding proofs in parts 
		(\ref{powermeanO-a}) and (\ref{powermeanO-b}).
		The identity \(\left.\der{}{\rno}\dmn{\rno,\mP}(\oev)\right\vert_{\rno=\rnf}=\ddmn{\rnf,\mP}(\oev)\) 
		follows from  \eqref{eq:ddnmn-boundedness} by applying \cite[Cor. 2.8.7.(ii)]{bogachev} for \(X=\oev\).

		\item[(\ref{lem:powermeanO}-\ref{powermeanO-d})] 
		For any \(\beta \in [0,1]\) and \(\rno_0,\rno_1 \in \reals{+} \) let \(\rno_{\beta}\) 
		be \(\rno_{\beta}=\beta \rno_{1}+(1-\beta)\rno_{0}\). Then as a result of the H\"{o}lder's inequality,
		\begin{align}
		\label{eq:mean-d:Y}
		\int 
		\left(\nmn{\rno_{1},\mP} \right)^{\frac{\beta \rno_1}{\rno_{\beta}}}
		\left(\nmn{\rno_{0},\mP} \right)^{\frac{(1-\beta)\rno_0}{\rno_{\beta}}}
		\mmn{1,\mP}(\dif{\dout})
		&\leq
		(\lon{\mmn{\rno_{1},\mP}})^{\frac{\beta \rno_1}{\rno_{\beta}}}
		(\lon{\mmn{\rno_{0},\mP}})^{\frac{(1-\beta)\rno_0}{\rno_{\beta}}}.
		\end{align}
		On the other hand by Lemma \ref{lem:powermeandensityO}-(\ref{powermeandensityO-c})
		\begin{align}
		\label{eq:mean-d:Z}
		\nmn{\rno_{\beta},\mP}
		&\leq
		(\nmn{\rno_{1},\mP})^{\frac{\beta\rno_{1}}{\rno_{\beta}}}(\nmn{\rno_{0},\mP})^{\frac{(1-\beta)\rno_{0}}{\rno_{\beta}}}.
		\end{align}
		Then the log-convexity of \(\lon{\mmn{\rno,\mP}}^{\rno} \) as a function of \(\rno\) follows from 
		\eqref{eq:mean-d:Y} and \eqref{eq:mean-d:Z}.

		If \(\mmn{1,\mP}(\cup_{\gamma\geq 1} \set{A}(\mP,\gamma))\!<\!1\), then the log-convexity of \(\lon{\mmn{\rno,\mP}}^{\rno}\) 
		is strict because the inequality in \eqref{eq:mean-d:Z} is strict for \(\dout\)'s that are not in 
		\(\cup_{\gamma\geq 1}\!\set{A}(\mP,\gamma)\)  by Lemma \ref{lem:powermeandensityO}-(\ref{powermeandensityO-c}).
		For \(\dout \in \set{A}(\mP,\gamma)\), the inequality in \eqref{eq:mean-d:Z} is an equality and  
		\(\nmn{\rno,\mP}=\gamma^{\frac{\rno-1}{\rno}}\) for all \(\rno\).
		Consequently if \(\mmn{1,\mP}(\cup_{\gamma\geq 1} \set{A}(\mP,\gamma))=1\),
		then the log-convexity of \(\lon{\mmn{\rno,\mP}}^{\rno}\) is strict iff the inequality in \eqref{eq:mean-d:Y} is strict. But  
		if \(\mmn{1,\mP}(\cup_{\gamma\geq 1} \set{A}(\mP,\gamma))=1\), then the H\"{o}lder's inequality in \eqref{eq:mean-d:Y} is strict unless 
		there exists a \(\gamma\geq1\) such that \(\mmn{1,\mP}(\set{A}(\mP,\gamma))=1\).
		
		We proceed with calculating the limit at zero.
		As a result of the expression for \(\nmn{\rno,\mP}\) given in part (\ref{powermeandensityO-a}) 
		we have,
		\begin{align}
		\notag 
		\left(\sum\nolimits_{\mW:\tpn{1}(\mW|\dout)>0} \mP(\mW)\right)^{\frac{\rno-1}{\rno}}\nmn{\rno,\mP}
		&=\left(\sum\nolimits_{\mW} \tpn{0}(\mW|\dout)  \left( \tfrac{\tpn{1}(\mW|\dout)}{\tpn{0}(\mW|\dout)}\right)^{\rno}\right)^{\frac{1}{\rno}}
		&
		&\mbox{where}
		&
		\tpn{0}(\mW|\dout)
		&=\tfrac{\mP(\mW)}{\sum\nolimits_{\tilde{\mW}:\tpn{1}(\tilde{\mW}|\dout)>0} \mP(\tilde{\mW})}.
		\end{align}
		
		Then using L'Hospital's rule \cite[Thm. 5.13]{rudin} for calculating limits and the H\"{o}lder's inequality we get,
		\begin{align}
		\label{eq:mean-d:A}
		\lim\limits_{\rno\to 0}
		\left(\sum\nolimits_{\mW:\tpn{1}(\mW|\dout)>0} \mP(\mW)\right)^{\frac{\rno-1}{\rno}}\nmn{\rno,\mP}
		&=e^{\sum_{\mW} \tpn{0}(\mW|\dout)  \ln \tfrac{\tpn{1}(\mW|\dout)}{\tpn{0}(\mW|\dout)}}
		&
		&\mmn{1,\mP}-a.e.
		\\
		\label{eq:mean-d:B}
		\left(\sum\nolimits_{\mW:\tpn{1}(\mW|\dout)>0} \mP(\mW)\right)^{\frac{\rno-1}{\rno}}\nmn{\rno,\mP}
		&\leq 1
		&
		&\forall \rno\in (0,1),~ \mmn{1,\mP}-a.e.
		\end{align}

		The sum \(\sum\nolimits_{\mW:\tpn{1}(\mW|\dout)>0} \mP(\mW)\) is a simple function of \(\dout\), i.e. its range is a finite set,
		because \(\supp{\mP}\) has a finite number of distinct subsets. Thus the essential supremum is the maximum value of 
		the sum with positive probability. Therefore 
		\begin{align}
		\label{eq:mean-d:C}
		\mmn{1,\mP}\left(\left\{\sum\nolimits_{\mW:\tpn{1}(\mW|\dout)>0} \mP(\mW)=\psi\right\}\right)
		&>0
		&
		&\mbox{where}
		&
		\psi
		&=\essup_{\mmn{1,\mP}}\sum\nolimits_{\mW:\tpn{1}(\mW|\dout)>0}\mP(\mW).
		\end{align}
		Then using  \eqref{eq:mean-d:A} we get
		\begin{align}
		\notag
		\lim_{\rno \to 0} \psi^{\frac{\rno-1}{\rno}}\nmn{\rno,\mP}
		&=\IND{\sum\nolimits_{\mW:\tpn{1}(\mW|\dout)>0}\mP(\mW)=\psi}e^{\sum_{\mW} \tpn{0}(\mW|\dout)  \ln \tfrac{\tpn{1}(\mW|\dout)}{\tpn{0}(\mW|\dout)}}
		&
		&\mmn{1,\mP}\mbox{-a.e.}
		\end{align}
		On the other hand \(\psi^{\frac{\rno-1}{\rno}}\nmn{\rno,\mP}\leq 1\) for all \(\rno\in (0,1)\), \(\mmn{1,\mP}-\)a.e. 
		by \eqref{eq:mean-d:B} and the definition of \(\psi\) given in \eqref{eq:mean-d:C}. 
		Thus we can apply the dominated convergence theorem \cite[2.8.1]{bogachev}:
		\begin{align}
		\label{eq:mean-d:D}
		\lim_{\rno\to 0} \int \abs{\psi^{\frac{\rno-1}{\rno}}\nmn{\rno,\mP}-\IND{\sum\nolimits_{\mW:\tpn{1}(\mW|\dout)>0}\mP(\mW)=\psi}e^{\sum_{\mW} \tpn{0}(\mW|\dout)  \ln \tfrac{\tpn{1}(\mW|\dout)}{\tpn{0}(\mW|\dout)}}}
		\mmn{1,\mP}(\dif{\dout})=0. 
		\end{align}
		Consequently,
		\begin{align}
		\label{eq:mean-d:E}
		\lim_{\rno\to 0} \left(\psi^{(\rno-1)} \lon{\nmn{\rno,\mP}}^{\rno} \right)^{\frac{1}{\rno}}  
		&=\int \IND{\sum\nolimits_{\mW:\tpn{1}(\mW|\dout)>0}\mP(\mW)=\psi}e^{\sum_{\mW} \tpn{0}(\mW|\dout)  \ln \tfrac{\tpn{1}(\mW|\dout)}{\tpn{0}(\mW|\dout)}}\mmn{1,\mP}(\dif{\dout}). 
		\end{align}
		The right hand side of \eqref{eq:mean-d:E} is a real number between \(0\)  and \(1\) by  \eqref{eq:mean-d:C}. Thus we have,
		\begin{align}
		\label{eq:mean-d:F}
		\lim_{\rno\to 0} \psi^{\rno-1} \lon{\nmn{\rno,\mP}}^{\rno} 
		&=1.
		\end{align}
		
		\item[(\ref{lem:powermeanO}-\ref{powermeanO-e})] 
		\(\lon{\mmn{\rno,\mP}}\leq \abs{\supp{\mP}}\) by Lemma \ref{lem:powermeanequivalence}-(\ref{powermeanequivalence-a}).
		The continuity of \(\lon{\mmn{\rno,\mP}}\) in \(\rno\) is implied 
		by the continuity of \(\mmn{\rno,\mP}\) in \(\rno\) for the total variation topology on \(\zmea{\outA}\),
		proved in part (\ref{powermeanO-a}).
		Furthermore, \(\lon{\mmn{\rno,\mP}}=\mmn{\rno,\mP}(\outS)\)
		because \(\mmn{\rno,\mP}\in\zmea{\outA}\) by part (\ref{powermeanO-a}).
		In addition \(\der{}{\rno}\mmn{\rno,\mP}(\outS)\geq 0\) 
		by part (\ref{powermeanO-b}).
		Hence \(\lon{\mmn{\rno,\mP}}\) is  a nondecreasing function of \(\rno\).

		Let \(\oev_{\mP}\) be \(\oev_{\mP}=\{\dout:\tpn{1}(\cdot|\dout)\neq \mP(\cdot)\}\).
		Then 
		\(\forall \dout \in \oev_{\mP}\), 
		\(\tpn{\rno}(\cdot|\dout)\neq \mP(\cdot)\) and  \(\nmn{\rno,\mP}\) is monotone increasing in \(\rno\) on \(\reals{\geq0} \). 
		On the other hand, if there are two or more distinct \(\mW\)'s in \(\supp{\mP}\), 
		then \(\mmn{1,\mP}(\oev_{\mP})>0\).
		Thus \(\lon{\mmn{\rno,\mP}}\) is monotone increasing if there exist 
		\(\mW,\widetilde{\mW}\in \supp{\mP}\) such that  \(\mW\neq \widetilde{\mW}\). Else \(\nmn{\rno,\mP}=1\) thus  
		\(\lon{\mmn{\rno,\mP}}=\int \nmn{\rno,\mP} \mmn{1,\mP}(\dif{\dout})=1\) for all \(\rno\in [0,\infty]\).
	\end{enumerate}
\end{proof}

\begin{proof}[Proof of Lemma \ref{lem:powermeanP}]~
	\begin{enumerate}
		\item[(\ref{lem:powermeanP}-\ref{powermeanP-a})]
		Let us start with \(\rno=0\) case. 
		Since the weighted arithmetic mean of any two non-negative real numbers is greater than their weighted geometric mean,
		for any reference measure \(\rfm\) for \(\mmn{1,\pmn{1}}\) and \(\mmn{1,\pmn{2}}\) we have, 
		\begin{align}
		\notag
		\beta\der{\mmn{0,\pmn{1}}}{\rfm}+(1-\beta)\der{\mmn{0,\pmn{2}}}{\rfm}
		&\geq  \left(\der{\mmn{0,\pmn{1}}}{\rfm}\right)^{\beta} \left(\der{\mmn{0,\pmn{2}}}{\rfm}\right)^{1-\beta}
		\\
		\notag
		&=  \der{\mmn{0,\pmn{\beta}}}{\rfm}. 
		\end{align} 
		
		For any \(\rno \in (0,1]\) the function \(\dinp^{\sfrac{1}{\rno}}\) is convex in \(\dinp\).
		Then for any reference measure \(\rfm\) for \(\mmn{\rno,\pmn{1}}\) and \(\mmn{\rno,\pmn{2}}\) 
		as a result of the Jensen's inequality we have,
		\begin{align}
		\notag
		\beta\der{\mmn{\rno,\pmn{1}}}{\rfm}+(1-\beta)\der{\mmn{\rno,\pmn{2}}}{\rfm}
		&\geq     \left(\sum\nolimits_{\mW} (\beta\pmn{1}(\mW)+(1-\beta)\pmn{2}(\mW))  \left(\der{\mW}{\rfm}\right)^{\rno} \right)^{\sfrac{1}{\rno}}
		\\
		\notag
		&=\der{\mmn{\rno,\pmn{\beta}}}{\rfm}.
		\end{align} 
		\(\lon{\mmn{\rno,\mP}}\) is convex in \(\mP\) because \(\der{\mmn{\rno,\mP}}{\rfm}\) is convex in \(\mP\) 
		and \(\der{\mmn{\rno,\mP}}{\rfm}\) is non-negative.
		
		\item[(\ref{lem:powermeanP}-\ref{powermeanP-b})]
		For \(\rno\in [1,\infty)\) the function \(\dinp^{\sfrac{1}{\rno}}\) is concave in \(\dinp\). 
		Thus the inequalities are reversed. 
		Hence both  the Radon-Nikodym derivative \(\der{\mmn{\rno,\mP}}{\rfm}\) and the norm \(\lon{\mmn{\rno,\mP}}\) are concave in \(\mP\).
		
		For any reference measure \(\rfm\) for \(\mmn{\infty,\pmn{1}}\) and \(\mmn{\infty,\pmn{2}}\) 
		by the definition of \(\der{\mmn{\infty,\mP}}{\rfm}\) given in \eqref{eq:def:powermean-density},
		we have  
		\begin{align}
		\notag
		\beta\der{\mmn{\infty,\pmn{1}}}{\rfm}+(1-\beta)\der{\mmn{\infty,\pmn{2}}}{\rfm}
		&\leq \der{\mmn{\infty,\pmn{\beta}}}{\rfm}.
		\end{align} 
		\(\lon{\mmn{\infty,\mP}}\) is concave in \(\mP\) because \(\der{\mmn{\infty,\mP}}{\rfm}\) is concave in \(\mP\) 
		and \(\der{\mmn{\infty,\mP}}{\rfm}\) is non-negative.
		
		\item[(\ref{lem:powermeanP}-\ref{powermeanP-c})]
		Identities are confirmed using the definitions of \(\smn{\wedge}\), \(\smn{1}\) and \(\smn{2}\) by 
		substitution. On the other hand,
		\begin{align}
		\notag
		\lon{\pmn{1}-\pmn{2}}
		&=\lon{\pmn{1}\vee\pmn{2}}-\lon{\pmn{1}\wedge\pmn{2}}
		\\
		\notag
		&=2-2\lon{\pmn{1}\wedge\pmn{2}}.
		\end{align} 
		Hence \(\smn{\wedge}\in\pdis{\pmea{\outA}}\). 
		Using the fist identity together with 
		\(\smn{\wedge}\in\pdis{\pmea{\outA}}\) and \(\pmn{1}\in\pdis{\pmea{\outA}}\) we get \(\smn{1}\in\pdis{\pmea{\outA}}\). 
		Similarly \(\smn{2}\in\pdis{\pmea{\outA}}\) follows from the second identity, \(\smn{\wedge}\in\pdis{\pmea{\outA}}\) 
		and \(\pmn{1}\in\pdis{\pmea{\outA}}\).   
		
		\item[(\ref{lem:powermeanP}-\ref{powermeanP-d})]
		Let \(\delta\) be \(\delta=\tfrac{\lon{\pmn{2}-\pmn{1}}}{2}\).
		For any reference measure \(\rfm\) for \(\mmn{\rno,\pmn{1}}\) and \(\mmn{\rno,\pmn{2}}\) and \(\rno \in  (0,1]\),
		\begin{align}
		\notag
		\der{\mmn{\rno,\pmn{1}}}{\rfm}- \der{\mmn{\rno,\pmn{2}}}{\rfm}
		&=\left[(1-\delta)(\der{\mmn{\rno,\smn{\wedge}}}{\rfm})^{\rno}+ \delta (\der{\mmn{\rno,\smn{1}}}{\rfm})^{\rno}\right]^{\sfrac{1}{\rno}}
		-\left[(1-\delta)(\der{\mmn{\rno,\smn{\wedge}}}{\rfm})^{\rno}+ \delta (\der{\mmn{\rno,\smn{2}}}{\rfm})^{\rno}\right]^{\sfrac{1}{\rno}}
		\\
		\notag
		&\leq
		\left[(1-\delta)(\der{\mmn{\rno,\smn{\wedge}}}{\rfm})^{\rno}+ \delta (\der{\mmn{\rno,\smn{1}}}{\rfm})^{\rno}\right]^{\sfrac{1}{\rno}}
		-(1-\delta)^{\frac{1}{\rno}}\der{\mmn{\rno,\smn{\wedge}}}{\rfm}
		\\
		\notag
		&\leq \left[(1-\delta)\der{\mmn{\rno,\smn{\wedge}}}{\rfm} + \delta \der{\mmn{\rno,\smn{1}}}{\rfm}\right]
		-(1-\delta)^{\sfrac{1}{\rno}}\der{\mmn{\rno,\smn{\wedge}}}{\rfm}.
		\end{align}
		where the last inequality follows from the Jensen's inequality and the convexity of \(\dinp^{\sfrac{1}{\rno}}\) in \(\dinp\) for \(\rno \in (0,1]\). 
		
		We bound \(\der{\mmn{\rno,\pmn{2}}}{\rfm}- \der{\mmn{\rno,\pmn{1}}}{\rfm}\) in a similarly way. 
		Using these two bounds we can bound \(\lon{\mmn{\rno,\pmn{1}}-\mmn{\rno,\pmn{2}}}\) as follows
		\begin{align}
		\notag
		\lon{\mmn{\rno,\pmn{1}}-\mmn{\rno,\pmn{2}}}
		&=\int\nolimits_{\der{\mmn{\rno,\pmn{1}}}{\rfm}>\der{\mmn{\rno,\pmn{2}}}{\rfm}} 
		(\der{\mmn{\rno,\pmn{1}}}{\rfm}- \der{\mmn{\rno,\pmn{2}}}{\rfm}) \rfm(\dif{\dout})
		+\int\nolimits_{\der{\mmn{\rno,\pmn{2}}}{\rfm}>\der{\mmn{\rno,\pmn{1}}}{\rfm}} 
		(\der{\mmn{\rno,\pmn{2}}}{\rfm}- \der{\mmn{\rno,\pmn{1}}}{\rfm}) \rfm(\dif{\dout})
		\\
		\notag
		&\leq 2\left[1-\delta-(1-\delta)^{\sfrac{1}{\rno}}\right] \lon{\mmn{\rno,\smn{\wedge}}}
		+\delta \lon{\mmn{\rno,\smn{1}}}+ \delta \lon{\mmn{\rno,\smn{2}}}
		\\
		\notag
		&\leq 2\left[1-(1-\delta)^{\sfrac{1}{\rno}}\right]
		\\
		\notag
		&\leq \tfrac{2}{\rno}\delta.
		\end{align}
		
		\item[(\ref{lem:powermeanP}-\ref{powermeanP-e})] 
		One can confirm using the derivative test that for any 
		\(\rno \in  [1,\infty)\),  \(\dinp_0\geq 0\) and  \(\dinp_1\geq \dinp_2\geq 0\) we have
		\begin{align}
		\notag
		\left[(1-\delta)\dinp_{0}^{\rno}+ \delta \dinp_{1}^{\rno}\right]^{\sfrac{1}{\rno}}
		-\left[(1-\delta)\dinp_{0}^{\rno}+ \delta \dinp_{2}^{\rno}\right]^{\sfrac{1}{\rno}}
		&\leq \delta^{\sfrac{1}{\rno}} (\dinp_1-\dinp_2).
		\end{align}
		Then for any reference measure \(\rfm\) for \(\mmn{\rno,\pmn{1}}\) and \(\mmn{\rno,\pmn{2}}\)
		we have
		\begin{align}
		\notag
		\der{\mmn{\rno,\pmn{1}}}{\rfm}- \der{\mmn{\rno,\pmn{2}}}{\rfm}
		&=\left[(1-\delta)(\der{\mmn{\rno,\smn{\wedge}}}{\rfm})^{\rno}+ \delta (\der{\mmn{\rno,\smn{1}}}{\rfm})^{\rno}\right]^{\sfrac{1}{\rno}}
		-\left[(1-\delta)(\der{\mmn{\rno,\smn{\wedge}}}{\rfm})^{\rno}+ \delta (\der{\mmn{\rno,\smn{2}}}{\rfm})^{\rno}\right]^{\sfrac{1}{\rno}}
		\\
		\notag
		&\leq \delta^{\sfrac{1}{\rno}}
		\left[\der{\mmn{\rno,\smn{1}}}{\rfm}-\der{\mmn{\rno,\smn{2}}}{\rfm}\right].
		\end{align}
		We can bound \(\der{\mmn{\rno,\pmn{2}}}{\rfm}- \der{\mmn{\rno,\pmn{1}}}{\rfm}\) in a similarly way. 
		On the other hand \(\der{\mmn{\rno,\pmn{1}}}{\rfm}\geq \der{\mmn{\rno,\pmn{2}}}{\rfm}\) iff \(\der{\mmn{\rno,\smn{1}}}{\rfm}\geq \der{\mmn{\rno,\smn{2}}}{\rfm}\).
		Thus we can bound \(\lon{\mmn{\rno,\pmn{1}}-\mmn{\rno,\pmn{2}}}\) using the bounds on 
		\(\der{\mmn{\rno,\pmn{1}}}{\rfm}- \der{\mmn{\rno,\pmn{2}}}{\rfm}\)
		and  \(\der{\mmn{\rno,\pmn{2}}}{\rfm}- \der{\mmn{\rno,\pmn{1}}}{\rfm}\): 
		\begin{align}
		\notag
		\lon{\mmn{\rno,\pmn{1}}-\mmn{\rno,\pmn{2}}}
		&=
		\int_{\der{\mmn{\rno,\pmn{1}}}{\rfm}>\der{\mmn{\rno,\pmn{2}}}{\rfm}}
		(\der{\mmn{\rno,\pmn{1}}}{\rfm}- \der{\mmn{\rno,\pmn{2}}}{\rfm}) \rfm(\dif{\dout})
		+
		\int_{\der{\mmn{\rno,\pmn{2}}}{\rfm}>\der{\mmn{\rno,\pmn{1}}}{\rfm}}
		(\der{\mmn{\rno,\pmn{2}}}{\rfm}- \der{\mmn{\rno,\pmn{1}}}{\rfm}) \rfm(\dif{\dout})
		\\
		\notag
		&\leq \delta^{\sfrac{1}{\rno}}
		\int_{\der{\mmn{\rno,\smn{1}}}{\rfm}>\der{\mmn{\rno,\smn{2}}}{\rfm}}
		(\der{\mmn{\rno,\smn{1}}}{\rfm}- \der{\mmn{\rno,\smn{2}}}{\rfm}) \rfm(\dif{\dout})
		+\delta^{\sfrac{1}{\rno}} 
		\int_{\der{\mmn{\rno,\smn{2}}}{\rfm}>\der{\mmn{\rno,\smn{1}}}{\rfm}}
		(\der{\mmn{\rno,\smn{2}}}{\rfm}- \der{\mmn{\rno,\smn{1}}}{\rfm}) \rfm(\dif{\dout})
		\\
		\notag
		&=\delta^{\sfrac{1}{\rno}} \lon{\mmn{\rno,\smn{1}}-\mmn{\rno,\smn{2}}}.
		\end{align}
	\end{enumerate}
\end{proof}

\subsection[The Proofs About \(\RMI{\rno}{\mP}{\cha}\)]{Proofs of the Lemmas on the \renyi Information}\label{sec:information-proofs}
\begin{proof}[Proof of Lemma \ref{lem:informationO}]
	\(\RMI{\infty}{\rno}{\cha}\leq \ln \abs{\supp{\mP}}\) because \(\lon{\mmn{\infty,\mP}}\leq \abs{\supp{\mP}}\) 
	by Lemma \ref{lem:powermeanO}-(\ref{powermeanO-e}). 
	
	\(\lim_{\rno \downarrow0}\RMI{\rno}{\mP}{\cha}=\RMI{0}{\mP}{\cha}\)
	follows from  Lemma \ref{lem:powermeanO}-(\ref{powermeanO-d})
	and the definition of \(\RMI{\rno}{\mP}{\cha}\) given in  \eqref{eq:def:information}.
	
	\(\lim_{\rno\uparrow\infty}\RMI{\rno}{\mP}{\cha}=\RMI{\infty}{\mP}{\cha}\)
	follows from the continuity of \(\lon{\mmn{\rno,\mP}}\) as a function of \(\rno\) 
	at infinity, i.e. Lemma \ref{lem:powermeanO}-(\ref{powermeanO-e}), 
	and the definition of 
	\(\RMI{\rno}{\mP}{\cha}\) given in  \eqref{eq:def:information}.

	Both  \(\lon{\mmn{\rno,\mP}}\) and \(\lon{\dmn{\rno,\mP}}\) are continuously differentiable on \(\reals{+} \),
	\(\der{}{\rno}\lon{\mmn{\rno,\mP}}=\lon{\dmn{\rno,\mP}}\) and \(\der{}{\rno}\lon{\dmn{\rno,\mP}}=\ddmn{\rno,\mP}(\outS)\)
	because of  Lemma \ref{lem:powermeanO}-(\ref{powermeanO-a},\ref{powermeanO-b},\ref{powermeanO-c}). Then as a result of its definition 
	given in  \eqref{eq:def:information},  \(\RMI{\rno}{\mP}{\cha}\) is continuously differentiable 
	in \(\rno\) on \((0,1)\) and \((1,\infty)\).
	The expression 
	for the derivative for \(\rno\neq1\) given in  \eqref{eq:lem:informationOder} follows from the chain rule.
	
	In order to extend the continuous differentiability to \(\rno=1\), first we establish that \(\RMI{\rno}{\mP}{\cha}\) is 
	continuous at \(\rno=1\). As a result of L'Hospital's rule \cite[Thm. 5.13]{rudin} and Lemma \ref{lem:powermeanO}-(\ref{powermeanO-b})
	\(\lim\limits_{\rno \to 1} \tfrac{\rno}{\rno-1} \ln \lon{\mmn{\rno,\mP}}=\lon{\dmn{1,\mP}}\).
	On the other hand \(\lon{\dmn{1,\mP}}=\RMI{1}{\mP}{\cha}\) as a result of  
	\eqref{eq:def:dpowermean}, Lemma \ref{lem:powermeandensityO}-(\ref{powermeandensityO-b}) and 
	the definition of \(\RMI{1}{\mP}{\cha}\). Thus \(\RMI{\rno}{\mP}{\cha}\) is continuous at \(\rno=1\).
	Then, 
	\begin{align}
	\notag
	\left.\der{}{\rno}\RMI{\rno}{\mP}{\cha}\right\vert_{\rno=1}
	&=\lim_{\rno \to 1} \tfrac{1}{1-\rno}
	\left[\lon{\dmn{1,\mP}}-\tfrac{\rno}{\rno-1}\ln \lon{\mmn{\rno,\mP}} \right].
	\end{align}
	\(\lon{\mmn{\rno,\mP}}\) and \(\lon{\dmn{\rno,\mP}}\) are continuously differentiable by Lemma \ref{lem:powermeanO}-(\ref{powermeanO-b},\ref{powermeanO-c}). 
	Then using L'Hospital's rule \cite[Thm. 5.13]{rudin} and the identity \(\der{}{\rno}\lon{\dmn{\rno,\mP}}=\ddmn{\rno,\mP}(\outS)\)
	we get 
	\begin{align}
	\notag
	\lim_{\rno \to 1} \tfrac{1}{1-\rno}
	\left[\lon{\dmn{1,\mP}}-\tfrac{\rno}{\rno-1}\ln \lon{\mmn{\rno,\mP}} \right]
	&=\tfrac{1}{2}\left[\ddmn{1,\mP}(\outS)+2\lon{\dmn{1,\mP}}-\lon{\dmn{1,\mP}}^2\right].
	\end{align}
	Hence \(\RMI{\rno}{\mP}{\cha}\) is differentiable at \(\rno=1\) and its derivative at \(\rno=1\) is the one given in 
	\eqref{eq:lem:informationOder}. 
	Finally, in order to show that \(\der{}{\rno}\RMI{\rno}{\mP}{\cha}\) is continuous at \(\rno=1\) 
	we apply L'Hospital's rule \cite[Thm. 5.13]{rudin} to confirm,
	\begin{align}
	\notag
	\lim_{\rno \to 1} 
	\der{}{\rno}\RMI{\rno}{\mP}{\cha}
	&=\lim_{\rno \to 1} 
	\tfrac{\rno(\rno-1)\lon{\dmn{\rno,\mP}}-\lon{\mmn{\rno,\mP}}\ln\lon{\mmn{\rno,\mP}}}{\lon{\mmn{\rno,\mP}}(\rno-1)^2} 
	\\
	\notag
	&=\tfrac{1}{2}\left[\ddmn{1,\mP}(\outS)+2\lon{\dmn{1,\mP}}-\lon{\dmn{1,\mP}}^2\right].
	\end{align}
	
	As a function \(\rno\) on \(\reals{+}\), \(\rno\ln\lon{\mmn{\rno,\mP}}\) is convex by Lemma \ref{lem:powermeanO}-(\ref{powermeanO-d})
	and  differentiable by Lemma \ref{lem:powermeanO}-(\ref{powermeanO-b}). 
	Then \(\rno\ln\lon{\mmn{\rno,\mP}}\) has a tangent at each \(\rno\in\reals{+}\) and it lays above all of its tangents, 
	i.e. for all \(\rno,\rnt\in\reals{+} \) such that \(\rno\neq \rnt\),
	\begin{align}
	\label{eq:tangent}
	\rnt\ln\lon{\mmn{\rnt,\mP}}  
	&\geq 
	\rno\ln\lon{\mmn{\rno,\mP}}+\left(\ln\lon{\mmn{\rno,\mP}}+\tfrac{\rno \lon{\dmn{\rno,\mP}}}{\lon{\mmn{\rno,\mP}}}\right)(\rnt-\rno).
	\end{align}
	Then for all \(\rno,\rnt\in\reals{+} \) such that \(\rno\neq \rnt\) we have
	\begin{align}
	\label{eq:tangent-new}
	\tfrac{\rno}{\rno-\rnt}
	\tfrac{\lon{\dmn{\rno,\mP}}}{\lon{\mmn{\rno,\mP}}}
	+\tfrac{\rnt}{(\rno-\rnt)^2}  \ln\tfrac{\lon{\mmn{\rnt,\mP}}}{\lon{\mmn{\rno,\mP}}}
	&\geq 0.
	\end{align}
	If we apply the above inequality at \(\rnt=1\) we can conclude, using  \eqref{eq:lem:informationOder}, that 
	\(\der{}{\rno} \RMI{\rno}{\mP}{\cha}\geq 0\) for \(\rno\neq1\). For \(\rno=1\) using 
	Lemma \ref{lem:powermeandensityO}-(\ref{powermeandensityO-b}) and  Lemma \ref{lem:powermeanO}-(\ref{powermeanO-c}) we get
	\begin{align}
	\notag
	\left.\der{}{\rno} \RMI{\rno}{\mP}{\cha} \right\vert_{\rno=1}
	&=\tfrac{1}{2} \int\sum\nolimits_{\mW} \tpn{1}(\mW|\dout) 
	\left(\ln \tfrac{\tpn{1}(\mW|\dout)}{\mP(\mW)}-\RMI{1}{\mP}{\cha}\right)^2
	\mmn{1,\mP}(\dif{\dout})
	\\
	\label{eq:informationderatone} 
	&\geq0.
	\end{align}
	Thus \(\der{}{\rno} \RMI{\rno}{\mP}{\cha}\) is non-negative for all \(\rno\in\reals{+} \) and
	\(\RMI{\rno}{\mP}{\cha}\) is a nondecreasing function of \(\rno\).
	Then \(\RMI{\rno}{\mP}{\cha}\) is non-negative as well because
	\(\RMI{0}{\rno}{\cha}\geq -\ln \sum\nolimits_{\mW}\mP(\mW)=0\)
	and \(\lim_{\rno \downarrow 0} \RMI{\rno}{\mP}{\cha}=\RMI{0}{\rno}{\cha}\).
	
	If \(\mmn{1,\mP}(\set{A}(\mP,\gamma))=1\) for a \(\gamma\), then 
	\(\lon{\mmn{\rno,\mP}}=\gamma^{\frac{\rno-1}{\rno}}\) for all  \(\rno\in\reals{+} \)
	and \(\RMI{\rno}{\mP}{\cha}=\ln\gamma\)  for all \(\rno\in [0,\infty]\), 
	because 
	\(\nmn{\rno,\mP}=\gamma^{\frac{\rno-1}{\rno}}\) for all \(\dout\in \set{A}(\mP,\gamma)\). 
	
	If there does not exist a \(\gamma\) such that \(\mmn{1,\mP}(\set{A}(\mP,\gamma))=1\), then 
	the convexity of \(\rno\ln\lon{\mmn{\rno,\mP}}\) is strict by Lemma \ref{lem:powermeanO}-(\ref{powermeanO-d})
	and the variance of the random variable \(\ln\tfrac{\tpn{1}(\mW|\dout)}{\mP(\mW)}\) is positive.
	Thus the inequalities  \eqref{eq:tangent}, \eqref{eq:tangent-new}, and \eqref{eq:informationderatone} 
	are strict and  \(\der{}{\rno}\RMI{\rno}{\mP}{\cha}\) is positive for all \(\rno\in\reals{+}\).
\end{proof}

\begin{proof}[Proof of Lemma \ref{lem:informationP}]~
	\begin{enumerate}
		\item[(\ref{lem:informationP}-\ref{informationP-a})]
		Let us start with the values of \(\rno\) in \((0,1)\). 
		Recall that \(\lon{\cdot}:\fmea{\outA}\to \reals{+} \) and \(\tfrac{\rno}{\rno-1}\ln (\cdot):\reals{+} \to \reals{}\) are continuous functions
		and the composition of two continuous functions is a continuous function, \cite[Thm. 18.2.c]{munkres}.
		Furthermore, the function \(\lon{\mmn{\rno,\mP}}\) is continuous in \(\mP\) on \(\pdis{\pmea{\outA}}\),
		---and hence on \(\pdis{\cha}\)--- for \(\rno\in(0,1)\) by Lemma \ref{lem:powermeanP}-(\ref{powermeanP-d}).
		Thus \(\RMI{\rno}{\mP}{\cha}\) is continuous in \(\mP\) on \(\pdis{\cha}\) for \(\rno \in (0,1)\).
		
		For any \(\pmn{1},\pmn{2} \in \pdis{\cha}\) and \(\beta\in [0,1]\) let 
		\(\pmn{\beta}=\beta \pmn{1}+(1-\beta)\pmn{2}\).
		Recall that \(\lon{\mmn{\rno,\mP}}\) is convex in \(\mP\) for \(\rno\in (0,1)\) by Lemma \ref{lem:powermeanP}-(\ref{powermeanP-a}).
		Then by the definition of \(\RMI{\rno}{\mP}{\cha}\) given in \eqref{eq:def:information} we have 
		\begin{align}
		\notag
		\RMI{\rno}{\pmn{\beta}}{\cha}
		&\geq \tfrac{\rno}{\rno-1} \ln \left( \beta\lon{\mmn{\rno,\pmn{1}}}+(1-\beta)\lon{\mmn{\rno,\pmn{2}}} \right)  
		\\
		\notag
		&\geq \tfrac{\rno}{\rno-1} \ln \left( \lon{\mmn{\rno,\pmn{1}}} \vee \lon{\mmn{\rno,\pmn{2}}} \right)
		\\
		\notag
		&\geq  \RMI{\rno}{\pmn{1}}{\cha} \wedge \RMI{\rno}{\pmn{2}}{\cha}
		\end{align}
		Thus \renyi information is continuous and quasi-concave in \(\mP\) for \(\rno\in (0,1)\).
		
		For \(\rno=0\) case, first note that \(\abs{e^{-\RMI{0}{\pmn{1}}{\cha}}-e^{-\RMI{0}{\pmn{2}}{\cha}}} \leq \lon{\pmn{1}-\pmn{2}}\). 
		Thus \(e^{-\RMI{0}{\mP}{\cha}}\) is continuous in \(\mP\). Since \(-\ln\dinp\) is continuous on \(\reals{+}\), 
		\(\RMI{0}{\mP}{\cha}\) is continuous in \(\mP\). In order to prove that \(\RMI{0}{\mP}{\cha}\) is quasi-concave, note that 
		for any \(\beta\in(0,1)\) and \(\pmn{1},\pmn{2}\in\pdis{\cha}\) we have, 
		\begin{align}
		\notag
		\RMI{0}{\pmn{\beta}}{\cha}
		&= - \ln  \essup_{\mmn{1,\pmn{\beta}}} 
		\left[
		\beta \sum\nolimits_{\mW:\pmn{1}(\mW|\dout)>0} \pmn{1}(\mW)+(1-\beta) \sum\nolimits_{\mW:\pmn{2}(\mW|\dout)>0} \pmn{2}(\mW)
		\right]
		\\
		\notag
		&\geq  - \ln  
		\left[  
		\essup_{\mmn{1,\pmn{1}}} \left(\sum\nolimits_{\mW:\pmn{1}(\mW|\dout)>0}\pmn{1}(\mW)\right)
		\bigvee 
		\essup_{\mmn{1,\pmn{2}}} \left(\sum\nolimits_{\mW:\pmn{2}(\mW|\dout)>0}\pmn{2}(\mW)\right)
		\right]
		\\
		\notag
		&=\RMI{0}{\pmn{1}}{\cha} \wedge \RMI{0}{\pmn{2}}{\cha}.
		\end{align} 
		
		\item[(\ref{lem:informationP}-\ref{informationP-b})]
		For any \(\pmn{1},\pmn{2} \in \pdis{\cha}\) and \(\beta\in [0,1]\) let 
		\(\pmn{\beta}=\beta \pmn{1}+(1-\beta)\pmn{2}\).
		Recall that \(\lon{\mmn{\rno,\mP}}\) is concave in \(\mP\) for \(\rno\in (1,\infty]\) by Lemma \ref{lem:powermeanP}-(\ref{powermeanP-b}).
		Then by the definition of \(\RMI{\rno}{\mP}{\cha}\) we have
		\begin{align}
		\notag
		\RMI{\rno}{\pmn{\beta}}{\cha}
		&\geq \tfrac{\rno}{\rno-1} \ln \left( \beta\lon{\mmn{\rno,\pmn{1}}}+(1-\beta)\lon{\mmn{\rno,\pmn{2}}} \right)  
		\\
		\notag
		&\geq \beta \tfrac{\rno}{\rno-1} \ln \lon{\mmn{\rno,\pmn{1}}}
		+(1-\beta)  \tfrac{\rno}{\rno-1} \ln \lon{\mmn{\rno,\pmn{2}}}
		\\
		\notag
		&= \beta \RMI{\rno}{\pmn{1}}{\cha} +(1-\beta) \RMI{\rno}{\pmn{2}}{\cha}
		\end{align}
		where the second inequality follows from the Jensen's inequality and the concavity of the logarithm function. 
		
		For \(\rno=1\) case, note that as a result of the definition of \(\RMI{\rno}{\mP}{\cha}\) we have
		\begin{align}
		\notag
		\RMI{1}{\pmn{\beta}}{\cha}
		&=\beta \RMI{1}{\pmn{1}}{\cha}+(1-\beta)\RMI{1}{\pmn{2}}{\cha}
		+\int \left[\beta \der{\mmn{1,\pmn{1}}}{\mmn{1,\pmn{\beta}}} \ln \der{\mmn{1,\pmn{1}}}{\mmn{1,\pmn{\beta}}}
		+(1-\beta) \der{\mmn{1,\pmn{2}}}{\mmn{1,\pmn{\beta}}} \ln \der{\mmn{1,\pmn{2}}}{\mmn{1,\pmn{\beta}}}
		\right] \mmn{1,\pmn{\beta}}(\dif{\dout})
		\\
		\notag
		&\geq \beta \RMI{1}{\pmn{1}}{\cha}+(1-\beta)\RMI{1}{\pmn{2}}{\cha}
		\end{align}
		where the inequality follows from 
		\(\beta \der{\mmn{1,\pmn{1}}}{\mmn{1,\pmn{\beta}}}+(1-\beta) \der{\mmn{1,\pmn{2}}}{\mmn{1,\pmn{\beta}}}=1\),
		the convexity of the function \(\dinp \ln \dinp\) and the Jensen's inequality.
	\end{enumerate}
\end{proof}

\subsection[The Proof About \(\qmn{\rno,\mP}\)]{Proof of the Lemma on the \renyi Mean}\label{sec:mean-proofs}
\begin{proof}[Proof of Lemma \ref{lem:information:def}] 
	For \(\rno=0\), as a result of the definition of the order zero \renyi information given in \eqref{eq:def:information}
	and the definition of the order zero \renyi divergence given in \eqref{eq:def:divergence} we have 
	\begin{align}
	\notag
	\RD{0}{\mP \mtimes \cha}{\mP\otimes\mQ}
	&=-\ln \int \sum\nolimits_{\mW} \mP(\mW) \der{\mQ}{\rfm} \IND{\tpn{1}(\mW|\dout) \der{\mmn{1,\mP}}{\rfm}>0}   \rfm(\dif{\dout}) 
	\\
	\notag
	&=\RMI{0}{\mP}{\cha}
	-\ln \int \tfrac{\sum\nolimits_{\mW:\tpn{1}(\mW|\dout)>0} \mP(\mW)}{\essup_{\mmn{1,\mP}} \sum\nolimits_{\mW:\tpn{1}(\mW|\dout)>0} \mP(\mW)} \IND{\der{\mmn{1,\mP}}{\rfm}>0} \der{\mQ}{\rfm} \rfm(\dif{\dout}). 
	\end{align}
	Then the definition of \(\qmn{0,\mP}\) given in  \eqref{eq:def:mean} implies
	\eqref{eq:lem:information:defA} and \eqref{eq:lem:information:defB}. 
	
	For \(\rno\in(0,\infty]\),  \eqref{eq:lem:information:defA} follows from the definitions of 
	the \renyi information, divergence, and mean given in \eqref{eq:def:information}, \eqref{eq:def:divergence}
	\eqref{eq:def:mean} by substitution. Using \eqref{eq:sibson} and \eqref{eq:lem:information:defA} we get,
	\begin{align}
	\label{eq:information:positive}
	\RD{\rno}{\mP \mtimes \cha}{\mP  \otimes \mQ}
	&=\RMI{\rno}{\mP}{\cha}
	+\RD{\rno}{\qmn{\rno,\mP}}{\mQ}
	&
	&\forall \rno\in (0,\infty].
	\end{align}
	On the other hand \(\qmn{\rno,\mP}\) is a probability measure by definition.
	Then  \eqref{eq:lem:information:defB} and uniqueness of \renyi mean as the minimizer follow
	from   \eqref{eq:lem:information:defA}, \eqref{eq:information:positive},
	 and Lemma \ref{lem:divergence-pinsker}. 
	
	The following identity and \eqref{eq:lem:information:defB} imply \eqref{eq:lem:information:defC}.
	\begin{align}
	\notag
	\RD{\rno}{\mP \mtimes \cha}{\mP  \otimes \mQ}
	&=\RD{\rno}{\mmn{\rno,\mP}}{\mQ}
	&
	&\forall \mP\in \pdis{\cha}, \mQ\in\fmea{\outA}, \rno\in(0,\infty]\setminus\{1\}.
	\end{align}
\end{proof}

\subsection[\(\RC{\rno}{\Scha{\fX}}\) via The Ergodic Thm.]{The Ergodic Theorem and the \renyi Capacity}\label{sec:shiftchannel}
For \(\Scha{\fX}\) described in Example \ref{eg:shiftchannel} 
we have \(\RD{\rno}{\mW}{\lbm}=\RD{\rno}{\wmn{\fX}}{\lbm}\) for all \(\mW\in \Scha{\fX}\) where \(\lbm\)
is the Lebesgue measure.
Thus by \eqref{eq:thm:minimaxradius} of Theorem \ref{thm:minimax} we have
\begin{align}
\notag
\RC{\rno}{\Scha{\fX}}
&\leq\RD{\rno}{\wmn{\fX}}{\lbm}.
\end{align}
We prove the reverse inequality, \(\RC{\rno}{\Scha{\fX}}\geq\RD{\rno}{\wmn{\fX}}{\lbm}\), 
using the Birkoff-Khinchin ergodic theorem \cite[8.4.1]{dudley}.
In particular, we show that there exists a sequence of priors \(\{\pmn{\ind}\}_{\ind\in\integers{+}}\)
such that\footnote{Finding a different sequence of priors for each order \(\rno\) in \((0,\infty]\) would have
	been sufficient for establishing \(\RC{\rno}{\Scha{\fX}}\geq\RD{\rno}{\wmn{\fX}}{\lbm}\). 
	The existence of a sequence of priors \(\{\pmn{\ind}\}_{\ind\in\integers{+}}\) such that \(\lim_{\ind \to \infty} \RMI{\rno}{\pmn{\ind}}{\Scha{\fX}}=\RC{\rno}{\Scha{\fX}}\) for all
	orders \(\rno\) in \(\reals{+}\) allows us to assert the convexity of
	\((\rno-1)\RC{\rno}{\Scha{\fX}}\) in \(\rno\) on \(\reals{+}\), rather than just \([1,\infty)\).} 
\(\lim_{\ind \to \infty} \RMI{\rno}{\pmn{\ind}}{\Scha{\fX}}\geq \RD{\rno}{\wmn{\fX}}{\lbm}\) for all 
\(\rno\in(0,\infty]\). 

For any \(\knd\in\integers{}\) and \(\dinp\in\reals{}\) let \(\trans{\dinp}^{\knd}\) be the transformation resulting from \(\knd\) 
successive applications of \(\trans{\dinp}\). As a result of the definition of \(\trans{\dinp}\) given in \eqref{eq:def:shiftchannel}, 
\(\trans{\dinp}^{\knd}=\trans{\knd\dinp}\) for any \(\knd\in\integers{}\) and \(\dinp\in\reals{}\).
For any \(\gX\in\Lon{\lbm}\), \(\dinp\in\reals{}\) and \(\ind\in\integers{+}\) let \(\eav{\gX}{\dinp}{\ind}\) be
\begin{align}
\notag
\eav{\gX}{\dinp}{\ind}(\dout)
&\DEF \tfrac{1}{\ind}\sum\nolimits_{\knd=0}^{\ind-1} \gX \circ \trans{\dinp}^{\knd}(\dout)
\\
\notag
&=\tfrac{1}{\ind}\sum\nolimits_{\knd=0}^{\ind-1} 
\gX(\dout-\knd \dinp-\lfloor\dout-\knd \dinp\rfloor).
\end{align}

For any \(\dinp\in\reals{}\), the transformation \(\trans{\dinp}\) is measure preserving for the measure space 
\((\outS,\outA,\lbm)\);
if in addition \(\dinp\) is irrational, then \(\trans{\dinp}\) is ergodic. 
Hence \(\eav{\gX}{\dinp}{\ind}\) converges to \(\int\gX(\dout)\dif{\dout}\) \(\lbm-\)a.e.
for any \(\gX\in\Lon{\lbm}\) and \(\dinp\in\reals{}\setminus \rationals{}\),
by the Birkoff-Khinchin Ergodic theorem \cite[8.4.1]{dudley}:
\begin{align}
\label{eq:ergodicity}
\eav{\gX}{\dinp}{\ind}
&\xrightarrow{\lbm-a.e.}\int\gX(\dout)\dif{\dout}
&
&\forall \gX\in\Lon{\lbm} \mbox{~and~}
\forall \dinp \in \reals{} \setminus\rationals{}.
\end{align}

Let \(\dinp\) be an irrational number that will be fixed for the rest of the proof. 
For any \(\ind\in\integers{+}\), let \(\pmn{\ind}\) be the prior that has equal 
probability mass on each probability measure corresponding to a \(\fX\circ\trans{\dinp}^{\knd}\) 
for some \(\knd\in\{0,\ldots,(\ind-1)\}\). Then 
\begin{align}
\notag
\der{\mmn{\rno,\pmn{\ind}}}{\lbm}(\dout)
&=\left(\tfrac{1}{\ind}\sum\nolimits_{\knd=0}^{\ind-1}  (\fX^{\rno}\circ\trans{\dinp}^{\knd})(\dout) \right)^{\sfrac{1}{\rno}}
\\
\notag
&=
\left(\sum\nolimits_{\knd=0}^{\ind-1} \tfrac{1}{\ind} 
\fX^{\rno}\left(\dout-\tfrac{\knd\dinp}{\ind}-\lfloor\dout-\tfrac{\knd\dinp}{\ind} \rfloor\right)\right)^{\sfrac{1}{\rno}}
&
&\forall \rno \in \reals{+}.
\end{align}
For \(\rno\in\reals{+}\setminus\{1\}\), we calculate the limit 
\(\lim_{\ind\to\infty} \RMI{\rno}{\pmn{\ind}}{\Scha{\fX}}\) by calculating 
the limit \(\lim_{\ind \to \infty} \lon{\mmn{\rno,\pmn{\ind}}}\). For \(\rno=1\) and  \(\rno=\infty\) the result follows from continuity 
arguments.

\begin{enumerate}[(a)]
	\item {\it \(\rno \in (0,1)\) case:}
	\(\int \fX^{\rno} \dif{\dout}\leq \left( \int \fX(\dout)\dif{\dout} \right)^{\rno}=1\) by the Jensen's inequality.
	Hence \(\fX^{\rno}\in\Lon{\lbm}\) as a result of  \eqref{eq:ergodicity} we have
	\begin{align}
	\notag
	\left(\der{\mmn{\rno,\pmn{\ind}}}{\lbm}\right)^{\rno}
	&\xrightarrow{\lbm-a.e.} \int \fX^{\rno}(\dout) \dif{\dout}
	&
	&\Rightarrow
	&
	\der{\mmn{\rno,\pmn{\ind}}}{\lbm}
	&\xrightarrow{\lbm-a.e.} \left(\int \fX^{\rno}(\dout) \dif{\dout}\right)^{\sfrac{1}{\rno}}.
	\end{align}
	
	For any \(\epsilon>0\) there exists a \(\delta>0\) such that 
	if \(\lbm(\oev)<\delta\) for a \(\oev \in \outA\), 
	then  \(\wmn{\fX}(\oev)<\epsilon\), 
	because \(\wmn{\fX}\AC\lbm\). Since \(\lbm\) is invariant under translations and \(\Scha{\fX}\) is
	the set of all mod one translations of \(\wmn{\fX}\), \(\mmn{1,\mP}(\oev)<\epsilon\) whenever \(\lbm(\oev)<\delta\), as well.
	Then \(\{\mmn{\rno,\mP}:\mP\in \pdis{\Scha{\fX}},\rno \in (0,1]\}\UAC \lbm\)
	and \(\{\der{\mmn{\rno,\pmn{\ind}}}{\lbm}\}_{\ind \in \integers{+}}\) is uniformly \(\lbm-\)integrable
	because \(\mmn{\rno,\mP}(\oev)\) is a nondecreasing function of \(\rno\) for all \(\oev \in \outA\) 
	by Lemma \ref{lem:powermeanO}-(\ref{powermeanO-b}). 
	Since almost everywhere convergence implies convergence in measure by 
	\cite[Thm. 2.2.3]{bogachev}, using Lebesgue-Vitali convergence 
	theorem \cite[4.5.4]{bogachev},  we can conclude that \(\der{\mmn{\rno,\pmn{\ind}}}{\lbm}\) converges to 
	\(\left(\int \fX^{\rno}(\dout) \dif{\dout}\right)^{\sfrac{1}{\rno}}\) in \(\Lon{\lbm}\), as well:
	\(\der{\mmn{\rno,\pmn{\ind}}}{\lbm}
	\xrightarrow{\Lon{\lbm}} \left(\int \fX^{\rno} (\dout) \dif{\dout}\right)^{\sfrac{1}{\rno}}\).
	Then \(\lim\nolimits_{\ind\to\infty} \lon{\mmn{\rno,\pmn{\ind}}}=\left(\int \fX^{\rno}(\dout) \dif{\dout}\right)^{\sfrac{1}{\rno}}\).
	Using the definition of \renyi information given in \eqref{eq:def:information} we get
	\begin{align}
	\notag
	\lim\nolimits_{\ind\to\infty} \RMI{\rno}{\pmn{\ind}}{\Scha{\fX}}
	&=\tfrac{1}{\rno-1}\ln \left(\int \fX^{\rno}(\dout) \dif{\dout}\right)
	\\
	\notag
	&=\RD{\rno}{\wmn{\fX}}{\lbm}
	&
	&\forall \rno \in (0,1).
	\end{align}
	\item{\it \(\rno=1\) case:} The \renyi information is a nondecreasing function of 
	the order by Lemma \ref{lem:informationO}. Then 
	\begin{align}
	\notag
	\liminf\nolimits_{\ind\to\infty} \RMI{1}{\pmn{\ind}}{\Scha{\fX}}
	&\geq \liminf\nolimits_{\ind\to\infty} \RMI{\rno}{\pmn{\ind}}{\Scha{\fX}}
	\\
	\notag
	&=\RD{\rno}{\wmn{\fX}}{\lbm}
	&
	&\forall \rno\in(0,1).
	\end{align}
	Since the \renyi divergence is a nondecreasing and lower semicontinuous function of the order by Lemma \ref{lem:divergence-order},
	we have
	\begin{align}
	\notag
	\liminf\nolimits_{\ind\to\infty} \RMI{1}{\pmn{\ind}}{\Scha{\fX}}
	&\geq \lim\nolimits_{\rno \uparrow 1}\RD{\rno}{\wmn{\fX}}{\lbm}
	\\
	\notag
	&=\RD{1}{\wmn{\fX}}{\lbm}.
	&
	&
	\end{align}
	\item{\it \(\rno\in(1,\infty)\) case:} We analyze the finite \(\int \fX^{\rno}(\dout) \dif{\dout}\) and infinite 
	\(\int \fX^{\rno}(\dout) \dif{\dout}\) cases separately.
	\begin{itemize}
		\item  If \(\int \fX^{\rno}(\dout) \dif{\dout}<\infty\), then \(\fX^{\rno}\in\Lon{\lbm}\) and 
		\(\der{\mmn{\rno,\pmn{\ind}}}{\lbm}\xrightarrow{\lbm} \left(\int \fX^{\rno}(\dout)\dif{\dout}\right)^{\sfrac{1}{\rno}}\)
		by  \eqref{eq:ergodicity} because almost everywhere convergence implies convergence in measure by \cite[Thm. 2.2.3]{bogachev}.
		On the other hand, as a result of the concavity of the function \(\dsta^{\sfrac{1}{\rno}}\) in \(\dsta\) for \(\rno\in (1,\infty)\) 
		and the Jensen's inequality we have
		\begin{align}
		\notag
		\mmn{\rno,\pmn{\ind}}(\oev)\leq \left(\tfrac{1}{\ind}\sum\nolimits_{\knd=0}^{\ind-1}  \int_{\trans{\dinp}^{\knd}\oev}  
		\fX^{\rno}(\dout) \dif{\dout}\right)^{\sfrac{1}{\rno}}.
		\end{align}
		Then the uniform  \(\lbm-\)integrability of \(\der{\mmn{\rno,\pmn{\ind}}}{\lbm}\)
		follows from the translational invariance of \(\lbm\) and the \(\lbm-\)integrability of \(\fX^{\rno} \) following
		an argument similar to the one we have for \(\rno\in (0,1)\) case. Thus using 
		Lebesgue-Vitali convergence theorem \cite[4.5.4]{bogachev}
		and the definition of \renyi information exactly the same way we did for \(\rno\in (0,1)\) case we get
		\begin{align}
		\notag
		\lim\nolimits_{\ind\to\infty} \RMI{\rno}{\pmn{\ind}}{\Scha{\fX}}
		&=\RD{\rno}{\wmn{\fX}}{\lbm}
		&
		&\mbox{if~}
		\int \fX^{\rno}(\dout) \dif{\dout}<\infty.
		\end{align}
		\item 
		If \(\int \fX^{\rno} \dif{\dout}=\infty\), then we repeat the above analysis for \(\fX\wedge \gamma\) for a \(\gamma\in \reals{+} \) instead 
		of \(\fX\). As a result we get,
		\begin{align}
		\notag
		\liminf\nolimits_{\ind\to\infty} \RMI{\rno}{\pmn{\ind}}{\Scha{\fX}}
		&\geq \tfrac{1}{\rno-1}\ln \left(\int (\fX(\dout) \wedge \gamma)^{\rno} \dif{\dout}\right)
		&
		&\forall \gamma \in \reals{+} 
		\end{align}
		Note that as \(\gamma \uparrow \infty\), \(\int (\fX(\dout)\wedge \gamma)^{\rno} \dif{\dout}\uparrow \int \fX^{\rno}(\dout) \dif{\dout}\). Thus
		\begin{align}
		\notag
		\lim\nolimits_{\ind\to\infty} \RMI{\rno}{\pmn{\ind}}{\Scha{\fX}}
		&= \infty
		&
		&\mbox{if~}
		\int \fX^{\rno}(\dout) \dif{\dout}=\infty.
		\end{align}
	\end{itemize}
	\item{\it \(\rno=\infty\) case:} Repeat the analysis for \(\rno=1\) case by replacing \(\rno=1\) and \((0,1)\) by \(\rno=\infty\) and \((1,\infty)\).
\end{enumerate}

We have used the ergodic theorem \cite[8.4.1]{dudley} in order to be able to conduct our analysis for arbitrary measurable functions. 
If we restrict our attention to functions that are bounded and continuous  at all but finite number of points, we can choose \(\pmn{\ind}\) 
to be the priors that have \(\sfrac{1}{2^{\ind}}\) probability mass on each probability measure corresponding to a \(\fX\circ \trans{2^{-\ind}}^{\knd}\) 
for \(\knd\in\{0,1,\ldots,(2^{\ind}-1)\}\).  
Then the identity \((\der{\mmn{\rno,\pmn{\ind}}}{\lbm})^{\rno}\xrightarrow{\lbm-a.e.} \left(\int \fX^{\rno}(\dout)\dif{\dout}\right)\) is a 
result of  Riemann integrability of \(\fX^{\rno}\) rather than the ergodicity. 

We have used the Lebesgue-Vitali convergence theorem \cite[4.5.4]{bogachev} instead of 
the dominated convergence theorem \cite[2.8.1]{bogachev}.
That is a matter of taste; one can prove the same statements using the dominated convergence theorem. 
First, do the analysis for \(\tilde{\fX}=\fX \wedge \gamma\), 
and then take the limit as \(\gamma\) diverges to infinity.
\addtocontents{toc}{\setcounter{tocdepth}{3}}
\end{document}